\documentclass{amsart}

\usepackage{amsmath,amssymb,amsfonts,amsthm,mathtools}
\usepackage{mathrsfs}
\usepackage{bm}
\usepackage{cancel}
\usepackage{xfrac}
\usepackage{xcolor}
\newcommand{\rank}{\operatorname{rank}}
\DeclareMathOperator{\tr}{tr}

\DeclareMathOperator{\Tr}{Tr}   
\DeclareMathOperator{\range}{range}
\DeclareMathOperator{\vectorise}{vec}

\usepackage{tikz-cd}

\usepackage[a4paper,margin=1.2in]{geometry}
\usepackage{hyperref}

\numberwithin{equation}{section}

\newtheorem{theorem}{Theorem}[section]
\newtheorem{lemma}[theorem]{Lemma}
\newtheorem{proposition}[theorem]{Proposition}
\newtheorem{corollary}[theorem]{Corollary}
\newtheorem{remark}[theorem]{Remark}

\theoremstyle{definition}
\newtheorem{definition}[theorem]{Definition}
\newtheorem{example}[theorem]{Example}

\theoremstyle{plain}
\usepackage[numbers]{natbib}

\title{Continuity of certain quantum channel divergences}
\author{Morgan Jones}
\address{Department of Mathematics, University of Oklahoma, Norman, Oklahoma, USA}
\email{morgan.d.jones-2@ou.edu}
\date{October 2026}
\begin{document}
\begin{abstract}
The geometric channel fidelity, the geometric R\'enyi divergences and the Belavkin--Staszewski relative entropy give converse bounds for capacities, discrimination and estimation. Channels are known only approximately, yet two pairs of channels with equal sums can be arbitrarily close in the diamond norm and have fidelities $0$ and $1$. The remedy is the Arveson--Radon--Nikodym derivative $T$ of one map with respect to the sum $\Omega$, whose transpose is the environment effect selecting that map in a Stinespring dilation of $\Omega$. In $T$, these quantities and all maximal $f$-divergences of operator convex $f$ are continuous, with explicit dimension-free moduli; for two pairs of channels with the same sum the fidelities differ by at most $2\|T_1-T_2\|_\infty^{1/2}$, which gives H\"older continuity in the diamond norm on such splittings. The fidelity is continuous at a pair of completely positive maps exactly when the Choi matrix of the sum is invertible or the fidelity vanishes: noisy channels of full Kraus rank are continuity points, unitary gates are not.
\end{abstract}
\maketitle

\section{Introduction}\label{sec:intro}

The geometric channel fidelity of $\Phi,\Psi\in\mathrm{CP}(\mathcal{X},\mathcal{Y})$ is $\widehat{F}(\Phi,\Psi)=\inf_{\rho}\tr(\Phi_\rho\#\Psi_\rho)$, where $\Phi_\rho=(\Phi\otimes\mathbf{1}_{\mathrm{L}(\mathcal{X})})(\rho)$ for states $\rho$ on $\mathcal{X}\otimes\mathcal{X}$ and $\#$ is the matrix geometric mean (Katariya and Wilde \cite{KatariyaWilde2021} use its square); the geometric R\'enyi divergences and the Belavkin--Staszewski relative entropy arise in the same way from maximal $f$-divergences \cite{Matsumoto2018newf}. These channel divergences give converse bounds for channel capacities \cite{FangFawzi2021}, discrimination and estimation \cite{KatariyaWilde2021}. In practice a channel is known only approximately, and the relevant worst-case error is the diamond distance $|||\cdot|||_1$ \cite{KuengLongDohertyFlammia2016}, so such bounds need continuity in this norm. It fails: if $\dim\mathcal{Y}\geq2$, for every $\eta>0$ there are replacement channels with $\Phi_1+\Psi_1=\Phi_2+\Psi_2$ and $|||\Phi_1-\Phi_2|||_1=|||\Psi_1-\Psi_2|||_1<\eta$, but $\widehat{F}(\Phi_1,\Psi_1)=0$ and $\widehat{F}(\Phi_2,\Psi_2)=1$ (Corollary~\ref{impossibility corollary}); on the same pairs every bounded maximal $f$-divergence with $f(1)=0$ jumps by $f(0^+)+f'(\infty)$ (Corollary~\ref{maxf impossibility}).

The remedy is to record how a pair splits its sum. Let $\Omega=\Phi+\Psi$ have linearly independent Kraus operators $V_1,\dots,V_r$, and let $W$ have columns $\vectorise(V_k)$. By Arveson's Radon--Nikodym theorem (Theorem~\ref{thm:arveson}), $\Phi(X)=\sum_{k,l}T_{kl}V_kXV_l^*$ for a unique $0\leq T\leq\mathbf{1}_{\mathbb{C}^r}$ on $\mathbb{C}^r$, and $\mathbf{1}_{\mathbb{C}^r}-T$ gives $\Psi$. Up to normalisation, $(\Phi,\Psi)$ is a two-outcome instrument implementing $\Omega$, with which-branch effect $T^\top$ on the environment of its Stinespring dilation (Theorem~\ref{effect theorem}). By an identity essentially due to Kian \cite{kian2026relative}, $J(\Phi)\,\sigma\,J(\Psi)=Wg_\sigma(T)W^*$ with $g_\sigma(t)=t\,\sigma\,(1-t)$ for every Kubo--Ando mean $\sigma$ (Theorem~\ref{fundamental principle}), so all quantities above are functions of $T$. Hence, by Theorem~\ref{main theorem},
\begin{equation}\label{eq:intro fidelity}
\big|\widehat{F}(\Phi_1,\Psi_1)-\widehat{F}(\Phi_2,\Psi_2)\big|\leq\|\Omega^*(\mathbf{1}_{\mathcal{Y}})\|_\infty\,\|T_1-T_2\|_\infty^{1/2},
\end{equation}
with no dimensional factor. Here $\|\Omega^*(\mathbf{1}_{\mathcal{Y}})\|_\infty=2$ if $\Phi_i$ and $\Psi_i$ are channels, and $=1$ if $\Omega$ itself is a channel, as for an instrument. In the counterexample $T_1$ is a projection, on which $g_\#$ vanishes, while $T_2=\frac12\mathbf{1}_{\mathbb{C}^r}$. Analogous bounds hold for all means, the geometric R\'enyi quasi-entropies $\widehat{Q}_\alpha$, $\alpha\in(0,1)$, and bounded maximal $f$-divergences (Theorems~\ref{Modulus of continuity results}, \ref{thm:renyi continuity}, \ref{thm:maxf}). As $\nu(\Omega)^2\|T_1-T_2\|_\infty\leq|||\Phi_1-\Phi_2|||_1$, with $\nu(\Omega)$ the least operator norm of a unit combination of the $V_k$ (Theorem~\ref{thm:two sided}), $\widehat{F}$ is H\"older continuous in the diamond norm on the splittings of $\Omega$, with constant $\|\Omega^*(\mathbf{1}_{\mathcal{Y}})\|_\infty/\nu(\Omega)$ (Corollary~\ref{cor:cb continuity}); in the counterexample $\nu(\Omega)\to0$.

On pairs of completely positive maps, Theorem~\ref{thm:dichotomy} shows exactly where continuity survives:
\begin{equation}\label{eq:intro dichotomy}
\widehat{F}\ \text{is continuous at }(\Phi,\Psi)\iff J(\Phi+\Psi)\ \text{is invertible or}\ \widehat{F}(\Phi,\Psi)=0 .
\end{equation}
Likewise for $\widehat{Q}_\alpha$, with analogues for all means other than the weighted arithmetic means and, under a rank condition, for bounded maximal $f$-divergences. Pairs with a channel of full Kraus rank, e.g.\ generalised amplitude damping with nonzero damping and temperature \cite{KhatriSharmaWilde2020}, are continuity points; but an ideal gate $U(\cdot)U^*$ has $\widehat{F}=1$ with itself and $\widehat{F}=0$ with every mis-rotated $U'(\cdot)U'^*$, $U'\notin\mathbb{C}U$, however close (Corollary~\ref{cor:channels dichotomy}).
\begin{itemize}
\item \emph{Unbounded divergences.} If $D_{\max}(\Phi_i\|\Psi_i)\leq\log((1-\epsilon)/\epsilon)$ for both pairs, as for a detector of efficiency $1-\epsilon$, the Belavkin--Staszewski relative entropy and $\widehat{Q}_\alpha$, $\alpha\in(1,2]$, obey continuity estimates in $T$ whose constants blow up as $\epsilon\downarrow0$ (Theorem~\ref{thm:unbounded}, Corollaries~\ref{cor:BS} and~\ref{cor:renyi large alpha}). Every maximal $f$-divergence with operator convex $f$ is continuous in $T$ with values in $(-\infty,+\infty]$ (Theorem~\ref{thm:extended continuity}).
\item \emph{Explicit moduli} for every mean and operator convex $f$, from exact operator Lipschitz constants of one-pole kernels; for operator convex $f$ on $[\epsilon,1-\epsilon]$ the Lipschitz constants obtained are within a factor $3$ of optimal, and in many cases exact (Section~\ref{sec:loewner}).
\item \emph{Different sums and supports.} Changes of the sum enter through the Thompson metric; for arbitrary pairs of channels, $|\widehat{F}(\Phi_1,\Psi_1)-\widehat{F}(\Phi_2,\Psi_2)|\leq(11/2)\,\delta^{1/2}$, with $\delta$ their distance relative to $\Phi_1+\Psi_1+\Phi_2+\Psi_2$ (Theorem~\ref{thm:four maps}, Corollary~\ref{cor:universal}). Support changes, as under leakage \cite{WoodGambetta2018}, enter through the Lebesgue decomposition of the sums (Theorem~\ref{thm:different supports}).
\end{itemize}

\emph{Related work.} Arveson \cite{arveson1969subalgebras} proved the Radon--Nikodym theorem for completely positive maps; Raginsky \cite{Raginsky2003} gave its Kraus form, the upper bound of Lemma~\ref{lipschitz inequality} and a dimension-dependent lower bound. Kian \cite{kian2026relative} defined Kubo--Ando means of such maps via these derivatives, and Okayasu studied their geometric mean and Lebesgue decomposition \cite{Okayasu2026}. Okayasu \cite{Okayasu2026maxf} also introduced operator-valued maximal $f$-divergences of completely positive maps, with criteria for finiteness and boundedness \cite[Proposition~4.11]{Okayasu2026maxf} and joint lower semicontinuity \cite[Theorem~4.7]{Okayasu2026maxf}; in finite dimensions our closed form is the largest eigenvalue of his divergence of the adjoint maps (Section~\ref{sec:all f}). Matsumoto \cite{matsumoto2010reverse,Matsumoto2018newf} introduced his fidelity and maximal $f$-divergences; Katariya and Wilde \cite{KatariyaWilde2021} and Fang and Fawzi \cite{FangFawzi2021} gave closed forms for geometric channel divergences. Maximal $f$-divergences and their lower semicontinuity are studied in \cite{HiaiMosonyi2017,Hiai2019maximal}, continuity of R\'enyi divergences in \cite{MosonyiHiai2024}, and dimension-dependent continuity bounds for Belavkin--Staszewski quantities of states in \cite{BluhmCapelGondolfPerezHernandez2023}. To our knowledge, the quantitative continuity of these divergences and means in $T$, the dimension-free lower bound in the two-sided estimate, the results for different sums and supports, and the dichotomy are new.

Sections~\ref{sec:prelim}--\ref{sec:fidelity} give preliminaries and the counterexample, Sections~\ref{sec:norms}--\ref{sec:all f} continuity in $T$, and Sections~\ref{sec:four maps}--\ref{sec:dichotomy} the diamond norm, explicit moduli, different sums and supports, and the dichotomy. Remarks headed ``Physical interpretation'' explain the main notions operationally.

\section{Preliminaries}\label{sec:prelim}
\emph{Notation.} We follow \cite{watrous2018theory}. Complex Euclidean spaces are denoted by calligraphic letters, $\mathbf{1}_{\mathcal{V}}$ is the identity operator on $\mathcal{V}$, and $\mathbf{1}_{\mathrm{L}(\mathcal{V})}$ is the identity map on $\mathrm{L}(\mathcal{V})$; the space is always indicated, for instance $\mathbf{1}_{\mathbb{C}^r}$ for the identity on the space $\mathbb{C}^r$ of an ARN derivative. We write $\mathbf{j}=(1,\dots,1)^\top$ for the all-ones vector, of the length dictated by the context, so that $\mathbf{j}\mathbf{j}^\top$ is the matrix all of whose entries equal $1$; it should not be confused with an identity operator. Finally, $\boldsymbol{\delta}_x$ denotes the unit point mass at a point $x$.

\subsection{The Arveson--Radon--Nikodym theorem}
Given two completely positive maps, \cite{arveson1969subalgebras} showed that one can construct a kind of Radon--Nikodym derivative analogue. This is used in \cite{kian2026relative} to construct the mean of completely positive maps. We will repeat the construction here in the finite-dimensional setting, adjusting the notation to be compatible with Watrous \cite{watrous2018theory}. For now we recall the important facts for the Arveson--Radon--Nikodym (ARN) theorem. The following follows from \cite{arveson1969subalgebras} with some objects changed to fit the conventions of \cite{kian2026relative}. We show the proof in the finite-dimensional case. This version uses the Kraus form of the ARN theorem \cite{Raginsky2003}. The following lemma is a special case of the Douglas factorisation lemma \cite{Douglas1966}; we include a proof for completeness.
 
\begin{lemma}\label{lem:factor}
Let $W\in \mathrm{L}(\mathbb{C}^r,\mathcal{V})$ be injective and let $Y\in \mathrm{L}(\mathcal{V})$ satisfy $0\leq Y \leq WW^*$.
Then there is a unique $T\in \mathrm{L}(\mathbb{C}^r)$ with $Y=WTW^*$, and it satisfies $0\leq T\leq\mathbf{1}_{\mathbb{C}^r}$.
\end{lemma}
 
\begin{proof}
Since $W$ has full column rank, $W^*W$ is invertible, and the Moore--Penrose pseudoinverse $W^+=(W^*W)^{-1}W^*$ satisfies $W^+W=\mathbf{1}_{\mathbb{C}^r}$. The operator
$P=WW^+$ is a self-adjoint orthogonal projection onto $\range(W)$.
 
If $W^*x=0$, then
$0\leq x^*Yx\leq x^*WW^*x=\|W^*x\|^2=0$, so $Y^{1/2}x=0$ and therefore $Yx=0$. Hence $(\range W)^\perp=\ker W^*\subseteq\ker Y$. This implies $Y(\mathbf{1}_{\mathcal{V}}-P)=0$, that is, $Y=YP$. Taking adjoints and using that $P$ and $Y$ are self-adjoint, we obtain $Y=PY$, so that $Y=PYP$.
 
Let $T=W^+Y(W^+)^*$. This satisfies
\[WTW^*=(WW^+)Y(WW^+)^*=PYP=Y.\]
Since conjugation preserves the Loewner order, $0\leq Y\leq WW^*$ implies
\[0\leq T\leq W^+WW^*(W^+)^*=(W^+W)(W^+W)^*=\mathbf{1}_{\mathbb{C}^r}.\]
Finally, suppose that $T,T'\in \mathrm{L}(\mathbb{C}^r)$ satisfy $WTW^*=WT'W^*$. Multiplying by $W^+$ on the left and by $(W^+)^*$ on the right, and using $W^+W=\mathbf{1}_{\mathbb{C}^r}$, we obtain $T=T'$.
\end{proof}
 
\begin{theorem}[Arveson's Radon--Nikodym theorem, Kraus form]\label{thm:arveson}
Let $\Phi\in\mathrm{CP}(\mathcal{X},\mathcal{Y})$ with $\Phi(X)=\sum_{k=1}^rV_kXV_k^*$, where $V_1,\dots,V_r\in \mathrm{L}(\mathcal{X},\mathcal{Y})$ are
linearly independent. For $T\in \mathrm{L}(\mathbb{C}^r)$ define the linear map
\[\Psi_T(X)=\sum_{k,l=1}^rT_{kl}V_kXV_l^*.\]
A linear map $\Psi:\mathrm{L}(\mathcal{X})\to \mathrm{L}(\mathcal{Y})$ satisfies 
\[0\leq \Psi\leq \Phi\] in the completely positive ordering (that is, $\Psi$ and $\Phi-\Psi$ are completely positive) if and only if $\Psi=\Psi_T$ for some $T\in \mathrm{L}(\mathbb{C}^r)$ satisfying $0\leq T\leq\mathbf{1}_{\mathbb{C}^r}$. This $T$ is unique. Moreover, with the injective operator $W=[\vectorise(V_1)\ \cdots\ \vectorise(V_r)]$, we have $J(\Psi_T)=WTW^*$ for every $T\in\mathrm{L}(\mathbb{C}^r)$.
\end{theorem}
 
\begin{proof}
Let $W=[\vectorise(V_1)\ \cdots\ \vectorise(V_r)]$. Since $\vectorise$ is a linear bijection and the $V_k$ are
linearly independent, $W$ has full column rank. As the map $X\mapsto AXB^*$ has Choi matrix $\vectorise(A)\vectorise(B)^*$ \cite[Section~2.2.2]{watrous2018theory},
linearity of $J$ gives $J(\Psi_T)=WTW^*$ for every $T$, and in particular $J(\Phi)=WW^*$. If $0\leq\Psi\leq\Phi$, then by Choi's
theorem \cite{Choi1975} $0\leq J(\Psi)\leq J(\Phi)=WW^*$, so Lemma~\ref{lem:factor} yields a unique $T$ with $0\leq T\leq\mathbf{1}_{\mathbb{C}^r}$ and
$J(\Psi)=WTW^*=J(\Psi_T)$; injectivity of $J$ gives $\Psi=\Psi_T$, and uniqueness of $T$ follows from the uniqueness in
Lemma~\ref{lem:factor}. Conversely, if $0\leq T\leq\mathbf{1}_{\mathbb{C}^r}$, then $J(\Psi_T)=WTW^*\geq0$ and $J(\Phi-\Psi_T)=W(\mathbf{1}_{\mathbb{C}^r}-T)W^*\geq0$,
so $\Psi_T$ and $\Phi-\Psi_T$ are completely positive by Choi's theorem, that is, $0\leq\Psi_T\leq\Phi$.
\end{proof}
This theorem really says that we have an order preserving bijection
\[
    \{T\in \mathrm{L}(\mathbb{C}^r):0\leq T\leq \mathbf{1}_{\mathbb{C}^r}\}\leftrightarrow \{\Psi \in \mathrm{CP}(\mathcal{X},\mathcal{Y}):\Psi\leq \Phi\}.
\]
 
We make a definition which differs slightly from the general $C^*$-algebra version, this is purely out of convenience because in the finite-dimensional case, the corresponding commutant elements are decided by positive semidefinite operators \cite{Okayasu2026}. 
\begin{definition}
    Let $\Phi$ and $\Psi$ be as above. We say that the matrix $T$ is the ARN derivative of $\Psi$ with respect to $\Phi$ (relative to the Kraus decomposition $V_1,\dots,V_r$ of $\Phi$). 
\end{definition}
 
\begin{remark}
    Two Kraus decompositions $(V_k)_{k=1}^r$ and $(V'_j)_{j=1}^r$ of $\Phi$ by linearly independent operators are related by a unitary $U\in \mathrm{L}(\mathbb{C}^r)$, $V'_j=\sum_k U_{jk}V_k$ \cite[Section~2.2.2]{watrous2018theory}. A direct computation shows that the corresponding ARN derivatives of a fixed $\Psi\leq\Phi$ are related by $T'=\overline{U}\,T\,U^\top$. In particular $T$ and $T'$ have the same spectrum, and all unitarily invariant norms of $T$, and of differences $T_1-T_2$ of derivatives computed in the same Kraus decomposition, are independent of the decomposition chosen. 
\end{remark}

 For a completely positive map $\Omega$ with Kraus decomposition $\Omega(X)=\sum_{k=1}^rV_kXV_k^*$ by linearly independent operators $V_1,\dots,V_r$, let
\[
\nu(\Omega)=\min\Big\{\Big\|\sum_{k=1}^rc_kV_k\Big\|_\infty:\ c\in\mathbb{C}^r,\ \|c\|_2=1\Big\}>0 .
\]
This quantity relates the operator norm of a derivative to the completely bounded trace norm of the corresponding map; its properties are collected in Lemma~\ref{lem:nu} and Theorem~\ref{thm:two sided} in Section~\ref{sec:four maps}.

\subsection{The effect point of view}
 
One may represent a completely positive map as the compression of an operator which acts on the tensor product of the underlying space and some ancillary space. One can read the ARN derivative as saying that if $\Psi\leq \Phi$, then $\Psi$ can be understood as applying an effect on the ancilla space with respect to which the Stinespring dilation \cite{Stinespring1955} of $\Phi$ is written. 
Let $\Phi\in \mathrm{CP}(\mathcal{X},\mathcal{Y})$ have Kraus decomposition $\Phi(X)=\sum_{k=1}^rV_kXV_k^*$ with linearly independent $V_k$, let $\mathcal{Z}=\mathbb{C}^r$ with standard basis $e_1,\dots,e_r$, and let $S\in \mathrm{L}(\mathcal{X},\mathcal{Y}\otimes\mathcal{Z})$ be the Stinespring operator written in the style of  \cite[Section~2.2.2]{watrous2018theory}
\[Sx=\sum_{k=1}^r V_kx\otimes e_k,\]
so that $\Phi(X)=\Tr_{\mathcal{Z}}(SXS^*)$. We can then consider the complementary completely positive map \cite{DevetakShor2005} 
$\Phi^c(X)=\Tr_{\mathcal{Y}}(SXS^*)$. 
\begin{theorem}\label{effect theorem}
    For Hermitian $T\in \mathrm{L}(\mathcal{Z})$ we have 
    \begin{enumerate}
        \item $\Psi_T(X)=\Tr_{\mathcal{Z}}\left( (\mathbf{1}_{\mathcal{Y}}\otimes T^\top)SXS^* \right)$ for all $X\in \mathrm{L}(\mathcal{X})$;
        \item $\tr(\Psi_T(X))=\tr(T^\top\Phi^c(X))$;
        \item $\Psi_T^*(\mathbf{1}_{\mathcal{Y}})=S^*(\mathbf{1}_{\mathcal{Y}}\otimes T^\top)S=(\Phi^c)^*(T^\top)$.
    \end{enumerate}
 If $0\leq T\leq \mathbf{1}_{\mathcal{Z}}$, one can consider $\Psi_T$ as the outcome of applying $\Phi$ and measuring an effect $T^\top$ on the ancillary environment generated by $\Phi$. In particular, one can decompose each channel $\Phi$ as a sum $\sum_{i} \Psi_{T_i}$ where the $T_i^\top$ form a POVM on the environment $\mathcal{Z}$. Moreover, for $T\geq 0$ the completely bounded trace norm (which coincides with the induced trace norm for completely positive maps \cite{Paulsen2003}) satisfies 
 \[
     |||\Psi_T|||_1=\max_{\rho\in \mathrm{D}(\mathcal{X})}\tr(T^\top\Phi^c(\rho)).
 \]
\end{theorem}
 \begin{proof}
     We have $SXS^*=\sum_{kl}V_kXV_l^*\otimes e_ke_l^*$, hence
     \[\Tr_{\mathcal{Z}}\left((\mathbf{1}_{\mathcal{Y}}\otimes T^\top)SXS^*\right)=\sum_{kl}\tr(T^\top e_ke_l^*)\,V_kXV_l^*=\sum_{kl}T_{kl}V_kXV_l^*=\Psi_T(X),\]
     since $\tr(T^\top e_ke_l^*)=(T^\top)_{lk}=T_{kl}$. This proves (1).

     Taking the full trace in (1) and using \[\tr((\mathbf{1}_{\mathcal{Y}}\otimes T^\top)SXS^*)=\tr(T^\top\Tr_{\mathcal{Y}}(SXS^*))\] gives (2).

     For (3), the adjoint of \[X\mapsto (\mathbf{1}_{\mathcal{Y}}\otimes T^\top)SXS^*\] is \[R\mapsto S^*(\mathbf{1}_{\mathcal{Y}}\otimes \overline{T})RS\] and the adjoint of $\Tr_{\mathcal{Z}}$ is \[Y\mapsto Y\otimes\mathbf{1}_{\mathcal{Z}}.\] Hence we have that $\Psi_T^*(Y)=S^*(\mathbf{1}_{\mathcal{Y}}\otimes \overline{T})(Y\otimes\mathbf{1}_{\mathcal{Z}})S=S^*(Y\otimes \overline{T})S=S^*(Y\otimes T^\top)S$ where we have used that $\overline{T}=T^\top$ for Hermitian $T$, and letting $Y=\mathbf{1}_{\mathcal{Y}}$ and comparing with $(\Phi^c)^*(Z)=S^*(\mathbf{1}_{\mathcal{Y}}\otimes Z)S$. Finally, if $T\geq 0$ then $\Psi_T$ is completely positive, so we may apply the Russo--Dye theorem \cite{RussoDye1966} in the form \cite[Proposition~3.6]{Paulsen2003} to conclude \[|||\Psi_T|||_1=\|\Psi_T^*(\mathbf{1}_{\mathcal{Y}})\|_\infty=\max_{\rho\in \mathrm{D}(\mathcal{X})}\tr(\Psi_T^*(\mathbf{1}_{\mathcal{Y}})\rho)=\max_{\rho\in \mathrm{D}(\mathcal{X})}\tr(\Psi_T(\rho)),\] and (2) gives the last formula.
 \end{proof}
\begin{remark}[Physical interpretation]\label{rem:phys-effect}
Applied to $\Omega=\Phi+\Psi$, the theorem shows that $\Phi=\Psi_T$ and $\Psi=\Psi_{\mathbf{1}_{\mathbb{C}^r}-T}$, where $T$ is the derivative of $\Phi$ with respect to $\Omega$, arise by running a Stinespring dilation of $\Omega$ and measuring the POVM $\{T^\top,\mathbf{1}_{\mathcal{Z}}-T^\top\}$ on its environment. So $(\Phi,\Psi)$ is a two-outcome instrument implementing $\Omega$ (when $\Omega$ is a channel), and $T^\top$ is the which-branch effect hidden in the environment. Classically, $t=d\mathbb{P}/d(\mathbb{P}+\mathbb{Q})$ is the posterior probability of $\mathbb{P}$ under a uniform prior and $D_f(\mathbb{P}\|\mathbb{Q})=\int h_f(t)\,d(\mathbb{P}+\mathbb{Q})$, with $h_f$ as in Lemma~\ref{lem:hf extended}; by analogy, Theorem~\ref{thm:all f}(b) is its channel version.
\end{remark}
 
\subsection{Kubo--Ando means of operators} 
In \cite{kian2026relative}, one considers means of channels through the theory of Kubo--Ando means. Hence we must take a quick detour and list some background facts. 
 
In order to consider any kind of notion of means of channels using Kubo--Ando theory, as performed in \cite{kian2026relative}, we must first recall the Kubo--Ando theory of means of positive semidefinite operators. The following is an amalgamation of facts from \cite{bhatia1997matrix}, \cite{hiai2010matrix}, \cite{hiaipetz2014intro}, and the original paper \cite{KuboAndo1980} and none of these results are claimed as ours. In order to make later steps more compatible with quantum information theory related texts, we rephrase some statements or change notation to be consistent with \cite{watrous2018theory}. This will be consistent across the rest of the text.
 
Matrix means are usually studied through operator connections. One astounding fact is that although operator connections need only satisfy a few general properties, there are several important and extremely useful classification theorems. We will be denoting complex Euclidean spaces using calligraphic letters, $\mathcal{X}, \mathcal{Y}$ and $\mathcal{Z}$, etc. Many of the philosophies in this section carry over to the analogous operator convex setting, where one defines relations between positive semidefinite matrices using operator convex functions (see \cite{hiaipetz2014intro},\cite{bhatia1997matrix}). This is perfectly fine, however we will see by appropriately shifting things around, we can get to the desired results for some divergences without considering the techniques of Section~\ref{sec:unbounded}.
 
\begin{definition}[Loewner order]
Given a complex Euclidean space $\mathcal{X}$, the Loewner order is the partial order on the Hermitian operators given by
\[
    A\leq B\iff B-A\geq 0
\]
in the sense that $B-A$ is positive semidefinite.
\end{definition}
 
\begin{definition}[Operator connection \cite{KuboAndo1980}]
A binary operation $(A,B)\mapsto A\,\sigma\,B$ on $\mathrm{Pos}(\mathcal{X})$ is an operator connection if it satisfies the following 
\begin{itemize}
    \item (monotonicity) $0\leq A\leq A'$ and $0\leq B\leq B'$ imply $A\,\sigma\,B\leq A'\,\sigma\,B'$;
    \item (transformer inequality) for $A,B\geq 0$ and arbitrary $C$,
    \[
        C^*(A\,\sigma\,B)C\leq (C^*AC)\,\sigma\,(C^*BC);
    \]
    \item (continuity from above) if $A_n\searrow A$ and $B_n\searrow B$ then 
    \[
        (A_n\,\sigma\,B_n)\searrow (A\,\sigma\,B),
    \] where $A_n\searrow A$ indicates that 
    $A_n\to A$ in the operator norm and $A_1\geq A_2\geq \dots$.
\end{itemize}
In particular, for $A,B \in \mathrm{Pos}(\mathcal{X})$ one has
\[
    A\,\sigma\,B=\lim_{\eta\downarrow0}A_\eta\,\sigma\,B_\eta,
\]
where $A_\eta= A+\eta \mathbf{1}_{\mathcal{X}}$ and $B_\eta=B+\eta \mathbf{1}_{\mathcal{X}}$, so that a connection is determined by its values on positive definite operators.
\end{definition}
\subsection{Constructing operator connections}
We will see that we can construct operator connections in a standard way. We need to introduce several definitions before we continue. 
 
\begin{definition}[Operator monotone function {\cite[Chapter~V]{bhatia1997matrix}}]
Let $(t_-,t_+)$ be an open interval. A function $f:(t_-,t_+)\to \mathbb{R}$ is said to be operator monotone (or matrix monotone) on $(t_-,t_+)$ if $f(A)\leq f(B)$ whenever $A\leq B$ are Hermitian operators, of any size, with spectra contained in $(t_-,t_+)$. 
\end{definition}
 
\begin{definition}[Pick function {\cite[Chapter~V]{bhatia1997matrix}}]
    Let $\mathbb{C}^+=\{z\in \mathbb{C}: \mathrm{Im}(z)>0\}$. A function $f:\mathbb{C}^+\to \mathbb{C}$ is called a Pick function if it is analytic and $\mathrm{Im}f(z)\geq 0$ for all $z\in\mathbb{C}^+$ (by the open mapping theorem, a nonconstant Pick function maps $\mathbb{C}^+$ into $\mathbb{C}^+$). We denote the set of all Pick functions by $\mathcal{P}$. Given a real interval $(t_-,t_+)$, we denote by $\mathcal{P}(t_-,t_+)$ the set of all Pick functions which extend analytically across $(t_-,t_+)$ with real values on $(t_-,t_+)$.
\end{definition}
 
Finding operator monotone functions may appear to be a tricky business; however, they are characterised by a theorem of Loewner \cite{Loewner1934}; see also \cite[Chapter~V]{bhatia1997matrix}. 
\begin{theorem}[Loewner \cite{Loewner1934}]\label{Loewner}
Let $(t_-,t_+)$ be a real interval (possibly all of $\mathbb{R}$) and $f$ a real valued function on $(t_-,t_+)$. Then $f$ is operator monotone on $(t_-,t_+)$ if and only if $f$ is the restriction to $(t_-,t_+)$ of a function in $\mathcal{P}(t_-,t_+)$. 
\end{theorem}
The next theorem tells us that operator connections correspond to operator monotone functions, and the general theory of operator means defined in this way shares the name Kubo--Ando with this theorem. 
 
\begin{theorem}[Kubo--Ando theorem \cite{KuboAndo1980}]\label{kubo-ando}
For each operator connection $\sigma$ there exists a unique operator monotone function $f_\sigma:(0,\infty)\to[0,\infty)$, the representing function of $\sigma$, such that 
\[
    f_\sigma(t)\mathbf{1}_{\mathcal{X}}=\mathbf{1}_{\mathcal{X}}\,\sigma\,(t\mathbf{1}_{\mathcal{X}}) 
\]
for every $t>0$, and, for $A>0$ and $B\geq 0$,
\[
    A\,\sigma\,B=A^{1/2}f_\sigma(A^{-1/2}BA^{-1/2})A^{1/2}=f_\sigma(BA^{-1})A
\]
via the functional calculus. If $A$ is not invertible, one perturbs it and takes limits from above. Conversely, every operator monotone $f:(0,\infty)\to[0,\infty)$ defines, by this formula, an operator connection whose representing function is $f$.
\end{theorem}
 
\begin{definition}
    If also $f_\sigma(1)=1$, that is $\mathbf{1}_{\mathcal{X}}\,\sigma\,\mathbf{1}_{\mathcal{X}}=\mathbf{1}_{\mathcal{X}}$, we call $\sigma$ an operator mean.
\end{definition}
\begin{remark}
    The term operator mean varies in the literature, and one must take care here. We do not impose on operator means that they be commutative unlike other authors.
\end{remark}
For example, the representing functions $\sqrt{t}$, $t^\alpha$ with $0<\alpha<1$, and $t/(1+t)$ give, respectively, the geometric mean \cite{PuszWoronowicz1975,KuboAndo1980}, the weighted geometric mean \cite{KuboAndo1980} and the parallel sum \cite{AndersonDuffin1969},
\begin{align*}
A\#B&=A^{1/2}\big(A^{-1/2}BA^{-1/2}\big)^{1/2}A^{1/2},\\
A\#_\alpha B&=A^{1/2}\big(A^{-1/2}BA^{-1/2}\big)^{\alpha}A^{1/2},\\
A:B&=\big(A^{-1}+B^{-1}\big)^{-1},
\end{align*}
for positive definite $A,B$, extended to positive semidefinite operators by the regularisation above.
 
We then see by combining Theorems~\ref{kubo-ando} and~\ref{Loewner} that every nonnegative function in $\mathcal{P}(0,\infty)$ gives rise to an operator connection.
 
\subsection{The mean of completely positive maps}
 
We now come to the construction developed by \cite{kian2026relative}. Given two completely positive maps $\Phi,\Psi\in \mathrm{CP}(\mathcal{X},\mathcal{Y})$, let $\Omega=\Phi+\Psi$. If $\Omega$ has Kraus decomposition 
$\Omega(X)=\sum_{k=1}^rV_kXV_k^*$ with linearly independent $V_k$, then \[(\Phi\,\sigma\,\Psi)(X)=\sum_{k,l=1}^r(T_1\,\sigma\,T_2)_{kl}V_kXV_l^*\]
where $T_1$ and $T_2$ are the ARN derivatives of $\Phi$ and $\Psi$ with respect to the Kraus decomposition of $\Omega$. Kian defines the mean of two maps dominated by a common map through their Arveson--Radon--Nikodym derivatives, so that the derivative of the mean is the mean of the derivatives \cite[Proposition~3.5]{kian2026relative}, and the formula above is the finite-dimensional Kraus form of this definition \cite[Proposition~3.13]{kian2026relative}, in our notation. The choice $\Omega=\Phi+\Psi$ gives what Kian calls the intrinsic Kubo--Ando mean \cite[Remark~3.3]{kian2026relative}. In general the relative mean depends on the dominating map \cite[Remark~3.3]{kian2026relative}, but not in our setting: the relative geometric mean does not depend on it \cite[Remark~4.2]{kian2026relative}, and for maps between matrix algebras the same holds for every Kubo--Ando mean \cite[Theorem~5.1]{kian2026relative}. One observation which is easily made is that under the ARN bijection, since the derivative of $\Omega$ with respect to itself is $\mathbf{1}_{\mathbb{C}^r}$, we must have $T_2=\mathbf{1}_{\mathbb{C}^r}-T_1$. We now formulate the fundamental theorem of means of completely positive maps which will inform the rest of the paper.
Let $\sigma$ be a Kubo--Ando mean, so that its representing function $f_\sigma$ is positive on $(0,\infty)$ and $f_\sigma(1)=1$. For scalars $s,t>0$ we have $s\,\sigma\,t=sf_\sigma(t/s)$. We define the function
\[
    g_\sigma(t)=t\,\sigma\,(1-t)=tf_\sigma\left(\frac{1-t}{t}\right),\qquad 0<t<1.
\]
For a function $f$ on $(0,\infty)$ we write $f(0^+)=\lim_{t\downarrow0}f(t)$ and $f'(\infty)=\lim_{t\to\infty}f(t)/t$ whenever these limits exist, following \cite{Okayasu2026maxf}. For $f=f_\sigma$ both limits exist in $[0,\infty)$, since $f_\sigma$ is nondecreasing and concave \cite[Theorem~V.2.5]{bhatia1997matrix}, and we set $g_\sigma(0)=f_\sigma'(\infty)$ and $g_\sigma(1)=f_\sigma(0^+)$, which are the limits of $g_\sigma$ at the endpoints. 
The following theorem is essentially due to Kian: by \cite[Proposition~3.13 and Theorem~5.1]{kian2026relative}, $J(\Phi)\,\sigma\,J(\Psi)=J(\Phi\,\sigma\,\Psi)=W\big(T\,\sigma\,(\mathbf{1}_{\mathbb{C}^r}-T)\big)W^*$ in the notation of the theorem, and for the geometric mean the identity $J(\Phi\#\Psi)=J(\Phi)\#J(\Psi)$ is also \cite[Proposition~4.1]{Okayasu2026}. We include a proof for completeness in the present notation; what we add is the identity $T\,\sigma\,(\mathbf{1}_{\mathbb{C}^r}-T)=g_\sigma(T)$, through the functional calculus, and the resulting description of the spectrum.
 
\begin{theorem}\label{fundamental principle}
    Let $\sigma$ be a Kubo--Ando mean with representing function $f_\sigma$, let $\Phi,\Psi\in\mathrm{CP}(\mathcal{X},\mathcal{Y})$, let $\Omega=\Phi+\Psi$ have a Kraus decomposition with linearly independent operators $V_k$, let $W=[\vectorise(V_1)\ \cdots\ \vectorise(V_r)]$, and let $T$ be the ARN derivative of $\Phi$ with respect to $\Omega$. Then at the Choi level
    \[
        J(\Phi)\,\sigma\,J(\Psi)=Wg_{\sigma}(T)W^*=W\left(T^{1/2}f_\sigma\big(T^{-1/2}(\mathbf{1}_{\mathbb{C}^r}-T)T^{-1/2}\big)T^{1/2}\right)W^*,
    \]
    where in the case of singular $T$ or $\mathbf{1}_{\mathbb{C}^r}-T$ one takes decreasing nonsingular regularisations from above, using the continuity from above of operator connections. Consequently the mean $\Phi\,\sigma\,\Psi$ is a submap of $\Phi+\Psi$ with ARN derivative $g_\sigma(T)$, whose spectrum is $g_\sigma(\mathrm{spec}(T))$.
\end{theorem}
\begin{proof}
Let $W^+=(W^*W)^{-1}W^*$ and $P=WW^+$ be as in the proof of Lemma~\ref{lem:factor}, so that $W^+W=\mathbf{1}_{\mathbb{C}^r}$ and $P$ is the orthogonal projection onto $\range(W)$. We first show that
\begin{equation}\label{eq:W transfer}
(WAW^*)\,\sigma\,(WBW^*)=W(A\,\sigma\,B)W^*\qquad\text{for all }A,B\in\mathrm{Pos}(\mathbb{C}^r).
\end{equation}
Let $M=(WAW^*)\,\sigma\,(WBW^*)$. We will show that transformer inequality for operators $C\in\mathrm{L}(\mathcal{V}_1,\mathcal{V}_2)$ between different spaces holds (for more suitable matrices this already holds \cite{bhatia1997matrix}). This follows from the square case by regarding all operators as acting on $\mathcal{V}_1\oplus\mathcal{V}_2$ as existing in a larger space by extending by zero. This uses that $(A\oplus0)\,\sigma\,(B\oplus0)=(A\,\sigma\,B)\oplus0$, by the Kubo--Ando formula for positive definite operators and regularisation. The transformer inequality with $C=W^*$ gives $W(A\,\sigma\,B)W^*\leq M$. With $c_{A,B}=\max\{\|A\|_\infty,\|B\|_\infty\}$, monotonicity gives $M\leq(c_{A,B}WW^*)\,\sigma\,(c_{A,B}WW^*)=c_{A,B}WW^*$, since $X\,\sigma\,X=X$ for a mean \cite{KuboAndo1980}; hence $\ker W^*\subseteq\ker M$, and $M=PMP$. The transformer inequality with $C=(W^+)^*$ gives
\[
W^+M(W^+)^*\leq\big(W^+WAW^*(W^+)^*\big)\,\sigma\,\big(W^+WBW^*(W^+)^*\big)=A\,\sigma\,B,
\]
and conjugating by $W$ yields $M=PMP\leq W(A\,\sigma\,B)W^*$, which proves \eqref{eq:W transfer}.

By the proof of Theorem~\ref{thm:arveson}, $J(\Phi)=WTW^*$ and $J(\Psi)=W(\mathbf{1}_{\mathbb{C}^r}-T)W^*$, so \eqref{eq:W transfer} gives $J(\Phi)\,\sigma\,J(\Psi)=W\big(T\,\sigma\,(\mathbf{1}_{\mathbb{C}^r}-T)\big)W^*$, and the second expression in the statement is the Kubo--Ando formula when $T$ and $\mathbf{1}_{\mathbb{C}^r}-T$ are invertible. It remains to show that $T\,\sigma\,(\mathbf{1}_{\mathbb{C}^r}-T)=g_\sigma(T)$. For $\eta>0$ the operators $T+\eta\mathbf{1}_{\mathbb{C}^r}$ and $(\mathbf{1}_{\mathbb{C}^r}-T)+\eta\mathbf{1}_{\mathbb{C}^r}$ are positive definite and commute, so by Theorem~\ref{kubo-ando} and the functional calculus, $(T+\eta\mathbf{1}_{\mathbb{C}^r})\,\sigma\,(\mathbf{1}_{\mathbb{C}^r}-T+\eta\mathbf{1}_{\mathbb{C}^r})=g_{\sigma,\eta}(T)$ with $g_{\sigma,\eta}(t)=(t+\eta)f_\sigma\big(\tfrac{1-t+\eta}{t+\eta}\big)$. As $\eta\downarrow0$ the left-hand side decreases to $T\,\sigma\,(\mathbf{1}_{\mathbb{C}^r}-T)$ by continuity from above, while $g_{\sigma,\eta}(t)\to g_\sigma(t)$ for every $t\in[0,1]$: for $0<t<1$ by continuity of $f_\sigma$, and at $t=0$ and $t=1$ by the definitions of $f_\sigma'(\infty)$ and $f_\sigma(0^+)$. Since $T$ is a positive semidefinite operator and $\sigma$ is a Kubo--Ando mean, $g_{\sigma,\eta}(T)\to g_\sigma(T)$, and hence $T\,\sigma\,(\mathbf{1}_{\mathbb{C}^r}-T)=g_\sigma(T)$.

Finally, for $0<t<1$ monotonicity gives $0\leq g_\sigma(t)=t\,\sigma\,(1-t)\leq 1\,\sigma\,1=1$, and the same holds at the endpoints by taking limits. Thus $0\leq g_\sigma(T)\leq\mathbf{1}_{\mathbb{C}^r}$, and $\Phi\,\sigma\,\Psi=\Psi_{T\,\sigma\,(\mathbf{1}_{\mathbb{C}^r}-T)}=\Psi_{g_\sigma(T)}$ is a submap of $\Omega$ with ARN derivative $g_\sigma(T)$ by Theorem~\ref{thm:arveson}. The statement about the spectrum is the spectral mapping theorem.
\end{proof}

\begin{remark}\label{rem:connections}
The proof of Theorem~\ref{fundamental principle} uses only monotonicity, the transformer inequality, continuity from above and the identity $X\,\sigma\,X=f_\sigma(1)X$. It therefore applies verbatim to every operator connection $\sigma$, with representing function $f_\sigma$ as in Theorem~\ref{kubo-ando}. Put
\[
\varphi_\sigma(t)=t\,\sigma\,(1-t)=tf_\sigma\Big(\frac{1-t}{t}\Big)\quad(0<t<1),\qquad\varphi_\sigma(0)=f_\sigma'(\infty),\qquad\varphi_\sigma(1)=f_\sigma(0^+),
\]
which are the limits of $\varphi_\sigma$ at the endpoints, since $f_\sigma$ is nonnegative, nondecreasing and concave. Then:
\begin{enumerate}
\item \eqref{eq:W transfer} holds for the connection $\sigma$, that is, $(WAW^*)\,\sigma\,(WBW^*)=W(A\,\sigma\,B)W^*$ for all $A,B\in\mathrm{Pos}(\mathbb{C}^r)$;
\item $T\,\sigma\,(\mathbf{1}_{\mathbb{C}^r}-T)=\varphi_\sigma(T)$ for $0\leq T\leq\mathbf{1}_{\mathbb{C}^r}$, in the functional calculus.
\end{enumerate}
If $\sigma$ is a mean, then $\varphi_\sigma=g_\sigma$.
\end{remark}

\section{The mean of channels and geometric channel fidelity}\label{sec:fidelity}
 
In this section, we present a concrete application of the means of completely positive maps in terms of the ARN derivative. In quantum information theory, one can consider many kinds of divergence. One of these is the geometric channel fidelity. This quantity has some annoying properties. In this section, we will build up the definition through the Matsumoto fidelity \cite{matsumoto2010reverse}. We will retranslate some definitions in \cite{matsumoto2010reverse} and \cite{KatariyaWilde2021} in the language of Watrous \cite{watrous2018theory}. 
 
\subsection{Matsumoto fidelity}
\subsubsection{Motivation of the definition}
In quantum information theory, one often wants to consider how much two quantum states ``overlap''. One ubiquitous example of this is the quantum (or Uhlmann) fidelity \cite{Uhlmann1976} of two states. Given a complex Euclidean space $\mathcal{X}$ and density operators $\rho, \xi\in \mathrm{D}(\mathcal{X})$, their fidelity is defined as 
\[F(\rho,\xi)=\tr\left(\sqrt{\sqrt{\xi}\rho\sqrt{\xi}}\right)=\|\sqrt{\rho}\sqrt{\xi}\|_1.\]
The quantum fidelity is the analogue of the Bhattacharyya coefficient in classical probability (see, for example, \cite[Section~3.2]{watrous2018theory}). The fidelity is closely related to the geometric mean: for invertible $\xi$ one checks directly that 
\[F(\rho,\xi)=\tr\big(\xi\,(\xi^{-1}\#\rho)\big).\] In \cite{matsumoto2010reverse}, Matsumoto considers flipping the definition for quantum fidelity, and instead takes the trace of the geometric mean of the operators themselves. The Matsumoto fidelity is defined on states $\rho, \xi\in \mathrm{D}(\mathcal{X})$ by 
\[
    F_M(\rho,\xi)=\tr(\rho\#\xi).
\]
The following interpretation is given in \cite{cree2020fidelity}. Two classical probability distributions $\{p_i\}, \{q_i\}$ prepare states $\rho$ and $\xi$ if there exists a quantum channel $\Phi$ such that $\Phi(\mathrm{diag}(\{p_i\}))=\rho$ and $\Phi(\mathrm{diag}(\{q_i\}))=\xi$. The Matsumoto fidelity is the maximal possible classical fidelity between classical probability distributions that can prepare states $\rho$ and $\xi$. The quantum fidelity is the minimum classical fidelity over quantum-to-classical measurement procedures; the Matsumoto fidelity is the maximum classical fidelity over classical-to-quantum preparations. In this sense, it is a kind of dual to the quantum fidelity \cite{cree2020fidelity}. Comprehensive properties of the Matsumoto fidelity can be found in \cite{cree2020fidelity} and the original paper by Matsumoto \cite{matsumoto2010reverse}.
 
\subsection{Geometric channel fidelity}
Given two channels $\Phi, \Psi\in \mathrm{C}(\mathcal{X},\mathcal{Y})$, one can consider the maximum output fidelity \cite[Section~3.3]{watrous2018theory} given by 
\[
    F_{\max}(\Phi,\Psi)=\max_{\rho,\xi\in \mathrm{D}(\mathcal{X})}F(\Phi(\rho),\Psi(\xi)).
\]
The geometric channel fidelity as defined in \cite{KatariyaWilde2021} (see also \cite[eq.~(19)]{NuradhaWilde2025}) is a kind of opposite to this, using the Matsumoto fidelity. They define the geometric channel fidelity as the square of the following quantity: 
\[
    \widehat{F}(\Phi,\Psi)=\inf_{\rho\in \mathrm{D}(\mathcal{X}\otimes\mathcal{X})}F_M\big((\Phi\otimes\mathbf{1}_{\mathrm{L}(\mathcal{X})})(\rho), (\Psi\otimes\mathbf{1}_{\mathrm{L}(\mathcal{X})})(\rho)\big).
\] 
We will not be using the squared quantity, and will refer to the above as the geometric channel fidelity. 
We notice in this definition that there is only one input state, shared by both channels. This is because if we allowed two, it would be very easy to make this quantity zero. 
 
In this section and the next we work under the assumption that we are considering completely positive maps, not necessarily channels. The definition of the geometric channel fidelity still makes sense, because complete positivity guarantees that $(\Phi\otimes\mathbf{1}_{\mathrm{L}(\mathcal{X})})(\rho)$ and $(\Psi\otimes\mathbf{1}_{\mathrm{L}(\mathcal{X})})(\rho)$ are positive semidefinite, which is all we require for the geometric mean to make sense; by Choi's theorem \cite{Choi1975} the same holds for the Choi matrices $J(\Phi)$ and $J(\Psi)$.  
 
\begin{proposition}\label{lambda min formula}
   For $\Phi,\Psi\in \mathrm{CP}(\mathcal{X},\mathcal{Y})$, $\widehat{F}(\Phi,\Psi)=\lambda_{\min}(\Tr_{\mathcal{Y}}(J(\Phi)\#J(\Psi)))$.
\end{proposition}
\begin{proof}
    For channels this is a consequence of \cite[Proposition~44]{KatariyaWilde2021} at $\alpha=\tfrac12$ (see also the semidefinite program for the geometric fidelity of channels in \cite[eq.~(20)]{NuradhaWilde2025}); we include the short argument since we need it for completely positive maps. Since $F_M$ is jointly concave \cite{cree2020fidelity} and $\rho\mapsto(\Phi\otimes\mathbf{1}_{\mathrm{L}(\mathcal{X})})(\rho)$ is linear, the function being minimised is concave on the compact convex set $\mathrm{D}(\mathcal{X}\otimes\mathcal{X})$, so the infimum may be taken over pure states $uu^*$. Every unit vector $u\in\mathcal{X}\otimes\mathcal{X}$ is of the form $u=\vectorise(A)=(\mathbf{1}_{\mathcal{X}}\otimes A^\top)\vectorise(\mathbf{1}_{\mathcal{X}})$ with $\|A\|_2=1$ \cite[Section~1.1.2]{watrous2018theory}, whence
    \begin{equation}\label{eq:vec identity}(\Phi\otimes\mathbf{1}_{\mathrm{L}(\mathcal{X})})(uu^*)=(\mathbf{1}_{\mathcal{Y}}\otimes A^\top)J(\Phi)(\mathbf{1}_{\mathcal{Y}}\otimes A^\top)^*,\end{equation}
    and similarly for $\Psi$. If $A$ is invertible, the transformer equality for invertible operators \cite{KuboAndo1980} gives
    \begin{align*}
    \tr\big((\Phi\otimes\mathbf{1}_{\mathrm{L}(\mathcal{X})})(uu^*)\#(\Psi\otimes\mathbf{1}_{\mathrm{L}(\mathcal{X})})(uu^*)\big)&=\tr\big((\mathbf{1}_{\mathcal{Y}}\otimes \overline{A}A^\top)\,J(\Phi)\#J(\Psi)\big)\\
    &=\tr\big((AA^*)^\top\,\Tr_{\mathcal{Y}}(J(\Phi)\#J(\Psi))\big),
    \end{align*}
    since $(\mathbf{1}_{\mathcal{Y}}\otimes A^\top)^*(\mathbf{1}_{\mathcal{Y}}\otimes A^\top)=\mathbf{1}_{\mathcal{Y}}\otimes(AA^*)^\top$. As $A$ ranges over invertible operators with $\|A\|_2=1$, $(AA^*)^\top$ ranges over the invertible density operators, so the infimum over such $u$ is
    \[\lambda_{\min}\big(\Tr_{\mathcal{Y}}(J(\Phi)\#J(\Psi))\big).\]
    Finally, $F_M$ is upper semicontinuous (it is a decreasing limit of the continuous functions $\tr(\rho_\eta\#\xi_\eta)$), so the infimum over all pure states equals the infimum over the dense set of pure states with invertible $A$.
\end{proof}
 
\subsection{An alarming property}
 
The Matsumoto fidelity is the trace of a geometric mean of positive semidefinite operators. It is well-known that if one takes the geometric mean of operators whose ranges have intersection $\{0\}$, then this mean is zero. This observation is stated in \cite{cree2020fidelity} as the distinct image property. This distinct image property of course extends to Choi matrices of completely positive maps. This leads to the following proposition. 
 
\begin{proposition}\label{impossibility proposition}
   Suppose $\dim\mathcal{Y}\geq 2$. There exist pairs of completely positive maps $(\Phi_1,\Psi_1)$, $(\Phi_2,\Psi_2)$ in $\mathrm{CP}(\mathcal{X},\mathcal{Y})$ such that $\Phi_1+\Psi_1=\Phi_2+\Psi_2$ but $\widehat{F}(\Phi_1,\Psi_1)=0$ and $\widehat{F}(\Phi_2,\Psi_2)=1$.  
\end{proposition}
\begin{proof}
Let $\phi=y_1y_1^*$ and $\psi=y_2y_2^*$ be distinct pure states on $\mathcal{Y}$, which exist since $\dim\mathcal{Y}\geq2$, and let $\Phi_1=\tr(\cdot)\,\phi$ and $\Psi_1=\tr(\cdot)\,\psi$ be the corresponding replacement maps. Their Choi matrices are
\[
J(\Phi_1)=\phi\otimes\mathbf{1}_{\mathcal{X}},\qquad J(\Psi_1)=\psi\otimes\mathbf{1}_{\mathcal{X}},
\]
with ranges $\mathrm{span}\{y_1\}\otimes\mathcal{X}$ and $\mathrm{span}\{y_2\}\otimes\mathcal{X}$. Since $\phi\neq\psi$, the unit vectors $y_1$ and $y_2$ are linearly independent, so these ranges intersect trivially. Hence $J(\Phi_1)\#J(\Psi_1)=0$ by the distinct image property \cite{cree2020fidelity}, and Proposition~\ref{lambda min formula} gives $\widehat{F}(\Phi_1,\Psi_1)=\lambda_{\min}(0)=0$.
 
Let $\Omega=\Phi_1+\Psi_1$ and $\Phi_2=\Psi_2=\tfrac12\Omega$, so that $\Phi_1+\Psi_1=\Phi_2+\Psi_2$ and $J(\Phi_2)=J(\Psi_2)=\tfrac12(\phi+\psi)\otimes\mathbf{1}_{\mathcal{X}}$. Since $A\#A=A$, we have $J(\Phi_2)\#J(\Psi_2)=J(\Phi_2)$, and therefore
\[
\widehat{F}(\Phi_2,\Psi_2)=\lambda_{\min}\Big(\Tr_{\mathcal{Y}}\big(\tfrac12(\phi+\psi)\otimes\mathbf{1}_{\mathcal{X}}\big)\Big)=\lambda_{\min}\big(\tfrac12\tr(\phi+\psi)\,\mathbf{1}_{\mathcal{X}}\big)=1,
\]
since $\phi$ and $\psi$ are states.
\end{proof}
 
Proposition~\ref{impossibility proposition} suggests something alarming. One may have two pairs of completely positive maps which sum to the same completely positive map, with each pair having opposite geometric fidelities. Stronger still, by rotating $\phi$ slightly, the two splittings can be made arbitrarily close in the completely bounded trace norm. 
\begin{corollary}[The completely bounded trace norm is ineffective for geometric fidelity]\label{impossibility corollary}
    Suppose $\dim\mathcal{Y}\geq 2$. For any $\eta>0$, there exist completely positive maps $\Phi_1,\Psi_1,\Phi_2, \Psi_2\in \mathrm{CP}(\mathcal{X},\mathcal{Y})$ such that $\Phi_1+\Psi_1=\Phi_2+\Psi_2$, \[|||\Phi_1-\Phi_2|||_1 =|||\Psi_1-\Psi_2|||_1<\eta\] and 
    \[\widehat{F}(\Phi_1,\Psi_1)=0, \,\, \widehat{F}(\Phi_2,\Psi_2)=1. \]
\end{corollary}
\begin{proof}
Choose $\theta\in(0,\pi/2)$ with $\sin\theta<\eta$, and let $\phi=y_1y_1^*$ and $\psi=y_2y_2^*$ be pure states on $\mathcal{Y}$ with $|\langle y_1,y_2\rangle|=\cos\theta$; they are distinct since $\theta>0$. Define $\Phi_1,\Psi_1,\Phi_2,\Psi_2$ as in the proof of Proposition~\ref{impossibility proposition}. That proof gives $\Phi_1+\Psi_1=\Phi_2+\Psi_2$, $\widehat{F}(\Phi_1,\Psi_1)=0$ and $\widehat{F}(\Phi_2,\Psi_2)=1$.
 
It remains to compute the distances. We have $\Phi_1-\Phi_2=\tfrac12(\Phi_1-\Psi_1)=-(\Psi_1-\Psi_2)$. For $X_1,X_2\in\mathrm{L}(\mathcal{X})$,
\[
\big((\Phi_1-\Psi_1)\otimes\mathbf{1}_{\mathrm{L}(\mathcal{X})}\big)(X_1\otimes X_2)=\tr(X_1)\,(\phi-\psi)\otimes X_2,
\]
so by linearity, for every $\rho\in\mathrm{D}(\mathcal{X}\otimes\mathcal{X})$,
\[
\big((\Phi_1-\Psi_1)\otimes\mathbf{1}_{\mathrm{L}(\mathcal{X})}\big)(\rho)=(\phi-\psi)\otimes\rho_R,\qquad \rho_R=(\tr\otimes\mathbf{1}_{\mathrm{L}(\mathcal{X})})(\rho)\in\mathrm{D}(\mathcal{X}).
\]
Since $\|A\otimes B\|_1=\|A\|_1\|B\|_1$ and $\|\rho_R\|_1=1$, this has trace norm $\|\phi-\psi\|_1=2\sqrt{1-|\langle y_1,y_2\rangle|^2}=2\sin\theta$ for every $\rho$. As $\Phi_1-\Psi_1$ is Hermitian-preserving, its completely bounded trace norm is the supremum of this quantity over $\rho\in\mathrm{D}(\mathcal{X}\otimes\mathcal{X})$ \cite[Section~3.3]{watrous2018theory}, so $|||\Phi_1-\Psi_1|||_1=2\sin\theta$ and
\[
|||\Phi_1-\Phi_2|||_1=|||\Psi_1-\Psi_2|||_1=\sin\theta<\eta. \qedhere
\]
\end{proof}
 
\begin{remark}
The mechanism behind Proposition~\ref{impossibility proposition} and Corollary~\ref{impossibility corollary} is known at the level of states: the Matsumoto fidelity of two distinct pure states vanishes \cite{cree2020fidelity}, and maximal $f$-divergences are lower semicontinuous \cite[Section~5]{Hiai2019maximal}, \cite[Theorem~A.4]{MosonyiHiai2024} (see also \cite{Matsumoto2018newf}) but in general not continuous. Section~\ref{sec:dichotomy} determines exactly where the quantities considered here are continuous.
\end{remark}
 
We seek to address this problem through the ARN derivative, in particular to explain why it is allowed to happen. First we must take a detour and develop some norm inequalities. 
 
\section{Norm and continuity estimates}\label{sec:norms}
 
Throughout this section $|||\cdot|||_1$ denotes the completely bounded trace norm \cite{Kitaev1997}, \cite[Section~3.3.2]{watrous2018theory}. For completely positive maps it coincides with the induced trace norm and equals $\|\Phi^*(\mathbf{1}_{\mathcal{Y}})\|_\infty$; this follows from the Russo--Dye theorem \cite{RussoDye1966}, in the form \cite[Proposition~3.6]{Paulsen2003}, together with the duality between the completely bounded trace norm and the completely bounded operator norm \cite[Theorem~3]{johnston2009computing}.
 
\begin{lemma}[Sharp Lipschitz inequality in ARN derivatives]\label{lipschitz inequality}
Let $0\leq \Psi_1,\Psi_2\leq \Phi$ with ARN derivatives $T_1,T_2$ with respect to a fixed Kraus decomposition of $\Phi$. Then 
\[
    |||\Psi_{T_1}-\Psi_{T_2}|||_1\leq |||\Psi_{|T_1-T_2|}|||_1\leq |||\Phi|||_1\|T_1-T_2\|_\infty.
\]
\end{lemma}
\begin{remark}
    The outer inequality was shown by Raginsky in \cite{Raginsky2003}.
\end{remark}
\begin{proof}
    Let $\Theta=\Psi_{T_1}-\Psi_{T_2}$ and write $T_1-T_2=A_+-A_-$ with $A_\pm\geq 0$ and $A_+A_-=0$ (the Jordan decomposition), so that $|T_1-T_2|=A_++A_-$. By linearity of $T\mapsto\Psi_T$ we have 
    \[
        \Psi_{T_1}-\Psi_{T_2}=\Psi_{A_+}-\Psi_{A_-}.
    \]
    The map $\Theta$ is Hermitian-preserving, so by \cite[Section~3.3]{watrous2018theory}, the completely bounded trace norm is attained on a pure state $uu^*$ with $u\in \mathcal{X}\otimes \mathcal{X}$ a unit vector. Since $uu^*$ is positive semidefinite, and each $\Psi_{A_{\pm}}$ is completely positive, we have 
    \[
        (\Theta\otimes \mathbf{1}_{\mathrm{L}(\mathcal{X})})(uu^*)=X_+-X_- 
    \]
    where $X_+=(\Psi_{A_+}\otimes \mathbf{1}_{\mathrm{L}(\mathcal{X})})(uu^*)\geq 0$ and $X_-=(\Psi_{A_-}\otimes \mathbf{1}_{\mathrm{L}(\mathcal{X})})(uu^*)\geq0$. Let $\rho$ be the reduced state of $uu^*$ on the first tensor factor, so that $\tr((\Xi\otimes\mathbf{1}_{\mathrm{L}(\mathcal{X})})(uu^*))=\tr(\Xi^*(\mathbf{1}_{\mathcal{Y}})\rho)$ for every map $\Xi$. By the triangle inequality and the fact that the trace norm of a positive operator is its trace, we obtain
    \[
        \begin{aligned}
        |||\Theta|||_1=\|X_+-X_-\|_1&\leq \tr(X_+)+\tr(X_-)=\tr\big((\Psi_{A_+}+\Psi_{A_-})^*(\mathbf{1}_{\mathcal{Y}})\rho\big)\\
        &=\tr\big(\Psi_{|T_1-T_2|}^*(\mathbf{1}_{\mathcal{Y}})\rho\big)\leq \|\Psi_{|T_1-T_2|}^*(\mathbf{1}_{\mathcal{Y}})\|_\infty,
        \end{aligned}
    \]
    the last step because $\Psi_{|T_1-T_2|}^*(\mathbf{1}_{\mathcal{Y}})\geq0$ and $\rho$ is a state. Since $\Psi_{|T_1-T_2|}$ is completely positive, \cite[Proposition~3.6]{Paulsen2003} and \cite[Theorem~3]{johnston2009computing} give
    \[
        \|\Psi_{|T_1-T_2|}^*(\mathbf{1}_{\mathcal{Y}})\|_\infty =|||\Psi_{|T_1-T_2|}|||_1,
    \]
    which proves the first inequality. For the second, Theorem~\ref{effect theorem} gives \[\Psi_{|T_1-T_2|}^*(\mathbf{1}_{\mathcal{Y}})=S^*(\mathbf{1}_{\mathcal{Y}}\otimes |T_1-T_2|^\top)S.\] Since $0\leq |T_1-T_2|^\top\leq \|T_1-T_2\|_\infty\mathbf{1}_{\mathcal{Z}}$ and conjugation preserves the Loewner order,
    \[0\leq S^*(\mathbf{1}_{\mathcal{Y}}\otimes|T_1-T_2|^\top)S\leq \|T_1-T_2\|_\infty S^*S=\|T_1-T_2\|_\infty\Phi^*(\mathbf{1}_{\mathcal{Y}}).\] By the Weyl monotonicity principle \cite[Section~III.2]{bhatia1997matrix}, this inequality between positive operators carries over to the operator norm, so 
    \[
        \|S^*(\mathbf{1}_{\mathcal{Y}}\otimes |T_1-T_2|^\top) S\|_\infty\leq \|T_1-T_2\|_\infty\|\Phi^*(\mathbf{1}_{\mathcal{Y}})\|_\infty=|||\Phi|||_1\|T_1-T_2\|_\infty,
    \]
    where we have applied \cite[Proposition~3.6]{Paulsen2003} again in the last equality. 
\end{proof}
 
\begin{proposition}\label{K6}
\begin{enumerate}
    \item For a complex Euclidean space $\mathcal{Z}$ and $0\leq E,E'\leq \mathbf{1}_{\mathcal{Z}}$, 
    \[\|(E-E^2)-(E'-E'^2)\|_\infty\leq \|E-E'\|_\infty.\]
    \item For $E,E'$ as above, with $g_\#(t)=\sqrt{t(1-t)}$, we have 
    \[
        \|g_\#(E)-g_\#(E')\|_\infty\leq \|E-E'\|_\infty^{1/2}.
    \]
\end{enumerate}
Hence, for completely positive pairs $(\Phi_1,\Psi_1), (\Phi_2,\Psi_2)$ with the same sum $\Omega$ and derivatives $T_1$ and $T_2$ of $\Phi_1$ and $\Phi_2$ with respect to $\Omega$, we have that 
\[
    |||\Phi_1\#\Psi_1-\Phi_2\#\Psi_2|||_1\leq |||\Omega|||_1\|T_1-T_2\|_\infty^{1/2}.
\]
\end{proposition}
\begin{proof}
    \begin{enumerate}
        \item Let $A=E-E'$ and $B=\mathbf{1}_{\mathcal{Z}}-E-E'$. Since 
        \[
            E^2-E'^2=\frac{1}{2}(A(E+E')+(E+E')A),
        \]
        we have 
        \[
            (E-E^2)-(E'-E'^2)=\frac{1}{2}(AB+BA).
        \]
        Since $0\leq E+E'\leq 2\cdot\mathbf{1}_{\mathcal{Z}}$, we have $-\mathbf{1}_{\mathcal{Z}}\leq B\leq\mathbf{1}_{\mathcal{Z}}$, so $\|B\|_\infty\leq 1$. By the triangle inequality and submultiplicativity of the operator norm, 
        \[\|(E-E^2)-(E'-E'^2)\|_\infty\leq \|A\|_\infty\|B\|_\infty\leq \|A\|_\infty=\|E-E'\|_\infty.\]
        \item Since $E$ and $\mathbf{1}_{\mathcal{Z}}-E$ commute, we have $g_{\#}(E)=(E-E^2)^{1/2}\geq 0$. Since the square root is operator monotone on $[0,\infty)$ and $\sqrt{0}=0$, we have by \cite{Ando1988} (see also \cite[Theorem~X.1.1]{bhatia1997matrix}) that $\|\sqrt{X}-\sqrt{Y}\|_\infty\leq\|X-Y\|_\infty^{1/2}$ for $X,Y\geq 0$, and therefore, using part (1),
        \[
            \|g_\#(E)-g_\#(E')\|_\infty\leq \|(E-E^2)-(E'-E'^2)\|_\infty^{1/2}\leq \|E-E'\|_\infty^{1/2}.
        \] 
    \end{enumerate}
    Suppose now that $\Phi_1+\Psi_1=\Phi_2+\Psi_2=\Omega$ and let $T_1$ and $T_2$ be the respective ARN derivatives of $\Phi_1$ and $\Phi_2$ with respect to $\Omega$. By Theorem~\ref{fundamental principle}, $\Phi_i\#\Psi_i=\Psi_{g_\#(T_i)}$, so Lemma~\ref{lipschitz inequality} (applied with $\Phi=\Omega$) gives
    \[
        |||\Phi_1\#\Psi_1-\Phi_2\#\Psi_2|||_1\leq |||\Omega|||_1\|g_{\#}(T_1)-g_{\#}(T_2)\|_\infty.
    \]
    Using part (2), we obtain 
    \[
        |||\Phi_1\#\Psi_1-\Phi_2\#\Psi_2|||_1\leq |||\Omega|||_1\|T_1-T_2\|_\infty^{1/2}.
    \]
\end{proof}
 
\begin{remark}\label{rem:sharp}
Both the constant and the exponent in part (2) are sharp: for $1\times1$ matrices $E=\delta$ and $E'=0$ one has $\|g_\#(E)-g_\#(E')\|_\infty=\sqrt{\delta(1-\delta)}$, so the ratio to $\|E-E'\|_\infty^{1/2}=\sqrt{\delta}$ tends to $1$ as $\delta\to 0$.
\end{remark}
 
\begin{theorem}\label{Modulus of continuity results}
    Let $\Omega=\Phi_1+\Psi_1=\Phi_2+\Psi_2$ be completely positive maps and let $T_i$ be the ARN derivative of $\Phi_i$ with respect to $\Omega$. For a continuous function $g:[0,1]\to\mathbb{R}$ let \[\omega_g(\delta)=\sup\big\{\|g(E)-g(E')\|_\infty:\ 0\leq E,E'\leq\mathbf{1}_{\mathcal{Z}},\ \|E-E'\|_\infty\leq\delta\big\},\qquad\delta\geq0,\] the supremum being over all finite-dimensional $\mathcal{Z}$. Then, for every Kubo--Ando mean $\sigma$,
    \[
        |||\Phi_1\,\sigma\,\Psi_1-\Phi_2\,\sigma\,\Psi_2|||_1\leq |||\Omega|||_1\,\omega_{g_\sigma}(\|T_1-T_2\|_\infty). 
    \] 
\end{theorem}
    \begin{proof}
        By Theorem~\ref{fundamental principle},
        \[
            \Phi_1\,\sigma\,\Psi_1-\Phi_2\,\sigma\,\Psi_2= \Psi_K
        \]
        where $K=g_\sigma(T_1)-g_\sigma(T_2)$. By definition of $\omega_{g_\sigma}$, 
    \[\|K\|_\infty\leq \omega_{g_\sigma}(\|T_1-T_2\|_\infty).\] Since $g_\sigma(T_1)$ and $g_\sigma(T_2)$ are effects, Lemma~\ref{lipschitz inequality} gives 
   $|||\Psi_K|||_1\leq |||\Omega|||_1\|K\|_\infty$,
     which implies the claim.
     \end{proof}

We now come to the main result, which suggests that in order to resolve this problematic behaviour of the geometric channel fidelity, one has to look at the ARN derivatives of the completely positive maps. 
\begin{theorem}\label{main theorem}
 Suppose we are given completely positive maps $\Phi_1, \Psi_1 ,\Phi_2, \Psi_2$ such that 
    $\Omega=\Phi_1+\Psi_1=\Phi_2+\Psi_2$. Let $\widehat{F}_1=\widehat{F}(\Phi_1,\Psi_1)$ and $\widehat{F}_2=\widehat{F}(\Phi_2,\Psi_2)$. Let $T_1$ and $T_2$ be the respective ARN derivatives of $\Phi_1$ and $\Phi_2$ with respect to $\Omega$. Then 
    \[
        |\widehat{F}_1-\widehat{F}_2|\leq \|\Omega^*(\mathbf{1}_{\mathcal{Y}})\|_\infty \|T_1-T_2\|_\infty^{1/2}.
    \]
\end{theorem}
\begin{proof}
    By Weyl's inequality \cite[Section~III.2]{bhatia1997matrix}, $\lambda_{\min}$ is $1$-Lipschitz with respect to the operator norm on Hermitian operators. Let $K=g_\#(T_1)-g_\#(T_2)$, which is Hermitian. By Proposition~\ref{lambda min formula} and Theorem~\ref{fundamental principle},
    \begin{align*}
           |\widehat{F}_1-\widehat{F}_2|&\leq \|\Tr_{\mathcal{Y}}(J(\Phi_1)\#J(\Psi_1)-J(\Phi_2)\#J(\Psi_2))\|_\infty\\
           &=\|\Tr_{\mathcal{Y}}(WKW^*)\|_\infty.
    \end{align*}
    Since $-\|K\|_\infty WW^*\leq WKW^*\leq\|K\|_\infty WW^*$ and the partial trace is a positive map, 
    \[-\|K\|_\infty\Tr_{\mathcal{Y}}(WW^*)\leq\Tr_{\mathcal{Y}}(WKW^*)\leq\|K\|_\infty\Tr_{\mathcal{Y}}(WW^*),\]
    and for a Hermitian $H$ and positive $M$ with $-sM\leq H\leq sM$ one has $\|H\|_\infty\leq s\|M\|_\infty$. Moreover $\Tr_{\mathcal{Y}}(WW^*)=\Tr_{\mathcal{Y}}(J(\Omega))=\Omega^*(\mathbf{1}_{\mathcal{Y}})^\top$, whose operator norm is $\|\Omega^*(\mathbf{1}_{\mathcal{Y}})\|_\infty$. Hence, by Proposition~\ref{K6}(2),
    \[
        |\widehat{F}_1-\widehat{F}_2|\leq\|K\|_\infty\|\Omega^*(\mathbf{1}_{\mathcal{Y}})\|_\infty\leq\|\Omega^*(\mathbf{1}_{\mathcal{Y}})\|_\infty\|T_1-T_2\|_\infty^{1/2}.\qedhere
    \]
\end{proof}
 
What Theorem~\ref{main theorem} says is that if two pairs of completely positive maps have the same sum, then their respective geometric channel fidelities are controlled by the operator norm of the difference of their ARN derivatives. Coming back to Corollary~\ref{impossibility corollary}, in the Kraus decomposition $\{V_i=y_1e_i^*,\ V'_i=y_2e_i^*\}_{i=1}^{n}$ of $\Omega$, where $n=\dim\mathcal{X}$ and $e_1,\dots,e_n$ is an orthonormal basis of $\mathcal{X}$ (these $2n$ operators are linearly independent because $y_1$ and $y_2$ are), the derivative of $\Phi_1$ is the projection $T_1=\mathbf{1}_{\mathbb{C}^n}\oplus 0_{\mathbb{C}^n}$, while the splitting $\left(\frac{\Omega}{2}, \frac{\Omega}{2}\right)$ has ARN derivative $T_2=\frac{1}{2}\mathbf{1}_{\mathbb{C}^{2n}}$. Hence $\|T_1-T_2\|_\infty=\frac12$ for every $\theta>0$, however small. The adjoint of $\Omega$ evaluated at $\mathbf{1}_{\mathcal{Y}}$ is 
\[\Omega^*(\mathbf{1}_{\mathcal{Y}})=(\tr(\phi)+\tr(\psi))\mathbf{1}_{\mathcal{X}}=2\cdot\mathbf{1}_{\mathcal{X}},\] so the right-hand side of Theorem~\ref{main theorem} equals
\[
    \|\Omega^*(\mathbf{1}_{\mathcal{Y}})\|_\infty\|T_1-T_2\|_\infty^{1/2}=2\cdot(\tfrac12)^{1/2}=\sqrt{2}\geq 1=|\widehat{F}_1-\widehat{F}_2|.
\]
The completely bounded trace norm between the splittings tends to zero, but the distance between the derivatives does not, which is exactly what allows the fidelity to jump.
\begin{remark}[Physical interpretation]\label{rem:phys-readout}
By Remark~\ref{rem:phys-effect}, pairs with the same sum are one device $\Omega$ with different which-branch effects. A classical readout that reports the outcome of $\Phi$ as that of $\Psi$ with probability $p_{01}$, and conversely with probability $p_{10}$, gives the pair $\big((1-p_{01})\Phi+p_{10}\Psi,\,p_{01}\Phi+(1-p_{10})\Psi\big)$. The sum is unchanged, and $T$ becomes $p_{10}\mathbf{1}_{\mathbb{C}^r}+(1-p_{01}-p_{10})T$, a change of norm at most $\max\{p_{01},p_{10}\}$ because $\mathrm{spec}(T)\subseteq[0,1]$. So, for every pair, Theorem~\ref{main theorem} bounds the change of $\widehat{F}$ by $\|\Omega^*(\mathbf{1}_{\mathcal{Y}})\|_\infty\max\{p_{01},p_{10}\}^{1/2}$. Such classical noise was found to dominate readout errors on superconducting devices \cite{MaciejewskiZimborasOszmaniec2020}. If $p_{01}+p_{10}\leq1$, the new spectrum lies in $[p_{10},1-p_{01}]$: imperfect readout gives the two-sided bounds of Sections~\ref{sec:all f} and~\ref{sec:loewner}.
\end{remark}
We now verify the modulus of continuity does indeed behave properly under the right circumstances.
\begin{lemma}\label{modulus to zero}
For a compact interval $I\subseteq\mathbb{R}$ and a continuous function $g:I\to\mathbb{R}$ put
\[
\omega^I_g(\delta)=\sup\big\{\|g(E)-g(E')\|_\infty:\ E,E'\text{ Hermitian with spectra in }I,\ \|E-E'\|_\infty\leq\delta\big\},\qquad\delta\geq0,
\]
the supremum being over all finite dimensions, so that $\omega_g=\omega^{[0,1]}_g$ is the modulus of Theorem~\ref{Modulus of continuity results}; for $\epsilon\in(0,1)$ we also write $\omega^{(\epsilon)}_g=\omega^{[0,1-\epsilon]}_g$. Then $\omega^I_g(\delta)\to0$ as $\delta\to0$.
In particular this holds for $g=g_\sigma$, where $\sigma$ is any Kubo--Ando mean, and for $g=h_f$, where $f$ is as in
Lemma~\ref{lem:maxf decomposition}.
\end{lemma}
\begin{proof}
Let $\eta>0$ and $R=\max\{|x|:x\in I\}$. By the Weierstrass approximation theorem there is a polynomial $\tilde g(t)=\sum_{k=0}^nc_kt^k$ with
$\sup_{t\in I}|g(t)-\tilde g(t)|\leq\eta$. If $E,E'$ are Hermitian with spectra in $I$, the spectral theorem gives $\|g(E)-\tilde g(E)\|_\infty\leq\eta$ and
$\|g(E')-\tilde g(E')\|_\infty\leq\eta$. Since $\|E\|_\infty,\|E'\|_\infty\leq R$ and
$E^k-E'^k=\sum_{j=0}^{k-1}E^j(E-E')E'^{\,k-1-j}$, we have $\|E^k-E'^k\|_\infty\leq kR^{k-1}\|E-E'\|_\infty$, and hence
$\|\tilde g(E)-\tilde g(E')\|_\infty\leq L_{\tilde g}\|E-E'\|_\infty$ with $L_{\tilde g}=\sum_{k}k|c_k|R^{k-1}$, a constant independent of the dimension.
Therefore
\[
\|g(E)-g(E')\|_\infty\leq2\eta+L_{\tilde g}\|E-E'\|_\infty,
\]
so $\omega^I_g(\delta)\leq2\eta+L_{\tilde g}\delta$ and $\limsup_{\delta\to0}\omega^I_g(\delta)\leq2\eta$. Since $\eta>0$
was arbitrary, $\omega^I_g(\delta)\to0$.
 
For the two special cases it remains to check continuity on $[0,1]$. The function $g_\sigma$ is continuous on $(0,1)$,
and at the endpoints its values were defined as the limits $g_\sigma(0)=f_\sigma'(\infty)$ and
$g_\sigma(1)=f_\sigma(0^+)$. The function $h_f$ is continuous on $[0,1]$ by the first part of the proof of Theorem~\ref{thm:maxf}.
\end{proof}
 
Continuity results regarding maximal divergences have also been studied by \cite{MosonyiHiai2024}, and more recently in the von Neumann algebra context via lower semicontinuity by Okayasu in \cite{Okayasu2026maxf}. 

 We conclude this section with a lower bound for the maps themselves.
 \begin{proposition}\label{bilipschitz}
Let $\Omega(X)=\sum_{k=1}^rV_kXV_k^*$ with linearly independent $V_1,\dots,V_r$, let
$W=[\vectorise(V_1)\ \cdots\ \vectorise(V_r)]$ and $G=W^*W$. For Hermitian $T_1,T_2\in\mathrm{L}(\mathbb{C}^r)$,
\begin{align*}
\frac{\lambda_{\min}(G)}{\dim\mathcal{X}}\,\|T_1-T_2\|_1&\leq|||\Psi_{T_1}-\Psi_{T_2}|||_1\\
&\leq\|\Omega^*(\mathbf{1}_{\mathcal{Y}})\|_\infty\,\|T_1-T_2\|_\infty .
\end{align*}
\end{proposition}
\begin{proof}
The upper bound follows from the proof of Lemma~\ref{lipschitz inequality} with $\Phi=\Omega$, which uses only that $T_1-T_2$ is Hermitian, together with $|||\Omega|||_1=\|\Omega^*(\mathbf{1}_{\mathcal{Y}})\|_\infty$. For the lower bound put $K=T_1-T_2$, so that
$\Psi_{T_1}-\Psi_{T_2}=\Psi_K$ and $J(\Psi_K)=WKW^*$. Since
\[
J(\Psi_K)=(\Psi_K\otimes\mathbf{1}_{\mathrm{L}(\mathcal{X})})\big(\vectorise(\mathbf{1}_{\mathcal{X}})\vectorise(\mathbf{1}_{\mathcal{X}})^*\big)
\]
and $\vectorise(\mathbf{1}_{\mathcal{X}})\vectorise(\mathbf{1}_{\mathcal{X}})^*/\dim\mathcal{X}$ is a density operator,
$\|WKW^*\|_1\leq\dim\mathcal{X}\,|||\Psi_K|||_1$. Write $W=UG^{1/2}$ with $U$ an isometry (polar decomposition); then
$\|WKW^*\|_1=\|G^{1/2}KG^{1/2}\|_1$, and since $G$ is invertible,
by H\"older's inequality \cite[Chapter~IV]{bhatia1997matrix},
\begin{align*}
\|K\|_1&=\|G^{-1/2}(G^{1/2}KG^{1/2})G^{-1/2}\|_1\\
&\leq\|G^{-1/2}\|_\infty^2\,\|G^{1/2}KG^{1/2}\|_1=\lambda_{\min}(G)^{-1}\|G^{1/2}KG^{1/2}\|_1.
\end{align*} Combining these estimates gives the lower bound.
\end{proof}

We record the estimate used in the proof of Theorem~\ref{main theorem} in the form in which it is used later.

\begin{lemma}\label{lem:partial trace}
Let $\Omega$, $V_1,\dots,V_r$ and $W$ be as in Proposition~\ref{bilipschitz}. For Hermitian $N_1,N_2\in\mathrm{L}(\mathbb{C}^r)$,
\[
\big\|\Tr_{\mathcal{Y}}(WN_1W^*)-\Tr_{\mathcal{Y}}(WN_2W^*)\big\|_\infty\leq\|\Omega^*(\mathbf{1}_{\mathcal{Y}})\|_\infty\,\|N_1-N_2\|_\infty ,
\]
and the same bound holds for $|\lambda_{\min}(\Tr_{\mathcal{Y}}(WN_1W^*))-\lambda_{\min}(\Tr_{\mathcal{Y}}(WN_2W^*))|$ and for the corresponding difference of largest eigenvalues.
\end{lemma}

\begin{proof}
Put $K=N_1-N_2$. Since $-\|K\|_\infty\mathbf{1}_{\mathbb{C}^r}\leq K\leq\|K\|_\infty\mathbf{1}_{\mathbb{C}^r}$ and the partial trace is positive, $-\|K\|_\infty\Tr_{\mathcal{Y}}(WW^*)\leq\Tr_{\mathcal{Y}}(WKW^*)\leq\|K\|_\infty\Tr_{\mathcal{Y}}(WW^*)$, and $\Tr_{\mathcal{Y}}(WW^*)=\Tr_{\mathcal{Y}}J(\Omega)=\Omega^*(\mathbf{1}_{\mathcal{Y}})^\top\geq0$. For a Hermitian $H$ and positive $M$ with $-sM\leq H\leq sM$ one has $\|H\|_\infty\leq s\|M\|_\infty$, which gives the first claim; the second follows from Weyl's inequality \cite[Section~III.2]{bhatia1997matrix}.
\end{proof}
\section{A general approach}\label{sec:general}
\subsection{Geometric R\'enyi relative entropy}
The problem presented in Corollary~\ref{impossibility corollary} is not unique to the geometric channel fidelity. In general, any channel divergence which can be expressed through a Choi-level Kubo--Ando mean that vanishes on pairs with trivially intersecting ranges will have this problem as well; Theorem~\ref{range theorem} below characterises these means.
 
In \cite{KatariyaWilde2021}, the following definitions are made which we recall (see also \cite{FangFawzi2021}; up to sign, the underlying state quantity $\widehat{Q}_\alpha$ is the maximal $f$-divergence of \cite{Matsumoto2018newf} for the power function $f(t)=t^\alpha$). 
\begin{definition}[Geometric R\'enyi relative entropy]
Let $\rho\in \mathrm{D}(\mathcal{V})$ be a quantum state and let $\xi\in\mathrm{Pos}(\mathcal{V})$ be a positive semidefinite operator. Let $\alpha\in (0,1)\cup(1,\infty)$. The geometric R\'enyi relative quasi-entropy is defined as 
\[
    \widehat{Q}_\alpha(\rho\|\xi)=\lim_{\eta\downarrow0}\tr\big(\xi_\eta^{1/2}(\xi_\eta^{-1/2}\rho\xi_\eta^{-1/2})^\alpha\xi_\eta^{1/2}\big)
\]
where $\xi_\eta=\xi+\eta \mathbf{1}_{\mathcal{V}}$. The geometric R\'enyi relative entropy is then defined as
\[
    \widehat{D}_\alpha(\rho\|\xi)=\frac{1}{\alpha-1}\log\widehat{Q}_\alpha(\rho\|\xi).
\]
\end{definition}
We note that in the above definition, we do not even require $\rho$ to be a state for the definition to make sense. We consequently do not require that $\rho$ is a state in the rest of this text.
 
If $\alpha\in (0,1)$, then $t\mapsto t^\alpha$ is operator monotone on $(0,\infty)$ and takes the value $1$ at $t=1$; it is the representing function of the weighted geometric mean $\#_\alpha$. In this notation $\widehat{Q}_\alpha(\rho\|\xi)=\tr(\xi\#_\alpha\rho)=\tr(\rho\#_{1-\alpha}\xi)$, with the regularisation from above built into the definition of the mean. This will prove important shortly. 
The next definition we recall from \cite{KatariyaWilde2021} is the following; again we drop the requirement that the maps be channels and replace it with complete positivity.
\begin{definition}[Geometric R\'enyi relative entropy of completely positive maps]
Let $\Phi,\Psi\in \mathrm{CP}(\mathcal{X},\mathcal{Y})$ be completely positive. The geometric R\'enyi relative entropy is defined for $\alpha\in (0,1)\cup (1,\infty)$ as 
\[
    \widehat{D}_\alpha(\Phi\|\Psi)=\sup_{\rho\in \mathrm{D}(\mathcal{X}\otimes \mathcal{X})}\widehat{D}_\alpha\big((\Phi\otimes\mathbf{1}_{\mathrm{L}(\mathcal{X})})(\rho)\,\|\,(\Psi\otimes\mathbf{1}_{\mathrm{L}(\mathcal{X})})(\rho)\big).
\]
\end{definition}
For $\alpha\in(0,1)\cup(1,2]$ this quantity admits a closed form in terms of the Choi matrices \cite{KatariyaWilde2021}. We state it for $\alpha\in (0,1)$, where, since $\frac{1}{\alpha-1}<0$, the supremum of $\widehat{D}_\alpha$ corresponds to an infimum of $\widehat{Q}_\alpha$. 
\begin{lemma}[{\cite[Proposition~44]{KatariyaWilde2021}}]\label{renyi closed form}
Let $\Phi,\Psi\in \mathrm{CP}(\mathcal{X},\mathcal{Y})$ and $\alpha\in (0,1)$, and define 
\[\widehat{Q}_\alpha(\Phi\|\Psi)=\inf_{\rho\in \mathrm{D}(\mathcal{X}\otimes \mathcal{X})}\widehat{Q}_\alpha\big((\Phi\otimes\mathbf{1}_{\mathrm{L}(\mathcal{X})})(\rho)\,\|\,(\Psi\otimes\mathbf{1}_{\mathrm{L}(\mathcal{X})})(\rho)\big).\] Then
\[
\widehat{D}_\alpha(\Phi\|\Psi)=\frac{1}{\alpha-1}\log\widehat{Q}_\alpha(\Phi\|\Psi),\qquad \widehat{Q}_\alpha(\Phi\|\Psi)=\lambda_{\min}\left(\Tr_{\mathcal{Y}}\left(J(\Psi)\#_\alpha J(\Phi)\right)\right),
\]
where the weighted geometric mean is understood with the regularisation from above.
\end{lemma}
\begin{proof}
The first identity follows from the definitions, since $t\mapsto\frac{1}{\alpha-1}\log t$ is decreasing for $\alpha\in(0,1)$. For the second, the proof of Proposition~\ref{lambda min formula} applies verbatim with $\#$ replaced by $\#_\alpha$: joint concavity of $(\rho,\xi)\mapsto\tr(\xi\#_\alpha\rho)$ \cite{KuboAndo1980} reduces the infimum to pure states, the transformer equality \cite{KuboAndo1980} handles invertible $A$, and upper semicontinuity handles the closure; compare \cite[Lemma~5]{FangFawzi2021} for the case $\alpha\in(1,2]$.
\end{proof}
The careful reader will notice that $\widehat{D}_\alpha(\Phi\|\Psi)$ may take the value $+\infty$ for $\alpha\in(0,1)$. This is not a problem; it corresponds to $\widehat{Q}_\alpha(\Phi\|\Psi)=0$, which occurs for instance when the Choi matrices have ranges with trivial intersection.
 
It is not much of a stretch to see that the problem of Proposition~\ref{impossibility proposition} extends to $\widehat{Q}_\alpha$, because this quantity is induced by a Kubo--Ando mean of Choi matrices; indeed the same replacement maps work, since $A\#_\alpha B=0$ whenever the ranges of $A$ and $B$ intersect trivially (Theorem~\ref{range theorem} below). We will again demonstrate that we can control such problems through the ARN derivative. 
Since $J(\Psi)\#_\alpha J(\Phi)=J(\Phi)\#_{1-\alpha}J(\Psi)$, Theorem~\ref{fundamental principle} applied to the mean $\#_{1-\alpha}$ shows that $J(\Psi)\#_\alpha J(\Phi)=Wh_\alpha(T)W^*$, where $T$ is the ARN derivative of $\Phi$ with respect to $\Omega=\Phi+\Psi$ and 
\begin{equation}\label{eq:h alpha}
    h_\alpha(t)=g_{\#_{1-\alpha}}(t)=t\#_{1-\alpha}(1-t)=t^{\alpha}(1-t)^{1-\alpha},\qquad t\in[0,1].
\end{equation}
\begin{lemma}\label{holder lemma}
    Let $\mathcal{Z}$ be a complex Euclidean space, let $0\leq E,E'\leq \mathbf{1}_{\mathcal{Z}}$ and let $\alpha\in (0,1)$. Then
    \[
        \|h_\alpha(E)-h_\alpha(E')\|_\infty\leq \|E-E'\|_\infty^\alpha+\|E-E'\|_\infty^{1-\alpha}.
    \]
\end{lemma}
\begin{proof}
    Since $E$ and $\mathbf{1}_{\mathcal{Z}}-E$ commute, we have
    \[
        h_\alpha(E) = E^{\alpha}(\mathbf{1}_{\mathcal{Z}}-E)^{1-\alpha} \geq 0,
    \]
    and therefore
    \[
        h_\alpha(E)-h_\alpha(E') = \big(E^{\alpha}-E'^{\alpha}\big)(\mathbf{1}_{\mathcal{Z}}-E)^{1-\alpha} + E'^{\alpha}\Big((\mathbf{1}_{\mathcal{Z}}-E)^{1-\alpha}-(\mathbf{1}_{\mathcal{Z}}-E')^{1-\alpha}\Big).
    \]
    By the triangle inequality and submultiplicativity of $\|\cdot\|_\infty$,
    \begin{align*}
        \|h_\alpha(E)-h_\alpha(E')\|_\infty \leq\;& \|E^{\alpha}-E'^{\alpha}\|_\infty\,\|(\mathbf{1}_{\mathcal{Z}}-E)^{1-\alpha}\|_\infty \\
        &+ \|E'^{\alpha}\|_\infty\,\big\|(\mathbf{1}_{\mathcal{Z}}-E)^{1-\alpha}-(\mathbf{1}_{\mathcal{Z}}-E')^{1-\alpha}\big\|_\infty.
    \end{align*}
    Since $\|(\mathbf{1}_{\mathcal{Z}}-E)^{1-\alpha}\|_\infty\leq 1$ and $\|E'^{\alpha}\|_\infty\leq 1$,
    \[
        \|h_\alpha(E)-h_\alpha(E')\|_\infty \leq \|E^{\alpha}-E'^{\alpha}\|_\infty + \big\|(\mathbf{1}_{\mathcal{Z}}-E)^{1-\alpha}-(\mathbf{1}_{\mathcal{Z}}-E')^{1-\alpha}\big\|_\infty.
    \]
    Now apply the operator-norm H\"older continuity of positive power functions \cite{Ando1988}, \cite[Theorem~X.1.1]{bhatia1997matrix}: for $A,B\geq 0$ and $\beta\in(0,1]$,
    \[
        \|A^\beta - B^\beta\|_\infty \leq \|A-B\|_\infty^\beta.
    \]
    Taking $\beta=\alpha$, $A=E$, $B=E'$ gives $\|E^{\alpha}-E'^{\alpha}\|_\infty \leq \|E-E'\|_\infty^{\alpha}$, and taking $\beta=1-\alpha$, $A=\mathbf{1}_{\mathcal{Z}}-E$, $B=\mathbf{1}_{\mathcal{Z}}-E'$ gives
    \[
        \big\|(\mathbf{1}_{\mathcal{Z}}-E)^{1-\alpha}-(\mathbf{1}_{\mathcal{Z}}-E')^{1-\alpha}\big\|_\infty \leq \|(\mathbf{1}_{\mathcal{Z}}-E)-(\mathbf{1}_{\mathcal{Z}}-E')\|_\infty^{1-\alpha} = \|E-E'\|_\infty^{1-\alpha}.
    \]
    Combining the two bounds,
    \[
        \|h_\alpha(E)-h_\alpha(E')\|_\infty \leq \|E-E'\|_\infty^{\alpha}+\|E-E'\|_\infty^{1-\alpha},
    \]
    as claimed.
\end{proof}
\begin{theorem}\label{thm:renyi continuity}
    Let $\alpha\in(0,1)$ and let $\Phi_1,\Psi_1,\Phi_2,\Psi_2\in \mathrm{CP}(\mathcal{X},\mathcal{Y})$ satisfy $\Omega=\Phi_1+\Psi_1=\Phi_2+\Psi_2$. Let $T_1$ and $T_2$ denote the ARN derivatives of $\Phi_1$ and $\Phi_2$ with respect to $\Omega$. Then
    \[
        |\widehat{Q}_\alpha(\Phi_1\|\Psi_1)-\widehat{Q}_\alpha(\Phi_2\|\Psi_2)|\leq \|\Omega^*(\mathbf{1}_{\mathcal{Y}})\|_\infty\big(\|T_1-T_2\|_\infty^\alpha+\|T_1-T_2\|_\infty^{1-\alpha}\big).
    \]
\end{theorem}
    \begin{proof}
        By Lemma~\ref{renyi closed form}, $\widehat{Q}_\alpha(\Phi_i\|\Psi_i)=\lambda_{\min}(\Tr_{\mathcal{Y}}(Wh_\alpha(T_i)W^*))$. The proof of Theorem~\ref{main theorem} now applies verbatim with $g_\#$ replaced by $h_\alpha$ and Proposition~\ref{K6}(2) replaced by Lemma~\ref{holder lemma}. 
    \end{proof}
    \subsection{The meta problem}
In this subsection we describe in general terms when the problem of Proposition~\ref{impossibility proposition} occurs. Let $\sigma$ be a Kubo--Ando mean with representing function $f_\sigma$. By the integral representation of operator connections \cite[Theorem~5.11]{hiaipetz2014intro}, for all positive semidefinite $A,B$,
\[
A\,\sigma\,B=a_\sigma\,A+b_\sigma\,B+\int_{(0,\infty)}\frac{1+\lambda}{\lambda}\,\big((\lambda A):B\big)\,dm_\sigma(\lambda),
\]
where $m_\sigma$ is a finite positive measure, $a_\sigma=f_\sigma(0^+)=1\,\sigma\,0$ and $b_\sigma=f_\sigma'(\infty)=0\,\sigma\,1$. Evaluating the representing function at $1$ gives $a_\sigma+b_\sigma+m_\sigma\big((0,\infty)\big)=f_\sigma(1)=1$. We also use that $\range(A:B)=\range(A)\cap\range(B)$ for positive semidefinite $A,B$ \cite{AndersonDuffin1969}.
 
\begin{theorem}\label{range theorem}
Let $\sigma$ be a Kubo--Ando mean with representing function $f_\sigma$.
\begin{enumerate}
\item If $f_\sigma(0^+)=f_\sigma'(\infty)=0$, then $\range(A\,\sigma\,B)=\range(A)\cap\range(B)$ for all positive semidefinite $A,B$. In particular, $A\,\sigma\,B=0$ if and only if $\range(A)\cap\range(B)=\{0\}$.
\item Conversely, if $A\,\sigma\,B=0$ for all positive semidefinite $A,B$ on a space of dimension at least $2$ with $\range(A)\cap\range(B)=\{0\}$, then $f_\sigma(0^+)=f_\sigma'(\infty)=0$.
\end{enumerate}
\end{theorem}
\begin{proof}
(1) Here $a_\sigma=b_\sigma=0$ and $m_\sigma\big((0,\infty)\big)=1$. Since $\range(\lambda A)=\range(A)$ for $\lambda>0$, each of the positive semidefinite operators $(\lambda A):B$ has range $R=\range(A)\cap\range(B)$, and hence kernel $R^\perp$. If $v\in R^\perp$ then $\langle v,(A\,\sigma\,B)v\rangle=0$. If $v\notin R^\perp$ then $\langle v,((\lambda A):B)v\rangle>0$ for every $\lambda>0$, and since $m_\sigma\neq0$ the integral is strictly positive. Thus $\ker(A\,\sigma\,B)=R^\perp$, that is, $\range(A\,\sigma\,B)=R$.
 
(2) Let $e_1,e_2$ be orthonormal and $A=e_1e_1^*$, $B=e_2e_2^*$. Then $\range(A)\cap\range(B)=\{0\}$, so every $(\lambda A):B$ vanishes and $A\,\sigma\,B=a_\sigma\,A+b_\sigma\,B$. If this is zero, then $a_\sigma=b_\sigma=0$.
\end{proof}
 
Consequently, suppose that a real-valued function $\mathsf{D}$ on $\mathrm{CP}(\mathcal{X},\mathcal{Y})\times\mathrm{CP}(\mathcal{X},\mathcal{Y})$ has the form $\mathsf{D}(\Phi,\Psi)=\mathsf{H}\big(J(\Phi)\,\sigma\,J(\Psi)\big)$ for a function $\mathsf{H}$ with $\mathsf{H}(0)=0$, where $\sigma$ satisfies the hypothesis of Theorem~\ref{range theorem}(1). Then $\mathsf{D}(\Phi,\Psi)=0$ whenever $\range(J(\Phi))\cap\range(J(\Psi))=\{0\}$. By Proposition~\ref{lambda min formula} this applies to $\widehat{F}$, with $\mathsf{H}(X)=\lambda_{\min}(\Tr_{\mathcal{Y}}X)$, and by Lemma~\ref{renyi closed form} it applies to $\widehat{Q}_\alpha$ with the roles of $\Phi$ and $\Psi$ exchanged. The converse implication fails in general, because $\lambda_{\min}(\Tr_{\mathcal{Y}}X)$ can vanish for nonzero $X\geq0$. For instance, if $\Phi=\Psi$ and $\Phi^*(\mathbf{1}_{\mathcal{Y}})$ is singular, then $\widehat{F}(\Phi,\Phi)=\lambda_{\min}(\Phi^*(\mathbf{1}_{\mathcal{Y}})^\top)=0$, although the Choi ranges coincide.
 
We can consider this in relation to the ARN derivative. Let $\Omega=\Phi+\Psi$ and let $T$ be the ARN derivative of $\Phi$ with respect to $\Omega$. The condition $f_\sigma(0^+)=f_\sigma'(\infty)=0$ says exactly that $g_\sigma(0)=g_\sigma(1)=0$, while $g_\sigma(t)=tf_\sigma\big(\tfrac{1-t}{t}\big)>0$ for $0<t<1$. Since $W$ is injective, Theorem~\ref{fundamental principle} shows that $J(\Phi)\,\sigma\,J(\Psi)=Wg_\sigma(T)W^*$ vanishes if and only if $\mathrm{spec}(T)\subseteq\{0,1\}$, that is, if and only if $T$ is a projection. For the geometric mean this vanishing criterion was obtained by Kian \cite{kian2026relative}, in the form that $\Phi\#\Psi=0$ exactly when $\Phi$ and $\Psi$ have no nonzero common completely positive submap; submaps whose derivative is a projection are the sharp submaps studied in \cite{Kian2026sharpness}. The replacement maps of Proposition~\ref{impossibility proposition} are of exactly this kind: in the Kraus decomposition described after Theorem~\ref{main theorem}, the ARN derivative of $\Phi_1$ is $\mathbf{1}_{\mathbb{C}^n}\oplus0_{\mathbb{C}^n}$.
 
\subsection{Operator convex functions and divergences}
In this subsection, we show that the Kubo--Ando framework from above extends further to other channel level divergences. Many channel divergences are defined through operator convex functions instead of Kubo--Ando means.
\begin{definition}[Operator convex function {\cite[Chapter~V]{bhatia1997matrix}}]
Let $(t_-,t_+)$ be an open interval. A function $f:(t_-,t_+)\to\mathbb{R}$ is said to be operator convex on $(t_-,t_+)$ if
\[
    f\big(tA+(1-t)B\big)\leq t f(A)+(1-t)f(B)
\]
for all $t\in[0,1]$ and all Hermitian operators $A,B$ of the same (arbitrary) size with spectra contained in $(t_-,t_+)$. It is operator concave if $-f$ is operator convex.
\end{definition}
 
Multiple quantum channel divergences can be framed in a similar way to those defined by Kubo--Ando means. A broad class of these channel divergences are the maximal $f$-divergences, introduced by Matsumoto \cite{Matsumoto2018newf} and studied systematically in \cite{HiaiMosonyi2017,Hiai2019maximal}. 
\begin{definition}[Maximal $f$-divergence \cite{Matsumoto2018newf,HiaiMosonyi2017}]\label{def:maxf}
Let $f:(0,\infty)\to\mathbb{R}$ be operator convex; the limits $f(0^+)$ and $f'(\infty)$ (notation as before Theorem~\ref{fundamental principle}) exist in $(-\infty,+\infty]$ (Lemma~\ref{lem:hf extended}).
For positive definite $A,B\in \mathrm{L}(\mathcal{V})$ the \emph{maximal $f$-divergence} is
\[
\widehat{D}_f(A\|B)=\tr\!\Big(B^{1/2}f\big(B^{-1/2}AB^{-1/2}\big)B^{1/2}\Big),
\]
and for positive semidefinite $A,B$ it is defined by
$\widehat{D}_f(A\|B)=\lim_{\eta\downarrow0}\widehat{D}_f(A+\eta\mathbf{1}_{\mathcal{V}}\|B+\eta\mathbf{1}_{\mathcal{V}})$.
We call $f$ \emph{bounded} if $f(0^+)<\infty$ and $f'(\infty)<\infty$.
\end{definition}
 
\begin{definition}[Maximal $f$-divergence of completely positive maps; cf.\ \cite{KatariyaWilde2021,FangFawzi2021}]\label{def:maxf maps}
For $\Phi,\Psi\in \mathrm{CP}(\mathcal{X},\mathcal{Y})$,
\[
\widehat{D}_f(\Phi\|\Psi)=\sup_{\rho\in \mathrm{D}(\mathcal{X}\otimes\mathcal{X})}
\widehat{D}_f\big(\Phi_\rho\,\|\,\Psi_\rho\big),
\]
where $\Phi_\rho=(\Phi\otimes\mathbf{1}_{\mathrm{L}(\mathcal{X})})(\rho)$ and $\Psi_\rho=(\Psi\otimes\mathbf{1}_{\mathrm{L}(\mathcal{X})})(\rho)$.
\end{definition}
 
\begin{definition}[Examples]
For positive semidefinite $A,B$ and $0<\alpha<1$:
\begin{enumerate}
\item \emph{Hellinger divergence} \cite{BhatiaGaubertJain2019,PitrikVirosztek2020}: $f_{\mathrm{H}}(t)=(1-\sqrt{t})^2$, and
\[
\widehat{D}_{\mathrm{H}}(A\|B)=\tr A+\tr B-2\tr(A\#B).
\]
\item \emph{Hellinger divergence of order $\alpha$} \cite{LieseVajda2006}: $f_{\mathrm{H},\alpha}(t)=\dfrac{1-t^\alpha}{1-\alpha}$, and
\[
\widehat{D}_{\mathrm{H},\alpha}(A\|B)=\frac{\tr B-\tr(B\#_\alpha A)}{1-\alpha}.
\]
\item \emph{Vincze--Le Cam (triangular) divergence} \cite{LieseVajda2006}: $f_{\mathrm{LC}}(t)=\dfrac{(t-1)^2}{t+1}$, and
\[
\widehat{D}_{\mathrm{LC}}(A\|B)=\tr A+\tr B-4\tr(A:B).
\]
\item \emph{Maximal Jensen--Shannon divergence} \cite{Lin1991}: $f_{\mathrm{JS}}(t)=t\log t-(1+t)\log\dfrac{1+t}{2}$, and
\[
\widehat{D}_{\mathrm{JS}}(A\|B)=\tr\!\Big(B^{1/2}f_{\mathrm{JS}}\big(B^{-1/2}AB^{-1/2}\big)B^{1/2}\Big).
\]
\end{enumerate}
Each $f$ above is bounded, with $(f(0^+),f'(\infty))$ equal to $(1,1)$, $(\tfrac{1}{1-\alpha},0)$, $(1,1)$ and $(\log2,\log2)$ respectively.
\end{definition}
 
\begin{remark}
If $\Phi=\Psi_T$ and $\Psi=\Psi_{\mathbf{1}_{\mathbb{C}^r}-T}$, where $T$ is the ARN derivative of $\Phi$ with respect to $\Omega=\Phi+\Psi$, then the maximal $f$-divergence of the pair is governed by $h_f(T)$, where $h_f(t)=(1-t)f\big(\tfrac{t}{1-t}\big)$ on $(0,1)$, extended by $h_f(0)=f(0^+)$ and $h_f(1)=f'(\infty)$. For the examples above,
\[
h_{\mathrm{H}}(T)=\mathbf{1}_{\mathbb{C}^r}-2\sqrt{T(\mathbf{1}_{\mathbb{C}^r}-T)},\qquad
h_{\mathrm{H},\alpha}(T)=\frac{(\mathbf{1}_{\mathbb{C}^r}-T)-T^\alpha(\mathbf{1}_{\mathbb{C}^r}-T)^{1-\alpha}}{1-\alpha},
\]
\[
h_{\mathrm{LC}}(T)=(2T-\mathbf{1}_{\mathbb{C}^r})^2,\qquad
h_{\mathrm{JS}}(T)=(\log 2)\,\mathbf{1}_{\mathbb{C}^r}+T\log T+(\mathbf{1}_{\mathbb{C}^r}-T)\log(\mathbf{1}_{\mathbb{C}^r}-T).
\]
\end{remark}
 
We find under certain circumstances, the issue of Proposition~\ref{impossibility proposition} rears its head. We first decompose a bounded maximal $f$-divergence into an affine part and a Kubo--Ando connection. Divergences of this form, built from Kubo--Ando means, appear in \cite{PitrikVirosztek2020}; the closed-form expression in Theorem~\ref{thm:maxf} extends those of \cite{KatariyaWilde2021,FangFawzi2021}, and an operator-valued analogue is the divergence of \cite[Definition~3.1]{Okayasu2026maxf}.

\begin{lemma}\label{lem:maxf decomposition}
Let $f:(0,\infty)\to\mathbb{R}$ be operator convex with $f(1)=0$ and with $f(0^+)$ and $f'(\infty)$ finite. Set
\[
h(t)=f(0^+)+f'(\infty)\,t-f(t),\qquad t>0.
\]
Then $h$ is a nonnegative operator monotone function on $(0,\infty)$ with
$h(0^+)=0$ and $h(t)/t\to0$ as $t\to\infty$. Writing $\sigma_h$ for the Kubo--Ando connection with representing
function $h$, we have, for all positive semidefinite $A,B\in \mathrm{L}(\mathcal{V})$,
\[
\widehat{D}_f(A\|B)=f(0^+)\tr B+f'(\infty)\tr A-\tr(B\,\sigma_h\,A).
\]
\end{lemma}
\begin{proof}
Extend $f$ to $[0,\infty)$ by $f(0)=f(0^+)$; the extension is convex, so $t\mapsto (f(t)-f(0))/t$ is
nondecreasing on $(0,\infty)$ with limit $f'(\infty)$. Hence $f(t)\leq f(0^+)+f'(\infty)t$, that is, $h\geq 0$.
As an affine function minus an operator convex function, $h$ is operator concave, and extended by $h(0)=0$ it is
continuous, nonnegative and operator concave on $[0,\infty)$; it is therefore operator monotone
\cite[Theorem~V.2.5]{bhatia1997matrix}. The two limits follow from the definitions of $f(0^+)$ and $f'(\infty)$.
For positive definite $A,B$, applying the functional calculus to $f=f(0^+)+f'(\infty)t-h$ and using
$B^{1/2}(B^{-1/2}AB^{-1/2})B^{1/2}=A$ and $B\,\sigma_h\,A=B^{1/2}h(B^{-1/2}AB^{-1/2})B^{1/2}$ gives
\[
B^{1/2}f\big(B^{-1/2}AB^{-1/2}\big)B^{1/2}=f(0^+)B+f'(\infty)A-B\,\sigma_h\,A .
\]
For positive semidefinite $A,B$, replace $A,B$ by $A+\eta\mathbf{1}_{\mathcal{V}}$, $B+\eta\mathbf{1}_{\mathcal{V}}$, take traces
and let $\eta\downarrow0$: the left-hand side converges to $\widehat{D}_f(A\|B)$ by definition, and
$(B+\eta\mathbf{1}_{\mathcal{V}})\,\sigma_h\,(A+\eta\mathbf{1}_{\mathcal{V}})$ decreases to $B\,\sigma_h\,A$ \cite{KuboAndo1980}.
\end{proof}
 
\begin{theorem}\label{thm:maxf}
Let $f$ be as in Lemma~\ref{lem:maxf decomposition}. Let $\Omega\in \mathrm{CP}(\mathcal{X},\mathcal{Y})$ have a Kraus
decomposition $\Omega(X)=\sum_{k=1}^rV_kXV_k^*$ with linearly independent $V_1,\dots,V_r$, and let $W=[\vectorise(V_1)\ \cdots\ \vectorise(V_r)]$.
Define $h_f:[0,1]\to\mathbb{R}$ by
\[
h_f(t)=(1-t)\,f\Big(\frac{t}{1-t}\Big)\quad(0<t<1),\qquad h_f(0)=f(0^+),\qquad h_f(1)=f'(\infty).
\]
\begin{enumerate}
\item[(a)] If $0\leq T\leq\mathbf{1}_{\mathbb{C}^r}$, $\Phi=\Psi_T$ and $\Psi=\Psi_{\mathbf{1}_{\mathbb{C}^r}-T}$, then
\[
\widehat{D}_f(\Phi\|\Psi)=\lambda_{\max}\Big(\Tr_{\mathcal{Y}}\big(W\,h_f(T)\,W^*\big)\Big).
\]
\item[(b)] If $0\leq T_1,T_2\leq\mathbf{1}_{\mathbb{C}^r}$, then
\[
\Big|\widehat{D}_f\big(\Psi_{T_1}\big\|\Psi_{\mathbf{1}_{\mathbb{C}^r}-T_1}\big)-\widehat{D}_f\big(\Psi_{T_2}\big\|\Psi_{\mathbf{1}_{\mathbb{C}^r}-T_2}\big)\Big|
\leq\|\Omega^*(\mathbf{1}_{\mathcal{Y}})\|_\infty\;\omega_{h_f}\big(\|T_1-T_2\|_\infty\big),
\]
where $\omega_{h_f}$ is as in Theorem~\ref{Modulus of continuity results},
and $\omega_{h_f}(\delta)\to0$ as $\delta\to0$.
\end{enumerate}
\end{theorem}
 
\begin{proof}
The argument follows the strategy of \cite[Lemma~5]{FangFawzi2021}. With $h$ as in Lemma~\ref{lem:maxf decomposition}, put
$k(t)=(1-t)\,h\big(\tfrac{t}{1-t}\big)$ for $0<t<1$ and $k(0)=k(1)=0$. Substituting $f=f(0^+)+f'(\infty)t-h$ into
the definition of $h_f$ gives
\[
h_f(t)=f(0^+)(1-t)+f'(\infty)\,t-k(t),\qquad 0\leq t\leq 1 .
\]
The function $k$ is continuous on $[0,1]$: as $t\downarrow0$ we have $k(t)\to h(0^+)=0$, and as $t\uparrow1$, writing
$x=t/(1-t)\to\infty$, we have $k(t)=t\,h(x)/x\to0$. Hence $h_f$ is continuous on $[0,1]$, and
$\omega_{h_f}(\delta)\to0$ by Lemma~\ref{modulus to zero}.
 
 The ARN derivatives of $\Psi$ and $\Phi$ with respect to $\Omega$ are
$\mathbf{1}_{\mathbb{C}^r}-T$ and $T$, and the scalar connection satisfies $s\,\sigma_h\,(1-s)=s\,h\big(\tfrac{1-s}{s}\big)$, which equals
$k(t)$ at $s=1-t$. Theorem~\ref{fundamental principle}, applied to the pair $(\Psi,\Phi)$, therefore gives
$J(\Psi)\,\sigma_h\,J(\Phi)=W\,k(T)\,W^*$. (If $h\not\equiv0$ then $h(1)>0$ and $\sigma_h=h(1)\,\sigma_{h/h(1)}$ with
$\sigma_{h/h(1)}$ a mean; if $h\equiv0$ there is nothing to prove.) Since $J(\Phi)=WTW^*$ and
$J(\Psi)=W(\mathbf{1}_{\mathbb{C}^r}-T)W^*$, we then have
\[
Z_f=f(0^+)J(\Psi)+f'(\infty)J(\Phi)-J(\Psi)\,\sigma_h\,J(\Phi)=W\,h_f(T)\,W^* .
\]
 
 Let $M=\Tr_{\mathcal{Y}}(W\,h_f(T)\,W^*)$, a Hermitian operator on $\mathcal{X}$. Let
$u=\vectorise(A)$ with $\|A\|_2=1$ and $A$ invertible, and let $C=\mathbf{1}_{\mathcal{Y}}\otimes A^\top$, which is
invertible. As in the proof of Proposition~\ref{lambda min formula},
$(\Phi\otimes\mathbf{1}_{\mathrm{L}(\mathcal{X})})(uu^*)=C\,J(\Phi)\,C^*$ and similarly for $\Psi$. By the transformer equality
for invertible operators \cite{KuboAndo1980},
\[
(CJ(\Psi)C^*)\,\sigma_h\,(CJ(\Phi)C^*)=C\big(J(\Psi)\,\sigma_h\,J(\Phi)\big)C^*,
\]
so Lemma~\ref{lem:maxf decomposition} gives
\begin{align*}
\widehat{D}_f\big((\Phi\otimes\mathbf{1}_{\mathrm{L}(\mathcal{X})})(uu^*)\big\|(\Psi\otimes\mathbf{1}_{\mathrm{L}(\mathcal{X})})(uu^*)\big)
&=\tr(CZ_fC^*)=\tr\big((\mathbf{1}_{\mathcal{Y}}\otimes(AA^*)^\top)Z_f\big)\\
&=\tr\big((AA^*)^\top M\big),
\end{align*}
using $C^*C=\mathbf{1}_{\mathcal{Y}}\otimes(AA^*)^\top$ and the previous part. 
 Let
$d(\rho)=\widehat{D}_f\big((\Phi\otimes\mathbf{1}_{\mathrm{L}(\mathcal{X})})(\rho)\big\|(\Psi\otimes\mathbf{1}_{\mathrm{L}(\mathcal{X})})(\rho)\big)$ for
$\rho\in \mathrm{D}(\mathcal{X}\otimes\mathcal{X})$. By Lemma~\ref{lem:maxf decomposition}, $d$ is a linear function minus
$\rho\mapsto\tr\big((\Psi\otimes\mathbf{1}_{\mathrm{L}(\mathcal{X})})(\rho)\,\sigma_h\,(\Phi\otimes\mathbf{1}_{\mathrm{L}(\mathcal{X})})(\rho)\big)$, which is concave because
$\sigma_h$ is jointly concave \cite{KuboAndo1980} and $\rho\mapsto\big((\Phi\otimes\mathbf{1}_{\mathrm{L}(\mathcal{X})})(\rho),(\Psi\otimes\mathbf{1}_{\mathrm{L}(\mathcal{X})})(\rho)\big)$
is linear. Thus $d$ is convex, and writing $\rho=\sum_ip_iu_iu_i^*$ in its spectral decomposition gives
$d(\rho)\leq\max_id(u_iu_i^*)$; so the supremum defining $\widehat{D}_f(\Phi\|\Psi)$ may be taken over pure states.
Moreover $d$ is lower semicontinuous: for $\eta>0$ the map $(A,B)\mapsto\tr\big((B+\eta\mathbf{1}_{\mathcal{Y}\otimes\mathcal{X}})\,\sigma_h\,(A+\eta\mathbf{1}_{\mathcal{Y}\otimes\mathcal{X}})\big)$
is continuous and decreases to $\tr(B\,\sigma_h\,A)$ as $\eta\downarrow0$, so $(A,B)\mapsto\tr(B\,\sigma_h\,A)$ is upper
semicontinuous. Every unit vector $u=\vectorise(A)$ is a limit of unit vectors $u_n=\vectorise(A_n)$ with $A_n$
invertible, and then $d(uu^*)\leq\liminf_nd(u_nu_n^*)$. Hence the supremum may be taken over $u=\vectorise(A)$ with
$A$ invertible. As $A$ ranges over such operators, $(AA^*)^\top$ ranges over the invertible density operators on
$\mathcal{X}$, which are dense in $\mathrm{D}(\mathcal{X})$. By the previous parts,
\[
\widehat{D}_f(\Phi\|\Psi)=\sup\big\{\tr(\xi M):\xi\in \mathrm{D}(\mathcal{X})\ \text{invertible}\big\}
=\max_{\xi\in \mathrm{D}(\mathcal{X})}\tr(\xi M)=\lambda_{\max}(M),
\]
which is (a).
 
Let $M_i=\Tr_{\mathcal{Y}}(W\,h_f(T_i)\,W^*)$ and $K=h_f(T_1)-h_f(T_2)$. By (a) and
Weyl's inequality \cite[Section~III.2]{bhatia1997matrix}, the left-hand side of (b) is
$|\lambda_{\max}(M_1)-\lambda_{\max}(M_2)|\leq\|M_1-M_2\|_\infty$. Since
$-\|K\|_\infty\mathbf{1}_{\mathbb{C}^r}\leq K\leq\|K\|_\infty\mathbf{1}_{\mathbb{C}^r}$, conjugating by $W$ and applying the positive map
$\Tr_{\mathcal{Y}}$ gives
\[
-\|K\|_\infty\Tr_{\mathcal{Y}}(WW^*)\leq M_1-M_2\leq\|K\|_\infty\Tr_{\mathcal{Y}}(WW^*),
\]
and $\Tr_{\mathcal{Y}}(WW^*)=\Tr_{\mathcal{Y}}J(\Omega)=\Omega^*(\mathbf{1}_{\mathcal{Y}})^\top\geq0$. Hence
\[
\|M_1-M_2\|_\infty\leq\|\Omega^*(\mathbf{1}_{\mathcal{Y}})\|_\infty\|K\|_\infty
\leq\|\Omega^*(\mathbf{1}_{\mathcal{Y}})\|_\infty\,\omega_{h_f}(\|T_1-T_2\|_\infty),
\]
which is (b).
\end{proof}
 
Theorem~\ref{thm:maxf} also explains how the problem of Proposition~\ref{impossibility proposition} arises for maximal $f$-divergences. By Lemma~\ref{lem:maxf decomposition}, the connection $\sigma_h$ inside $\widehat{D}_f$ has representing function $h$ with $h(0^+)=0$ and $h'(\infty)=0$, so, by Theorem~\ref{range theorem}, it vanishes on Choi matrices whose ranges intersect trivially.
 
\begin{corollary}\label{maxf impossibility}
Let $f$ be as in Lemma~\ref{lem:maxf decomposition} and let $\Phi,\Psi\in\mathrm{CP}(\mathcal{X},\mathcal{Y})$. If $\range(J(\Phi))\cap\range(J(\Psi))=\{0\}$, then (cf.\ \cite{Okayasu2026maxf})
\[
\widehat{D}_f(\Phi\|\Psi)=\lambda_{\max}\big(f(0^+)\,\Psi^*(\mathbf{1}_{\mathcal{Y}})+f'(\infty)\,\Phi^*(\mathbf{1}_{\mathcal{Y}})\big).
\]
In particular, for the maps of Corollary~\ref{impossibility corollary} we have $\widehat{D}_f(\Phi_1\|\Psi_1)=f(0^+)+f'(\infty)$ and $\widehat{D}_f(\Phi_2\|\Psi_2)=0$, although $|||\Phi_1-\Phi_2|||_1=|||\Psi_1-\Psi_2|||_1<\eta$.
\end{corollary}
\begin{proof}
Apply Theorem~\ref{thm:maxf} with $\Omega=\Phi+\Psi$. Its proof shows that $\widehat{D}_f(\Phi\|\Psi)=\lambda_{\max}(\Tr_{\mathcal{Y}}Z_f)$, where
\[
Z_f=f(0^+)J(\Psi)+f'(\infty)J(\Phi)-J(\Psi)\,\sigma_h\,J(\Phi).
\]
If $h\not\equiv0$, then $\sigma_h=h(1)\,\sigma_{h/h(1)}$ and $\sigma_{h/h(1)}$ is a mean satisfying the hypothesis of Theorem~\ref{range theorem}(1); if $h\equiv0$ then $\sigma_h=0$. In either case $J(\Psi)\,\sigma_h\,J(\Phi)=0$, so $\Tr_{\mathcal{Y}}Z_f=\big(f(0^+)\Psi^*(\mathbf{1}_{\mathcal{Y}})+f'(\infty)\Phi^*(\mathbf{1}_{\mathcal{Y}})\big)^\top$, using $\Tr_{\mathcal{Y}}J(\Xi)=\Xi^*(\mathbf{1}_{\mathcal{Y}})^\top$. Transposition does not change eigenvalues, which proves the formula.
 
For the maps of Corollary~\ref{impossibility corollary}, the Choi ranges of $\Phi_1$ and $\Psi_1$ intersect trivially by the proof of Proposition~\ref{impossibility proposition}, and $\Phi_1^*(\mathbf{1}_{\mathcal{Y}})=\Psi_1^*(\mathbf{1}_{\mathcal{Y}})=\mathbf{1}_{\mathcal{X}}$, so $\widehat{D}_f(\Phi_1\|\Psi_1)=f(0^+)+f'(\infty)$. For $\Phi_2=\Psi_2$ we have $J(\Phi_2)\,\sigma_h\,J(\Phi_2)=h(1)J(\Phi_2)$, hence $Z_f=\big(f(0^+)+f'(\infty)-h(1)\big)J(\Phi_2)=f(1)J(\Phi_2)=0$ and $\widehat{D}_f(\Phi_2\|\Psi_2)=0$.
\end{proof}
 
Thus every bounded maximal $f$-divergence exhibits the phenomenon of Corollary~\ref{impossibility corollary}, while Theorem~\ref{thm:maxf}(b) shows that it is controlled by the ARN derivative.

\begin{remark}\label{rem:normalisation}
The normalisation $f(1)=0$ is used only in the last statement of Corollary~\ref{maxf impossibility}. The proofs of Lemma~\ref{lem:maxf decomposition}, of Theorem~\ref{thm:maxf} and of the first statement of Corollary~\ref{maxf impossibility} do not use it. These results therefore hold for every bounded operator convex $f$, that is, every operator convex $f$ with $f(0^+)$ and $f'(\infty)$ finite, and we use them in this generality.
\end{remark}

\section{Unbounded maximal \texorpdfstring{$f$}{f}-divergences}\label{sec:unbounded}

Theorem~\ref{thm:maxf} requires both $f(0^+)$ and $f'(\infty)$ to be finite. Several of the most widely used maximal $f$-divergences have $f'(\infty)=+\infty$: the Belavkin--Staszewski relative entropy, with $f(t)=t\log t$ \cite{BelavkinStaszewski1982}; the geometric R\'enyi quasi-entropies with $\alpha\in(1,2]$, with $f(t)=t^\alpha$ \cite{KatariyaWilde2021,FangFawzi2021}; and the maximal $\chi^2$-divergence, with $f(t)=(t-1)^2$. For these the function $h_f$ of Theorem~\ref{thm:maxf} is continuous on $[0,1)$ but unbounded near $t=1$, which is where the ARN derivative $\mathbf{1}_{\mathbb{C}^r}-T$ of $\Psi$ degenerates. In this section we show that the closed form and the continuity estimates of Theorem~\ref{thm:maxf} persist under a bound on the max-relative entropy, which keeps the spectrum of $T$ away from $1$, and we compute explicit moduli of continuity for the first two examples. We first make a note of why the max-relative entropy is the natural object to discuss. Let $\Phi$ and $\Psi$ be completely positive maps and suppose that $\Psi\leq \Phi$. Let $W$ be the column matrix of a Kraus decomposition of $\Phi$ by linearly independent operators, as in Theorem~\ref{thm:arveson}, and let $T$ be the ARN derivative of $\Psi$ with respect to $\Phi$. Then the top eigenvalue $t_1$ of $T$ is the quantity \[\min\{\lambda\geq 0:\Psi\leq \lambda \Phi\}.\] This is because by Choi's theorem, we have 
\[
    \Psi\leq \lambda\Phi\iff J(\Psi)\leq \lambda J(\Phi).
\]
By transferring to the positive semidefinite picture by the ARN correspondence, this means that 
\[
    \Psi\leq\lambda\Phi \iff W(\lambda \mathbf{1}_{\mathbb{C}^r}-T)W^*\geq 0.
\]
But this is only true if and only if $\lambda \mathbf{1}_{\mathbb{C}^r}-T\geq 0$. Since $T$ is positive semidefinite, we then have $\lambda\mathbf{1}_{\mathbb{C}^r}-T\geq 0$ if and only if $\lambda\geq t$ for every eigenvalue $t$ of $T$, that is, if and only if $\lambda \geq t_1$. The max-relative entropy $D_{\max}(\Psi\|\Phi)$ is the logarithm of this value. So this is a natural object with which to control things.

Throughout this section $f:(0,\infty)\to\mathbb{R}$ is operator convex with $f(0^+)<\infty$. We do not assume that $f(1)=0$ or that $f'(\infty)<\infty$, and we extend $f$ continuously to $[0,\infty)$ by $f(0)=f(0^+)$. The maximal $f$-divergences $\widehat{D}_f(A\|B)$ and $\widehat{D}_f(\Phi\|\Psi)$ are defined by the formulas of Definitions~\ref{def:maxf} and~\ref{def:maxf maps}; Lemma~\ref{lem:unb basic} below shows that the limit in Definition~\ref{def:maxf} exists in $(-\infty,+\infty]$. For positive definite $A,B$ we write
\[
\mathsf{P}_f(A,B)=B^{1/2}f\big(B^{-1/2}AB^{-1/2}\big)B^{1/2}
\]
for the operator perspective of $f$, so that $\widehat{D}_f(A\|B)=\tr\mathsf{P}_f(A,B)$ in that case.

\subsection{The operator perspective on supports}
In this subsection, we show similar results to the Kubo--Ando case for operator convex $f$. This is necessary for the unbounded case, as we cannot exploit the rewriting of unbounded maximal $f$-divergences in the same way as the bounded case through the mean. We present a few well-known facts. The first is the invariance under congruences of operator perspective functions (the analogous objects to the connections induced by operator monotone functions).
\begin{lemma}\label{lem:unb congruence pd}
Let $A,B,R\in\mathrm{L}(\mathcal{V})$ with $A\geq0$, $B>0$ and $R$ invertible. Then
\[
\mathsf{P}_f(RAR^*,RBR^*)=R\,\mathsf{P}_f(A,B)\,R^*.
\]
\end{lemma}

\begin{proof}
Let $Y=RB^{1/2}$, which is invertible, and let $Y=(YY^*)^{1/2}U$ be its polar decomposition, with $U$ unitary. Since $YY^*=RBR^*$ and $RAR^*=YXY^*$ with $X=B^{-1/2}AB^{-1/2}$, we have $(RBR^*)^{-1/2}RAR^*(RBR^*)^{-1/2}=UXU^*$, and hence $f\big((RBR^*)^{-1/2}RAR^*(RBR^*)^{-1/2}\big)=Uf(X)U^*$. Therefore
\begin{align*}
\mathsf{P}_f(RAR^*,RBR^*)&=(YY^*)^{1/2}Uf(X)U^*(YY^*)^{1/2}\\
&=Yf(X)Y^*=R\,\mathsf{P}_f(A,B)\,R^*.\qedhere
\end{align*}
\end{proof}

\begin{definition}\label{def:unb perspective}
Let $A,B\in\mathrm{Pos}(\mathcal{V})$ with $\range(A)\subseteq\range(B)$ and $B\neq0$. Let $m=\dim\range(B)$, let $U\in\mathrm{L}(\mathbb{C}^m,\mathcal{V})$ be an isometry with range $\range(B)$, and put $A_0=U^*AU$ and $B_0=U^*BU$, so that $B_0>0$, $A=UA_0U^*$ and $B=UB_0U^*$. We set
\[
\mathsf{P}_f(A,B)=U\,\mathsf{P}_f(A_0,B_0)\,U^*,
\]
and $\mathsf{P}_f(0,0)=0$.
\end{definition}

The identity $A=UA_0U^*$ holds because $A$ is Hermitian with $\range(A)\subseteq\range(U)$, so that $A=UU^*AUU^*$. Two isometries with range $\range(B)$ differ by a unitary on $\mathbb{C}^m$, so by Lemma~\ref{lem:unb congruence pd} the definition does not depend on the choice of $U$; for $B>0$ we may take $U=\mathbf{1}_{\mathcal{V}}$, and the definition agrees with the previous one.

\begin{lemma}\label{lem:unb congruence}
Let $A,B\in\mathrm{Pos}(\mathcal{V})$ with $\range(A)\subseteq\range(B)$, and let $C\in\mathrm{L}(\mathcal{V},\mathcal{V}')$ be injective on $\range(B)$. Then $\range(CAC^*)\subseteq\range(CBC^*)$ and
\[
\mathsf{P}_f(CAC^*,CBC^*)=C\,\mathsf{P}_f(A,B)\,C^*.
\]
\end{lemma}

\begin{proof}
We may assume $B\neq0$. With $U$, $A_0$ and $B_0$ as in Definition~\ref{def:unb perspective}, the operator $Y=CU$ is injective. Its polar decomposition is $Y=U'R$, where $R=(Y^*Y)^{1/2}$ is invertible and $U'=YR^{-1}$ is an isometry with range $\range(Y)$. Then $CBC^*=U'(RB_0R)U'^*$ and $CAC^*=U'(RA_0R)U'^*$ with $RB_0R>0$, so $\range(CBC^*)=\range(U')\supseteq\range(CAC^*)$. Using Definition~\ref{def:unb perspective} with the isometry $U'$ and then Lemma~\ref{lem:unb congruence pd},
\begin{align*}
\mathsf{P}_f(CAC^*,CBC^*)&=U'\,\mathsf{P}_f(RA_0R,RB_0R)\,U'^*=U'R\,\mathsf{P}_f(A_0,B_0)\,RU'^*\\
&=Y\,\mathsf{P}_f(A_0,B_0)\,Y^*=C\,\mathsf{P}_f(A,B)\,C^*.\qedhere
\end{align*}
\end{proof}

\begin{lemma}\label{lem:unb basic}
Let $A,B\in\mathrm{Pos}(\mathcal{V})$ and $n_{\mathcal{V}}=\dim\mathcal{V}$.
\begin{enumerate}
\item The function $\varphi_\eta(A,B)=\tr\mathsf{P}_f(A+\eta\mathbf{1}_{\mathcal{V}},B+\eta\mathbf{1}_{\mathcal{V}})-\eta n_{\mathcal{V}}\,f(1)$ is nonincreasing in $\eta\in(0,\infty)$. Consequently the limit defining $\widehat{D}_f(A\|B)$ exists in $(-\infty,+\infty]$, and $\widehat{D}_f(A\|B)=\sup_{\eta>0}\varphi_\eta(A,B)$.
\item The map $(A,B)\mapsto\widehat{D}_f(A\|B)$ is jointly convex and lower semicontinuous on $\mathrm{Pos}(\mathcal{V})\times\mathrm{Pos}(\mathcal{V})$.
\item If $\range(A)\subseteq\range(B)$, then $\widehat{D}_f(A\|B)=\tr\mathsf{P}_f(A,B)$; in particular, $\widehat{D}_f(A\|B)$ is finite.
\end{enumerate}
\end{lemma}

\begin{proof}
(1) On pairs of positive definite operators, $\mathsf{P}_f$ is jointly operator convex \cite{EbadianNikoufarGordji2011} (see also \cite{Effros2009}) and positively homogeneous, $\mathsf{P}_f(\lambda A,\lambda B)=\lambda\,\mathsf{P}_f(A,B)$ for $\lambda>0$. It is therefore subadditive:
\[
\mathsf{P}_f(A_1+A_2,B_1+B_2)=2\,\mathsf{P}_f\big(\tfrac12(A_1+A_2),\tfrac12(B_1+B_2)\big)\leq\mathsf{P}_f(A_1,B_1)+\mathsf{P}_f(A_2,B_2).
\]
For $0<\eta'<\eta$, apply this with $(A_1,B_1)=(A+\eta'\mathbf{1}_{\mathcal{V}},B+\eta'\mathbf{1}_{\mathcal{V}})$ and $A_2=B_2=(\eta-\eta')\mathbf{1}_{\mathcal{V}}$, and use $\mathsf{P}_f(c\mathbf{1}_{\mathcal{V}},c\mathbf{1}_{\mathcal{V}})=cf(1)\mathbf{1}_{\mathcal{V}}$; taking traces gives $\varphi_\eta(A,B)\leq\varphi_{\eta'}(A,B)$. Hence $\varphi_\eta(A,B)$ increases as $\eta\downarrow0$, and since $\eta n_{\mathcal{V}} f(1)\to0$, the limit of $\tr\mathsf{P}_f(A+\eta\mathbf{1}_{\mathcal{V}},B+\eta\mathbf{1}_{\mathcal{V}})$ exists in $(-\infty,+\infty]$ and equals $\sup_{\eta>0}\varphi_\eta(A,B)$.

(2) For fixed $\eta>0$, the map $\varphi_\eta$ is continuous, by continuity of square roots, inverses and the continuous functional calculus on positive definite operators, and it is jointly convex, as the trace of a jointly operator convex map plus a constant. A pointwise supremum of continuous jointly convex functions is lower semicontinuous and jointly convex, so (1) gives (2).

(3) We may assume $B\neq0$. Decompose $\mathcal{V}=\range(B)\oplus\ker(B)$. Since $\range(A)\subseteq\range(B)$ and $A$ is Hermitian, $\ker(B)\subseteq\ker(A)$, so $A=A_0\oplus0$ and $B=B_0\oplus0$ with $B_0>0$, identifying $\range(B)$ with $\mathbb{C}^m$ as in Definition~\ref{def:unb perspective}. The functional calculus respects direct sums, so
\[
\tr\mathsf{P}_f(A+\eta\mathbf{1}_{\mathcal{V}},B+\eta\mathbf{1}_{\mathcal{V}})=\tr\mathsf{P}_f(A_0+\eta\mathbf{1}_{\mathbb{C}^m},B_0+\eta\mathbf{1}_{\mathbb{C}^m})+\eta f(1)\dim\ker(B).
\]
As $\eta\downarrow0$, the operators $X_\eta=(B_0+\eta\mathbf{1}_{\mathbb{C}^m})^{-1/2}(A_0+\eta\mathbf{1}_{\mathbb{C}^m})(B_0+\eta\mathbf{1}_{\mathbb{C}^m})^{-1/2}$ converge to $X_0=B_0^{-1/2}A_0B_0^{-1/2}$, and for $\eta\leq1$ their spectra lie in a fixed compact interval $I_0\subseteq[0,\infty)$. Since $f$ is continuous on $I_0$, $f(X_\eta)\to f(X_0)$, and therefore $\mathsf{P}_f(A_0+\eta\mathbf{1}_{\mathbb{C}^m},B_0+\eta\mathbf{1}_{\mathbb{C}^m})\to\mathsf{P}_f(A_0,B_0)$. The second term tends to $0$, so $\widehat{D}_f(A\|B)=\tr\mathsf{P}_f(A_0,B_0)=\tr\mathsf{P}_f(A,B)$.
\end{proof}

\subsection{Closed form and continuity}

We now come to a statement further relating the max-relative entropy and the ARN derivative. We notice that this then relates to the Lebesgue decomposition of Proposition~\ref{prop:lebesgue cp}. 
\begin{lemma}\label{lem:dmax}
Let $\Omega$ and $W$ be as in Theorem~\ref{thm:maxf}, let $0\leq T\leq\mathbf{1}_{\mathbb{C}^r}$, and put $\Phi=\Psi_T$ and $\Psi=\Psi_{\mathbf{1}_{\mathbb{C}^r}-T}$. For $c\geq0$, we have $\Phi\leq c\Psi$ in the completely positive order if and only if $T\leq\frac{c}{1+c}\mathbf{1}_{\mathbb{C}^r}$. Consequently,
\[
D_{\max}(\Phi\|\Psi):=\log\inf\{c\geq0:\ \Phi\leq c\Psi\}=\log\frac{\|T\|_\infty}{1-\|T\|_\infty},
\]
with $\log0=-\infty$ and $D_{\max}(\Phi\|\Psi)=+\infty$ when $\|T\|_\infty=1$. In particular, for $\epsilon\in(0,1)$ we have $D_{\max}(\Phi\|\Psi)\leq\log\frac{1-\epsilon}{\epsilon}$ if and only if $T\leq(1-\epsilon)\mathbf{1}_{\mathbb{C}^r}$.
\end{lemma}

\begin{proof}
By linearity of $T\mapsto\Psi_T$, $c\Psi-\Phi=\Psi_{c\mathbf{1}_{\mathbb{C}^r}-(1+c)T}$, whose Choi matrix is \[W\big(c\mathbf{1}_{\mathbb{C}^r}-(1+c)T\big)W^*\] by Theorem~\ref{thm:arveson}. By Choi's theorem, $c\Psi-\Phi$ is completely positive if and only if this operator is positive semidefinite, and this happens if and only if $c\mathbf{1}_{\mathbb{C}^r}-(1+c)T\geq0$.

Reasoning about the spectrum of $T$, we then obtain that this translates to 
\[
    \frac{\|T\|_\infty}{1-\|T\|_\infty}\leq c< \infty.
\]
This then implies that \[\{ c\geq 0:\Phi\leq c\Psi\}=\left[\frac{\|T\|_\infty}{1-\|T\|_\infty},\infty\right)\]
if $\|T\|_\infty<1$, and is empty otherwise. This gives the formula for $D_{\max}$. The last statement follows because $t\mapsto t/(1-t)$ is increasing on $[0,1)$.
\end{proof}
\begin{remark}[Physical interpretation]\label{rem:phys-dmax}
For channels and $c>1$, $\Phi\leq c\Psi$ holds if and only if $\Psi=c^{-1}\Phi+(1-c^{-1})\Xi$ with $\Xi=(c\Psi-\Phi)/(c-1)$ a channel, so $e^{-D_{\max}(\Phi\|\Psi)}$ is the largest probability with which $\Psi$ secretly runs $\Phi$. Equivalently, by complete positivity of $c\Psi-\Phi$, on every input, ancilla allowed, every measurement outcome is at most $e^{D_{\max}(\Phi\|\Psi)}$ times likelier under $\Phi$ than under $\Psi$, as for states \cite{ReebKastoryanoWolf2011}. Detector inefficiency $\epsilon\in(0,1)$ gives such bounds: an instrument $(\Phi,\Psi)$ is observed as $((1-\epsilon)\Phi,\Psi+\epsilon\Phi)$, the case $p_{01}=\epsilon$, $p_{10}=0$ of Remark~\ref{rem:phys-readout}, with derivative $(1-\epsilon)T\leq(1-\epsilon)\mathbf{1}_{\mathbb{C}^r}$, so $D_{\max}((1-\epsilon)\Phi\|\Psi+\epsilon\Phi)\leq\log((1-\epsilon)/\epsilon)$ by the lemma, whatever the ideal pair.
\end{remark}

 For $\epsilon\in(0,1)$ and a continuous function $g:[0,1-\epsilon]\to\mathbb{R}$ we use the modulus $\omega^{(\epsilon)}_g=\omega^{[0,1-\epsilon]}_g$ of Lemma~\ref{modulus to zero}, which tends to $0$ as $\delta\to0$.

\begin{theorem}\label{thm:unbounded}
Let $\Omega\in\mathrm{CP}(\mathcal{X},\mathcal{Y})$ have linearly independent Kraus operators $V_1,\dots,V_r$, so that $\Omega(X)=\sum_{k=1}^rV_kXV_k^*$, let $W$ be the matrix with columns $\vectorise(V_1)$, \dots, $\vectorise(V_r)$, and define the continuous function $h_f:[0,1)\to\mathbb{R}$ by
\[
h_f(t)=(1-t)\,f\Big(\frac{t}{1-t}\Big),\qquad\text{so that } h_f(0)=f(0^+).
\]
Let $\epsilon\in(0,1)$.
\begin{enumerate}
\item[(a)] If $0\leq T\leq(1-\epsilon)\mathbf{1}_{\mathbb{C}^r}$, $\Phi=\Psi_T$ and $\Psi=\Psi_{\mathbf{1}_{\mathbb{C}^r}-T}$, then $\widehat{D}_f(\Phi\|\Psi)$ is finite and
\[
\widehat{D}_f(\Phi\|\Psi)=\lambda_{\max}\Big(\Tr_{\mathcal{Y}}\big(W\,h_f(T)\,W^*\big)\Big).
\]
\item[(b)] If $0\leq T_1,T_2\leq(1-\epsilon)\mathbf{1}_{\mathbb{C}^r}$, then
\[
\Big|\widehat{D}_f\big(\Psi_{T_1}\big\|\Psi_{\mathbf{1}_{\mathbb{C}^r}-T_1}\big)-\widehat{D}_f\big(\Psi_{T_2}\big\|\Psi_{\mathbf{1}_{\mathbb{C}^r}-T_2}\big)\Big|
\leq\|\Omega^*(\mathbf{1}_{\mathcal{Y}})\|_\infty\;\omega^{(\epsilon)}_{h_f}\big(\|T_1-T_2\|_\infty\big),
\]
and $\omega^{(\epsilon)}_{h_f}(\delta)\to0$ as $\delta\to0$.
\end{enumerate}
By Lemma~\ref{lem:dmax}, the hypotheses say exactly that the pairs involved satisfy $D_{\max}\leq\log\frac{1-\epsilon}{\epsilon}$.
\end{theorem}

\begin{proof}
(a) The structure of the first part is that we aim to use the conditions to transform into the absolute continuity form of Lemma~\ref{lem:unb basic}(3). Let $c=(1-\epsilon)/\epsilon$. We can clearly then see that $\Phi\leq c\Psi$ by Lemma~\ref{lem:dmax}. We proceed in a similar way to the bounded case by letting $\rho\in\mathrm{D}(\mathcal{X}\otimes\mathcal{X})$ write $\Phi_\rho=(\Phi\otimes\mathbf{1}_{\mathrm{L}(\mathcal{X})})(\rho)$ and $\Psi_\rho=(\Psi\otimes\mathbf{1}_{\mathrm{L}(\mathcal{X})})(\rho)$, and let $d(\rho)=\widehat{D}_f(\Phi_\rho\|\Psi_\rho)$.

Since $c\Psi-\Phi$ is completely positive, $\Phi_\rho\leq c\Psi_\rho$; if $\Psi_\rho v=0$ then $0\leq v^*\Phi_\rho v\leq c\,v^*\Psi_\rho v=0$, so $\ker(\Psi_\rho)\subseteq\ker(\Phi_\rho)$ and $\range(\Phi_\rho)\subseteq\range(\Psi_\rho)$. By Lemma~\ref{lem:unb basic}(3), $d(\rho)=\tr\mathsf{P}_f(\Phi_\rho,\Psi_\rho)$ is finite for every $\rho$.

We now need to inspect these quantities at the Choi level, following that we can reduce the maximal $f$-divergences to partial traces of Choi relations. By Theorem~\ref{thm:arveson}, $J(\Phi)=WTW^*$ and $J(\Psi)=W(\mathbf{1}_{\mathbb{C}^r}-T)W^*$. Since $\mathbf{1}_{\mathbb{C}^r}-T\geq\epsilon\mathbf{1}_{\mathbb{C}^r}$ is invertible and $W$ is injective, Lemma~\ref{lem:unb congruence} with $C=W$ gives $\range(J(\Phi))\subseteq\range(J(\Psi))$ and
\[
\mathsf{P}_f\big(J(\Phi),J(\Psi)\big)=W\,\mathsf{P}_f\big(T,\mathbf{1}_{\mathbb{C}^r}-T\big)\,W^*.
\]
As $T$ and $\mathbf{1}_{\mathbb{C}^r}-T$ commute, $(\mathbf{1}_{\mathbb{C}^r}-T)^{-1/2}T(\mathbf{1}_{\mathbb{C}^r}-T)^{-1/2}=T(\mathbf{1}_{\mathbb{C}^r}-T)^{-1}$, and the functional calculus for operator perspectives gives $\mathsf{P}_f(T,\mathbf{1}_{\mathbb{C}^r}-T)=(\mathbf{1}_{\mathbb{C}^r}-T)f\big(T(\mathbf{1}_{\mathbb{C}^r}-T)^{-1}\big)=h_f(T)$. Hence $\mathsf{P}_f(J(\Phi),J(\Psi))=W\,h_f(T)\,W^*$. 

We want to show that by considering the definition of the maximal $f$-divergence, by considering this over density operators, that this corresponds to the largest eigenvalue of the relation of the Choi matrices. This is the same line of reasoning as in the Kubo--Ando case previously. We again look at invertible operators first.

Let $M=\Tr_{\mathcal{Y}}(W\,h_f(T)\,W^*)$, a Hermitian operator on $\mathcal{X}$. Let $u=\vectorise(A)$ with $\|A\|_2=1$ and $A$ invertible, and put $C=\mathbf{1}_{\mathcal{Y}}\otimes A^\top$, which is invertible. By \eqref{eq:vec identity}, $\Phi_{uu^*}=C\,J(\Phi)\,C^*$ and $\Psi_{uu^*}=C\,J(\Psi)\,C^*$. By Lemma~\ref{lem:unb basic}(3), Lemma~\ref{lem:unb congruence} and the previous step,
\begin{align*}
d(uu^*)&=\tr\big(C\,\mathsf{P}_f(J(\Phi),J(\Psi))\,C^*\big)\\&=\tr\big((\mathbf{1}_{\mathcal{Y}}\otimes(AA^*)^\top)\,W\,h_f(T)\,W^*\big)\\
&=\tr\big((AA^*)^\top M\big),
\end{align*}
using $C^*C=\mathbf{1}_{\mathcal{Y}}\otimes(AA^*)^\top$.

 By Lemma~\ref{lem:unb basic}(2), $d$ is convex and lower semicontinuous on $\mathrm{D}(\mathcal{X}\otimes\mathcal{X})$, being the composition of a jointly convex, lower semicontinuous function with the linear map $\rho\mapsto(\Phi_\rho,\Psi_\rho)$. Writing $\rho=\sum_ip_iu_iu_i^*$ in its spectral decomposition, convexity gives $d(\rho)\leq\max_id(u_iu_i^*)$, so the supremum defining $\widehat{D}_f(\Phi\|\Psi)$ may be taken over pure states. Every unit vector $u=\vectorise(A)$ is a limit of unit vectors $u_n=\vectorise(A_n)$ with $A_n$ invertible, and lower semicontinuity gives $d(uu^*)\leq\liminf_nd(u_nu_n^*)$, so the supremum may be taken over $u=\vectorise(A)$ with $A$ invertible. As $A$ ranges over such operators, $(AA^*)^\top$ ranges over the invertible density operators on $\mathcal{X}$, which are dense in $\mathrm{D}(\mathcal{X})$. Hence
\[
\widehat{D}_f(\Phi\|\Psi)=\sup\big\{\tr(\xi M):\ \xi\in\mathrm{D}(\mathcal{X})\text{ invertible}\big\}=\max_{\xi\in\mathrm{D}(\mathcal{X})}\tr(\xi M)=\lambda_{\max}(M).
\]

(b) Let $M_i=\Tr_{\mathcal{Y}}(W\,h_f(T_i)\,W^*)$ and $K=h_f(T_1)-h_f(T_2)$, so that $\|K\|_\infty\leq\omega^{(\epsilon)}_{h_f}(\|T_1-T_2\|_\infty)$ by definition. By (a), the left-hand side of (b) is $|\lambda_{\max}(M_1)-\lambda_{\max}(M_2)|$, which is at most $\|\Omega^*(\mathbf{1}_{\mathcal{Y}})\|_\infty\|K\|_\infty$ by Lemma~\ref{lem:partial trace}. Finally, $h_f$ is continuous on $[0,1-\epsilon]$, so $\omega^{(\epsilon)}_{h_f}(\delta)\to0$ as $\delta\to0$.
\end{proof}
\begin{remark}[Physical interpretation]\label{rem:phys-reverse}
For states, $\widehat{D}_f$ is the least classical $f$-divergence of a pair of distributions that prepares them through one classical-to-quantum channel \cite{Matsumoto2018newf}. For channels and an input $u=\vectorise(A)$ with $A$ invertible, let $T=\sum_it_iP_i$ spectrally, $\widetilde{W}=(\mathbf{1}_{\mathcal{Y}}\otimes A^\top)W$ and $w_i=\tr(\widetilde{W}P_i\widetilde{W}^*)$. The distributions $p_i=t_iw_i$, $q_i=(1-t_i)w_i$ and the preparation $\boldsymbol{\delta}_i\mapsto \widetilde{W}P_i\widetilde{W}^*/w_i$ give the outputs $\widetilde{W}T\widetilde{W}^*$ and $\widetilde{W}(\mathbf{1}_{\mathbb{C}^r}-T)\widetilde{W}^*$ (see \eqref{eq:vec identity}), and $\sum_iq_if(p_i/q_i)=\sum_iw_ih_f(t_i)$ is the quantum value by Lemma~\ref{lem:mixing}, so they are optimal. Their likelihood ratios $t_i/(1-t_i)$ do not depend on the input, and $\widehat{D}_f(\Phi\|\Psi)$ is the worst case over the weights $w_i$. The bound $D_{\max}\leq\log((1-\epsilon)/\epsilon)$ means that all ratios lie in $[0,(1-\epsilon)/\epsilon]$, and the constants of Corollaries~\ref{cor:BS} and~\ref{cor:renyi large alpha} measure how steep $h_f$ is on $[0,1-\epsilon]$.
\end{remark}

\subsection{The relative R\'enyi entropies continued}
In Section~\ref{sec:general}, we considered the bounds for the R\'enyi relative entropies for $\alpha\in (0,1)$. We now have the machinery in place to compute explicit bounds for $\alpha\in (1,2]$. We first need to list some useful bounds. Using the subsequent results, we will also show continuity bounds for the Belavkin--Staszewski divergence. It is important to note that in the following lemma, we are considering integrals of matrices in the functional calculus sense. 

\begin{lemma}\label{lem:unb moduli}
Let $\mathcal{V}$ be a complex Euclidean space, $\delta\in(0,1]$, $\beta\in(0,1]$ and $\epsilon>0$.
\begin{enumerate}
\item If $X_1,X_2\geq\epsilon\mathbf{1}_{\mathcal{V}}$, then $\|\log X_1-\log X_2\|_\infty\leq\|X_1-X_2\|_\infty/\epsilon$.
\item If $X_1,X_2\geq\epsilon\mathbf{1}_{\mathcal{V}}$, then $\|X_1^{-\beta}-X_2^{-\beta}\|_\infty\leq\beta\,\epsilon^{-\beta-1}\|X_1-X_2\|_\infty$.
\item If $0\leq A,B\leq\mathbf{1}_{\mathcal{V}}$ and $\|A-B\|_\infty\leq\delta$, then \[\|A\log A-B\log B\|_\infty\leq\delta\log\frac1\delta+(1+\delta)\log(1+\delta)\leq\delta\log\frac1\delta+2\log2\cdot\delta\]
\item If $0\leq A,B\leq\mathbf{1}_{\mathcal{V}}$, then $\|A^{1+\beta}-B^{1+\beta}\|_\infty\leq(3-\beta)\|A-B\|_\infty$.
\end{enumerate}
\end{lemma}

\begin{proof}
(1) For $X>0$, the scalar identity \[\log x=\int_0^\infty\big(\frac{1}{1+s}-\frac{1}{x+s}\big)\,ds\] and the functional calculus give \[\log X=\int_0^\infty\big((1+s)^{-1}\mathbf{1}_{\mathcal{V}}-(X+s\mathbf{1}_{\mathcal{V}})^{-1}\big)\,ds.\] Hence
\[
\log X_1-\log X_2=\int_0^\infty(X_1+s\mathbf{1}_{\mathcal{V}})^{-1}(X_1-X_2)(X_2+s\mathbf{1}_{\mathcal{V}})^{-1}\,ds,
\]
and $\|(X_i+s\mathbf{1}_{\mathcal{V}})^{-1}\|_\infty\leq(\epsilon+s)^{-1}$ gives the bound $\|X_1-X_2\|_\infty\int_0^\infty(\epsilon+s)^{-2}\,ds=\|X_1-X_2\|_\infty/\epsilon$.

(2) For $\beta=1$ this follows from $X_1^{-1}-X_2^{-1}=X_1^{-1}(X_2-X_1)X_2^{-1}$. For $0<\beta<1$ we use the representation $x^{-\beta}=\frac{\sin\beta\pi}{\pi}\int_0^\infty s^{-\beta}(x+s)^{-1}\,ds$, valid for $x>0$ \cite[Chapter~V]{bhatia1997matrix}, which gives
\[
X_1^{-\beta}-X_2^{-\beta}=\frac{\sin\beta\pi}{\pi}\int_0^\infty s^{-\beta}(X_1+s\mathbf{1}_{\mathcal{V}})^{-1}(X_2-X_1)(X_2+s\mathbf{1}_{\mathcal{V}})^{-1}\,ds .
\]
Since $\|(X_i+s\mathbf{1}_{\mathcal{V}})^{-1}\|_\infty\leq(\epsilon+s)^{-1}$ and \[\int_0^\infty s^{-\beta}(\epsilon+s)^{-2}\,ds=\epsilon^{-\beta-1}\,\Gamma(1-\beta)\Gamma(1+\beta)=\epsilon^{-\beta-1}\,\frac{\beta\pi}{\sin\beta\pi}\] the norm of the right-hand side is at most $\beta\epsilon^{-\beta-1}\|X_1-X_2\|_\infty$ where $\Gamma$ is the gamma function.

(3) The case $A=B$ is trivial, and since the right-hand side is increasing in $\delta$ we may assume $\delta=\|A-B\|_\infty>0$. For $s>0$ let $g_s(t)=\frac{t}{1+s}-\frac{t}{t+s}$. The scalar identity $t\log t=\int_0^\infty g_s(t)\,ds$, valid for $t\geq0$, gives $A\log A-B\log B=\int_0^\infty\big(g_s(A)-g_s(B)\big)\,ds$. We claim $\|g_s(A)-g_s(B)\|_\infty\leq\min\big\{\frac{\delta}{s(1+s)},\frac1{1+s}\big\}$. Indeed, the first bound is $\mathrm{Lip}^{[0,1]}(g_s)=\frac1{s(1+s)}$ from Corollary~\ref{cor:kernels}(d), whose proof does not use the present lemma, and the second holds because, by the same corollary, $g_s(A)$ and $g_s(B)$ have spectra in $[-\frac1{1+s},0]$. Integrating,
\[
\|A\log A-B\log B\|_\infty\leq\int_0^\delta\frac{ds}{1+s}+\int_\delta^\infty\frac{\delta\,ds}{s(1+s)}=\log(1+\delta)+\delta\log\frac{1+\delta}{\delta},
\]
which is the first bound. The second follows from $(1+\delta)\log(1+\delta)\leq2\log2\cdot\delta$ for $\delta\in[0,1]$, by convexity.

(4) For $\beta=1$ the claim is $\|A^2-B^2\|_\infty\leq2\|A-B\|_\infty$, which follows from $A^2-B^2=\frac12\big((A+B)(A-B)+(A-B)(A+B)\big)$. For $0<\beta<1$, the representation \[x^\beta=\frac{\sin\beta\pi}{\pi}\int_0^\infty s^{\beta-1}\,x(x+s)^{-1}\,ds\] for $x\geq0$ \cite[Chapter~V]{bhatia1997matrix} gives \[A^{1+\beta}=\frac{\sin\beta\pi}{\pi}\int_0^\infty s^{\beta-1}A^2(A+s\mathbf{1}_{\mathcal{V}})^{-1}\,ds\] and similarly for $B$. Writing \[A^2(A+s\mathbf{1}_{\mathcal{V}})^{-1}=A-s\mathbf{1}_{\mathcal{V}}+s^2(A+s\mathbf{1}_{\mathcal{V}})^{-1}\] and \[Y_s=A(A+s\mathbf{1}_{\mathcal{V}})^{-1}\] \[Z_s=B(B+s\mathbf{1}_{\mathcal{V}})^{-1}\] so that \[s(A+s\mathbf{1}_{\mathcal{V}})^{-1}=\mathbf{1}_{\mathcal{V}}-Y_s\] and \[s(B+s\mathbf{1}_{\mathcal{V}})^{-1}=\mathbf{1}_{\mathcal{V}}-Z_s\] we obtain with $H=A-B$
\[
A^2(A+s\mathbf{1}_{\mathcal{V}})^{-1}-B^2(B+s\mathbf{1}_{\mathcal{V}})^{-1}=H-(\mathbf{1}_{\mathcal{V}}-Y_s)H(\mathbf{1}_{\mathcal{V}}-Z_s)=Y_sH+HZ_s-Y_sHZ_s .
\]
Since $0\leq A,B\leq\mathbf{1}_{\mathcal{V}}$, we have $\|Y_s\|_\infty,\|Z_s\|_\infty\leq(1+s)^{-1}$, so the norm of this operator is at most $\big(\frac{2}{1+s}+\frac{1}{(1+s)^2}\big)\|H\|_\infty$. As $\int_0^\infty\frac{s^{\beta-1}}{1+s}\,ds=\frac{\pi}{\sin\beta\pi}$ and $\int_0^\infty\frac{s^{\beta-1}}{(1+s)^2}\,ds=\Gamma(\beta)\Gamma(2-\beta)=\frac{(1-\beta)\pi}{\sin\beta\pi}$, integrating gives $\|A^{1+\beta}-B^{1+\beta}\|_\infty\leq(2+1-\beta)\|H\|_\infty$.
\end{proof}

\begin{remark}\label{rem:bhatia sinha}
For $f(t)=t^\alpha$, Bhatia and Sinha \cite{BhatiaSinha1994} showed that the Fr\'echet derivative of $A\mapsto A^\alpha$ at a positive operator $A\in\mathrm{L}(\mathcal{V})$ has norm $\|Df(A)\|=\|\alpha A^{\alpha-1}\|_\infty$ when $\alpha<0$ or $\alpha\geq2$, extending the case $0\leq\alpha\leq1$ obtained earlier, but that this fails when $1<\alpha<\sqrt2$. Since $\|Df(A)\|\geq\|Df(A)(\mathbf{1}_{\mathcal{V}})\|_\infty=\|\alpha A^{\alpha-1}\|_\infty$ and $Df(cA)=c^{\alpha-1}Df(A)$ for $c>0$, for each $1<\alpha<\sqrt2$ there is a positive $A$ with $\|A\|_\infty=1$ and $\|Df(A)\|>\alpha$; as the best constant in Lemma~\ref{lem:unb moduli}(4) is $\mathrm{Lip}^{[0,1]}(t^{1+\beta})\geq\|Df(A)\|$ for such $A$ and $\alpha=1+\beta$ (Lemma~\ref{lem:DK} and the Daleckii--Krein formula), the constant $3-\beta$ there cannot be lowered to $1+\beta$, the maximum of the derivative of $t^{1+\beta}$ on $[0,1]$, when $0<\beta<\sqrt2-1$. For $\beta=1$ the constant $3-\beta=2$ is optimal, since $Df(\mathbf{1}_{\mathcal{V}})(H)=2H$ for $f(t)=t^2$.
\end{remark}

Using these bounds, we now obtain continuity bounds for the Belavkin--Staszewski relative entropy for channels with the same sum with max-relative entropy upper bounds on the ARN derivative.
\begin{corollary}\label{cor:BS}
Let $\widehat{D}_{\mathrm{BS}}=\widehat{D}_f$ with $f(t)=t\log t$, the Belavkin--Staszewski relative entropy \cite{BelavkinStaszewski1982}, and let $\Omega$, $W$, $\epsilon$, $T_1$, $T_2$ be as in Theorem~\ref{thm:unbounded}(b). Let $\Phi_i=\Psi_{T_i}$, $\Psi_i=\Psi_{\mathbf{1}_{\mathbb{C}^r}-T_i}$ and $\delta=\|T_1-T_2\|_\infty$. Then
\[
\big|\widehat{D}_{\mathrm{BS}}(\Phi_1\|\Psi_1)-\widehat{D}_{\mathrm{BS}}(\Phi_2\|\Psi_2)\big|\leq\|\Omega^*(\mathbf{1}_{\mathcal{Y}})\|_\infty\,\delta\Big(\log\frac{4}{\epsilon\delta}+\frac1\epsilon\Big),
\]
the right-hand side being interpreted as $0$ for $\delta=0$.
\end{corollary}

\begin{proof}
Here $f(0^+)=0$ and $h_f(t)=t\log\frac{t}{1-t}=t\log t-t\log(1-t)$. With $L_i=\log(\mathbf{1}_{\mathbb{C}^r}-T_i)$,
\[
h_f(T_1)-h_f(T_2)=\big(T_1\log T_1-T_2\log T_2\big)-(T_1-T_2)L_1-T_2(L_1-L_2).
\]
Since $\delta\leq1$, the first term has $\infty$-norm at most $\delta\log\frac1\delta+2\log2\cdot\delta$ by Lemma~\ref{lem:unb moduli}(3). The spectrum of $\mathbf{1}_{\mathbb{C}^r}-T_1$ lies in $[\epsilon,1]$, so $\|L_1\|_\infty\leq\log\frac1\epsilon$, and Lemma~\ref{lem:unb moduli}(1) with $X_i=\mathbf{1}_{\mathbb{C}^r}-T_i$ gives $\|L_1-L_2\|_\infty\leq\delta/\epsilon$; also $\|T_2\|_\infty\leq1$. Thus $\|h_f(T_1)-h_f(T_2)\|_\infty\leq\delta\log\frac1\delta+2\log2\cdot\delta+\log\frac1\epsilon\cdot\delta+\frac{\delta}{\epsilon}=\delta\big(\log\frac{4}{\epsilon\delta}+\frac1\epsilon\big)$, and the proof of Theorem~\ref{thm:unbounded}(b) gives the claim.
\end{proof}

\begin{corollary}\label{cor:renyi large alpha}
Let $\alpha\in(1,2]$ and $\beta=\alpha-1$, and let $\Omega$, $W$, $\epsilon$, $T_1$, $T_2$, $\Phi_i$, $\Psi_i$ and $\delta$ be as in Corollary~\ref{cor:BS}. Define $\widehat{Q}_\alpha(\Phi\|\Psi)=\sup_{\rho\in\mathrm{D}(\mathcal{X}\otimes\mathcal{X})}\widehat{Q}_\alpha(\Phi_\rho\|\Psi_\rho)$, so that $\widehat{D}_\alpha(\Phi\|\Psi)=\frac{1}{\alpha-1}\log\widehat{Q}_\alpha(\Phi\|\Psi)$. Then
\[
\widehat{Q}_\alpha(\Phi_i\|\Psi_i)=\lambda_{\max}\Big(\Tr_{\mathcal{Y}}\big(W\,h_\alpha(T_i)\,W^*\big)\Big),\qquad h_\alpha(t)=t^\alpha(1-t)^{1-\alpha},
\]
and
\[
\big|\widehat{Q}_\alpha(\Phi_1\|\Psi_1)-\widehat{Q}_\alpha(\Phi_2\|\Psi_2)\big|\leq\|\Omega^*(\mathbf{1}_{\mathcal{Y}})\|_\infty\Big((4-\alpha)\,\epsilon^{1-\alpha}+(\alpha-1)\,\epsilon^{-\alpha}\Big)\,\delta .
\]
If moreover $\Phi_1,\Psi_1,\Phi_2,\Psi_2$ are channels, then $|\widehat{D}_\alpha(\Phi_1\|\Psi_1)-\widehat{D}_\alpha(\Phi_2\|\Psi_2)|$ is at most $\frac{1}{\alpha-1}$ times the right-hand side.
\end{corollary}

\begin{proof}
The function $f(t)=t^\alpha$ is operator convex on $(0,\infty)$ for $\alpha\in[1,2]$ \cite[Chapter~V]{bhatia1997matrix}, with $f(0^+)=0$, and $h_f=h_\alpha$ is given by the formula \eqref{eq:h alpha}, now considered on $[0,1)$. The definition of $\widehat{Q}_\alpha(\rho\|\xi)$ regularises only $\xi$, but when $\range(\rho)\subseteq\range(\xi)$ the direct-sum argument in the proof of Lemma~\ref{lem:unb basic}(3) shows that it also equals $\tr\mathsf{P}_f(\rho,\xi)$, since $0^\alpha=0$. By the proof of Theorem~\ref{thm:unbounded}(a), $\range(\Phi_{i,\rho})\subseteq\range(\Psi_{i,\rho})$ for every $\rho$, so $\widehat{Q}_\alpha(\Phi_i\|\Psi_i)=\widehat{D}_f(\Phi_i\|\Psi_i)$, and Theorem~\ref{thm:unbounded}(a) gives the closed form.

For the estimate, put $X_i=\mathbf{1}_{\mathbb{C}^r}-T_i\geq\epsilon\mathbf{1}_{\mathbb{C}^r}$, so that $h_\alpha(T_i)=T_i^\alpha X_i^{-\beta}$ and
\[
h_\alpha(T_1)-h_\alpha(T_2)=\big(T_1^\alpha-T_2^\alpha\big)X_1^{-\beta}+T_2^\alpha\big(X_1^{-\beta}-X_2^{-\beta}\big).
\]
By Lemma~\ref{lem:unb moduli}(4) and~(2), together with $\|X_1^{-\beta}\|_\infty\leq\epsilon^{-\beta}$ and $\|T_2^\alpha\|_\infty\leq1$, the norm of this operator is at most $(3-\beta)\epsilon^{-\beta}\delta+\beta\epsilon^{-\beta-1}\delta=\big((4-\alpha)\epsilon^{1-\alpha}+(\alpha-1)\epsilon^{-\alpha}\big)\delta$, and the proof of Theorem~\ref{thm:unbounded}(b) gives the estimate. If all four maps are channels, then $\Phi_{i,\rho}$ and $\Psi_{i,\rho}$ are states, so $\widehat{Q}_\alpha(\Phi_{i,\rho}\|\Psi_{i,\rho})\geq1$ by the data-processing inequality for $\widehat{D}_\alpha$ under the trace \cite{KatariyaWilde2021}; hence $\widehat{Q}_\alpha(\Phi_i\|\Psi_i)\geq1$, and the last claim follows from $|\log x-\log y|\leq|x-y|$ for $x,y\geq1$.
\end{proof}

\begin{remark}
The same argument applies to the maximal $\chi^2$-divergence, $f(t)=(t-1)^2$, for which $h_f(t)=(2t-1)^2(1-t)^{-1}$. Writing $h_f(T_1)-h_f(T_2)=\big((2T_1-\mathbf{1}_{\mathbb{C}^r})^2-(2T_2-\mathbf{1}_{\mathbb{C}^r})^2\big)X_1^{-1}+(2T_2-\mathbf{1}_{\mathbb{C}^r})^2\big(X_1^{-1}-X_2^{-1}\big)$ with $X_i=\mathbf{1}_{\mathbb{C}^r}-T_i$, the identity $A^2-B^2=\frac12\big((A+B)(A-B)+(A-B)(A+B)\big)$ gives $\omega^{(\epsilon)}_{h_f}(\delta)\leq(4\epsilon^{-1}+\epsilon^{-2})\delta$. By Corollary~\ref{cor:kernels}(b), this improves to the optimal $\omega^{(\epsilon)}_{h_f}(\delta)\leq(\epsilon^{-2}-4)\delta$ for $\epsilon\leq\frac13$.
\end{remark}

\begin{remark}
The hypothesis $D_{\max}\leq\log\frac{1-\epsilon}{\epsilon}$ cannot simply be dropped when $f'(\infty)=+\infty$, since then $\widehat{D}_f(A\|B)=+\infty$ whenever $\range(A)\not\subseteq\range(B)$ \cite{HiaiMosonyi2017}. In the example of Corollary~\ref{impossibility corollary}, the derivative $T_1$ of $\Phi_1$ is a projection, so $D_{\max}(\Phi_1\|\Psi_1)=+\infty$ by Lemma~\ref{lem:dmax}, and $\widehat{D}_f(\Phi_1\|\Psi_1)=+\infty$, whereas $\widehat{D}_f(\Phi_2\|\Psi_2)=f(1)$ by Theorem~\ref{thm:unbounded}(a) with $T=\frac12\mathbf{1}_{\mathbb{C}^r}$. The constants in Corollaries~\ref{cor:BS} and~\ref{cor:renyi large alpha} blow up as $\epsilon\downarrow0$, that is, as the max-relative entropy within the pairs grows, and the impossibility example is the limiting case. If instead $f(0^+)=+\infty$ and $f'(\infty)<\infty$, the transpose $\tilde f(t)=tf(1/t)$ is operator convex with $\tilde f(0^+)=f'(\infty)$, and $\widehat{D}_f(A\|B)=\widehat{D}_{\tilde f}(B\|A)$ \cite{HiaiMosonyi2017}; the results of this section then apply with the roles of $\Phi$ and $\Psi$ exchanged, under the hypothesis $T\geq\epsilon\mathbf{1}_{\mathbb{C}^r}$.
\end{remark}

\section{All operator convex functions}\label{sec:all f}

Section~\ref{sec:unbounded} assumed $f(0^+)<\infty$, and the remark at its end treated the mirror case $f'(\infty)<\infty$. Several natural divergences have $f(0^+)=f'(\infty)=+\infty$. Examples are the Jeffreys divergence, with $f(t)=(t-1)\log t$; the symmetrised $\chi^2$-divergence, with $f(t)=(t-1)^2(1+t^{-1})$; and $f(t)=t^{\alpha}+t^{-\beta}$ with $\alpha\in(1,2]$, $\beta\in(0,1]$. In this section we assume only that $f$ is operator convex. The closed form of Theorem~\ref{thm:unbounded} then holds verbatim, with $h_f$ extended to $[0,1]$ by the boundary values $f(0^+)$ and $f'(\infty)$. It turns out, as we  will show, that finiteness is the key to determining when values of $\infty$ are attained.

\subsection{The Lebesgue decomposition of a pair}

For $A,B\in\mathrm{Pos}(\mathcal{V})$, Ando's Lebesgue decomposition \cite{Ando1976} writes $A=A_{\rm ac}+A_{\rm s}$. Here $A_{\rm ac}$ is the largest positive operator below $A$ that is dominated by a multiple of $B$; in finite dimensions it is the shorted operator of $A$ to $\range(B)$ \cite{Anderson1971,AndersonTrapp1975}, and $\range(A_{\rm s})\cap\range(B)=\{0\}$. For completely positive maps we apply this to Choi matrices. Okayasu has recently considered the Lebesgue decomposition in \cite{Okayasu2026}.

\begin{proposition}\label{prop:lebesgue cp}
Let $\Omega$, $V_1,\dots,V_r$ and $W$ be as in Theorem~\ref{fundamental principle}, let $0\leq T\leq\mathbf{1}_{\mathbb{C}^r}$, $\Phi=\Psi_T$ and $\Psi=\Psi_{\mathbf{1}_{\mathbb{C}^r}-T}$, and let $P_0$ and $P_1$ be the spectral projections of $T$ for the eigenvalues $0$ and $1$ (zero if these are not eigenvalues). Relative to $J(\Psi)$, the Lebesgue decomposition of $J(\Phi)$ is $J(\Phi)_{\rm ac}=W\,T(\mathbf{1}_{\mathbb{C}^r}-P_1)\,W^*$ and $J(\Phi)_{\rm s}=WP_1W^*$. Symmetrically, relative to $J(\Phi)$, $J(\Psi)_{\rm s}=WP_0W^*$. We write $\Phi_{\rm s}=\Psi_{P_1}$ and $\Psi_{\rm s}=\Psi_{P_0}$. In particular, the following are equivalent: $\Phi_{\rm s}=0$; $1\notin\mathrm{spec}(T)$; $\range J(\Phi)\subseteq\range J(\Psi)$; and $D_{\max}(\Phi\|\Psi)<\infty$.
\end{proposition}

\begin{proof}
We repeat a similar argument to the Douglas factorisation used to prove Theorem~\ref{thm:arveson}.
If $0\leq C\leq WTW^*$, then $\range C\subseteq\range W$, so $C=WC'W^*$ with $C'=W^+C(W^+)^*$ and $0\leq C'\leq T$. Since $W$ is injective, the short of $J(\Phi)$ to $\range J(\Psi)=W\range(\mathbf{1}_{\mathbb{C}^r}-T)$ is $W$ times the short of $T$ to $\range(\mathbf{1}_{\mathbb{C}^r}-T)=\range(\mathbf{1}_{\mathbb{C}^r}-P_1)$, times $W^*$. Write $P_1^\perp=\mathbf{1}_{\mathbb{C}^r}-P_1$; it commutes with $T$. If $0\leq C'\leq T$ and $\range C'\subseteq\range P_1^\perp$, then $C'=P_1^\perp C'P_1^\perp\leq P_1^\perp TP_1^\perp=TP_1^\perp$. Moreover $TP_1^\perp\leq T$, because $T-TP_1^\perp=TP_1\geq0$. So the short is $TP_1^\perp$, and $J(\Phi)_{\rm s}=WTP_1W^*=WP_1W^*$. The last equivalences follow from \[\range J(\Phi)\subseteq\range J(\Psi)\iff\range T\subseteq\range(\mathbf{1}_{\mathbb{C}^r}-T)\iff P_1=0\] and from Lemma~\ref{lem:dmax}.
\end{proof}

\subsection{Boundary values}

\begin{lemma}\label{lem:hf extended}
Let $f:(0,\infty)\to\mathbb{R}$ be operator convex. The limits $f(0^+)=\lim_{t\downarrow0}f(t)$ and $f'(\infty)=\lim_{t\to\infty}f(t)/t$ exist in $(-\infty,+\infty]$. The function $h_f(t)=(1-t)f\big(\frac{t}{1-t}\big)$ on $(0,1)$, extended by $h_f(0)=f(0^+)$ and $h_f(1)=f'(\infty)$, is convex and continuous as a map $[0,1]\to(-\infty,+\infty]$. In particular $m_f=\min_{[0,1]}h_f$ is finite.
\end{lemma}

\begin{proof}
For $t\in(0,1)$, the convex function $f$ satisfies $f(t)\geq f(1)+f'_-(1)(t-1)$, so $\lim_{t\downarrow0}f(t)$ exists in $(-\infty,+\infty]$. The chord slopes $(f(t)-f(1))/(t-1)$ increase in $t$, so their limit as $t\to\infty$, which is also the limit of $f(t)/t$, exists in $(-\infty,+\infty]$. On $(0,1)$ we have $h_f(t)=g(t,1-t)$, where $g(x,y)=yf(x/y)$ is the perspective of $f$, which is jointly convex on $(0,\infty)^2$; hence $h_f$ is convex on $(0,1)$. As $t\downarrow0$, $h_f(t)\to f(0^+)$. As $t\uparrow1$, $h_f(t)=t\,f(x)/x$ with $x=t/(1-t)\to\infty$, so $h_f(t)\to f'(\infty)$. A convex function on $(0,1)$ with limits in $(-\infty,+\infty]$ at both ends is bounded below and extends continuously to $[0,1]$ with values in $(-\infty,+\infty]$.
\end{proof}

For Hermitian $T$ with spectrum in $[0,1]$, we say that $h_f(T)$ is finite if $h_f$ is finite on $\mathrm{spec}(T)$; then $h_f(T)$ is defined by the functional calculus. By Proposition~\ref{prop:lebesgue cp}, $h_f(T)$ is finite exactly when $f'(\infty)<\infty$ or $\Phi_{\rm s}=0$, and $f(0^+)<\infty$ or $\Psi_{\rm s}=0$.

\subsection{The closed form}
We want to derive a closed form for all operator convex functions with appropriate conditions. Lemma~\ref{lem:unb basic}(3) says that we have a closed form when $\range(A)\subseteq \range(B)$. This is not always true, but we can perform congruence transformations to fix this issue as the below lemma shows.
\begin{lemma}\label{lem:equal range}
Let $f$ be operator convex. Parts~(1) and~(2) of Lemma~\ref{lem:unb basic} hold for $f$, with the same proofs. If $A,B\in\mathrm{Pos}(\mathcal{V})$ satisfy $\range(A)=\range(B)$, then, with $U,A_0,B_0$ as in Definition~\ref{def:unb perspective}, both $A_0$ and $B_0$ are invertible and $\widehat{D}_f(A\|B)=\tr\mathsf{P}_f(A,B)$. If moreover $C$ is injective on $\range(B)$, then $\widehat{D}_f(CAC^*\|CBC^*)=\tr\big(C\,\mathsf{P}_f(A,B)\,C^*\big)$.
\end{lemma}

\begin{proof}
Parts~(1) and~(2) of Lemma~\ref{lem:unb basic} use only the joint operator convexity and positive homogeneity of $\mathsf{P}_f$ on positive definite pairs and the continuity of $\varphi_\eta$, which hold for every operator convex $f$. For the second statement we follow the proof of part~(3): now $X_0=B_0^{-1/2}A_0B_0^{-1/2}$ is invertible, so for small $\eta$ the spectra of $X_\eta$ lie in a fixed compact interval $I_0\subseteq[\frac12\lambda_{\min}(X_0),\infty)$, where $f$ is continuous, and $\mathsf{P}_f(A_0+\eta\mathbf{1}_{\mathbb{C}^m},B_0+\eta\mathbf{1}_{\mathbb{C}^m})\to\mathsf{P}_f(A_0,B_0)$. For the last statement, $CAC^*$ and $CBC^*$ again have equal ranges. The proofs of Lemmas~\ref{lem:unb congruence pd} and~\ref{lem:unb congruence} use $f$ only on the spectrum of $B_0^{-1/2}A_0B_0^{-1/2}$, which is now contained in $(0,\infty)$, so $\mathsf{P}_f(CAC^*,CBC^*)=C\,\mathsf{P}_f(A,B)\,C^*$, and the second statement applies.
\end{proof}

Since $T$ is positive semidefinite, and congruence transformations are injective, we obtain a specific closed form for the above by applying the positive semidefinite functional calculus.  

\begin{lemma}\label{lem:mixing}
Let $f$ be operator convex, let $\widetilde{W}\in\mathrm{L}(\mathbb{C}^r,\mathcal{V})$ be injective, and let $0\leq T\leq\mathbf{1}_{\mathbb{C}^r}$ have spectral decomposition $T=\sum_\lambda\lambda P_\lambda$. Then $\tr(\widetilde{W}P_\lambda \widetilde{W}^*)>0$ for every $\lambda\in\mathrm{spec}(T)$, and
\[
\widehat{D}_f\big(\widetilde{W}T\widetilde{W}^*\,\big\|\,\widetilde{W}(\mathbf{1}_{\mathbb{C}^r}-T)\widetilde{W}^*\big)=\sum_{\lambda\in\mathrm{spec}(T)}h_f(\lambda)\,\tr(\widetilde{W}P_\lambda \widetilde{W}^*),
\]
with values in $(-\infty,+\infty]$.
In particular the left-hand side is finite if and only if $h_f(T)$ is finite, and then it equals $\tr\big(\widetilde{W}\,h_f(T)\,\widetilde{W}^*\big)$.
\end{lemma}

\begin{proof}
Since $\widetilde{W}$ is injective, $\tr(\widetilde{W}P_\lambda \widetilde{W}^*)=\|\widetilde{W}P_\lambda\|_2^2>0$. Put $A=\widetilde{W}T\widetilde{W}^*$, $B=\widetilde{W}(\mathbf{1}_{\mathbb{C}^r}-T)\widetilde{W}^*$ and, for $\theta\in(0,1]$, $T_\theta=(1-\theta)T+\frac\theta2\mathbf{1}_{\mathbb{C}^r}$. Then $T_\theta$ and $\mathbf{1}_{\mathbb{C}^r}-T_\theta$ are invertible, so $A_\theta=\widetilde{W}T_\theta \widetilde{W}^*$ and $B_\theta=\widetilde{W}(\mathbf{1}_{\mathbb{C}^r}-T_\theta)\widetilde{W}^*$ both have range $\range \widetilde{W}$. Lemma~\ref{lem:equal range}, applied with $C=\widetilde{W}$ to the invertible pair $T_\theta$, $\mathbf{1}_{\mathbb{C}^r}-T_\theta$, and the commuting positive semidefinite calculus as in the proof of Theorem~\ref{thm:unbounded}(a), give
\[
\widehat{D}_f(A_\theta\|B_\theta)=\tr\big(\widetilde{W}\,\mathsf{P}_f(T_\theta,\mathbf{1}_{\mathbb{C}^r}-T_\theta)\,\widetilde{W}^*\big)=\tr\big(\widetilde{W}\,h_f(T_\theta)\,\widetilde{W}^*\big)=\sum_\lambda h_f\big((1-\theta)\lambda+\tfrac\theta2\big)\tr(\widetilde{W}P_\lambda \widetilde{W}^*).
\]
As $\theta\downarrow0$, the right-hand side tends to $\sum_\lambda h_f(\lambda)\tr(\widetilde{W}P_\lambda \widetilde{W}^*)$ in $(-\infty,+\infty]$, by Lemma~\ref{lem:hf extended}. On the other hand, $(A_\theta,B_\theta)=(1-\theta)(A,B)+\theta(N,N)$ with $N=\frac12\widetilde{W}\widetilde{W}^*$, and $\widehat{D}_f(N\|N)=f(1)\tr N$ because $\mathsf{P}_f(X,X)=f(1)X$ for $X>0$. Joint convexity (Lemma~\ref{lem:unb basic}(2), valid by Lemma~\ref{lem:equal range}) gives $\widehat{D}_f(A_\theta\|B_\theta)\leq(1-\theta)\widehat{D}_f(A\|B)+\theta f(1)\tr N$, and lower semicontinuity gives $\liminf_{\theta\downarrow0}\widehat{D}_f(A_\theta\|B_\theta)\geq\widehat{D}_f(A\|B)$. Hence $\widehat{D}_f(A_\theta\|B_\theta)\to\widehat{D}_f(A\|B)$ in $(-\infty,+\infty]$, and the two limits agree.
\end{proof}
We now combine the previous two lemmas in the following formula. 
\begin{theorem}\label{thm:all f}
Let $f:(0,\infty)\to\mathbb{R}$ be operator convex, let $\Omega$, $V_1,\dots,V_r$ and $W$ be as in Theorem~\ref{fundamental principle}, let $0\leq T\leq\mathbf{1}_{\mathbb{C}^r}$, and put $\Phi=\Psi_T$ and $\Psi=\Psi_{\mathbf{1}_{\mathbb{C}^r}-T}$.
\begin{enumerate}
\item[(a)] $\widehat{D}_f(\Phi\|\Psi)<\infty$ if and only if $h_f(T)$ is finite, that is, if and only if
\[
\big[f'(\infty)<\infty\ \text{or}\ \Phi_{\rm s}=0\big]\quad\text{and}\quad\big[f(0^+)<\infty\ \text{or}\ \Psi_{\rm s}=0\big].
\]
\item[(b)] In that case $\widehat{D}_f(\Phi\|\Psi)=\lambda_{\max}\Big(\Tr_{\mathcal{Y}}\big(W\,h_f(T)\,W^*\big)\Big)$.
\item[(c)] In that case, with $P_{\rm a}=\mathbf{1}_{\mathbb{C}^r}-P_0-P_1$, $\Phi_{\rm ac}=\Psi_{TP_{\rm a}}$ and $\Psi_{\rm ac}=\Psi_{(\mathbf{1}_{\mathbb{C}^r}-T)P_{\rm a}}$,
\[
\Tr_{\mathcal{Y}}\big(W\,h_f(T)\,W^*\big)=\Tr_{\mathcal{Y}}\mathsf{P}_f\big(J(\Phi_{\rm ac}),J(\Psi_{\rm ac})\big)+f'(\infty)\Tr_{\mathcal{Y}}J(\Phi_{\rm s})+f(0^+)\Tr_{\mathcal{Y}}J(\Psi_{\rm s}),
\]
where a term with an infinite coefficient is absent because the corresponding singular part vanishes.
\end{enumerate}
\end{theorem}

\begin{proof}
We prove (a) and (b) together in a very similar manner to the Kubo--Ando case.

Let $u=\vectorise(A)$ with $\|A\|_2=1$ and $A$ invertible, put \[C=\mathbf{1}_{\mathcal{Y}}\otimes A^\top\] and $\widetilde{W}=CW$, which is injective. By \eqref{eq:vec identity} and Theorem~\ref{thm:arveson}, $\Phi_{uu^*}=\widetilde{W}T\widetilde{W}^*$ and \[\Psi_{uu^*}=\widetilde{W}(\mathbf{1}_{\mathbb{C}^r}-T)\widetilde{W}^*\] If $h_f(T)$ is not finite, Lemma~\ref{lem:mixing} gives $\widehat{D}_f(\Phi_{uu^*}\|\Psi_{uu^*})=+\infty$, hence $\widehat{D}_f(\Phi\|\Psi)=+\infty$. If $h_f(T)$ is finite, Lemma~\ref{lem:mixing} gives \[\widehat{D}_f(\Phi_{uu^*}\|\Psi_{uu^*})=\tr(\widetilde{W}\,h_f(T)\,\widetilde{W}^*)=\tr\big((AA^*)^\top M\big)\] with $M=\Tr_{\mathcal{Y}}(W\,h_f(T)\,W^*)$, as in the proof of Theorem~\ref{thm:unbounded}(a). The rest of that proof uses only the convexity and lower semicontinuity of $\rho\mapsto\widehat{D}_f(\Phi_\rho\|\Psi_\rho)$, which hold for every operator convex $f$ by Lemma~\ref{lem:equal range}, and gives $\widehat{D}_f(\Phi\|\Psi)=\lambda_{\max}(M)$. The reformulation in (a) is Proposition~\ref{prop:lebesgue cp}.

(c) \[h_f(T)=h_f(T)P_{\rm a}+f'(\infty)P_1+f(0^+)P_0\] and $J(\Phi_{\rm s})=WP_1W^*$, $J(\Psi_{\rm s})=WP_0W^*$. On $\range P_{\rm a}$ the operators $T$ and $\mathbf{1}_{\mathbb{C}^r}-T$ are invertible. Lemma~\ref{lem:equal range}, with $C=W$ restricted to $\range P_{\rm a}$, gives $\mathsf{P}_f(WTP_{\rm a}W^*,W(\mathbf{1}_{\mathbb{C}^r}-T)P_{\rm a}W^*)=W\,h_f(T)P_{\rm a}\,W^*$.
\end{proof}

Part~(c) is the channel analogue of the classical formula \cite{LieseVajda2006}
\[
D_f(\mathbb{P}\|\mathbb{Q})=\int f\Big(\frac{d\mathbb{P}_{\rm ac}}{d\mathbb{Q}}\Big)d\mathbb{Q}+f'(\infty)\mathbb{P}_{\rm s}(\mathcal{S})+f(0)\mathbb{Q}_{\rm s}(\mathcal{S}).
\] Theorem~\ref{thm:all f} contains Theorem~\ref{thm:maxf}(a), where both boundary values are finite and $f(1)=0$, and Theorem~\ref{thm:unbounded}(a). It also contains the mirror case discussed after Corollary~\ref{cor:renyi large alpha}. When both boundary values are infinite, a case not covered by those results, part~(a) requires $\Phi_{\rm s}=\Psi_{\rm s}=0$, that is, $\range J(\Phi)=\range J(\Psi)$. The operator-valued counterpart of part~(a) is \cite[Proposition~4.11]{Okayasu2026maxf}, which characterises when the divergence of \cite[Definition~3.1]{Okayasu2026maxf} has dense domain and gives sufficient conditions for it to be bounded; when $f(0^+)=f'(\infty)=+\infty$, a dense domain requires $\Phi$ and $\Psi$ to be mutually absolutely continuous \cite[Definition~4.10]{Okayasu2026maxf}, and boundedness holds if each is dominated by a multiple of the other, which in our setting is again the condition $\range J(\Phi)=\range J(\Psi)$, by Proposition~\ref{prop:lebesgue cp}. In finite dimensions the two are directly related. In the Stinespring representation $Y\mapsto S^*(Y\otimes\mathbf{1}_{\mathcal{Z}})S$ of $\Omega^*$ given by Theorem~\ref{effect theorem}, the Radon--Nikodym derivatives of $\Phi^*$ and $\Psi^*$ are $\mathbf{1}_{\mathcal{Y}}\otimes T^\top$ and $\mathbf{1}_{\mathcal{Y}}\otimes(\mathbf{1}_{\mathbb{C}^r}-T)^\top$, so, when $h_f(T)$ is finite, the divergence of \cite[Definition~3.1]{Okayasu2026maxf} of the adjoint maps is $S^*(\mathbf{1}_{\mathcal{Y}}\otimes h_f(T)^\top)S=\Psi_{h_f(T)}^*(\mathbf{1}_{\mathcal{Y}})$, the transpose of $\Tr_{\mathcal{Y}}(W\,h_f(T)\,W^*)$. Part~(b) thus identifies $\widehat{D}_f(\Phi\|\Psi)$ with the largest eigenvalue of Okayasu's operator-valued divergence of $\Phi^*$ and $\Psi^*$; to our knowledge this scalar identification, and the closed form when both boundary values are infinite, are new. Quantitatively, $\epsilon\mathbf{1}_{\mathbb{C}^r}\leq T\leq(1-\epsilon)\mathbf{1}_{\mathbb{C}^r}$ is equivalent to $\max\{D_{\max}(\Phi\|\Psi),D_{\max}(\Psi\|\Phi)\}\leq\log\frac{1-\epsilon}{\epsilon}$ by Lemma~\ref{lem:dmax}, so the natural parameter for such $f$ is the Thompson distance between $\Phi$ and $\Psi$.

\subsection{Continuity in the derivative}
We now show continuity results in this more general setting described in the previous theorem. 
\begin{theorem}\label{thm:extended continuity}
Let $f$ be operator convex and let $\Omega$ and $W$ be as in Theorem~\ref{fundamental principle}. Write $\mathcal{D}(T)=\widehat{D}_f(\Psi_T\|\Psi_{\mathbf{1}_{\mathbb{C}^r}-T})$ for $0\leq T\leq\mathbf{1}_{\mathbb{C}^r}$.
\begin{enumerate}
\item[(a)] Let $I\subseteq[0,1]$ be a compact interval on which $h_f$ is finite. If $\mathrm{spec}(T_1),\mathrm{spec}(T_2)\subseteq I$, then
\[
|\mathcal{D}(T_1)-\mathcal{D}(T_2)|\leq\|\Omega^*(\mathbf{1}_{\mathcal{Y}})\|_\infty\,\omega^I_{h_f}(\|T_1-T_2\|_\infty).
\]
\item[(b)] For every $T$ and every $\lambda\in\mathrm{spec}(T)$,
\[
\mathcal{D}(T)\geq\nu(\Omega)^2\big(h_f(\lambda)-m_f\big)+\min\{m_f,0\}\,\|\Omega^*(\mathbf{1}_{\mathcal{Y}})\|_\infty .
\]
\item[(c)] The map $T\mapsto\mathcal{D}(T)$ is continuous from $\{0\leq T\leq\mathbf{1}_{\mathbb{C}^r}\}$ to $(-\infty,+\infty]$.
\end{enumerate}
\end{theorem}

\begin{proof}
(a) By Theorem~\ref{thm:all f}(b) and Lemma~\ref{lem:partial trace}, as in the proof of Theorem~\ref{thm:unbounded}(b).

(b) We may assume $\mathcal{D}(T)<\infty$. Let $v$ be a unit eigenvector of $T$ for $\lambda$. Then \[h_f(T)\geq m_f\mathbf{1}_{\mathbb{C}^r}+(h_f(\lambda)-m_f)vv^*\] so, since conjugation by $W$ and the partial trace are positive,
\[
M=\Tr_{\mathcal{Y}}(W\,h_f(T)\,W^*)\geq m_f\Tr_{\mathcal{Y}}(WW^*)+(h_f(\lambda)-m_f)\Tr_{\mathcal{Y}}(Wvv^*W^*).
\]
With $B=\sum_kv_kV_k$ we have $Wv=\vectorise(B)$ and $\Tr_{\mathcal{Y}}(\vectorise(B)\vectorise(B)^*)=(B^*B)^\top$, whose largest eigenvalue is $\|B\|_\infty^2\geq\nu(\Omega)^2$. Moreover $\Tr_{\mathcal{Y}}(WW^*)=\Omega^*(\mathbf{1}_{\mathcal{Y}})^\top$. The inequality \[\lambda_{\max}(X+Y)\geq\lambda_{\max}(X)+\lambda_{\min}(Y)\] and Theorem~\ref{thm:all f}(b) give the claim.

(c) Let $\mathcal{F}_f=\{t\in[0,1]:h_f(t)<\infty\}$, which is $[0,1]$ minus at most the endpoints. If $h_f(T)$ is finite, choose a compact interval $I\subseteq \mathcal{F}_f$ whose relative interior in $[0,1]$ contains $\mathrm{spec}(T)$. By Weyl's inequality, $\mathrm{spec}(T')\subseteq I$ for $T'$ close to $T$, and (a) together with Lemma~\ref{modulus to zero} gives continuity at $T$. If $h_f(T)$ is not finite, then $\mathcal{D}(T)=+\infty$, and some $\lambda\in\mathrm{spec}(T)$ is an endpoint with $h_f(\lambda)=+\infty$. Every $T'$ close to $T$ has an eigenvalue $\lambda'$ close to $\lambda$, and $h_f(\lambda')\to+\infty$ by Lemma~\ref{lem:hf extended}. So $\mathcal{D}(T')\to+\infty$ by (b), since $\nu(\Omega)>0$.
\end{proof}

Thus the ARN derivative is a coordinate in which every maximal $f$-divergence is continuous, finite or not. This stands in contrast with the completely bounded trace norm (Corollary~\ref{impossibility corollary}). Part~(b) also shows how fast the divergence grows with the max-relative entropy, since $\lambda_{\max}(T)=\vartheta(D_{\max}(\Phi\|\Psi))$ with $\vartheta(x)=e^x/(1+e^x)$ by Lemma~\ref{lem:dmax}.

\begin{corollary}\label{cor:dmax growth}
For every operator convex $f$ and all $\Phi,\Psi\in\mathrm{CP}(\mathcal{X},\mathcal{Y})$ with $\Omega=\Phi+\Psi$,
\[
\widehat{D}_f(\Phi\|\Psi)\geq\nu(\Omega)^2\,h_f\big(\vartheta(D_{\max}(\Phi\|\Psi))\big)-|m_f|\big(\nu(\Omega)^2+\|\Omega^*(\mathbf{1}_{\mathcal{Y}})\|_\infty\big),
\]
and the same holds with $h_f\big(1-\vartheta(D_{\max}(\Psi\|\Phi))\big)$. For the Belavkin--Staszewski relative entropy, $h_f(\vartheta(x))=\vartheta(x)\,x\geq x/2$ for $x\geq0$, and $|m_f|<0.28$, so $\widehat{D}_{\mathrm{BS}}(\Phi\|\Psi)\geq\frac12\nu(\Omega)^2D_{\max}(\Phi\|\Psi)-0.28\big(\nu(\Omega)^2+\|\Omega^*(\mathbf{1}_{\mathcal{Y}})\|_\infty\big)$ whenever $D_{\max}(\Phi\|\Psi)\geq0$.
\end{corollary}

\begin{proof}
Apply Theorem~\ref{thm:extended continuity}(b) with $\lambda=\lambda_{\max}(T)$, respectively $\lambda=\lambda_{\min}(T)=1-\vartheta(D_{\max}(\Psi\|\Phi))$, and use $-\nu(\Omega)^2m_f+\min\{m_f,0\}\|\Omega^*(\mathbf{1}_{\mathcal{Y}})\|_\infty\geq-|m_f|(\nu(\Omega)^2+\|\Omega^*(\mathbf{1}_{\mathcal{Y}})\|_\infty)$. For $f(t)=t\log t$, $h_f(t)=t\log\frac{t}{1-t}$, and its minimum on $[0,1]$ is $-0.2785\ldots$.
\end{proof}

\section{Continuity in the completely bounded trace norm}\label{sec:four maps}

All estimates so far are stated in terms of ARN derivatives. Fault-tolerance thresholds, however, refer to additive worst-case errors, for which the completely bounded trace (diamond) norm is the standard measure \cite{KuengLongDohertyFlammia2016}. In this section we show that, for a fixed $\Omega$, the operator norm distance between derivatives and the completely bounded trace norm distance between the corresponding submaps are equivalent, with constants that do not depend on any dimension (Theorem~\ref{thm:two sided}). All continuity estimates of the previous sections therefore hold in the completely bounded trace norm on the splittings of a fixed map, and the impossibility example of Corollary~\ref{impossibility corollary} is identified as the case in which the constant degenerates. Pairs with different sums are treated in Section~\ref{sec:different sums}.

Let $\Omega\in\mathrm{CP}(\mathcal{X},\mathcal{Y})$ be nonzero with Kraus decomposition $\Omega(X)=\sum_{k=1}^rV_kXV_k^*$, $V_1,\dots,V_r$ linearly independent, and $W=[\vectorise(V_1)\ \cdots\ \vectorise(V_r)]$. Recall from Section~\ref{sec:prelim} the quantity $\nu(\Omega)$,
the smallest operator norm of a normalised linear combination of the Kraus operators. The minimum is attained and positive, since the unit sphere is compact and the $V_k$ are linearly independent. We first record the intrinsic meaning of the operator norm of a derivative.

\begin{lemma}\label{lem:relative distance}
For Hermitian $K\in\mathrm{L}(\mathbb{C}^r)$ and every unitarily invariant norm,
\[
\|K\|_\infty=\big\|J(\Omega)^{+1/2}\,J(\Psi_K)\,J(\Omega)^{+1/2}\big\|_\infty,
\]
where $J(\Omega)^{+1/2}$ is the square root of the Moore--Penrose inverse of $J(\Omega)$. Moreover, for $d\geq0$ we have $\|K\|_\infty\leq d$ if and only if $-d\,\Omega\leq\Psi_K\leq d\,\Omega$ in the completely positive order. In particular, if $\Phi_1+\Psi_1=\Phi_2+\Psi_2=\Omega$ with derivatives $T_1,T_2$, then $\|T_1-T_2\|_\infty$ is the smallest $d\geq0$ such that $-d\,\Omega\leq\Phi_1-\Phi_2\leq d\,\Omega$.
\end{lemma}

\begin{proof}
Write $W=UG^{1/2}$ with $G=W^*W$ invertible and $U$ an isometry (polar decomposition). Then $J(\Omega)=WW^*=UGU^*$, so $J(\Omega)^{+1/2}=UG^{-1/2}U^*$, and $J(\Psi_K)=WKW^*=UG^{1/2}KG^{1/2}U^*$. Hence $J(\Omega)^{+1/2}J(\Psi_K)J(\Omega)^{+1/2}=UKU^*$, which has the same singular values as $K$. For the second statement, $d\,\Omega\pm\Psi_K=\Psi_{d\mathbf{1}_{\mathbb{C}^r}\pm K}$ has Choi matrix $W(d\mathbf{1}_{\mathbb{C}^r}\pm K)W^*$. By Choi's theorem \cite{Choi1975}, both maps are completely positive if and only if $W(d\mathbf{1}_{\mathbb{C}^r}\pm K)W^*\geq0$, and conjugating by $W^+=(W^*W)^{-1}W^*$ and by $W$ shows that this holds if and only if $d\mathbf{1}_{\mathbb{C}^r}\pm K\geq0$, that is, $\|K\|_\infty\leq d$. The last statement follows with $K=T_1-T_2$, since $\Phi_1-\Phi_2=\Psi_{T_1-T_2}$.
\end{proof}
One of the principal pains with continuity results for channel divergences is dimensional constraints. The next lemma shows that in the same-sum case, we have a natural candidate which depends on the sum being considered. The primary application of this will be control from below of quantum channel divergences which consider two different decompositions of the same completely positive map.

\begin{lemma}\label{lem:nu}
Let $\Omega$ be as above and $J=J(\Omega)$, with $\lambda^+_{\min}(J)$ its smallest nonzero eigenvalue.
\begin{enumerate}
\item $\nu(\Omega)^2=\min\big\{\|\vectorise^{-1}(w)\|_\infty^2\,/\,\langle w,J^+w\rangle:\ w\in\range(J)\setminus\{0\}\big\}$, where $J^+$ is the Moore--Penrose inverse of $J$. In particular $\nu(\Omega)$ does not depend on the Kraus decomposition.
\item $\lambda^+_{\min}(J)/\min\{\dim\mathcal{X},\dim\mathcal{Y}\}\leq\nu(\Omega)^2\leq\lambda^+_{\min}(J)$.
\item If $\Omega'\in\mathrm{CP}(\mathcal{X},\mathcal{Y})$ satisfies $\Omega'\geq\Omega$ in the completely positive order and $\range J(\Omega')=\range J(\Omega)$, then $\nu(\Omega')\geq\nu(\Omega)$.
\end{enumerate}
\end{lemma}

\begin{proof}
(1) By linearity of $\vectorise$, $\sum_kc_kV_k=\vectorise^{-1}(Wc)$. As $c$ ranges over the unit sphere of $\mathbb{C}^r$, the vector $w=Wc$ ranges over $\range(W)=\range(WW^*)=\range(J)$, and $c=W^+w$ with $W^+=(W^*W)^{-1}W^*$, so that
\[
\|c\|_2^2=\langle w,(W^+)^*W^+w\rangle=\langle w,J^+w\rangle,
\]
because $(W^+)^*W^+=W(W^*W)^{-2}W^*=(WW^*)^+$. Hence $\nu(\Omega)^2$ is the minimum of $\|\vectorise^{-1}(w)\|_\infty^2$ over $w\in\range(J)$ with $\langle w,J^+w\rangle=1$, and the formula follows by homogeneity. It involves only $J(\Omega)$.

(2) For $w\in\range(J)$ we have $\|w\|_2^2/\lambda_{\max}(J)\leq\langle w,J^+w\rangle\leq\|w\|_2^2/\lambda^+_{\min}(J)$, while
\[
\frac{\|w\|_2^2}{\min\{\dim\mathcal{X},\dim\mathcal{Y}\}}\leq\frac{\|w\|_2^2}{\rank(\vectorise^{-1}(w))}\leq\|\vectorise^{-1}(w)\|_\infty^2\leq\|\vectorise^{-1}(w)\|_2^2=\|w\|_2^2 .
\] The lower bound follows from (1). For the upper bound take $w$ to be an eigenvector of $J$ for the eigenvalue $\lambda^+_{\min}(J)$.

(3) Put $\mathcal{R}=\range J(\Omega)=\range J(\Omega')$. By Choi's theorem \cite{Choi1975}, $J(\Omega')\geq J(\Omega)$; both operators are positive definite on $\mathcal{R}$ and vanish on $\mathcal{R}^\perp$, so operator monotonicity of inversion on $\mathcal{R}$ gives $J(\Omega')^+\leq J(\Omega)^+$, hence $\langle w,J(\Omega')^+w\rangle\leq\langle w,J(\Omega)^+w\rangle$ for $w\in \mathcal{R}$, and (1) gives $\nu(\Omega')^2\geq\nu(\Omega)^2$.
\end{proof}

\begin{remark}[Physical interpretation]\label{rem:phys-nu}
Since $\Omega^c(\rho)=\sum_{k,l}\tr(V_k\rho V_l^*)\,e_ke_l^*$ (proof of Theorem~\ref{effect theorem}), $\langle u,\Omega^c(\rho)u\rangle=\tr(V_u\rho V_u^*)$ with $V_u=\sum_k\overline{u_k}V_k$, whose maximum over states $\rho$ is $\|V_u\|_\infty^2$. Hence $\nu(\Omega)^2=\min_{\|u\|_2=1}\max_\rho\langle u,\Omega^c(\rho)u\rangle$, the excitability of the least excitable environment direction. As $\Psi_{\overline{u}\,\overline{u}^*}=V_u(\cdot)V_u^*$, changing $T$ by $\overline{u}\,\overline{u}^*$ for a minimising $u$, a change of norm $1$, changes the submaps by only $\nu(\Omega)^2$ in $|||\cdot|||_1$: this is the mechanism behind Corollary~\ref{impossibility corollary}, where $\nu(\Omega)\to0$ as $\theta\to0$. Noise helps: if $\Phi$ and $\Psi$ each contain a fraction $p$ of the completely depolarising channel, whose Choi matrix is $\mathbf{1}_{\mathcal{Y}\otimes\mathcal{X}}/n_{\mathcal{Y}}$ with $n_{\mathcal{Y}}=\dim\mathcal{Y}$, then $J(\Omega)\geq(2p/n_{\mathcal{Y}})\mathbf{1}_{\mathcal{Y}\otimes\mathcal{X}}$, and Lemma~\ref{lem:nu}(2) gives $\nu(\Omega)^2\geq2p/(n_{\mathcal{Y}}\min\{\dim\mathcal{X},n_{\mathcal{Y}}\})$; in particular Theorem~\ref{thm:dichotomy}(a) applies.
\end{remark}

\begin{theorem}\label{thm:two sided}
Let $\Omega$, $V_1,\dots,V_r$ be as above. For every Hermitian $K\in\mathrm{L}(\mathbb{C}^r)$,
\[
\nu(\Omega)^2\,\|K\|_\infty\leq|||\Psi_K|||_1\leq\|\Omega^*(\mathbf{1}_{\mathcal{Y}})\|_\infty\,\|K\|_\infty .
\]
The upper bound is attained for $K=\mathbf{1}_{\mathbb{C}^r}$, and the lower bound is attained for the completely dephasing map $\Omega(X)=\sum_{i=1}^nE_{ii}XE_{ii}$ on $\mathbb{C}^n$ and $K=\frac1n\mathbf{j}\mathbf{j}^*$, where $\mathbf{j}$ is the all-ones vector, one has $\nu(\Omega)^2=\frac1n$, $\|K\|_\infty=1$ and $|||\Psi_K|||_1=\frac1n$.
\end{theorem}

\begin{proof}
The upper bound is the proof of Lemma~\ref{lipschitz inequality} with $\Phi=\Omega$, which uses only that $K$ is Hermitian, together with $|||\Omega|||_1=\|\Omega^*(\mathbf{1}_{\mathcal{Y}})\|_\infty$.

For the lower bound let $v\in\mathbb{C}^r$ be a unit eigenvector of $K$ with $Kv=\lambda v$ and $|\lambda|=\|K\|_\infty$. On the subspace $\mathcal{S}=\mathrm{span}\{V_1,\dots,V_r\}$ of $\mathrm{L}(\mathcal{X},\mathcal{Y})$ define the linear functional $\varphi_0(\sum_kc_kV_k)=\sum_kc_k\overline{v_k}$, which is well defined by linear independence. By the Cauchy--Schwarz inequality and the definition of $\nu(\Omega)$,
\[
|\varphi_0(\textstyle\sum_kc_kV_k)|\leq\|c\|_2\|v\|_2=\|c\|_2\leq\nu(\Omega)^{-1}\|\textstyle\sum_kc_kV_k\|_\infty,
\]
so $\varphi_0$ has norm at most $\nu(\Omega)^{-1}$ with respect to the operator norm. By the Hahn--Banach theorem it extends to a functional on $\mathrm{L}(\mathcal{X},\mathcal{Y})$ of the same norm, and by the duality between the operator norm and the trace norm \cite[Section~1.1.3]{watrous2018theory} there is $Z\in\mathrm{L}(\mathcal{X},\mathcal{Y})$ with $\|Z\|_1\leq\nu(\Omega)^{-1}$ and $\tr(Z^*V_k)=\overline{v_k}$ for $k=1,\dots,r$.

Let $Z=U|Z|$ be a polar decomposition and put $A=|Z|^{1/2}\in\mathrm{L}(\mathcal{X})$ and $M=U|Z|^{1/2}\in\mathrm{L}(\mathcal{X},\mathcal{Y})$, so that $Z=MA^*$, $Z^*=AM^*$ and $\|A\|_2^2=\|M\|_2^2=\tr|Z|=\|Z\|_1$. Consider the density operator $\rho=\vectorise(A)\vectorise(A)^*/\|A\|_2^2$ on $\mathcal{X}\otimes\mathcal{X}$ and the operator $N=\vectorise(M)\vectorise(M)^*/\|M\|_2^2$ on $\mathcal{Y}\otimes\mathcal{X}$, which satisfies $\|N\|_\infty=1$. Since $\vectorise(A)=(\mathbf{1}_{\mathcal{X}}\otimes A^\top)\vectorise(\mathbf{1}_{\mathcal{X}})$ and $(\mathbf{1}_{\mathcal{Y}}\otimes A^\top)\vectorise(V_k)=\vectorise(V_kA)$ \cite[Section~1.1.2]{watrous2018theory}, the identity $J(\Psi_K)=WKW^*$ from the proof of Theorem~\ref{thm:arveson} gives
\[
(\Psi_K\otimes\mathbf{1}_{\mathrm{L}(\mathcal{X})})\big(\vectorise(A)\vectorise(A)^*\big)=\sum_{k,l=1}^rK_{kl}\,\vectorise(V_kA)\vectorise(V_lA)^*,
\]
and therefore
\begin{align*}
&\tr\Big(N\,(\Psi_K\otimes\mathbf{1}_{\mathrm{L}(\mathcal{X})})\big(\vectorise(A)\vectorise(A)^*\big)\Big)\,\|M\|_2^2\\
&\qquad=\sum_{k,l}K_{kl}\,\langle\vectorise(M),\vectorise(V_kA)\rangle\,\overline{\langle\vectorise(M),\vectorise(V_lA)\rangle}.
\end{align*}
Here $\langle\vectorise(M),\vectorise(V_kA)\rangle=\tr(M^*V_kA)=\tr(V_kAM^*)=\tr(V_kZ^*)=\tr(Z^*V_k)=\overline{v_k}$, so the right-hand side equals $\sum_{k,l}K_{kl}\overline{v_k}v_l=\langle v,Kv\rangle=\lambda$. Since $\|N\|_\infty=1$, the duality between the trace norm and the operator norm gives
\[
\big\|(\Psi_K\otimes\mathbf{1}_{\mathrm{L}(\mathcal{X})})(\rho)\big\|_1\geq\frac{|\lambda|}{\|A\|_2^2\|M\|_2^2}=\frac{\|K\|_\infty}{\|Z\|_1^2}\geq\nu(\Omega)^2\|K\|_\infty ,
\]
and the left-hand side is at most $|||\Psi_K|||_1$ because $\rho$ is a density operator.

For the example, $\|\sum_ic_iE_{ii}\|_\infty=\max_i|c_i|\geq n^{-1/2}$ for unit $c$, with equality for $c=n^{-1/2}\mathbf{j}$, so $\nu(\Omega)^2=\frac1n$; and $\Psi_K(X)=\frac1n\sum_{i,j}E_{ii}XE_{jj}=\frac1nX$, so $|||\Psi_K|||_1=\frac1n$, while $K$ is a rank-one projection.
\end{proof}

\begin{remark}
By Lemma~\ref{lem:nu}(2), Theorem~\ref{thm:two sided} implies
\[
\|K\|_\infty\leq\frac{\min\{\dim\mathcal{X},\dim\mathcal{Y}\}}{\lambda^+_{\min}(J(\Omega))}\,|||\Psi_K|||_1 ,
\]
which is never weaker than the operator norm consequence of Proposition~\ref{bilipschitz}, and is attained, up to the value of $\nu(\Omega)^2$, in the dephasing example. In practice $\lambda^+_{\min}(J(\Omega))$ is computable from the Choi matrix, whereas $\nu(\Omega)^2$ is the minimum of a nonconvex function on a sphere; Lemma~\ref{lem:nu}(2) provides the computable substitute. An earlier lower bound of this kind, with the dimensional factor $(\dim\mathcal{H})^{-2}$ in place of $\nu(\Omega)^2$, is due to Raginsky \cite[eq.~(17)]{Raginsky2003}, for maps on $\mathrm{B}(\mathcal{H})$ and derivatives taken with respect to a fixed reference channel.
\end{remark}

\begin{corollary}\label{cor:cb continuity}
Let $\Phi_1+\Psi_1=\Phi_2+\Psi_2=\Omega\neq0$ with $\Phi_i,\Psi_i\in\mathrm{CP}(\mathcal{X},\mathcal{Y})$, let $T_i$ be the ARN derivative of $\Phi_i$ with respect to $\Omega$, and put $\delta=|||\Phi_1-\Phi_2|||_1/\nu(\Omega)^2=|||\Psi_1-\Psi_2|||_1/\nu(\Omega)^2$. Then $\|T_1-T_2\|_\infty\leq\delta$, and consequently
\[
\big|\widehat{F}(\Phi_1,\Psi_1)-\widehat{F}(\Phi_2,\Psi_2)\big|\leq\frac{\|\Omega^*(\mathbf{1}_{\mathcal{Y}})\|_\infty}{\nu(\Omega)}\,|||\Phi_1-\Phi_2|||_1^{1/2},
\]
while $|||\Phi_1\,\sigma\,\Psi_1-\Phi_2\,\sigma\,\Psi_2|||_1\leq\|\Omega^*(\mathbf{1}_{\mathcal{Y}})\|_\infty\,\omega_{g_\sigma}(\delta)$ for every Kubo--Ando mean $\sigma$. More generally, each of the estimates of Theorems~\ref{Modulus of continuity results}, \ref{main theorem}, \ref{thm:renyi continuity} and~\ref{thm:maxf}(b) remains valid, under its hypotheses, with $\|T_1-T_2\|_\infty$ replaced by $\delta$.
\end{corollary}

\begin{proof}
Since $\Phi_1-\Phi_2=\Psi_{T_1-T_2}$ and $\Psi_1-\Psi_2=-(\Phi_1-\Phi_2)$, Theorem~\ref{thm:two sided} gives $\|T_1-T_2\|_\infty\leq\delta$. The right-hand side of each cited estimate is a nondecreasing function of $\|T_1-T_2\|_\infty$, since the modulus $\omega_g$ is a supremum over sets that increase with $\delta$ and $\delta\mapsto\delta^{1/2}$ and $\delta\mapsto\delta^\alpha+\delta^{1-\alpha}$ are increasing. The displayed bound is Theorem~\ref{main theorem} with $\|T_1-T_2\|_\infty^{1/2}\leq\delta^{1/2}$.
\end{proof}

\section{Explicit and sharp moduli for all operator convex functions and all means}\label{sec:loewner}
 
Lemma~\ref{modulus to zero} shows that $\omega^I_g(\delta)\to0$ for every continuous $g$, but gives no rate. Explicit moduli were computed above only for individual functions: for the geometric mean in Proposition~\ref{K6}, for the weighted geometric means in Lemma~\ref{holder lemma}, and for the Belavkin--Staszewski relative entropy and the geometric R\'enyi divergences with $\alpha\in(1,2]$ in Section~\ref{sec:unbounded}. In this section we obtain explicit, dimension-free moduli for every operator convex $f$ and for every Kubo--Ando mean from their integral representations. The kernels of these representations are rational functions with a single pole. We compute their operator Lipschitz constants exactly (Theorem~\ref{thm:one pole}), and this makes the resulting constants optimal in many cases and within a factor $3$ of optimal in general (Theorem~\ref{thm:sharp lipschitz}).
 
\subsection{Operator convex functions}
 We first recall some well-known facts of operator convex functions. This is merely a gathering of facts, many of them can be found in the paper of \cite{HiaiMosonyi2017}. The theorems listed in this paper form the general backbone of the following subsections.

Every operator convex $f:(0,\infty)\to\mathbb{R}$ has a representation \cite[eq.~(2.2)]{HiaiMosonyi2017}
\begin{equation}\label{eq:loewner}
f(x)=a+b(x-1)+c(x-1)^2+\int_{[0,\infty)}\frac{(x-1)^2}{x+s}\,d\mu(s),\qquad x>0,
\end{equation}
with $a=f(1)$, $b=f'(1)$, $c\geq0$ and a positive measure $\mu$ on $[0,\infty)$ with $\int(1+s)^{-1}d\mu(s)<\infty$; conversely every such function is operator convex. Very similarly to the Kubo--Ando case, the value $c$ is the limit $c=\lim_{x\to \infty}f(x)/x^2$. We call $(a,b,c,\mu)$ the \emph{Loewner data} of $f$. Substituting $x=t/(1-t)$ gives, for $0<t<1$,
\begin{equation}\label{eq:hf kernels}
h_f(t)=a(1-t)+b(2t-1)+c\,\kappa_\infty(t)+\int_{[0,\infty)}\kappa_s(t)\,d\mu(s),
\end{equation}
where
\[
\kappa_s(t)=\frac{(2t-1)^2}{s(1-t)+t}\quad(s\geq0),\qquad\kappa_\infty(t)=\frac{(2t-1)^2}{1-t}.
\]
Since $\kappa_s(0)=1/s$ and $\kappa_s(1)=1$, monotone convergence gives $f(0^+)=a-b+c+\int s^{-1}d\mu(s)$ and $f'(\infty)=b+\mu([0,\infty))$ if $c=0$, while $f'(\infty)=+\infty$ if $c>0$. Thus $f$ is bounded in the sense of Definition~\ref{def:maxf} if and only if $c=0$, $\int_{[0,1]}s^{-1}d\mu(s)<\infty$ (in particular $\mu(\{0\})=0$) and $\mu$ is finite. The affine part of \eqref{eq:hf kernels} has slope $2b-a$. Every kernel $\kappa_s$, $s\in[0,\infty]$, is convex on its domain, vanishes to second order at $t=\frac12$, and for $s\neq1$ has a single pole, at $t=s/(s-1)\notin(0,1)$ (at $t=1$ for $s=\infty$), while $\kappa_1(t)=(2t-1)^2$.
 
\subsection{Operator Lipschitz constants and Schur multipliers}
 We will produce concrete, explicit bounds in a similar manner to the methods of Section~\ref{sec:general}. We will use continuity constants in a similar manner, and use that these have to converge to zero when positive semidefinite operators have close operator norms.

We first recall some useful facts which can be found in any foundational textbooks on matrix analysis  (e.g. \cite{hiaipetz2014intro}, \cite{bhatia1997matrix}). 

\begin{definition}[{\cite[Section~1.6]{bhatia1997matrix}}]
    The Schur product or Hadamard product of two matrices $A$ and $B$ is
defined to be the matrix $A \circ B$ whose $(i,j)$-entry is
$A_{ij}B_{ij}$. This has the following significant properties.
\begin{enumerate}
    \item $\|A \circ B\|_\infty \leq \|A\|_\infty\|B\|_\infty$ for all $A,B$.
    \item If $A,B \geq 0$, then $A \circ B \geq 0$.
\end{enumerate}
\end{definition}

\begin{theorem}[Daleckii--Krein formula {\cite[Theorem~3.25]{hiaipetz2014intro}}]
Assume that $f\colon (t_-,t_+)\to\mathbb{R}$ is a $C^1$-function and
$A=\mathrm{diag}(t_1,t_2,\dots,t_n)$ with $t_-<t_i<t_+$ $(1\leq i\leq n)$.
If $B=B^*$, then the derivative $t\mapsto f(A+tB)$ is a Schur product:
\begin{equation}\label{eq:3.16}
  \left.\frac{d}{dt}\, f(A+tB)\right|_{t=0} = M\circ B,
\end{equation}
where $M$ is the divided difference matrix:
\[
  M_{ij}=
  \begin{cases}
    \dfrac{f(t_i)-f(t_j)}{t_i-t_j} & \text{if } t_i-t_j\neq 0,\\[2ex]
    f'(t_i) & \text{if } t_i-t_j=0.
  \end{cases}
\]
\end{theorem}

For a compact interval $I$ and a function $g$ that is continuously differentiable on a neighbourhood of $I$, put
\[
\mathrm{Lip}^I(g)=\sup\Big\{\frac{\|g(X)-g(Y)\|_\infty}{\|X-Y\|_\infty}:\ X\neq Y\in\mathrm{L}(\mathcal{V})\text{ Hermitian with spectra in }I\Big\},
\]
the supremum being over all finite dimensions, so that $\omega^I_g(\delta)\leq\mathrm{Lip}^I(g)\,\delta$. For a finite set $\Lambda\subseteq I$ let $g^{[1]}_\Lambda=[g^{[1]}(\lambda,\lambda')]_{\lambda,\lambda'\in\Lambda}$ be the divided difference matrix, with $g^{[1]}(\lambda,\lambda)=g'(\lambda)$. For a real matrix $M$ let $\|M\|_{\mathrm S}=\sup\{\|M\circ Z\|_\infty:\|Z\|_\infty\leq1\}$ be the norm of the Schur multiplier $Z\mapsto M\circ Z$.

We can express the Lipschitz continuity modulus in terms of the divided difference matrix.

\begin{lemma}\label{lem:DK}
$\mathrm{Lip}^I(g)=\sup_\Lambda\|g^{[1]}_\Lambda\|_{\mathrm S}$, the supremum being over finite subsets $\Lambda\subseteq I$. In particular $\mathrm{Lip}^I(g)\geq\max_I|g'|$.
\end{lemma}

\begin{remark}[Schur products and multipliers]\label{rem:schur}
The following lemma and proof uses four standard facts. First, the Schur product depends on the choice of basis, but operator norms do not: if $\Xi$ is a linear map on matrices and $U$ is unitary, then $X\mapsto U^*XU$ is an isometric bijection for the operator norm, so the induced norm of $\Xi$ is the same whether $\Xi$ is written in the standard basis or in the basis given by the columns of $U$. Second (Grothendieck--Haagerup), $\|M\|_{\mathrm S}\leq C$ if and only if there are vectors $x_i,y_j$ in a Hilbert space with $M_{ij}=\langle x_i,y_j\rangle$ and $\sup_i\|x_i\|\sup_j\|y_j\|\leq C$, and $\|M\|_{\mathrm S}$ is the least such $C$. Third, if $M'_{ij}=M_{\iota_1(i)\iota_2(j)}$ for surjective maps $\iota_1,\iota_2$, so that $M'$ is obtained from $M$ by repeating rows and columns, then $\|M'\|_{\mathrm S}=\|M\|_{\mathrm S}$. Indeed, a factorization $M_{ij}=\langle x_i,y_j\rangle$ gives $M'_{ij}=\langle x_{\iota_1(i)},y_{\iota_2(j)}\rangle$ with the same bound, so $\|M'\|_{\mathrm S}\leq\|M\|_{\mathrm S}$. Conversely, choose one index in each fibre of $\iota_1$ and of $\iota_2$; for $\|Z\|_\infty\leq1$, let $Z'$ be $Z$ placed in the chosen rows and columns and zero elsewhere, so that $\|Z'\|_\infty=\|Z\|_\infty$ and $M'\circ Z'$ is $M\circ Z$ placed in the same positions, whence $\|M\|_{\mathrm S}\leq\|M'\|_{\mathrm S}$. Fourth, if $M$ is real symmetric and $H$ is Hermitian, then $M\circ H$ is Hermitian, since $\overline{M_{ji}H_{ji}}=M_{ij}H_{ij}$.
\end{remark}

\begin{proof}[Proof of Lemma~\ref{lem:DK}]
Let $Z=\sum_{\lambda}\lambda P_\lambda$ be Hermitian with spectrum $\Lambda\subseteq I$. The Daleckii--Krein formula \cite[Theorem~V.3.3]{bhatia1997matrix} (see also \cite[Theorem~3.25]{hiaipetz2014intro}) gives
\[
  Dg(Z)[H]=\sum_{\lambda,\lambda'}g^{[1]}(\lambda,\lambda')\,P_\lambda HP_{\lambda'}.
\]
Choose an orthonormal basis adapted to the $P_\lambda$, and let $\lambda_i$ be the eigenvalue of $Z$ belonging to the $i$-th basis vector. The $(i,j)$ entry of $Dg(Z)[H]$ is then $g^{[1]}(\lambda_i,\lambda_j)H_{ij}$, with the convention $g^{[1]}(\lambda,\lambda)=g'(\lambda)$, that is,
\[
  Dg(Z)[H]=M\circ H,\qquad M_{ij}=g^{[1]}(\lambda_i,\lambda_j).
\]
The matrix $M$ is obtained from $g^{[1]}_\Lambda=\big(g^{[1]}(\lambda,\lambda')\big)_{\lambda,\lambda'\in\Lambda}$ by repeating each row and column according to the multiplicity of the corresponding eigenvalue. By the repetition and basis-independence facts of Remark~\ref{rem:schur}, $\|M\|_{\mathrm S}=\|g^{[1]}_\Lambda\|_{\mathrm S}$ and the choice of adapted basis is irrelevant, so
\[
  \|Dg(Z)\|\leq\|g^{[1]}_\Lambda\|_{\mathrm S}.
\]
The set of Hermitian operators with spectra in $I$ is convex, so for $X,Y$ in this set the segment $Y+u(X-Y)$, $u\in[0,1]$, stays in it, and
\[
  g(X)-g(Y)=\int_0^1Dg\big(Y+u(X-Y)\big)[X-Y]\,du.
\]
Therefore $\|g(X)-g(Y)\|_\infty\leq\sup_\Lambda\|g^{[1]}_\Lambda\|_{\mathrm S}\,\|X-Y\|_\infty$, which gives $\mathrm{Lip}^I(g)\leq\sup_\Lambda\|g^{[1]}_\Lambda\|_{\mathrm S}$.

For the reverse inequality, let $\Lambda=\{\lambda_1,\dots,\lambda_n\}$ lie in the interior of $I$, and put $Z=\mathrm{diag}(\lambda_1,\dots,\lambda_n)$ and $M=g^{[1]}_\Lambda$, which is real symmetric. For Hermitian $H$ of the same size and small $u$, the spectrum of $Z+uH$ lies in $I$, and by \cite[Theorem~3.25]{hiaipetz2014intro}, $Dg(Z)[H]=M\circ H$, so
\[
  M\circ H=\lim_{u\to0}u^{-1}\big(g(Z+uH)-g(Z)\big).
\]
Hence
\begin{align}\label{eq:hermitian-bound}
  \|M\circ H\|_\infty\leq\mathrm{Lip}^I(g)\,\|H\|_\infty\qquad\text{for all Hermitian }H.
\end{align}
To pass from Hermitian $H$ to arbitrary matrices, let $N$ be any $n\times n$ matrix and set
\[
  \widehat M=\begin{bmatrix}M&M\\M&M\end{bmatrix},\qquad
  \widehat N=\begin{bmatrix}0&N\\N^*&0\end{bmatrix}.
\]
Then $\widehat N$ is Hermitian with $\|\widehat N\|_\infty=\|N\|_\infty$, and
\[
  \widehat M\circ\widehat N=\begin{bmatrix}0&M\circ N\\M\circ N^*&0\end{bmatrix}
  =\begin{bmatrix}0&M\circ N\\(M\circ N)^*&0\end{bmatrix},
\]
where the lower-left block satisfies $M\circ N^*=\overline{M}\circ N^*=(M\circ N)^*$ because $M=M^\top=\overline M$. This matrix is Hermitian with norm $\|M\circ N\|_\infty$. Moreover $\widehat M$ is the matrix associated in the same way with $\mathrm{diag}(Z,Z)$, whose eigenvalues are those of $Z$ repeated twice and still lie in the interior of $I$, so the argument leading to \eqref{eq:hermitian-bound} applies to it and gives
\[
  \|M\circ N\|_\infty=\|\widehat M\circ\widehat N\|_\infty\leq\mathrm{Lip}^I(g)\|\widehat N\|_\infty=\mathrm{Lip}^I(g)\|N\|_\infty.
\]
Taking the supremum over $\|N\|_\infty\leq1$ yields $\|g^{[1]}_\Lambda\|_{\mathrm S}\leq\mathrm{Lip}^I(g)$ for $\Lambda$ in the interior of $I$, and continuity of $g^{[1]}$, and hence of $\Lambda\mapsto\|g^{[1]}_\Lambda\|_{\mathrm S}$, extends this to all finite $\Lambda\subseteq I$.

Finally, take $\Lambda=\{t\}$ with $t\in I$. Then $g^{[1]}_\Lambda$ is the $1\times1$ matrix $(g'(t))$, whose Schur norm is $|g'(t)|$, so $\mathrm{Lip}^I(g)\geq|g'(t)|$ for every $t\in I$, that is, $\mathrm{Lip}^I(g)\geq\max_I|g'|$.
\end{proof}
 
\begin{lemma}\label{lem:schur}
Let $M$ be a real matrix.
\begin{enumerate}
\item[(a)] $\|M\|_{\mathrm S}\geq\max_{i,j}|M_{ij}|$, and $\|M\|_{\mathrm S}=\max_iM_{ii}$ if $M\geq0$. For vectors $x,y$, $\|xy^\top\|_{\mathrm S}=\max_i|x_i|\max_j|y_j|$.
\item[(b)] \emph{(Compression.)} If $M$ is $k\times k$ and $\Pi$ is a real $n\times k$ matrix whose rows have $\ell^1$-norm at most $1$, then $\|\Pi M\Pi^\top\|_{\mathrm S}\leq\|M\|_{\mathrm S}$.
\item[(c)] \emph{($2\times2$ matrices.)} For real $q_1,q,q_2$ put $L=\max\{|q_1|,|q_2|\}$ and, for $t\geq L$, $A(t)=\sqrt{(t-q_1)(t-q_2)}$ and $B(t)=\sqrt{(t+q_1)(t+q_2)}$. Then $\big\|\big[\begin{smallmatrix}q_1&q\\q&q_2\end{smallmatrix}\big]\big\|_{\mathrm S}$ is the smallest $t\geq L$ with $A(t)+B(t)\geq2|q|$.
\item[(d)] If $q_1\leq q\leq q_2$, then $\big\|\big[\begin{smallmatrix}q_1&q\\q&q_2\end{smallmatrix}\big]\big\|_{\mathrm S}\leq\sqrt2\,\max\{|q_1|,|q_2|\}$.
\end{enumerate}
\end{lemma}
 
\begin{proof}
We use Haagerup's characterisation \cite[Chapter~8]{Paulsen2003}: $\|M\|_{\mathrm S}\leq t$ if and only if there are vectors $x_i,y_j$ in a Hilbert space with $M_{ij}=\langle x_i,y_j\rangle$ and $\|x_i\|^2,\|y_j\|^2\leq t$; equivalently, if and only if there are positive semidefinite $Y_1,Y_2$ with diagonal entries at most $t$ such that $\big[\begin{smallmatrix}Y_1&M\\M^\top&Y_2\end{smallmatrix}\big]\geq0$.
 
(a) The first claim follows from $|\langle x_i,y_j\rangle|\leq t$. If $M\geq0$, write $M_{ij}=\langle g_i,g_j\rangle$ with $\|g_i\|^2=M_{ii}$. For $xy^\top$ use one-dimensional vectors, rescaled.
 
(b) Put $x'_i=\sum_l\Pi_{il}x_l$ and $y'_j=\sum_{l'}\Pi_{jl'}y_{l'}$. Then $\langle x'_i,y'_j\rangle=(\Pi M\Pi^\top)_{ij}$ and $\|x'_i\|\leq\sum_l|\Pi_{il}|\,\|x_l\|\leq\sqrt t$.
 
(c) If $(Y_1,Y_2)$ is an admissible pair for the block decomposition from $M$, so is $(Y_2,Y_1)$, by conjugation with the flip of the two blocks, and hence so is $(\frac{Y_1+Y_2}2,\frac{Y_1+Y_2}2)$. Moreover \[\begin{pmatrix}Y_1&M\\M&Y_1\end{pmatrix}\geq0\] if and only if $Y_1\pm M\geq0$, and raising the diagonal of $Y_1$ to $t$ preserves both conditions. Hence $\|M\|_{\mathrm S}\leq t$ if and only if there is a real $\tilde q$ with \[\begin{pmatrix}t\pm q_1&\tilde q\pm q\\\tilde q\pm q&t\pm q_2\end{pmatrix}\geq0\] for both signs, that is, $t\geq L$, $|\tilde q-q|\leq A(t)$ and $|\tilde q+q|\leq B(t)$. Such a $\tilde q$ exists if and only if $2|q|\leq A(t)+B(t)$. Since $A$ and $B$ are nondecreasing on $[L,\infty)$, the claim follows.
 
(d) Replacing $M$ by $-M$ and exchanging the two indices preserves the Schur multiplier norm and maps $(q_1,q,q_2)$ to $(-q_2,-q,-q_1)$, so we may assume $q\geq0$, and by scaling $L=1$; then $0\leq q\leq q_2\leq1$ and $-1\leq q_1\leq q_2$. Let $t=\sqrt2$. Since $\partial_{q_1}(A+B)=\frac12\big(\sqrt{(t+q_2)/(t+q_1)}-\sqrt{(t-q_2)/(t-q_1)}\big)\geq0$ for $q_1\leq q_2$, we have $A(t)+B(t)\geq\varphi(q_2):=\sqrt{(t+1)(t-q_2)}+\sqrt{(t-1)(t+q_2)}$. The function $\varphi$ is concave, $\varphi(0)>0$ and $\varphi(1)=2\sqrt{t^2-1}=2$, so $\varphi(q_2)\geq2q_2\geq2q$ on $[0,1]$, and (c) gives the claim.
\end{proof}

We now combine the previous lemmas into a theorem which will begin to allow us to express the continuity moduli for the operator convex functions in terms of their integral representations. We notice that for the below theorem, the structure of the function is that of \eqref{eq:loewner}.
 
\begin{theorem}\label{thm:one pole}
Let $I=[t_-,t_+]$ with $t_-<t_+$, let $p\in\mathbb{R}\setminus I$, $\wp(t)=|t-p|^{-1}$, and $g=C+\beta t+w\wp$ with $C,\beta,w\in\mathbb{R}$. Then:
\begin{enumerate}
\item[(a)] \[\mathrm{Lip}^I(g)=\big\|g^{[1]}_{\{t_-,t_+\}}\big\|_{\mathrm S}=\Big\|\begin{bmatrix}g'(t_-)&g^{[1]}(t_-,t_+)\\g^{[1]}(t_-,t_+)&g'(t_+)\end{bmatrix}\Big\|_{\mathrm S},\] which is given explicitly by Lemma~\ref{lem:schur}(c);
\item[(b)] $\max_I|g'|\leq\mathrm{Lip}^I(g)\leq\sqrt2\,\max_I|g'|$;
\item[(c)] $\mathrm{Lip}^I(g)=\max_I|g'|$ if $\beta w\,(p-t_-)\geq0$, or if $g'(t_0)=0$ for some $t_0\in I$ with $\wp(t_0)\leq\frac12(\wp(t_-)+\wp(t_+))$. The latter holds whenever $t_0$ lies on the far side of the midpoint of $I$ from $p$, the midpoint included;
\item[(d)] the constant $\sqrt2$ in (b) is optimal: $k_z(t)=t+2z^2/(t+z)$ satisfies $\max_{[0,1]}|k_z'|=1$ and $\mathrm{Lip}^{[0,1]}(k_z)\to\sqrt2$ as $z\downarrow0$.
\end{enumerate}
\end{theorem}
 
\begin{proof}
Let $\zeta=+1$ if $p>t_+$ and $\zeta=-1$ if $p<t_-$. We then have $\wp^{[1]}(x,y)=\zeta\wp(x)\wp(y)$ for $x,y\in I$. This is because if $p\notin I$, the sign of $t-p$ is constant on $I$, equal to $-\zeta$, so $\wp(t)=-\zeta/(t-p)$ and $\wp(x)\wp(y)=1/((x-p)(y-p))$. For $x\neq y$ in $I$,
\[
  \wp(x)-\wp(y)=-\zeta\Big(\frac1{x-p}-\frac1{y-p}\Big)=\zeta\,\frac{x-y}{(x-p)(y-p)},
\]
so $\wp^{[1]}(x,y)=\zeta\wp(x)\wp(y)$. For $x=y$, $\wp'(x)=\zeta/(x-p)^2=\zeta\wp(x)^2$, so the formula holds on the diagonal too.
\[
g^{[1]}(x,y)=\beta+\zeta w\,\wp(x)\wp(y).
\]
(a) The function $\wp$ is continuous and monotone on $I$, so for $x\in I$ there is $\lambda(x)\in[0,1]$ with $\wp(x)=\lambda(x)\wp(t_-)+(1-\lambda(x))\wp(t_+)$. Put $\lambda_{t_-}=\lambda$ and $\lambda_{t_+}=1-\lambda$. Since $\lambda_{t_-}+\lambda_{t_+}=1$ and the kernel $\beta+\zeta w\,u_1u_2$ is affine in $u_1$ and in $u_2$,
\[
g^{[1]}(x,y)=\sum_{t,t'\in\{t_-,t_+\}}\lambda_t(x)\,\lambda_{t'}(y)\,g^{[1]}(t,t'),\qquad x,y\in I.
\]
Thus $g^{[1]}_\Lambda=\Pi g^{[1]}_{\{t_-,t_+\}}\Pi^\top$, where $\Pi$ has the nonnegative entries $\lambda_t(x)$ and rows summing to $1$. Lemma~\ref{lem:schur}(b) and Lemma~\ref{lem:DK} give $\mathrm{Lip}^I(g)\leq\|g^{[1]}_{\{t_-,t_+\}}\|_{\mathrm S}$, and the reverse inequality is Lemma~\ref{lem:DK} with $\Lambda=\{t_-,t_+\}$.
 
(b) The first inequality is in Lemma~\ref{lem:DK}. Since $g''$ has constant sign on $I$, $g^{[1]}(t_-,t_+)$ lies between $g'(t_-)$ and $g'(t_+)$, and $g'$ is monotone, so $\max_I|g'|=\max\{|g'(t_-)|,|g'(t_+)|\}$. Lemma~\ref{lem:schur}(d), applied to $g$ or to $-g$, gives the second inequality.
 
(c) If $\beta\zeta w\geq0$, then $\zeta\,\mathrm{sgn}(w)\,g^{[1]}_\Lambda=|\beta|\mathbf{j}\mathbf{j}^\top+|w|\wp\wp^\top\geq0$, where $\wp$ also denotes the vector $(\wp(x))_{x\in\Lambda}$, so Lemma~\ref{lem:schur}(a) gives $\|g^{[1]}_\Lambda\|_{\mathrm S}=\max_\Lambda|g'|$. Otherwise put $c_0=\sqrt{-\beta/(\zeta w)}>0$, so that $g'=\zeta w(\wp^2-c_0^2)$ and, with $v=\wp-c_0\mathbf{j}$,
\[
g^{[1]}_\Lambda=\zeta w\big(\wp\wp^\top-c_0^2\mathbf{j}\mathbf{j}^\top\big)=\zeta w\big(c_0\,v\mathbf{j}^\top+c_0\,\mathbf{j}v^\top+vv^\top\big).
\]
By Lemma~\ref{lem:schur}(a), $\|g^{[1]}_\Lambda\|_{\mathrm S}\leq|w|(2c_0N+N^2)=|w|\big((c_0+N)^2-c_0^2\big)$ with $N=\max_\Lambda|\wp-c_0|$. If $c_0\leq\frac12(\wp_{\min}+\wp_{\max})$, where $\wp_{\min},\wp_{\max}$ are the values of $\wp$ at the endpoints, then $N\leq\wp_{\max}-c_0$ and the bound is $|w|(\wp_{\max}^2-c_0^2)=|g'(t_1)|$, where $t_1$ is the endpoint closest to $p$. If $g'(t_0)=0$ then $\wp(t_0)=c_0$, which gives the stated condition. Finally, $\wp$ is convex and decreases away from $p$; if $t_0$ lies on the far side of the midpoint $t_{\mathrm{mid}}$, then $\wp(t_0)\leq\wp(t_{\mathrm{mid}})\leq\frac12(\wp(t_-)+\wp(t_+))$.
 
(d) Here $k_z'(t)=1-2z^2/(t+z)^2$ increases from $-1$ to $1-2z^2/(1+z)^2$, and $k_z^{[1]}(0,1)=1-2z/(1+z)$. The endpoint matrix tends to $\big[\begin{smallmatrix}-1&1\\1&1\end{smallmatrix}\big]$, whose Schur multiplier norm is $\sqrt2$ by Lemma~\ref{lem:schur}(c), since $A(t)+B(t)=2\sqrt{t^2-1}$. The norm in Lemma~\ref{lem:schur}(c) depends continuously on $(q_1,q,q_2)$, and (a) gives the claim.
\end{proof}
 
Theorem~\ref{thm:one pole}(d) shows that operator convexity does not imply $\mathrm{Lip}^I(g)=\max_I|g'|$, even for a single pole. The operator Lipschitz constant can exceed the scalar one because $g'$ changes sign on $I$, from $-1$ at the endpoint near the pole to almost $1$ on the rest of the interval.

We now construct sharp constants for the kernels described in \eqref{eq:hf kernels}. This will allow us to define continuity constants for operator divergences defined through convex perspectives. 
\begin{corollary}[Sharp constants for the kernels]\label{cor:kernels}
Let $I\subseteq[0,1]$ be a compact interval and $\varrho_I=\max_{t\in I}|2t-1|$.
\begin{enumerate}
\item[(a)]
Let $s\ge0$ and $\varpi_s(t)=s(1-t)+t$, and assume $d_s=\min_{t\in I}\varpi_s(t)>0$
(for $t\in[0,1]$ the function $\varpi_s(t)=(1-t)s+t$ is a convex combination of $s$ and $1$,
so this fails only if $s=0$ and $0\in I$). Then for $t\in I$,
\[
  |\kappa_s'(t)|=|2t-1|\Big(\frac{2}{\varpi_s(t)}+\frac{1+s}{\varpi_s(t)^2}\Big),
\]
the bracket is positive, and therefore
\[
  |\kappa_s'(t)|
  \leq\varrho_I\Big(\frac{2}{d_s}+\frac{1+s}{d_s^2}\Big).
\] and
\[
\mathrm{Lip}^I(\kappa_s)\leq\varrho_I\Big(\frac{2}{d_s}+\frac{1+s}{d_s^2}\Big),
\]
with equality, and equal to $\max_I|\kappa_s'|$, if $I$ is symmetric about $\frac12$. The same holds for $\kappa_\infty$, with $\varpi_\infty(t)=1-t$ and $1+s$ replaced by $1$. In addition $\|\kappa_s(X)-\kappa_s(Y)\|_\infty\leq\mathrm{osc}_I\kappa_s=\max_I\kappa_s-\min_I\kappa_s$.
\item[(b)] $\mathrm{Lip}^{[0,1-\epsilon]}(\kappa_\infty)=\epsilon^{-2}-4$ for $0<\epsilon\leq\frac13$.
\item[(c)] For $s>0$, the kernel $\eta_s(t)=(1+s)t(1-t)/(1+(s-1)t)$ of \eqref{eq:mean kernels} below satisfies $\mathrm{Lip}^{[0,1]}(\eta_s)=\max_{[0,1]}|\eta_s'|=(1+s)\max\{1,s^{-1}\}$.
\item[(d)] For $s>0$, $g_s(t)=\frac{t}{1+s}-\frac{t}{t+s}$ satisfies $\mathrm{Lip}^{[0,1]}(g_s)=\frac{1}{s(1+s)}$ and $-\frac{1}{1+s}\leq g_s\leq0$ on $[0,1]$.
\item[(e)] $d_s\geq(1+s)\min_{t\in I}\min\{t,1-t\}$.
\end{enumerate}
\end{corollary}
 
\begin{proof}
(a) The identity \[(1-s)(2t-1)=2\varpi_s(t)-(1+s)\] gives \[\kappa_s'=(2t-1)\big(2\varpi_s+(1+s)\big)/\varpi_s^2\] which is the formula. For the bound, with $X,Y\in\mathrm{L}(\mathcal{V})$ Hermitian with spectra in $I$, put \[\varpi_s(X)=s(\mathbf{1}_{\mathcal{V}}-X)+X\] and \[R(X)=(2X-\mathbf{1}_{\mathcal{V}})\varpi_s(X)^{-1}\] so that \[\kappa_s(X)=(2X-\mathbf{1}_{\mathcal{V}})R(X)\] \[\|\varpi_s(X)^{-1}\|_\infty\leq d_s^{-1}\] and $\|R(X)\|_\infty\leq\varrho_I/d_s$. Expanding, \[(2X-\mathbf{1}_{\mathcal{V}})\varpi_s(Y)-\varpi_s(X)(2Y-\mathbf{1}_{\mathcal{V}})=(1+s)(X-Y)\] since the terms in $XY$ cancel. Hence \[R(X)-R(Y)=(1+s)\,\varpi_s(X)^{-1}(X-Y)\varpi_s(Y)^{-1}\] using that $\varpi_s(X)^{-1}$ commutes with $2X-\mathbf{1}_{\mathcal{V}}$, and
\[
\kappa_s(X)-\kappa_s(Y)=2(X-Y)R(Y)+(1+s)\,R(X)(X-Y)\varpi_s(Y)^{-1}
\]
which gives \[\|\kappa_s(X)-\kappa_s(Y)\|_\infty\leq\varrho_I\big(\frac2{d_s}+\frac{1+s}{d_s^2}\big)\|X-Y\|_\infty\] If $I$ is symmetric about $\frac{1}{2}$, then at the endpoint where $\varpi_s$ is minimal both $|2t-1|=\varrho_I$ and $\varpi_s=d_s$, so $\max_I|\kappa_s'|$ equals the bound, and Lemma~\ref{lem:DK} gives equality. For $\kappa_\infty$ use \[\varpi_\infty(X)=\mathbf{1}_{\mathcal{V}}-X\] and \[(2X-\mathbf{1}_{\mathcal{V}})(\mathbf{1}_{\mathcal{V}}-Y)-(\mathbf{1}_{\mathcal{V}}-X)(2Y-\mathbf{1}_{\mathcal{V}})=X-Y\] The last statement holds because $\kappa_s(X)$ and $\kappa_s(Y)$ have spectra in $[\min_I\kappa_s,\max_I\kappa_s]$.
 
(b) We have \[\kappa_\infty(t)=(1-t)^{-1}-4t\] a one-pole function with $p=1$, $\beta=-4$ and $w=1$, and $\kappa_\infty'$ vanishes at $t_0=\frac12$. On $[0,1-\epsilon]$, $\wp(t)=(1-t)^{-1}$ takes the values $1$ and $\epsilon^{-1}$ at the endpoints and $\wp(t_0)=2$, so the condition of Theorem~\ref{thm:one pole}(c) is \[2\leq\frac12(1+\epsilon^{-1})\] that is, $\epsilon\leq\frac13$. Then \[\mathrm{Lip}=\max\{|\kappa_\infty'(0)|,\kappa_\infty'(1-\epsilon)\}=\max\{3,\epsilon^{-2}-4\}=\epsilon^{-2}-4\]
 
(c) For $s=1$, $\eta_1(t)=2t(1-t)$, and for every quadratic $g$ the matrix $g^{[1]}_\Lambda$ is a compression of $g^{[1]}_{\{t_-,t_+\}}$ as in the proof of Theorem~\ref{thm:one pole}(a), with $\lambda$ affine; here \[g^{[1]}_{\{0,1\}}=\begin{pmatrix}2&0\\0&-2\end{pmatrix}\] has $\infty$-norm $2$.

For $s>1$, $-\eta_s$ is a one-pole function with pole $p=-1/(s-1)<0$; then \[\wp(t)=(s-1)/(1+(s-1)t)\] with $\wp(0)=s-1$ and $\wp(1)=(s-1)/s$. The derivative $\eta_s'$ vanishes at $t_0=1/(1+\sqrt s)$, where $\wp(t_0)=(s-1)/\sqrt s$. The condition $\wp(t_0)\leq\frac12(\wp(0)+\wp(1))$ reads $2\sqrt s\leq1+s$, so Theorem~\ref{thm:one pole}(c) gives \[\mathrm{Lip}=\max\{\eta_s'(0),|\eta_s'(1)|\}=\max\{1+s,(1+s)/s\}=1+s\] For $s<1$ use $\eta_s(1-t)=\eta_{1/s}(t)$.
 
(d) We have \[g_s(t)=\frac{t}{1+s}-1+\frac{s}{t+s}\] which is a one-pole function with $p=-s$, $\beta=\frac1{1+s}$ and $w=s$. Its derivative \[g_s'(t)=\frac1{1+s}-\frac{s}{(t+s)^2}\] vanishes at \[t_0=\sqrt{s(1+s)}-s\in(0,1)\] where $\wp(t_0)=(s(1+s))^{-1/2}$. The condition of Theorem~\ref{thm:one pole}(c) is \[(s(1+s))^{-1/2}\leq\frac12\big(\frac1s+\frac1{1+s}\big)\] that is, $2\sqrt{s(1+s)}\leq1+2s$, which holds. Hence \[\mathrm{Lip}=\max\{|g_s'(0)|,g_s'(1)\}=\max\{\frac1{s(1+s)},\frac1{(1+s)^2}\}=\frac1{s(1+s)}\] Finally $|g_s(t)|=\frac{t(1-t)}{(1+s)(t+s)}\leq\frac{1-t}{1+s}$.
 
(e) $s(1-t)+t\geq(1+s)\min\{t,1-t\}$.
\end{proof}
 We finally come to the explicit modulus data for the integral representative of the operator perspective channel divergences. 
\begin{theorem}\label{thm:integral modulus}
Let $f$ be operator convex with Loewner data $(a,b,c,\mu)$, and let $I\subseteq[0,1]$ be a compact interval on which $h_f$ is finite. With $\varrho_I$, $d_s$ and $d_\infty$ as in Corollary~\ref{cor:kernels}, for $\delta\geq0$,
\[
\omega^I_{h_f}(\delta)\leq|2b-a|\,\delta+c\,\varrho_I\Big(\frac{2}{d_\infty}+\frac{1}{d_\infty^2}\Big)\delta+\int_{[0,\infty)}\min\Big\{\varrho_I\Big(\frac{2}{d_s}+\frac{1+s}{d_s^2}\Big)\delta,\ \mathrm{osc}_I\kappa_s\Big\}\,d\mu(s),
\]
where the $c$-term is absent if $c=0$. Consequently, Theorems~\ref{thm:maxf}(b), \ref{thm:unbounded}(b) and~\ref{thm:extended continuity}(a) hold with this explicit modulus.
\end{theorem}
 
\begin{proof}
If $1\in I$ then $c=0$, and if $0\in I$ then $\mu(\{0\})=0$, since $h_f$ is finite on $I$; so all constants are finite where they occur. For $X\in\mathrm{L}(\mathcal{V})$ with spectrum in $I$, \eqref{eq:hf kernels} and the spectral theorem give $h_f(X)=(a-b)\mathbf{1}_{\mathcal{V}}+(2b-a)X+c\,\kappa_\infty(X)+\int\kappa_s(X)\,d\mu(s)$, the integral being a finite sum of convergent scalar integrals times spectral projections. Subtract the same identity for $Y$, take norms inside the integral, and apply Corollary~\ref{cor:kernels}(a).
\end{proof}
 
\subsection{Sharp Lipschitz constants}
 
On the intervals $[\epsilon,1-\epsilon]$ that correspond to a two-sided bound on the max-relative entropy, the kernel constants of Corollary~\ref{cor:kernels} are attained. For $0<\epsilon\leq\frac12$ put $\varrho=1-2\epsilon$,
\[
M_s(\epsilon)=\varrho\Big(\frac{2}{d_s}+\frac{1+s}{d_s^2}\Big),\quad d_s=\epsilon(1+s)+\varrho\min\{1,s\},\qquad M_\infty(\epsilon)=\frac{1-4\epsilon^2}{\epsilon^2},
\]
so that $M_s(\epsilon)=\max_{[\epsilon,1-\epsilon]}|\kappa_s'|$ for $s\in[0,\infty]$ by Corollary~\ref{cor:kernels}(a), and
\[
L_f(\epsilon)=|2b-a|+c\,M_\infty(\epsilon)+\int_{[0,\infty)}M_s(\epsilon)\,d\mu(s).
\]
 
\begin{theorem}\label{thm:sharp lipschitz}
Let $f$ be operator convex with Loewner data $(a,b,c,\mu)$, let $0<\epsilon\leq\frac12$ and $I=[\epsilon,1-\epsilon]$.
\begin{enumerate}
\item[(a)] $\max_I|h_f'|\leq\mathrm{Lip}^I(h_f)\leq L_f(\epsilon)\leq3\max_I|h_f'|$.
\item[(b)] $\mathrm{Lip}^I(h_f)=L_f(\epsilon)=h_f'(1-\epsilon)$ if $\mu$ is carried by $[1,\infty)$ and $2b\geq a$; and $\mathrm{Lip}^I(h_f)=L_f(\epsilon)=-h_f'(\epsilon)$ if $c=0$, $\mu$ is carried by $[0,1]$ and $2b\leq a$.
\item[(c)] If $h_f(1-t)=h_f(t)$ for all $t$, equivalently $f(x)=xf(1/x)$, then $\mathrm{Lip}^I(h_f)=\max_I|h_f'|=h_f'(1-\epsilon)$.
\end{enumerate}
\end{theorem}
 
\begin{proof}
(a) The first inequality is Lemma~\ref{lem:DK}, and the second follows from \eqref{eq:hf kernels} as in the proof of Theorem~\ref{thm:integral modulus}.

For the third, write $t_-=\epsilon$ and $t_+=1-\epsilon$. Each $\kappa_s$, $s\in[0,\infty]$, is convex with $\kappa_s'(\frac12)=0$, so \[\kappa_s'(t_-)\leq0\leq\kappa_s'(t_+)\] and $M_s(\epsilon)=\max\{|\kappa_s'(t_-)|,\kappa_s'(t_+)\}$.

Writing \[\Sigma_{+}=c\,\kappa_\infty'(t_+)+\int\kappa_s'(t_+)\,d\mu(s)\] and \[\Sigma_{-}=c\,|\kappa_\infty'(t_-)|+\int|\kappa_s'(t_-)|\,d\mu(s)\] so that \[L_f(\epsilon)\leq|2b-a|+\Sigma_{+}+\Sigma_{-}\] \[h_f'(t_+)=(2b-a)+\Sigma_{+}\] and \[h_f'(t_-)=(2b-a)-\Sigma_{-}.\] Since $h_f$ is convex, $N:=\max_I|h_f'|=\max\{|h_f'(t_-)|,|h_f'(t_+)|\}$ . Hence \[N\geq(2b-a)+\Sigma_{+}\] and $N\geq \Sigma_{-}-(2b-a)$, so $\Sigma_{+}+\Sigma_{-}\leq2N$. Moreover $N\geq|2b-a|$: if $2b\geq a$ this follows from the first inequality, and otherwise from the second. Thus $L_f(\epsilon)\leq3N$.
 
(b) For $s\geq1$ the function $\varpi_s$ is nonincreasing, so $\max_I|\kappa_s'|$ is attained at $t_+$, by the formula in Corollary~\ref{cor:kernels}(a); the same holds for $\kappa_\infty$. If $\mu$ is carried by $[1,\infty)$ and $2b\geq a$, then $L_f(\epsilon)=(2b-a)+c\,\kappa_\infty'(t_+)+\int\kappa_s'(t_+)\,d\mu(s)=h_f'(t_+)\leq\max_I|h_f'|$, and (a) gives equality throughout. The second case is the mirror image.
 
(c) Write $\check g(t)=g(1-t)$. Since $h_f=\frac12(h_f+\check h_f)$ and the affine part of the right-hand side has slope $0$, \eqref{eq:hf kernels} gives
\[
h_f=\mathrm{const}+\tfrac c2\,\kappa^{\mathrm{sym}}_\infty+\tfrac12\int_{[0,\infty)}\kappa^{\mathrm{sym}}_s\,d\mu(s),\qquad \kappa^{\mathrm{sym}}_s(t)=\kappa_s(t)+\kappa_s(1-t).
\]
We claim that $\mathrm{Lip}^I(\kappa^{\mathrm{sym}}_s)=(\kappa^{\mathrm{sym}}_s)'(t_+)$ for every $s\in[0,\infty]$. Then $\mathrm{Lip}^I(h_f)\leq\frac c2(\kappa^{\mathrm{sym}}_\infty)'(t_+)+\frac12\int(\kappa^{\mathrm{sym}}_s)'(t_+)\,d\mu(s)=h_f'(t_+)$, and Lemma~\ref{lem:DK} gives equality. For $s=1$, $\kappa^{\mathrm{sym}}_1=2\kappa_1$ is quadratic, and the claim follows as in the proof of Corollary~\ref{cor:kernels}(c). For $s\neq1$, $\kappa_s=C+\beta t+w\wp_p$ with one pole $p\notin(0,1)$ and $\wp_p(t)=|t-p|^{-1}$; then $\kappa_s(1-t)=C'+\beta' t+w\wp_{1-p}$, with the mirrored pole on the other side of $[0,1]$, and the slopes cancel because $\kappa^{\mathrm{sym}}_s$ is symmetric. By the computation in the proof of Theorem~\ref{thm:one pole}, with $u=(\wp_p(x))_{x\in\Lambda}$ and $\tilde u=(\wp_{1-p}(x))_{x\in\Lambda}$,
\[
\pm(\kappa^{\mathrm{sym}}_s)^{[1]}_\Lambda=w\big(uu^\top-\tilde u\tilde u^\top\big)=\tfrac w2\big((u-\tilde u)(u+\tilde u)^\top+(u+\tilde u)(u-\tilde u)^\top\big).
\]
By Lemma~\ref{lem:schur}(a) its Schur multiplier norm is at most $|w|\max_\Lambda|u-\tilde u|\,\max_\Lambda(u+\tilde u)$. On the symmetric interval $I$, the function $\wp_p-\wp_{1-p}$ is monotone and antisymmetric about $\frac12$, and $\wp_p+\wp_{1-p}$ is convex and symmetric, so both maxima are attained at $t_+$, and the bound equals $|w|\,|\wp_p(t_+)^2-\wp_{1-p}(t_+)^2|=|(\kappa^{\mathrm{sym}}_s)'(t_+)|$. Finally $(\kappa^{\mathrm{sym}}_s)'(t_+)\geq0$, because $\kappa^{\mathrm{sym}}_s$ is convex and symmetric about $\frac12$.
\end{proof}
 
Part (c) applies to the Jeffreys, symmetrised $\chi^2$, Le Cam, Hellinger and Jensen--Shannon divergences, and part (b) to the Pearson and Neyman $\chi^2$-divergences. For these, Table~\ref{tab:lipschitz} gives the exact Lipschitz constants. For the Belavkin--Staszewski and R\'enyi divergences the Loewner data are two-sided and the theorem gives only $\max_I|h_f'|\leq\mathrm{Lip}^I(h_f)\leq L_f(\epsilon)$.

The next proposition shows that if the identity $\mathrm{Lip}^I(g)=\|g^{[1]}_{\{t_-,t_+\}}\|_{\mathrm S}$ of Theorem~\ref{thm:one pole}(a) holds for all combinations of two kernels $\kappa_s$, then it holds for every $h_f$.
 
\begin{proposition}\label{prop:two poles}
Let $I=[t_-,t_+]\subseteq(0,1)$. If $\mathrm{Lip}^I(g)=\|g^{[1]}_{\{t_-,t_+\}}\|_{\mathrm S}$ for every $g=\beta t+w_1\kappa_{s_1}+w_2\kappa_{s_2}$ with $\beta\in\mathbb{R}$, $w_1,w_2\geq0$ and $s_1,s_2\in[0,\infty]$, then $\mathrm{Lip}^I(h_f)=\|(h_f)^{[1]}_{\{t_-,t_+\}}\|_{\mathrm S}$ for every operator convex $f$.
\end{proposition}
 
\begin{proof}
Suppose first that $h=\beta t+\sum_{k=1}^Nw_k\kappa_{s_k}$ with $w\in[0,\infty)^N$; the $c$-term is the case $s_k=\infty$, and constants do not affect divided differences. Put $u_k=\kappa_{s_k}^{[1]}(t_-,t_+)-\kappa_{s_k}'(t_-)$ and $u'_k=\kappa_{s_k}'(t_+)-\kappa_{s_k}^{[1]}(t_-,t_+)$, which are positive by strict convexity. The endpoint data $\big(h'(t_-),h^{[1]}(t_-,t_+),h'(t_+)\big)$ are determined by $h'(t_-)$, $\sum_kw_ku_k$ and $\sum_kw_ku'_k$. The set $\mathcal{P}$ of $w'\in[0,\infty)^N$ with $\sum_kw'_ku_k=\sum_kw_ku_k$ and $\sum_kw'_ku'_k=\sum_kw_ku'_k$ is a compact polytope containing $w$, whose vertices have at most two nonzero coordinates. Write $w=\sum_j\lambda_jv^{(j)}$ as a convex combination of vertices, and let $h_j=\beta_jt+\sum_kv^{(j)}_k\kappa_{s_k}$, with $\beta_j$ chosen so that $h_j'(t_-)=h'(t_-)$. Then each $h_j$ has the same endpoint data as $h$, and $h-\sum_j\lambda_jh_j$ is affine with derivative $0$ at $t_-$, hence constant. For every finite $\Lambda$, $\|h^{[1]}_\Lambda\|_{\mathrm S}\leq\sum_j\lambda_j\mathrm{Lip}^I(h_j)=\|h^{[1]}_{\{t_-,t_+\}}\|_{\mathrm S}$.
 
In general, $\int\kappa_s\,d\mu(s)=\int(1+s)\kappa_s\,d\tilde\mu(s)$ with the finite measure $d\tilde\mu=(1+s)^{-1}d\mu$, and $s\mapsto(1+s)\kappa_s$ is continuous from $[0,\infty]$ into $C^1(I)$, with value $\kappa_\infty$ at $s=\infty$, because $I\subseteq(0,1)$. Approximating $\tilde\mu+c\,\boldsymbol{\delta}_\infty$ by finitely supported positive measures along finer and finer partitions of $[0,\infty]$, uniform continuity gives functions $h_n$ of the above form with $h_n\to h_f$ in $C^1(I)$. Then $(h_n)^{[1]}\to(h_f)^{[1]}$ uniformly on $I\times I$, and for every finite $\Lambda$, $\|(h_f)^{[1]}_\Lambda\|_{\mathrm S}=\lim_n\|(h_n)^{[1]}_\Lambda\|_{\mathrm S}\leq\lim_n\|(h_n)^{[1]}_{\{t_-,t_+\}}\|_{\mathrm S}=\|(h_f)^{[1]}_{\{t_-,t_+\}}\|_{\mathrm S}$.
\end{proof}

\begin{corollary}[Lipschitz moduli under a Thompson bound]\label{cor:thompson lipschitz}
Let $f$ be operator convex with Loewner data $(a,b,c,\mu)$ and $\epsilon\in(0,\frac12]$. Then $\omega^{[\epsilon,1-\epsilon]}_{h_f}(\delta)\leq L_f(\epsilon)\,\delta$, where $L_f(\epsilon)$ is within a factor $3$ of the optimal constant $\mathrm{Lip}^{[\epsilon,1-\epsilon]}(h_f)$, and
\[
L_f(\epsilon)\leq|2b-a|+\frac{1-4\epsilon^2}{\epsilon^2}\Big(c+\int_{[0,\infty)}\frac{d\mu(s)}{1+s}\Big).
\]
Hence, if $\Phi_i=\Psi_{T_i}$ and $\Psi_i=\Psi_{\mathbf{1}_{\mathbb{C}^r}-T_i}$ satisfy $\max\{D_{\max}(\Phi_i\|\Psi_i),D_{\max}(\Psi_i\|\Phi_i)\}\leq\log\frac{1-\epsilon}{\epsilon}$ for $i=1,2$, then
\[
\big|\widehat{D}_f(\Phi_1\|\Psi_1)-\widehat{D}_f(\Phi_2\|\Psi_2)\big|\leq\|\Omega^*(\mathbf{1}_{\mathcal{Y}})\|_\infty\,L_f(\epsilon)\,\|T_1-T_2\|_\infty .
\]
\end{corollary}
 
\begin{proof}
The first statement is Theorem~\ref{thm:sharp lipschitz}(a). Corollary~\ref{cor:kernels}(e) gives $d_s\geq\epsilon(1+s)$, hence $M_s(\epsilon)\leq\varrho\,\frac{1+2\epsilon}{\epsilon^2(1+s)}=\frac{1-4\epsilon^2}{\epsilon^2(1+s)}$. The last statement follows that we identify the hypothesis with $\epsilon\mathbf{1}_{\mathbb{C}^r}\leq T_i\leq(1-\epsilon)\mathbf{1}_{\mathbb{C}^r}$, and Theorem~\ref{thm:extended continuity}(a).
\end{proof}
 
Every operator convex $f$ therefore gives a divergence that is Lipschitz in the derivative under a two-sided bound on the max-relative entropy. The factor $\epsilon^{-2}$ is attained for the $\chi^2$-divergences, for which $L_f(\epsilon)$ is exact. For logarithmic $f$ the order is $\epsilon^{-1}$ (Table~\ref{tab:lipschitz}).
 
\begin{corollary}[Bounded $f$: rates]\label{cor:rates}
If $f$ is bounded, then for all $\delta\geq0$,
\[
\omega_{h_f}(\delta)\leq|2b-a|\,\delta+\int_{(0,1]}\min\Big\{\frac{(1+3s)\,\delta}{s^2},\frac1s\Big\}d\mu(s)+\int_{(1,\infty)}\min\{(3+s)\delta,1\}\,d\mu(s),
\]
with the exact kernel constants of Corollary~\ref{cor:kernels}(a), and the right-hand side tends to $0$ as $\delta\to0$. If moreover $d\mu(s)\leq Cs^{\lambda}ds$ on $(0,1]$ and $d\mu(s)\leq Cs^{-1-\lambda'}ds$ on $(1,\infty)$ for some $\lambda,\lambda'>0$, then $\omega_{h_f}(\delta)\leq C'\delta^{\min\{\lambda,\lambda',1\}}$, with an extra factor $\log\frac1\delta$ if $\min\{\lambda,\lambda'\}=1$. Here $C'$ depends only on $a,b,C,\lambda,\lambda'$.
\end{corollary}
 
\begin{proof}
Take $I=[0,1]$ in Theorem~\ref{thm:integral modulus}, so that $\varrho_I=1$ and $d_s=\min\{1,s\}$. For $s\leq1$, $\frac2s+\frac{1+s}{s^2}=\frac{1+3s}{s^2}$, and for $s>1$, $2+1+s=3+s$. Since $\kappa_s$ is convex and nonnegative, $\mathrm{osc}\,\kappa_s\leq\max\{\kappa_s(0),\kappa_s(1)\}=\max\{\frac1s,1\}$. The limit is dominated convergence, since $\int_{(0,1]}s^{-1}d\mu+\mu((1,\infty))<\infty$. For the rate, use $(1+3s)s^{-2}\leq4s^{-2}$ on $(0,1]$ and $3+s\leq4s$ on $(1,\infty)$, and split $\int_0^1\min\{4\delta s^{-2},s^{-1}\}s^\lambda ds$ at $s=4\delta$, and $\int_1^\infty\min\{4s\delta,1\}s^{-1-\lambda'}ds$ at $s=1/(4\delta)$.
\end{proof}
 
For $f(t)=-t^\alpha$, $\alpha\in(0,1)$, one has $d\mu=\frac{\sin\pi\alpha}{\pi}\frac{s^{\alpha}}{(1+s)^2}ds$, so $\lambda=\alpha$ and $\lambda'=1-\alpha$. Corollary~\ref{cor:rates} then recovers the order $\delta^{\min\{\alpha,1-\alpha\}}$ of Lemma~\ref{holder lemma}, with a worse constant.

\begin{example}
    Table~\ref{tab:loewner} collects the operator perspective divergences and their Loewner data, along with the measures their integrals are taken with respect to in their integral representations, as well as the function emerging from the operator perspective version of Theorem~\ref{fundamental principle}. Table~\ref{tab:lipschitz} is its continuation, collecting where the divergences are finite (in terms of the spectrum of the ARN derivatives) and where in the paper the moduli can be found, along with the optimal Lipschitz constants in the case of bounds in terms of the max-relative entropy. 
\end{example}

\begin{table}[htbp]
\centering\small
\begin{tabular}{@{}lllll@{}}
\hline
divergence & $f(x)$ & $(a,b,c)$ & $d\mu(s)$ & $h_f(t)$ \\
\hline
Belavkin--Staszewski & $x\log x$ & $(0,1,0)$ & $\frac{s}{(1+s)^2}ds$ & $t\log\frac{t}{1-t}$ \\[2pt]
reverse BS & $-\log x$ & $(0,-1,0)$ & $\frac{1}{(1+s)^2}ds$ & $(1-t)\log\frac{1-t}{t}$ \\[2pt]
Jeffreys & $(x-1)\log x$ & $(0,0,0)$ & $\frac{1}{1+s}ds$ & $(2t-1)\log\frac{t}{1-t}$ \\[2pt]
Pearson $\chi^2$ & $(x-1)^2$ & $(0,0,1)$ & $0$ & $\frac{(2t-1)^2}{1-t}$ \\[2pt]
Neyman $\chi^2$ & $\frac{(x-1)^2}{x}$ & $(0,0,0)$ & $\boldsymbol{\delta}_0$ & $\frac{(2t-1)^2}{t}$ \\[2pt]
symmetric $\chi^2$ & $(x-1)^2(1+\frac1x)$ & $(0,0,1)$ & $\boldsymbol{\delta}_0$ & $\frac{(2t-1)^2}{t(1-t)}$ \\[2pt]
Le Cam & $\frac{(x-1)^2}{x+1}$ & $(0,0,0)$ & $\boldsymbol{\delta}_1$ & $(2t-1)^2$ \\[2pt]
Hellinger & $(\sqrt x-1)^2$ & $(0,0,0)$ & $\frac2\pi\frac{\sqrt s}{(1+s)^2}ds$ & $(\sqrt t-\sqrt{1-t})^2$ \\[2pt]
Jensen--Shannon & $x\log x-(1+x)\log\frac{1+x}{2}$ & $(0,0,0)$ & $\frac{\min\{s,1\}}{(1+s)^2}ds$ & $\log2-H_{\mathrm{bin}}(t)$ \\[2pt]
R\'enyi, $\alpha\in(1,2)$ & $x^\alpha$ & $(1,\alpha,0)$ & $\frac{\sin\pi(\alpha-1)}{\pi}\frac{s^\alpha}{(1+s)^2}ds$ & $t^\alpha(1-t)^{1-\alpha}$ \\[2pt]
R\'enyi, $\alpha\in(-1,0)$ & $x^\alpha$ & $(1,\alpha,0)$ & $\frac{-\sin\pi\alpha}{\pi}\frac{s^\alpha}{(1+s)^2}ds$ & $t^\alpha(1-t)^{1-\alpha}$ \\
\hline
\end{tabular}
\caption{Loewner data and extended $h_f$ for the standard examples. For the Jensen--Shannon divergence, $H_{\mathrm{bin}}$ is the binary entropy in nats.}
\label{tab:loewner}
\end{table}
 
\begin{table}[htbp]
\centering\small
\begin{tabular}{@{}llll@{}}
\hline
divergence & finite on & modulus on that interval & $\mathrm{Lip}^{[\epsilon,1-\epsilon]}(h_f)$ \\
\hline
Belavkin--Staszewski & $[0,1)$ & Cor.~\ref{cor:BS} & $\log\frac{1-\epsilon}{\epsilon}+\frac1\epsilon$\ \ $\dagger$ \\[2pt]
reverse BS & $(0,1]$ & mirror of Cor.~\ref{cor:BS} & $\log\frac{1-\epsilon}{\epsilon}+\frac1\epsilon$\ \ $\dagger$ \\[2pt]
Jeffreys & $(0,1)$ & last column & $2\log\frac{1-\epsilon}{\epsilon}+\frac{1-2\epsilon}{\epsilon(1-\epsilon)}$ \\[2pt]
Pearson $\chi^2$ & $[0,1)$ & $(\epsilon^{-2}-4)\delta$ if $\epsilon\leq\frac13$ & $\epsilon^{-2}-4$ \\[2pt]
Neyman $\chi^2$ & $(0,1]$ & $(\epsilon^{-2}-4)\delta$ if $\epsilon\leq\frac13$ & $\epsilon^{-2}-4$ \\[2pt]
symmetric $\chi^2$ & $(0,1)$ & last column & $\epsilon^{-2}-(1-\epsilon)^{-2}$ \\[2pt]
Le Cam & $[0,1]$ & $4\delta$ & $4(1-2\epsilon)$ \\[2pt]
Hellinger & $[0,1]$ & $2\sqrt\delta$ & $\frac{1-2\epsilon}{\sqrt{\epsilon(1-\epsilon)}}$ \\[2pt]
Jensen--Shannon & $[0,1]$ & $2\delta\log\frac1\delta+4\log2\cdot\delta$ & $\log\frac{1-\epsilon}{\epsilon}$ \\[2pt]
R\'enyi, $\alpha\in(1,2)$ & $[0,1)$ & Cor.~\ref{cor:renyi large alpha} & $(\alpha-1+\epsilon)\,\frac{(1-\epsilon)^{\alpha-1}}{\epsilon^{\alpha}}$\ \ $\dagger$ \\[2pt]
R\'enyi, $\alpha\in(-1,0)$ & $(0,1]$ & mirror of Cor.~\ref{cor:renyi large alpha} & $(\epsilon-\alpha)\,\frac{(1-\epsilon)^{-\alpha}}{\epsilon^{1-\alpha}}$\ \ $\dagger$ \\
\hline
\end{tabular}
\caption{Moduli and optimal Lipschitz constants. In the third column, $\epsilon$ is the distance of the spectrum of $T$ from the infinite endpoints, and $\delta=\|T_1-T_2\|_\infty$. The last column is the exact operator Lipschitz constant of $h_f$ on $[\epsilon,1-\epsilon]$, except for the entries marked $\dagger$. These are lower bounds.}
\label{tab:lipschitz}
\end{table}
 
\subsection{All Kubo--Ando means}
 We now repeat the above process using the integral representation of the Kubo--Ando mean, inspired by our foray into the operator perspectives.
Every operator connection $\sigma$ has a representing function of the form \cite{KuboAndo1980}
\[
f_\sigma(x)=a_\sigma+b_\sigma x+\int_{(0,\infty)}\frac{x(1+s)}{x+s}\,dm_\sigma(s),\qquad a_\sigma,b_\sigma\geq0,
\]
with a finite positive measure $m_\sigma$; then $a_\sigma=f_\sigma(0^+)$, $b_\sigma=f_\sigma'(\infty)$, and $\sigma$ is a mean if and only if $a_\sigma+b_\sigma+m_\sigma((0,\infty))=1$. The function $\varphi_\sigma(t)=tf_\sigma((1-t)/t)$  is therefore
\begin{equation}\label{eq:mean kernels}
\varphi_\sigma(t)=a_\sigma\,t+b_\sigma\,(1-t)+\int_{(0,\infty)}\eta_s(t)\,dm_\sigma(s),
\end{equation}
where
\[
\eta_s(t)=\frac{(1+s)\,t(1-t)}{1+(s-1)t}=\Big(\frac{1}{1+s}\,t^{-1}+\frac{s}{1+s}\,(1-t)^{-1}\Big)^{-1},
\]
a mixture of weighted harmonic means of $t$ and $1-t$.
 
By Corollary~\ref{cor:kernels}(c), $\mathrm{Lip}^{[0,1]}(\eta_s)=\max_{[0,1]}|\eta_s'|=(1+s)\max\{1,s^{-1}\}$ for $s>0$.

\begin{lemma}\label{lem:harmonic kernels}
For $s>0$, $0\leq\eta_s\leq\frac{1+s}{(1+\sqrt s)^2}\leq1$ on $[0,1]$.
\end{lemma}
 
\begin{proof}
The maximum of the weighted harmonic mean is attained at $t=1/(1+\sqrt s)$.
\end{proof}
 
\begin{theorem}\label{thm:means modulus}
For every operator connection $\sigma$ with representation as above,
\[
\omega_{\varphi_\sigma}(\delta)\leq|a_\sigma-b_\sigma|\,\delta+\int_{(0,1]}\min\Big\{\frac{(1+s)\,\delta}{s},1\Big\}\,dm_\sigma(s)+\int_{(1,\infty)}\min\{(1+s)\delta,1\}\,dm_\sigma(s)\ \xrightarrow[\delta\to0]{}\ 0 .
\]
Consequently Theorem~\ref{Modulus of continuity results} and Theorem~\ref{thm:four maps}(a) hold with this explicit modulus for every mean. If $dm_\sigma(s)\leq Cs^{\lambda-1}ds$ on $(0,1]$ and $dm_\sigma(s)\leq Cs^{-\lambda'-1}ds$ on $(1,\infty)$ with $\lambda,\lambda'\in(0,1)$, then $\omega_{\varphi_\sigma}(\delta)\leq C'\delta^{\min\{\lambda,\lambda'\}}$.
\end{theorem}
 
\begin{proof}
Insert \eqref{eq:mean kernels} into $\varphi_\sigma(X)-\varphi_\sigma(Y)$; the affine part has slope $a_\sigma-b_\sigma$. Use Corollary~\ref{cor:kernels}(c), together with the bound $\|\eta_s(X)-\eta_s(Y)\|_\infty\leq\max\eta_s\leq1$ of Lemma~\ref{lem:harmonic kernels}. The rate follows as in Corollary~\ref{cor:rates}.
\end{proof}

\section{Pairs with different sums}\label{sec:different sums}
We now come to the more general and noncommutative setting. In general, it is a lot to ask four completely positive maps to have the same sum. In this section we collect the results of what happens when we use the four maps together to dominate. Due to the noncommutative setting, we cannot quite hope for inequalities as nice as before, hence we pick up an extra term in the Thompson metric. What helps us here is that this behaves well enough as a continuity constant.

Let $\Phi_1,\Psi_1,\Phi_2,\Psi_2\in\mathrm{CP}(\mathcal{X},\mathcal{Y})$, put $\Omega_i=\Phi_i+\Psi_i$, $\Omega'=\Omega_1+\Omega_2$ and $\overline{\Omega}=\frac12\Omega'$, and assume $\Omega'\neq0$. Let $\Omega'(X)=\sum_{k=1}^rV_kXV_k^*$ with linearly independent $V_1,\dots,V_r$ and $W=[\vectorise(V_1)\ \cdots\ \vectorise(V_r)]$. Each of the four maps lies below $\Omega'$ in the completely positive order, so by Theorem~\ref{thm:arveson} there are unique $0\leq A_i,B_i\leq\mathbf{1}_{\mathbb{C}^r}$ with $\Phi_i=\Psi_{A_i}$ and $\Psi_i=\Psi_{B_i}$, all taken with respect to $V_1,\dots,V_r$. We put
\[
C_i=A_i+B_i,\qquad\gamma=\|C_1-\tfrac12\mathbf{1}_{\mathbb{C}^r}\|_\infty,\qquad\delta=\max\big\{\|A_1-A_2\|_\infty,\ \|B_1-B_2\|_\infty\big\}.
\]
Since $\Psi_{A_1+B_1+A_2+B_2}=\Omega'=\Psi_{\mathbf{1}_{\mathbb{C}^r}}$, uniqueness in Theorem~\ref{thm:arveson} gives $C_1+C_2=\mathbf{1}_{\mathbb{C}^r}$; in particular $C_1$ and $C_2$ commute and $\gamma=\|C_2-\frac12\mathbf{1}_{\mathbb{C}^r}\|_\infty=\frac12\|C_1-C_2\|_\infty$. Note that $\|\Omega'^*(\mathbf{1}_{\mathcal{Y}})\|_\infty=2\|\overline{\Omega}^*(\mathbf{1}_{\mathcal{Y}})\|_\infty$. By Lemma~\ref{lem:relative distance}, applied to $\Omega'$ with $K=A_1-A_2$ and $K=B_1-B_2$, $\delta$ is the smallest $d\geq0$ such that $-d\,\Omega'\leq\Phi_1-\Phi_2\leq d\,\Omega'$ and $-d\,\Omega'\leq\Psi_1-\Psi_2\leq d\,\Omega'$.

Thus $\delta$ measures the distance between the two pairs \emph{relative} to their total $\Omega'$, in the completely positive order. The second parameter is the distance between the two sums. With $D_{\max}$ as in Lemma~\ref{lem:dmax}, we define the Thompson distance
\[
d_T(\Omega_1,\Omega_2)=\max\big\{D_{\max}(\Omega_1\|\Omega_2),\ D_{\max}(\Omega_2\|\Omega_1)\big\}\in[0,+\infty].
\]
This is the natural object to introduce, as in the commuting case, the infinite behaviour of the divergences was controlled by the max-relative entropy. Physically (Remark~\ref{rem:phys-dmax}), $d_T(\Omega_1,\Omega_2)\leq d$ says that every outcome probability of the device with its outcome forgotten, on every input with ancilla, is reproduced within a factor $e^{\pm d}$: a multiplicative calibration error. Theorem~\ref{thm:four maps} separates the in-pair change $\tau$ from this drift, and Corollary~\ref{cor:four maps cb} and Lemma~\ref{lem:one parameter} bound both by the additive errors $|||\Phi_1-\Phi_2|||_1$ and $|||\Psi_1-\Psi_2|||_1$ (through $\delta$, when $\delta<\frac12$).
\begin{lemma}\label{lem:thompson}
The following are equivalent:
\begin{enumerate}
\item[(i)] $d_T(\Omega_1,\Omega_2)<\infty$;
\item[(ii)] $\range J(\Omega_1)=\range J(\Omega_2)$;
\item[(iii)] $C_1$ and $C_2$ are invertible;
\item[(iv)] $\gamma<\frac12$.
\end{enumerate}
When they hold,
\[
2\gamma=\|C_1-C_2\|_\infty=\tanh\big(\tfrac12d_T(\Omega_1,\Omega_2)\big).
\]
\end{lemma}

\begin{proof}
Let $\lambda_-$ and $\lambda_+$ be the smallest and largest eigenvalues of $C_1$, so that $\gamma=\max\{\lambda_+-\frac12,\frac12-\lambda_-\}$ and the spectrum of $C_2=\mathbf{1}_{\mathbb{C}^r}-C_1$ is $\{1-\lambda:\ \lambda\in\mathrm{spec}(C_1)\}$. As in the proof of Lemma~\ref{lem:dmax}, $\Omega_1\leq c\,\Omega_2$ if and only if $C_1\leq cC_2=c(\mathbf{1}_{\mathbb{C}^r}-C_1)$, that is, $(1+c)C_1\leq c\mathbf{1}_{\mathbb{C}^r}$, that is, $\lambda_+\leq c/(1+c)$. Hence $D_{\max}(\Omega_1\|\Omega_2)=\log\frac{\lambda_+}{1-\lambda_+}$ if $\lambda_+<1$ and $+\infty$ if $\lambda_+=1$; similarly $D_{\max}(\Omega_2\|\Omega_1)=\log\frac{1-\lambda_-}{\lambda_-}$ if $\lambda_->0$ and $+\infty$ if $\lambda_-=0$. Thus (i) holds if and only if $0<\lambda_-\leq\lambda_+<1$, which is (iii), and (iii) is equivalent to (iv) because $\gamma<\frac12$ means exactly $\lambda_->0$ and $\lambda_+<1$.

For (ii), Theorem~\ref{thm:arveson} gives $J(\Omega_i)=WC_iW^*=(WC_i^{1/2})(WC_i^{1/2})^*$, whose range is $W(\range C_i)$; since $W$ is injective, (ii) holds if and only if $\range C_1=\range C_2$. If $C_1$ and $C_2$ are invertible this is clear. Conversely, if $\range C_1=\range C_2=\mathcal{R}$, then $\mathcal{R}\supseteq\range(C_1+C_2)=\mathbb{C}^r$, so both are invertible. Finally, $\tanh(\frac12\log x)=\frac{x-1}{x+1}$ gives $\tanh(\frac12\log\frac{\lambda_+}{1-\lambda_+})=2\lambda_+-1$ and $\tanh(\frac12\log\frac{1-\lambda_-}{\lambda_-})=1-2\lambda_-$; since $\tanh$ is increasing, $\tanh(\frac12d_T)=\max\{2\lambda_+-1,1-2\lambda_-\}=2\gamma$, and $\|C_1-C_2\|_\infty=\|2C_1-\mathbf{1}_{\mathbb{C}^r}\|_\infty=2\gamma$.
\end{proof}

For the rest of this section we assume that $\range J(\Omega_1)=\range J(\Omega_2)$, that is, that the two sums have the same Choi support. By Lemma~\ref{lem:thompson}, $C_1$ and $C_2$ are then invertible, and we put
\[
a_1=C_1^{1/2},\qquad a_2=C_2^{1/2},\qquad T_i=C_i^{-1/2}A_iC_i^{-1/2},\qquad\tau=\|T_1-T_2\|_\infty .
\]

\begin{lemma}\label{lem:inpair}
Under this assumption, $0\leq T_i\leq\mathbf{1}_{\mathbb{C}^r}$ and $\mathbf{1}_{\mathbb{C}^r}-T_i=C_i^{-1/2}B_iC_i^{-1/2}$. Moreover, $T_i$ is the ARN derivative of $\Phi_i$ with respect to $\Omega_i$, relative to the Kraus decomposition of $\Omega_i$ whose column matrix is $WC_i^{1/2}$. Finally, $a_1^2+a_2^2=\mathbf{1}_{\mathbb{C}^r}$, the operators $a_1$ and $a_2$ commute, and
\[
\|a_1-a_2\|_\infty\leq2\gamma,\qquad\|a_1+a_2\|_\infty\leq\sqrt2 .
\]
\end{lemma}

\begin{proof}
Positivity of $A_i$ and $B_i$ gives $T_i\geq0$ and $C_i^{-1/2}B_iC_i^{-1/2}\geq0$, and these two operators add up to $C_i^{-1/2}C_iC_i^{-1/2}=\mathbf{1}_{\mathbb{C}^r}$. The columns of $WC_i^{1/2}$ are $\vectorise(\widetilde V_j)$, where $\widetilde V_j=\sum_k(C_i^{1/2})_{kj}V_k$; these are linearly independent because $C_i^{1/2}$ is invertible. By Theorem~\ref{thm:arveson}, $J(\Omega_i)=WC_iW^*=(WC_i^{1/2})(WC_i^{1/2})^*$, so the $\widetilde V_j$ form a Kraus decomposition of $\Omega_i$, and $J(\Phi_i)=WA_iW^*=(WC_i^{1/2})\,T_i\,(WC_i^{1/2})^*$ identifies $T_i$ as the derivative of $\Phi_i$ in this decomposition. The operators $a_1$ and $a_2$ are functions of $C_1$, hence commute, and $a_1^2+a_2^2=C_1+C_2=\mathbf{1}_{\mathbb{C}^r}$. If $\lambda$ runs over the spectrum of $C_1$, the eigenvalues of $a_1\mp a_2$ are $\sqrt\lambda\mp\sqrt{1-\lambda}$, and
\[
|\sqrt\lambda-\sqrt{1-\lambda}|=\frac{|2\lambda-1|}{\sqrt\lambda+\sqrt{1-\lambda}}\leq|2\lambda-1|\leq2\gamma,\qquad\sqrt\lambda+\sqrt{1-\lambda}\leq\sqrt2,
\]
by $\sqrt\lambda+\sqrt{1-\lambda}\geq1$ and the Cauchy--Schwarz inequality.
\end{proof}

\begin{lemma}\label{lem:sandwich}
Let $\varphi:[0,1]\to\mathbb{R}$ be continuous, with $\|\varphi\|_\infty=\max_{[0,1]}|\varphi|$. Then
\[
\big\|a_1\,\varphi(T_1)\,a_1-a_2\,\varphi(T_2)\,a_2\big\|_\infty\leq\tfrac12\,\omega_\varphi(\tau)+2\sqrt2\,\gamma\,\|\varphi\|_\infty .
\]
If $T_1,T_2\leq(1-\epsilon)\mathbf{1}_{\mathbb{C}^r}$ for some $\epsilon\in(0,1)$ and $\varphi$ is only assumed continuous on $[0,1-\epsilon]$, the same holds with $\omega_\varphi$ replaced by $\omega^{(\epsilon)}_\varphi$ and $\|\varphi\|_\infty$ by $\max_{[0,1-\epsilon]}|\varphi|$. More generally, if $\mathrm{spec}(T_1),\mathrm{spec}(T_2)\subseteq I$ for a compact interval $I$ on which $\varphi$ is continuous, the same holds with $\omega^I_\varphi$ and $\max_I|\varphi|$.
\end{lemma}

\begin{proof}
Put $H_-=\varphi(T_1)-\varphi(T_2)$ and $H_+=\varphi(T_1)+\varphi(T_2)$, both Hermitian, with $\|H_-\|_\infty\leq\omega_\varphi(\tau)$ and $\|H_+\|_\infty\leq2\|\varphi\|_\infty$. Then
\[
a_1\varphi(T_1)a_1-a_2\varphi(T_2)a_2=\tfrac12\big(a_1H_-a_1+a_2H_-a_2\big)+\tfrac14\Big((a_1-a_2)H_+(a_1+a_2)+(a_1+a_2)H_+(a_1-a_2)\Big),
\]
as one checks by expanding the right-hand side. Since $-\|H_-\|_\infty\mathbf{1}_{\mathbb{C}^r}\leq H_-\leq\|H_-\|_\infty\mathbf{1}_{\mathbb{C}^r}$ and $a_1^2+a_2^2=\mathbf{1}_{\mathbb{C}^r}$, conjugating by $a_1$ and by $a_2$ and adding gives $-\|H_-\|_\infty\mathbf{1}_{\mathbb{C}^r}\leq a_1H_-a_1+a_2H_-a_2\leq\|H_-\|_\infty\mathbf{1}_{\mathbb{C}^r}$, so the first term has norm at most $\frac12\|H_-\|_\infty\leq\frac12\omega_\varphi(\tau)$. By Lemma~\ref{lem:inpair}, the second term has norm at most $\frac12\|a_1-a_2\|_\infty\|H_+\|_\infty\|a_1+a_2\|_\infty\leq\frac12\cdot2\gamma\cdot2\|\varphi\|_\infty\cdot\sqrt2=2\sqrt2\,\gamma\,\|\varphi\|_\infty$. The second statement is proved in the same way.
\end{proof}

\begin{theorem}\label{thm:four maps}
Let $\Phi_1,\Psi_1,\Phi_2,\Psi_2\in\mathrm{CP}(\mathcal{X},\mathcal{Y})$ be such that $\range J(\Phi_1+\Psi_1)=\range J(\Phi_2+\Psi_2)$, and let $\overline{\Omega}$, $\tau$ and $d_T=d_T(\Phi_1+\Psi_1,\Phi_2+\Psi_2)$ be as above.
\begin{enumerate}
\item[(a)] For every Kubo--Ando mean $\sigma$,
\[
|||\Phi_1\,\sigma\,\Psi_1-\Phi_2\,\sigma\,\Psi_2|||_1\leq\|\overline{\Omega}^*(\mathbf{1}_{\mathcal{Y}})\|_\infty\Big(\omega_{g_\sigma}(\tau)+2\sqrt2\,\|g_\sigma\|_\infty\tanh\tfrac{d_T}{2}\Big).
\]
\item[(b)] For the geometric channel fidelity,
\[
\big|\widehat{F}(\Phi_1,\Psi_1)-\widehat{F}(\Phi_2,\Psi_2)\big|\leq\|\overline{\Omega}^*(\mathbf{1}_{\mathcal{Y}})\|_\infty\Big(\tau^{1/2}+\sqrt2\tanh\tfrac{d_T}{2}\Big).
\]
\item[(c)] For $\alpha\in(0,1)$,
\[
\big|\widehat{Q}_\alpha(\Phi_1\|\Psi_1)-\widehat{Q}_\alpha(\Phi_2\|\Psi_2)\big|\leq\|\overline{\Omega}^*(\mathbf{1}_{\mathcal{Y}})\|_\infty\Big(\tau^\alpha+\tau^{1-\alpha}+2\sqrt2\tanh\tfrac{d_T}{2}\Big).
\]
\item[(d)] For $f$ as in Lemma~\ref{lem:maxf decomposition},
\[
\big|\widehat{D}_f(\Phi_1\|\Psi_1)-\widehat{D}_f(\Phi_2\|\Psi_2)\big|\leq\|\overline{\Omega}^*(\mathbf{1}_{\mathcal{Y}})\|_\infty\Big(\omega_{h_f}(\tau)+2\sqrt2\,\|h_f\|_\infty\tanh\tfrac{d_T}{2}\Big).
\]
\item[(e)] For $f$ as in Section~\ref{sec:unbounded}, if moreover $D_{\max}(\Phi_i\|\Psi_i)\leq\log\frac{1-\epsilon}{\epsilon}$ for $i=1,2$ and some $\epsilon\in(0,1)$, then (d) holds with $\omega_{h_f}$ replaced by $\omega^{(\epsilon)}_{h_f}$ and $\|h_f\|_\infty$ by $\max_{[0,1-\epsilon]}|h_f|$.
\item[(f)] For every operator convex $f$ and every compact interval $I\subseteq[0,1]$ on which $h_f$ is finite, if $\mathrm{spec}(T_i)\subseteq I$ for $i=1,2$, then
\[
\big|\widehat{D}_f(\Phi_1\|\Psi_1)-\widehat{D}_f(\Phi_2\|\Psi_2)\big|\leq\|\overline{\Omega}^*(\mathbf{1}_{\mathcal{Y}})\|_\infty\Big(\omega^I_{h_f}(\tau)+2\sqrt2\,\max_I|h_f|\,\tanh\tfrac{d_T}{2}\Big).
\]
For $I=[\epsilon,1-\epsilon]$ the hypothesis says $\max\{D_{\max}(\Phi_i\|\Psi_i),D_{\max}(\Psi_i\|\Phi_i)\}\leq\log\frac{1-\epsilon}{\epsilon}$, and Corollary~\ref{cor:thompson lipschitz} gives $\omega^I_{h_f}(\tau)\leq L_f(\epsilon)\tau$.
\end{enumerate}
If $\Phi_1+\Psi_1=\Phi_2+\Psi_2$, then $d_T=0$, $\overline{\Omega}=\Phi_1+\Psi_1$, and $\tau=\|T_1-T_2\|_\infty$ for the derivatives $T_i$ of $\Phi_i$ with respect to this common sum, so that (a)--(e) reduce to Theorems~\ref{Modulus of continuity results}, \ref{main theorem}, \ref{thm:renyi continuity}, \ref{thm:maxf}(b) and~\ref{thm:unbounded}(b) with the same constants.
\end{theorem}

\begin{proof}
By Lemma~\ref{lem:inpair}, $A_i=a_1T_ia_1$, $B_i=a_1(\mathbf{1}_{\mathbb{C}^r}-T_i)a_1$ for $i=1$, and likewise with $a_2$ for $i=2$. By Lemma~\ref{lem:thompson}, $2\gamma=\tanh\frac{d_T}{2}$, so Lemma~\ref{lem:sandwich} reads
\begin{equation}\label{eq:sandwich}
\big\|a_1\varphi(T_1)a_1-a_2\varphi(T_2)a_2\big\|_\infty\leq\tfrac12\,\omega_\varphi(\tau)+\sqrt2\,\|\varphi\|_\infty\tanh\tfrac{d_T}{2}.
\end{equation}

 Let $\sigma'$ be an operator connection. By \eqref{eq:W transfer}, which holds for connections by Remark~\ref{rem:connections}, $J(\Phi_i)\,\sigma'\,J(\Psi_i)=W(A_i\,\sigma'\,B_i)W^*$. The transformer equality for the invertible operator $a_i$ \cite{KuboAndo1980}, followed by Remark~\ref{rem:connections}(2) for the commuting pair $T_i$, $\mathbf{1}_{\mathbb{C}^r}-T_i$, gives
\[
A_i\,\sigma'\,B_i=a_i\big(T_i\,\sigma'\,(\mathbf{1}_{\mathbb{C}^r}-T_i)\big)a_i=a_i\,\varphi_{\sigma'}(T_i)\,a_i,
\]
and in the same way $B_i\,\sigma'\,A_i=a_i\,\varphi_{\sigma'}(\mathbf{1}_{\mathbb{C}^r}-T_i)\,a_i$. For a mean $\sigma$ we have $\varphi_\sigma=g_\sigma$.

 If $N_1,N_2\in\mathrm{L}(\mathbb{C}^r)$ are Hermitian, Lemma~\ref{lem:partial trace}, applied to $\Omega'$, gives
\begin{align*}
\big\|\Tr_{\mathcal{Y}}(WN_1W^*)-\Tr_{\mathcal{Y}}(WN_2W^*)\big\|_\infty&\leq\|\Omega'^*(\mathbf{1}_{\mathcal{Y}})\|_\infty\,\|N_1-N_2\|_\infty\\
&=2\|\overline{\Omega}^*(\mathbf{1}_{\mathcal{Y}})\|_\infty\|N_1-N_2\|_\infty,
\end{align*}
and the same bound holds for the difference of the smallest, or of the largest, eigenvalues. The factor $2$ combines with the factor $\frac12$ in \eqref{eq:sandwich}.

(a) By definition of the mean and the first step, $J(\Phi_i\,\sigma\,\Psi_i)=J(\Phi_i)\,\sigma\,J(\Psi_i)=W(A_i\,\sigma\,B_i)W^*$, so $\Phi_i\,\sigma\,\Psi_i=\Psi_{N_i}$ with respect to $V_1,\dots,V_r$, where $N_i=a_ig_\sigma(T_i)a_i$. Theorem~\ref{thm:two sided}, applied to $\Omega'$, gives $|||\Psi_{N_1}-\Psi_{N_2}|||_1\leq\|\Omega'^*(\mathbf{1}_{\mathcal{Y}})\|_\infty\|N_1-N_2\|_\infty$. Now apply \eqref{eq:sandwich} with $\varphi=g_\sigma$.

(b) By Proposition~\ref{lambda min formula} and the first step, $\widehat{F}(\Phi_i,\Psi_i)=\lambda_{\min}\big(\Tr_{\mathcal{Y}}(WN_iW^*)\big)$ with $N_i=a_ig_\#(T_i)a_i$. Apply the partial-trace estimate and \eqref{eq:sandwich} with $\varphi=g_\#$, using $\omega_{g_\#}(t)\leq t^{1/2}$ (Proposition~\ref{K6}(2)) and $\|g_\#\|_\infty=\frac12$.

(c) By Lemma~\ref{renyi closed form}, $\widehat{Q}_\alpha(\Phi_i\|\Psi_i)=\lambda_{\min}\big(\Tr_{\mathcal{Y}}(J(\Psi_i)\,\#_\alpha\,J(\Phi_i))\big)$, and since $g_{\#_{1-\alpha}}=h_\alpha$, the first step gives
\[
J(\Psi_i)\,\#_\alpha\,J(\Phi_i)=J(\Phi_i)\,\#_{1-\alpha}\,J(\Psi_i)=W\,a_ih_\alpha(T_i)a_i\,W^*.
\]
Apply the partial-trace estimate and \eqref{eq:sandwich} with $\varphi=h_\alpha$, using $\omega_{h_\alpha}(t)\leq t^\alpha+t^{1-\alpha}$ (Lemma~\ref{holder lemma}) and $\|h_\alpha\|_\infty\leq1$.

(d) By Lemma~\ref{lem:inpair}, $T_i$ is the derivative of $\Phi_i$ with respect to $\Omega_i$, relative to the Kraus decomposition with column matrix $Wa_i$. Theorem~\ref{thm:maxf}(a), applied to $\Omega_i$ with this matrix in place of $W$, gives $\widehat{D}_f(\Phi_i\|\Psi_i)=\lambda_{\max}\big(\Tr_{\mathcal{Y}}(W\,a_ih_f(T_i)a_i\,W^*)\big)$. Apply the partial-trace estimate and \eqref{eq:sandwich} with $\varphi=h_f$.

(e) Put $c=(1-\epsilon)/\epsilon$. The set of $c'\geq0$ with $\Phi_i\leq c'\Psi_i$ is closed and upward closed, so the definition of $D_{\max}$ in Lemma~\ref{lem:dmax} gives $\Phi_i\leq c\Psi_i$, that is, $W(cB_i-A_i)W^*\geq0$ by Choi's theorem; conjugating by $W^+$ gives $cB_i\geq A_i$, and conjugating by $a_i^{-1}$ gives $c(\mathbf{1}_{\mathbb{C}^r}-T_i)\geq T_i$, that is, $T_i\leq(1-\epsilon)\mathbf{1}_{\mathbb{C}^r}$. In particular $B_i$ is invertible, so Lemma~\ref{lem:unb congruence} with $C=W$, Lemma~\ref{lem:unb congruence pd} with $R=a_i$, and the computation in the proof of Theorem~\ref{thm:unbounded}(a) give
\[
\mathsf{P}_f\big(J(\Phi_i),J(\Psi_i)\big)=W\,\mathsf{P}_f(A_i,B_i)\,W^*=W\,a_ih_f(T_i)a_i\,W^*.
\]
The remainder of the proof of Theorem~\ref{thm:unbounded}(a) uses only $\Phi_i\leq c\Psi_i$ and this representation, and shows that $\widehat{D}_f(\Phi_i\|\Psi_i)=\lambda_{\max}\big(\Tr_{\mathcal{Y}}(W\,a_ih_f(T_i)a_i\,W^*)\big)$. Apply the partial-trace estimate and the second statement of Lemma~\ref{lem:sandwich}, in the form \eqref{eq:sandwich}, with $\varphi=h_f$.

(f) By Lemma~\ref{lem:inpair}, $T_i$ is the derivative of $\Phi_i$ with respect to $\Omega_i$, relative to the Kraus decomposition with column matrix $Wa_i$. Theorem~\ref{thm:all f}(b), applied to $\Omega_i$ with this matrix in place of $W$, gives $\widehat{D}_f(\Phi_i\|\Psi_i)=\lambda_{\max}\big(\Tr_{\mathcal{Y}}(W\,a_ih_f(T_i)a_i\,W^*)\big)$, since $h_f(T_i)$ is finite. Apply the partial-trace estimate and the last statement of Lemma~\ref{lem:sandwich} with $\varphi=h_f$. The identification of the hypothesis for $I=[\epsilon,1-\epsilon]$ is Lemma~\ref{lem:dmax} applied within each pair, in the frame $Wa_i$, to $T_i$ and to $\mathbf{1}_{\mathbb{C}^r}-T_i$.

Finally, if $\Phi_1+\Psi_1=\Phi_2+\Psi_2=\Omega$, then $\Omega'=2\Omega$ and $d_T=0$; the derivatives with respect to $\Omega'$ are $A_i=\frac12T_i$ and $B_i=\frac12(\mathbf{1}_{\mathbb{C}^r}-T_i)$, where $T_i$ is the derivative of $\Phi_i$ with respect to $\Omega$, so $C_1=C_2=\frac12\mathbf{1}_{\mathbb{C}^r}$ and the operators $T_i$ defined above coincide with these derivatives.
\end{proof}

\begin{lemma}\label{lem:one parameter}
Without any assumption on the supports, $\gamma\leq\delta$. If $\delta<\frac12$, then $\range J(\Omega_1)=\range J(\Omega_2)$ and
\[
\tau\leq\frac{2\sqrt2\,\delta}{\sqrt{1-4\gamma^2}}\leq\frac{2\sqrt2\,\delta}{\sqrt{1-4\delta^2}} .
\]
The constant $2\sqrt2$ is optimal: in Example~\ref{ex:one parameter sharp}, $\tau/\delta\to2\sqrt2$ as $\delta\to0$.
\end{lemma}

\begin{proof}
We have $2\gamma=\|C_1-C_2\|_\infty\leq\|A_1-A_2\|_\infty+\|B_1-B_2\|_\infty\leq2\delta$. If $\delta<\frac12$ then $\gamma<\frac12$, and Lemma~\ref{lem:thompson} gives the support condition. Put $u=a_1+a_2$ and $v=a_1-a_2$, which commute with $C_1$, and
\[
\bar T=\tfrac12(T_1+T_2),\qquad\bar T'=\mathbf{1}_{\mathbb{C}^r}-\bar T,\qquad M=\bar T-\bar T',\qquad T_-=T_1-T_2 .
\]
Since $0\leq T_i\leq\mathbf{1}_{\mathbb{C}^r}$, also $0\leq\bar T,\bar T'\leq\mathbf{1}_{\mathbb{C}^r}$. By Lemma~\ref{lem:inpair}, $A_i-B_i=a_i(2T_i-\mathbf{1}_{\mathbb{C}^r})a_i$, and $2T_1-\mathbf{1}_{\mathbb{C}^r}=M+T_-$, $2T_2-\mathbf{1}_{\mathbb{C}^r}=M-T_-$, so
\[
(A_1-A_2)-(B_1-B_2)=(a_1Ma_1-a_2Ma_2)+(a_1T_-a_1+a_2T_-a_2).
\]
For every $X$, $a_1Xa_1+a_2Xa_2=\frac12(uXu+vXv)$ and $a_1Xa_1-a_2Xa_2=\frac12(uXv+vXu)$, as one checks by expanding. Moreover $uv=a_1^2-a_2^2=C_1-C_2=(A_1-A_2)+(B_1-B_2)$. Put
\[
\tilde E=u^{-1}(A_1-A_2)u^{-1},\qquad\tilde F=u^{-1}(B_1-B_2)u^{-1},\qquad N=u^{-1}v .
\]
Since $u$ and $v$ commute, $N$ is Hermitian and $N=u^{-1}(uv)u^{-1}=\tilde E+\tilde F$. Conjugating the displayed identity by $u^{-1}$ gives $\tilde E-\tilde F=\frac12(NM+MN)+\frac12(T_-+NT_-N)$. Since $M=\bar T-\bar T'$ and $\mathbf{1}_{\mathbb{C}^r}=\bar T+\bar T'$, this is the exact identity
\begin{equation}\label{eq:one parameter identity}
T_-+NT_-N=2Y,\qquad Y=\tilde E\bar T'+\bar T'\tilde E-\tilde F\bar T-\bar T\tilde F .
\end{equation}
Write $Y=\begin{bmatrix}\bar T'&-\bar T\end{bmatrix}\begin{bmatrix}\tilde E\\\tilde F\end{bmatrix}+\begin{bmatrix}\tilde E&-\tilde F\end{bmatrix}\begin{bmatrix}\bar T'\\\bar T\end{bmatrix}$ as a sum of products of a row and a column of operators. We have $\|[\bar T'\ \ -\bar T]\|_\infty^2=\|\bar T'^2+\bar T^2\|_\infty\leq\|\bar T'+\bar T\|_\infty=1$, and likewise for the column $(\bar T',\bar T)$. With $\tilde\delta=\max\{\|\tilde E\|_\infty,\|\tilde F\|_\infty\}$, the column $(\tilde E,\tilde F)$ and the row $[\tilde E\ \ -\tilde F]$ have norm at most $\sqrt2\,\tilde\delta$. Hence $\|Y\|_\infty\leq2\sqrt2\,\tilde\delta$.

Let $\lambda$ run over the spectrum of $C_1$, which lies in $[\frac12-\gamma,\frac12+\gamma]$, and put $\theta_\lambda=2\sqrt{\lambda(1-\lambda)}$ and $s_\gamma=\sqrt{1-4\gamma^2}$. Since $4\lambda(1-\lambda)=1-(2\lambda-1)^2$, we have $\theta_\lambda\geq s_\gamma$. The eigenvalues of $u^2=\mathbf{1}_{\mathbb{C}^r}+2a_1a_2$ are $1+\theta_\lambda\geq1+s_\gamma$, so $\tilde\delta\leq\delta/(1+s_\gamma)$. Those of $N^2=v^2u^{-2}$ are $(1-\theta_\lambda)/(1+\theta_\lambda)\leq(1-s_\gamma)/(1+s_\gamma)$. By \eqref{eq:one parameter identity}, $\|T_-\|_\infty\leq2\|Y\|_\infty+\|N\|_\infty^2\|T_-\|_\infty$, and therefore
\[
\tau=\|T_-\|_\infty\leq\frac{2\|Y\|_\infty}{1-\|N\|_\infty^2}\leq2\cdot\frac{2\sqrt2\,\delta}{1+s_\gamma}\cdot\frac{1+s_\gamma}{2s_\gamma}=\frac{2\sqrt2\,\delta}{s_\gamma}.
\]
Finally $\gamma\leq\delta$.
\end{proof}

\begin{example}\label{ex:one parameter sharp}
Let $r=2$, $Z=\mathrm{diag}(1,-1)$, $X=\big[\begin{smallmatrix}0&1\\1&0\end{smallmatrix}\big]$, and let $U_0=(Z+X)/\sqrt2$ and $U_1=(Z-X)/\sqrt2$, which are symmetric unitaries. For $0<\delta<\frac12$ put $T_0=\mathrm{diag}(\delta,1-\delta)$ and
\[
A_{1,2}=\tfrac12\big(T_0\pm\delta U_0\big),\qquad B_{1,2}=\tfrac12\big(\mathbf{1}_{\mathbb{C}^2}-T_0\pm\delta U_1\big).
\]
These are positive semidefinite, since all four have positive diagonal entries and determinant at least $\frac{\sqrt2-1}{4\sqrt2}\,\delta(1-2\delta)>0$. They sum to $\mathbf{1}_{\mathbb{C}^2}$, and $\|A_1-A_2\|_\infty=\|B_1-B_2\|_\infty=\delta$; by Theorem~\ref{thm:arveson} they are the derivatives of four completely positive maps with respect to any linearly independent pair $V_1,V_2$. Here $C_1=\frac12(\mathbf{1}_{\mathbb{C}^2}+\sqrt2\,\delta Z)$, so $u=\sqrt2\,\mathbf{1}_{\mathbb{C}^2}+O(\delta)$, $N=O(\delta)$, $\tilde E=\frac\delta2U_0+O(\delta^2)$, $\tilde F=\frac\delta2U_1+O(\delta^2)$ and $\bar T=\mathrm{diag}(0,1)+O(\delta)$. A direct computation gives
\[
U_0\bar T'_0+\bar T'_0U_0-U_1\bar T_0-\bar T_0U_1=\sqrt2\,\mathbf{j}\mathbf{j}^\top,\qquad\bar T_0=\mathrm{diag}(0,1),\ \bar T'_0=\mathrm{diag}(1,0),
\]
so $Y=\frac{\delta}{\sqrt2}\mathbf{j}\mathbf{j}^\top+O(\delta^2)$, and \eqref{eq:one parameter identity} gives $\tau=2\sqrt2\,\delta+O(\delta^2)$. Since $2\sqrt2>2$ and $\frac{2\delta}{1-2\delta}=2\delta+O(\delta^2)$, we get $\tau>\frac{2\delta}{1-2\delta}$ for small $\delta$, so the bound $2\delta/(1-2\delta)$ of the commuting case does not hold in general.
\end{example}

\begin{corollary}\label{cor:universal}
For all $\Phi_1,\Psi_1,\Phi_2,\Psi_2\in\mathrm{CP}(\mathcal{X},\mathcal{Y})$ with $\Omega'\neq0$,
\[
\big|\widehat{F}(\Phi_1,\Psi_1)-\widehat{F}(\Phi_2,\Psi_2)\big|\leq\tfrac{11}{4}\,\|\overline{\Omega}^*(\mathbf{1}_{\mathcal{Y}})\|_\infty\,\delta^{1/2},
\]
where $\delta$ is the relative distance of Lemma~\ref{lem:relative distance}. For channels the constant is $\frac{11}{2}$.
\end{corollary}

\begin{proof}
Suppose first that $\delta\leq\frac2{15}$. By Lemma~\ref{lem:one parameter}, $\tau\leq2\sqrt2\,\delta/\sqrt{1-16/225}<2.935\,\delta$. By Lemmas~\ref{lem:thompson} and~\ref{lem:one parameter}, $\tanh\frac{d_T}{2}=2\gamma\leq2\delta\leq2(\frac2{15})^{1/2}\delta^{1/2}$. Theorem~\ref{thm:four maps}(b) therefore gives the bound with the constant
\[
2.935^{1/2}+2\sqrt2\,\big(\tfrac2{15}\big)^{1/2}<1.714+1.033<\tfrac{11}{4}.
\]
If $\delta>\frac2{15}$, we use that $0\leq\widehat{F}(\Phi_i,\Psi_i)\leq\|\overline{\Omega}^*(\mathbf{1}_{\mathcal{Y}})\|_\infty$: indeed, $J(\Phi_i)\#J(\Psi_i)\leq\frac12(J(\Phi_i)+J(\Psi_i))=\frac12J(\Omega_i)$ by the arithmetic--geometric mean inequality \cite{KuboAndo1980}, so Proposition~\ref{lambda min formula} gives $\widehat{F}(\Phi_i,\Psi_i)\leq\frac12\lambda_{\min}(\Omega_i^*(\mathbf{1}_{\mathcal{Y}})^\top)\leq\frac12\|\Omega_i^*(\mathbf{1}_{\mathcal{Y}})\|_\infty\leq\|\overline{\Omega}^*(\mathbf{1}_{\mathcal{Y}})\|_\infty$, using $\Omega_i\leq\Omega'=2\overline{\Omega}$. Hence the left-hand side is at most $\|\overline{\Omega}^*(\mathbf{1}_{\mathcal{Y}})\|_\infty<(\frac{15}{2})^{1/2}\|\overline{\Omega}^*(\mathbf{1}_{\mathcal{Y}})\|_\infty\delta^{1/2}$, and $(\frac{15}{2})^{1/2}<2.74$. For channels, $\overline{\Omega}^*(\mathbf{1}_{\mathcal{Y}})=2\cdot\mathbf{1}_{\mathcal{X}}$.
\end{proof}

\begin{corollary}\label{cor:four maps cb}
For all $\Phi_1,\Psi_1,\Phi_2,\Psi_2\in\mathrm{CP}(\mathcal{X},\mathcal{Y})$ with $\Omega'\neq0$,
\[
\delta\leq\frac{\max\big\{|||\Phi_1-\Phi_2|||_1,\ |||\Psi_1-\Psi_2|||_1\big\}}{\nu(\Omega')^2},
\]
and if $\range J(\Phi_2+\Psi_2)\subseteq\range J(\Phi_1+\Psi_1)$, then $\nu(\Omega')\geq\nu(\Phi_1+\Psi_1)$. Consequently, in the latter case,
\[
\big|\widehat{F}(\Phi_1,\Psi_1)-\widehat{F}(\Phi_2,\Psi_2)\big|\leq\frac{11\,\|\overline{\Omega}^*(\mathbf{1}_{\mathcal{Y}})\|_\infty}{4\,\nu(\Phi_1+\Psi_1)}\,\max\big\{|||\Phi_1-\Phi_2|||_1,\ |||\Psi_1-\Psi_2|||_1\big\}^{1/2}.
\]
\end{corollary}

\begin{proof}
With respect to the Kraus decomposition of $\Omega'$ we have $\Phi_1-\Phi_2=\Psi_{A_1-A_2}$ and $\Psi_1-\Psi_2=\Psi_{B_1-B_2}$, so the first inequality is Theorem~\ref{thm:two sided} applied to $\Omega'$. If $\range J(\Phi_2+\Psi_2)\subseteq\range J(\Phi_1+\Psi_1)$, then, since the range of a sum of positive semidefinite operators is the sum of their ranges, $\range J(\Omega')=\range J(\Phi_1+\Psi_1)$, and $\Omega'\geq\Phi_1+\Psi_1$ in the completely positive order, so Lemma~\ref{lem:nu}(3) gives $\nu(\Omega')\geq\nu(\Phi_1+\Psi_1)$. The last statement follows from Corollary~\ref{cor:universal}.
\end{proof}

\begin{remark}
Corollary~\ref{cor:four maps cb} shows that the geometric channel fidelity is locally H\"older continuous of order $\frac12$ in the completely bounded trace norm at every pair $(\Phi_1,\Psi_1)$, along perturbations that do not enlarge the Choi support of the sum, with a constant that depends on the pair only through $\|\overline{\Omega}^*(\mathbf{1}_{\mathcal{Y}})\|_\infty$ and $\nu(\Phi_1+\Psi_1)$, and not on any dimension. The same argument applies to the other quantities in Theorem~\ref{thm:four maps}, with the corresponding moduli. The support condition is needed here only to control $\tau$ by $\delta$: if $\Phi_2$ adds a small multiple $t\Xi$ of a map whose Choi matrix has range outside $\range J(\Phi_1+\Psi_1)$, then $\nu(\Omega')^2$ may be of order $t$, and Corollary~\ref{impossibility corollary} shows that some dependence on $\nu$ is unavoidable. Section~\ref{sec:different supports} extends Theorem~\ref{thm:four maps} itself to arbitrary supports. In the commuting case $\|T_1-T_2\|_\infty\leq2\delta/(1-2\delta)$, as a scalar computation shows. Example~\ref{ex:one parameter sharp} shows that noncommutativity costs a factor $\sqrt2$ in the leading constant of Lemma~\ref{lem:one parameter}, and no more.
\end{remark}

\section{Sums with different Choi supports}\label{sec:different supports}

We now drop the hypothesis $\range J(\Omega_1)=\range J(\Omega_2)$ of Section~\ref{sec:different sums}. Its only role there was to make $C_1$ and $C_2$ invertible. Lemma~\ref{lem:sandwich} then measures the change of the sums by $\gamma=\|C_1-\frac12\mathbf{1}_{\mathbb{C}^r}\|_\infty$, a supremum over the spectrum of $C_1$. A single direction that is present in one sum and absent from the other gives $\gamma=\frac12$, however small its weight, and the estimate becomes trivial. We replace this supremum by weighted estimates, organised by the Lebesgue decomposition of Section~\ref{sec:all f}.

We keep the notation of Section~\ref{sec:different sums}: $\Omega'$, $V_1,\dots,V_r$, $W$, $A_i$, $B_i$, $C_i$, $a_1=C_1^{1/2}$ and $a_2=C_2^{1/2}$. Put
\[
\Delta(Z)=\Tr_{\mathcal{Y}}(WZW^*),\qquad Z\in\mathrm{L}(\mathbb{C}^r),
\]
a completely positive map. For $K\geq0$, $\Delta(K)=\Tr_{\mathcal{Y}}J(\Psi_K)$ is the transpose of $\Psi_K^*(\mathbf{1}_{\mathcal{Y}})$, so $\|\Delta(\mathbf{1}_{\mathbb{C}^r})\|_\infty=\|\Omega'^*(\mathbf{1}_{\mathcal{Y}})\|_\infty=2\|\overline{\Omega}^*(\mathbf{1}_{\mathcal{Y}})\|_\infty$ and $\|\Delta(C_i)\|_\infty=\|\Omega_i^*(\mathbf{1}_{\mathcal{Y}})\|_\infty$. We use two elementary facts. If $Z$ is Hermitian and $-K\leq Z\leq K$, then
\begin{equation}\label{eq:cp domination}
\|\Delta(Z)\|_\infty\leq\|\Delta(K)\|_\infty,\qquad|||\Psi_Z|||_1\leq\|\Psi_K^*(\mathbf{1}_{\mathcal{Y}})\|_\infty .
\end{equation}
The first holds because $\Delta$ is positive. For the second, $\Psi_K\pm\Psi_Z=\Psi_{K\pm Z}$ is completely positive, so $-(\Psi_K\otimes\mathbf{1}_{\mathrm{L}(\mathcal{X})})(\rho)\leq(\Psi_Z\otimes\mathbf{1}_{\mathrm{L}(\mathcal{X})})(\rho)\leq(\Psi_K\otimes\mathbf{1}_{\mathrm{L}(\mathcal{X})})(\rho)$ for every state $\rho$, and therefore $\|(\Psi_Z\otimes\mathbf{1}_{\mathrm{L}(\mathcal{X})})(\rho)\|_1\leq\tr(\Psi_K\otimes\mathbf{1}_{\mathrm{L}(\mathcal{X})})(\rho)\leq\|\Psi_K^*(\mathbf{1}_{\mathcal{Y}})\|_\infty$.

\subsection{The Lebesgue decomposition of the sums}

Let $Q_1$, $Q_R$ and $Q_2$ be the spectral projections of $C_1$ for $\{1\}$, $(0,1)$ and $\{0\}$. They commute with $a_1$ and $a_2$, and $P_1=Q_1+Q_R$ and $P_2=Q_R+Q_2$ are the range projections of $C_1$ and $C_2=\mathbf{1}_{\mathbb{C}^r}-C_1$. By Proposition~\ref{prop:lebesgue cp}, applied to the pair $(\Omega_1,\Omega_2)$, whose derivative with respect to $\Omega'$ is $C_1$, the singular part of $\Omega_1$ relative to $\Omega_2$ is $(\Omega_1)_{\rm s}=\Psi_{Q_1}$, and that of $\Omega_2$ relative to $\Omega_1$ is $(\Omega_2)_{\rm s}=\Psi_{Q_2}$. The absolutely continuous parts $\Omega_i^{\rm ac}=\Omega_i-(\Omega_i)_{\rm s}=\Psi_{Q_RC_i}$ have the same Choi range, and the supports of the sums agree if and only if $Q_1=Q_2=0$.

\begin{lemma}\label{lem:inpair general}
Let $T_1=a_1^+A_1a_1^+$ and $T_2=a_2^+A_2a_2^+$, where $^+$ denotes the Moore--Penrose inverse, regarded as operators on $\range P_1$ and $\range P_2$.
\begin{enumerate}
\item[(a)] $0\leq T_i\leq P_i$, $A_i=a_iT_ia_i$ and $B_i=a_i(P_i-T_i)a_i$.
\item[(b)] If $U_i$ is an isometry with range $\range P_i$, then $U_i^*T_iU_i$ is the derivative of $\Phi_i$ with respect to $\Omega_i$, relative to the Kraus decomposition of $\Omega_i$ whose column matrix is $Wa_iU_i$.
\item[(c)] For every operator connection $\sigma'$, $A_i\,\sigma'\,B_i=a_i\,\varphi_{\sigma'}(T_i)\,a_i$, where $\varphi_{\sigma'}(T_i)$ is formed by the functional calculus in $\mathrm{L}(\range P_i)$ and extended by $0$.
\item[(d)] Let $N_i=a_i\varphi(T_i)a_i$ in this sense. Then $J(\Phi_i)\,\sigma\,J(\Psi_i)=WN_iW^*$ for a mean $\sigma$ and $\varphi=g_\sigma$, so that $\widehat F(\Phi_i,\Psi_i)=\lambda_{\min}(\Delta(N_i))$ for $\varphi=g_\#$ and $\widehat Q_\alpha(\Phi_i\|\Psi_i)=\lambda_{\min}(\Delta(N_i))$ for $\varphi=h_\alpha$. For every operator convex $f$ with $h_f(T_i)$ finite, $\widehat D_f(\Phi_i\|\Psi_i)=\lambda_{\max}(\Delta(N_i))$ for $\varphi=h_f$.
\end{enumerate}
\end{lemma}

\begin{proof}
(a) Since $0\leq A_i\leq C_i=a_i^2$, we have $\range A_i\subseteq\range C_i=\range P_i$, so $A_i=P_iA_iP_i=a_ia_i^+A_ia_i^+a_i=a_iT_ia_i$. Moreover $T_i\leq a_i^+C_ia_i^+=P_i$, and $B_i=C_i-A_i=a_i(P_i-T_i)a_i$.

(b) $J(\Omega_i)=WC_iW^*=(Wa_iU_i)(Wa_iU_i)^*$, and $Wa_iU_i$ is injective because $W$ is injective and $a_i$ is injective on $\range P_i$. By (a), $J(\Phi_i)=WA_iW^*=(Wa_iU_i)\,U_i^*T_iU_i\,(Wa_iU_i)^*$, and Theorem~\ref{thm:arveson} gives the claim.

(c) The operators $A_i$, $B_i$ and $a_i$ vanish on $\ker P_i$ and leave $\range P_i$ invariant. A connection respects such direct sums, by its defining formula for invertible arguments and continuity from above, and $0\,\sigma'\,0=0$. On $\range P_i$ the operator $a_i$ is invertible, so the transformer equality and Remark~\ref{rem:connections}(2), applied in $\mathrm{L}(\range P_i)$, give $A_i\,\sigma'\,B_i=a_i\big(T_i\,\sigma'\,(P_i-T_i)\big)a_i=a_i\varphi_{\sigma'}(T_i)a_i$.

(d) The first statements follow from (c) and \eqref{eq:W transfer}, exactly as in the proof of Theorem~\ref{thm:four maps}, and from Proposition~\ref{lambda min formula} and Lemma~\ref{renyi closed form}. For the last, apply Theorem~\ref{thm:all f}(b) to $\Omega_i$ with the Kraus decomposition of (b).
\end{proof}

We now define the quantities that replace $\tau$ and $\gamma$. Put
\[
\tau_{\rm ac}=\|Q_R(T_1-T_2)Q_R\|_\infty,\qquad\varkappa_i=\|Q_iT_iQ_R\|_\infty\ (i=1,2),\qquad\gamma_{\rm ac}=\|Q_R(C_1-\tfrac12\mathbf{1}_{\mathbb{C}^r})Q_R\|_\infty .
\]
These are the in-pair distance on the common part of the supports, the coupling inside each pair between the singular and the common directions, and the half-distance of $C_1$ from $\frac12$ on the common part. Since $Q_R$ is the spectral projection for $(0,1)$ and $r<\infty$, we have $\gamma_{\rm ac}<\frac12$. The proof of Lemma~\ref{lem:thompson}, applied to $\Omega_1^{\rm ac}$ and $\Omega_2^{\rm ac}$, whose derivatives with respect to their sum $\Psi_{Q_R}$ are the restrictions of $C_1$ and $C_2$ to $\range Q_R$, gives
\[
\tanh\big(\tfrac12d_{\rm ac}\big)=2\gamma_{\rm ac},\qquad d_{\rm ac}=d_T\big(\Omega_1^{\rm ac},\Omega_2^{\rm ac}\big)<\infty .
\]
Let $\Omega_1\#\Omega_2$ be the completely positive map with Choi matrix $J(\Omega_1)\#J(\Omega_2)$. By \eqref{eq:W transfer} this Choi matrix is $W(C_1\#C_2)W^*=Wa_1a_2W^*$, since $C_1$ and $C_2$ commute, and $a_1a_2=Q_Ra_1a_2Q_R$. Hence $\Omega_1^{\rm ac}+\Omega_2^{\rm ac}-2\,\Omega_1\#\Omega_2=\Psi_{Q_R(a_1-a_2)^2}$ is completely positive, and we put
\[
h_{\rm ac}=\big\|\big(\Omega_1^{\rm ac}+\Omega_2^{\rm ac}-2\,\Omega_1\#\Omega_2\big)^*(\mathbf{1}_{\mathcal{Y}})\big\|_\infty^{1/2},\qquad m_{\rm s}=\big\|\big((\Omega_1)_{\rm s}+(\Omega_2)_{\rm s}\big)^*(\mathbf{1}_{\mathcal{Y}})\big\|_\infty .
\]
Thus $h_{\rm ac}$ is a Hellinger-type distance between the absolutely continuous parts of the sums, and $m_{\rm s}$ is the size of their singular parts.

\begin{lemma}\label{lem:support quantities}
\begin{enumerate}
\item[(a)] If $\range J(\Omega_1)=\range J(\Omega_2)$, then $m_{\rm s}=\varkappa_1=\varkappa_2=0$, $\tau_{\rm ac}=\tau$, $\gamma_{\rm ac}=\gamma$ and $d_{\rm ac}=d_T$.
\item[(b)] $\varkappa_1,\varkappa_2\leq\frac12$.
\item[(c)] \[h_{\rm ac}^2\leq\min\big\{\|\Delta(Q_R|C_1-C_2|)\|_\infty,\ 4\gamma_{\rm ac}^2\|\Delta(Q_R)\|_\infty\big\},\] and $4\gamma_{\rm ac}^2\|\Delta(Q_R)\|_\infty\leq8\gamma_{\rm ac}^2\|\overline{\Omega}^*(\mathbf{1}_{\mathcal{Y}})\|_\infty$. Here $\Delta(Q_R|C_1-C_2|)$ is the transpose of $\Psi_{Q_R|C_1-C_2|}^*(\mathbf{1}_{\mathcal{Y}})$, the total variation of the absolutely continuous parts measured in the frame of $\Omega'$.
\end{enumerate}
\end{lemma}

\begin{proof}
(a) is immediate from the definitions. (b) Let $u\in\range Q_1$ and $v\in\range Q_R$ be unit vectors with $|\langle u,T_1v\rangle|=\varkappa_1$. The compression of $T_1$ to $\mathrm{span}\{u,v\}$ is a $2\times2$ matrix $\big[\begin{smallmatrix}s&x\\\bar x&t\end{smallmatrix}\big]$ with eigenvalues in $[0,1]$ and $|x|=\varkappa_1$. Then $|x|^2\leq st$ and $|x|^2\leq(1-s)(1-t)$, so $|x|^4\leq s(1-s)t(1-t)\leq\frac1{16}$. The same argument applies to $\varkappa_2$. (c) For commuting $a_1,a_2\geq0$, $|a_1-a_2|\leq a_1+a_2$, so $(a_1-a_2)^2\leq|a_1-a_2|(a_1+a_2)=|a_1^2-a_2^2|=|C_1-C_2|$. Also $|\sqrt\lambda-\sqrt{1-\lambda}|\leq|2\lambda-1|\leq2\gamma_{\rm ac}$ for $\lambda$ in the spectrum of $C_1$ on $\range Q_R$, so $Q_R(a_1-a_2)^2\leq4\gamma_{\rm ac}^2Q_R$. Apply $\Delta$, and use $\|\Delta(Q_R)\|_\infty\leq\|\Delta(\mathbf{1}_{\mathbb{C}^r})\|_\infty=2\|\overline{\Omega}^*(\mathbf{1}_{\mathcal{Y}})\|_\infty$.
\end{proof}

\subsection{The estimate}

\begin{theorem}\label{thm:different supports}
Let $\Phi_1,\Psi_1,\Phi_2,\Psi_2\in\mathrm{CP}(\mathcal{X},\mathcal{Y})$ be arbitrary, with $\Omega'\neq0$, and let $T_1,T_2$ be as in Lemma~\ref{lem:inpair general}. Let $I\subseteq[0,1]$ be a compact interval containing $\mathrm{spec}(T_1)\cup\mathrm{spec}(T_2)$, let $\varphi:I\to\mathbb{R}$ be continuous, $\|\varphi\|_I=\max_I|\varphi|$, and $N_i=a_i\varphi(T_i)a_i$. Then $\|\Delta(N_1)-\Delta(N_2)\|_\infty$ and $|||\Psi_{N_1}-\Psi_{N_2}|||_1$ are both at most
\begin{align*}
R_\varphi={}&\|\overline{\Omega}^*(\mathbf{1}_{\mathcal{Y}})\|_\infty\,\omega^I_\varphi(\tau_{\rm ac})+\|\varphi\|_I\Big(m_{\rm s}+\min\Big\{2\sqrt2\,\|\overline{\Omega}^*(\mathbf{1}_{\mathcal{Y}})\|_\infty\tanh\tfrac{d_{\rm ac}}{2},\ 2\|\overline{\Omega}^*(\mathbf{1}_{\mathcal{Y}})\|_\infty^{1/2}h_{\rm ac}\Big\}\Big)\\
&+\|\Omega_1^*(\mathbf{1}_{\mathcal{Y}})\|_\infty\,\omega^I_\varphi(\varkappa_1)+\|\Omega_2^*(\mathbf{1}_{\mathcal{Y}})\|_\infty\,\omega^I_\varphi(\varkappa_2).
\end{align*}
Consequently:
\begin{enumerate}
\item[(a)] $|||\Phi_1\,\sigma\,\Psi_1-\Phi_2\,\sigma\,\Psi_2|||_1\leq R_{g_\sigma}$ for every Kubo--Ando mean $\sigma$, with $I=[0,1]$;
\item[(b)] $|\widehat F(\Phi_1,\Psi_1)-\widehat F(\Phi_2,\Psi_2)|\leq R_{g_\#}$, and with $I=[0,1]$, $\omega_{g_\#}(t)\leq t^{1/2}$ and $\|g_\#\|_\infty=\frac12$,
\[
\begin{aligned}
R_{g_\#}\leq{}&\|\overline{\Omega}^*(\mathbf{1}_{\mathcal{Y}})\|_\infty\tau_{\rm ac}^{1/2}+\tfrac12m_{\rm s}+\sum_{i=1,2}\|\Omega_i^*(\mathbf{1}_{\mathcal{Y}})\|_\infty\varkappa_i^{1/2}\\
&+\min\Big\{\sqrt2\,\|\overline{\Omega}^*(\mathbf{1}_{\mathcal{Y}})\|_\infty\tanh\tfrac{d_{\rm ac}}{2},\ \|\overline{\Omega}^*(\mathbf{1}_{\mathcal{Y}})\|_\infty^{1/2}h_{\rm ac}\Big\};
\end{aligned}
\]
\item[(c)] $|\widehat Q_\alpha(\Phi_1\|\Psi_1)-\widehat Q_\alpha(\Phi_2\|\Psi_2)|\leq R_{h_\alpha}$ for $\alpha\in(0,1)$, with $I=[0,1]$;
\item[(d)] $|\widehat D_f(\Phi_1\|\Psi_1)-\widehat D_f(\Phi_2\|\Psi_2)|\leq R_{h_f}$ for every operator convex $f$ such that $h_f$ is finite on $I$.
\end{enumerate}
If $\range J(\Omega_1)=\range J(\Omega_2)$, then, taking the first entry of the minimum, $R_\varphi$ equals the bound of Theorem~\ref{thm:four maps} for the corresponding quantity.
\end{theorem}

\begin{proof}
Let $T_1^{\rm d}=Q_1T_1Q_1+Q_RT_1Q_R$ and $T_2^{\rm d}=Q_RT_2Q_R+Q_2T_2Q_2$ be the block-diagonal parts of $T_i$ with respect to the spectral decomposition of $C_1$. Compressions of an operator with spectrum in the interval $I$ have spectra in $I$, so $\varphi(T_i^{\rm d})$ is defined, and $\|T_i-T_i^{\rm d}\|_\infty=\varkappa_i$, since $T_1-T_1^{\rm d}=Q_1T_1Q_R+Q_RT_1Q_1$ and similarly for $T_2$. Put $N_i^{\rm d}=a_i\varphi(T_i^{\rm d})a_i$. By \eqref{eq:cp domination}, all three steps below give the same bound for $\|\Delta(\cdot)\|_\infty$ and for $|||\Psi_{(\cdot)}|||_1$.

 Since $\|\varphi(T_1)-\varphi(T_1^{\rm d})\|_\infty\leq\omega^I_\varphi(\varkappa_1)$, we have $\pm(N_1-N_1^{\rm d})\leq\omega^I_\varphi(\varkappa_1)\,a_1^2=\omega^I_\varphi(\varkappa_1)\,C_1$, and $\|\Delta(C_1)\|_\infty=\|\Psi_{C_1}^*(\mathbf{1}_{\mathcal{Y}})\|_\infty=\|\Omega_1^*(\mathbf{1}_{\mathcal{Y}})\|_\infty$. The same holds for the second pair.

 On $\range Q_1$ we have $a_1=\mathbf{1}_{\range Q_1}$ and $a_2=0$, and on $\range Q_2$ we have $a_1=0$ and $a_2=\mathbf{1}_{\range Q_2}$. Since $T_i^{\rm d}$ is block diagonal,
\[
\begin{aligned}N_1^{\rm d}-N_2^{\rm d}&=Q_1\varphi(Q_1T_1Q_1)Q_1-Q_2\varphi(Q_2T_2Q_2)Q_2+Z_R,\\ Z_R&=a_1\varphi(Q_RT_1Q_R)a_1-a_2\varphi(Q_RT_2Q_R)a_2,\end{aligned}
\]
where in $Z_R$ the functional calculus is taken in $\mathrm{L}(\range Q_R)$. The first two terms lie between $\mp\|\varphi\|_I(Q_1+Q_2)$, and $\|\Delta(Q_1+Q_2)\|_\infty=m_{\rm s}$.

 On $\range Q_R$ put $H_-=\varphi(Q_RT_1Q_R)-\varphi(Q_RT_2Q_R)$ and $H_+=\varphi(Q_RT_1Q_R)+\varphi(Q_RT_2Q_R)$, so that $\|H_-\|_\infty\leq\omega^I_\varphi(\tau_{\rm ac})$ and $\|H_+\|_\infty\leq2\|\varphi\|_I$. As in the proof of Lemma~\ref{lem:sandwich},
\[
Z_R=\tfrac12(a_1H_-a_1+a_2H_-a_2)+Y,\qquad Y=\tfrac14\big((a_1-a_2)H_+(a_1+a_2)+(a_1+a_2)H_+(a_1-a_2)\big),
\]
with all operators restricted to $\range Q_R$. Since $a_1^2+a_2^2=Q_R$ there, $\pm\frac12(a_1H_-a_1+a_2H_-a_2)\leq\frac12\|H_-\|_\infty Q_R$, and $\|\Delta(Q_R)\|_\infty\leq2\|\overline{\Omega}^*(\mathbf{1}_{\mathcal{Y}})\|_\infty$. For $Y$ we give two bounds. First, as in Lemma~\ref{lem:sandwich}, on $\range Q_R$ we have
\[
\|Y\|_\infty\leq\tfrac12\|a_1-a_2\|_\infty\|H_+\|_\infty\|a_1+a_2\|_\infty\leq\tfrac12\cdot2\gamma_{\rm ac}\cdot2\|\varphi\|_I\cdot\sqrt2,
\]
so $\pm Y\leq2\sqrt2\,\gamma_{\rm ac}\|\varphi\|_IQ_R$. Second, let $X,X',H_+$ be Hermitian and $\eta>0$. Expanding $(\eta^{1/2}X\mp\eta^{-1/2}X'H_+)(\eta^{1/2}X\mp\eta^{-1/2}H_+X')\geq0$ gives
\[
\pm(XH_+X'+X'H_+X)\leq\eta X^2+\eta^{-1}X'H_+^2X'\leq\eta X^2+\eta^{-1}\|H_+\|_\infty^2X'^2 .
\]
With $X=a_1-a_2$ and $X'=a_1+a_2$, and since $(a_1+a_2)^2=Q_R+2a_1a_2\leq2Q_R$ on $\range Q_R$, this gives $\pm4Y\leq\eta\,Q_R(a_1-a_2)^2+2\eta^{-1}\|H_+\|_\infty^2Q_R$. By definition $\|\Delta(Q_R(a_1-a_2)^2)\|_\infty=h_{\rm ac}^2$, so optimising over $\eta$, the bound is $\frac14\cdot2\|H_+\|_\infty h_{\rm ac}(2\|\Delta(Q_R)\|_\infty)^{1/2}\leq2\|\varphi\|_I\|\overline{\Omega}^*(\mathbf{1}_{\mathcal{Y}})\|_\infty^{1/2}h_{\rm ac}$. Since $2\gamma_{\rm ac}=\tanh\frac{d_{\rm ac}}2$, the first bound contributes at most $2\sqrt2\,\|\varphi\|_I\|\overline{\Omega}^*(\mathbf{1}_{\mathcal{Y}})\|_\infty\tanh\frac{d_{\rm ac}}2$.

Adding the three steps gives $R_\varphi$. Parts (a)--(d) follow from Lemma~\ref{lem:inpair general}(d) and $|\lambda_{\min}(M)-\lambda_{\min}(M')|,|\lambda_{\max}(M)-\lambda_{\max}(M')|\leq\|M-M'\|_\infty$; for (a), $\Phi_i\,\sigma\,\Psi_i=\Psi_{N_i}$ with respect to $V_1,\dots,V_r$. If the supports agree, Lemma~\ref{lem:support quantities}(a) shows that $R_\varphi$ with the first entry of the minimum is $\|\overline{\Omega}^*(\mathbf{1}_{\mathcal{Y}})\|_\infty\big(\omega_\varphi(\tau)+2\sqrt2\|\varphi\|_I\tanh\frac{d_T}{2}\big)$, which is the bound of Theorem~\ref{thm:four maps}.
\end{proof}

Theorem~\ref{thm:different supports} separates three effects. The singular parts of the sums enter linearly, through their size $m_{\rm s}$. The change of the common part enters through the in-pair distance $\tau_{\rm ac}$ and through either the Thompson or the Hellinger distance of the absolutely continuous parts, whichever is smaller. The Hellinger form is better when some common directions carry very different relative weights in the two sums but little absolute weight. The couplings $\varkappa_i$ measure how strongly each pair mixes the directions it does not share with the other sum into the common ones. The next corollary is the cleanest case, in which only the singular part survives.

\begin{corollary}[Adding an independent branch]\label{cor:independent branch}
Let $\Phi_1,\Psi_1,\Theta_\Phi,\Theta_\Psi\in\mathrm{CP}(\mathcal{X},\mathcal{Y})$, put $\Omega_1=\Phi_1+\Psi_1$ and $\Theta=\Theta_\Phi+\Theta_\Psi$, and assume $\range J(\Omega_1)\cap\range J(\Theta)=\{0\}$. Let $\Phi_2=\Phi_1+\Theta_\Phi$ and $\Psi_2=\Psi_1+\Theta_\Psi$. Then for every operator connection $\sigma'$,
\[
J(\Phi_2)\,\sigma'\,J(\Psi_2)=J(\Phi_1)\,\sigma'\,J(\Psi_1)+J(\Theta_\Phi)\,\sigma'\,J(\Theta_\Psi),
\]
and consequently
\[
\widehat F(\Phi_1,\Psi_1)+\widehat F(\Theta_\Phi,\Theta_\Psi)\leq\widehat F(\Phi_2,\Psi_2)\leq\widehat F(\Phi_1,\Psi_1)+\tfrac12\|\Theta^*(\mathbf{1}_{\mathcal{Y}})\|_\infty .
\]
For every mean $\sigma$, $|||\Phi_2\,\sigma\,\Psi_2-\Phi_1\,\sigma\,\Psi_1|||_1=|||\Theta_\Phi\,\sigma\,\Theta_\Psi|||_1\leq\|\Theta^*(\mathbf{1}_{\mathcal{Y}})\|_\infty$. In the notation of Theorem~\ref{thm:different supports}, $(\Omega_1)_{\rm s}=0$, $(\Omega_2)_{\rm s}=\Theta$, and $\tau_{\rm ac}=\varkappa_1=\varkappa_2=d_{\rm ac}=0$.
\end{corollary}

\begin{proof}
Let $\mathcal{V}_1=\range J(\Omega_1)$ and $\mathcal{V}_2=\range J(\Theta)$. Since $\mathcal{V}_1\cap\mathcal{V}_2=\{0\}$, there is an invertible $Z$ on $\mathcal{Y}\otimes\mathcal{X}$ with $Z\mathcal{V}_1\perp Z\mathcal{V}_2$. The Choi matrices of $\Phi_1,\Psi_1$ have ranges in $\mathcal{V}_1$, and those of $\Theta_\Phi,\Theta_\Psi$ in $\mathcal{V}_2$, because $J(\Phi_1),J(\Psi_1)\leq J(\Omega_1)$ and $J(\Theta_\Phi),J(\Theta_\Psi)\leq J(\Theta)$. After conjugation by $Z$ the two pairs live on orthogonal subspaces, where the connection acts blockwise. The transformer equality for the invertible $Z$ gives the identity. Taking partial traces, $\Tr_{\mathcal{Y}}\big(J(\Phi_2)\#J(\Psi_2)\big)=M_1+M_\Theta$ with $M_1=\Tr_{\mathcal{Y}}(J(\Phi_1)\#J(\Psi_1))$ and $M_\Theta=\Tr_{\mathcal{Y}}(J(\Theta_\Phi)\#J(\Theta_\Psi))$. Then $\lambda_{\min}(M_1)+\lambda_{\min}(M_\Theta)\leq\lambda_{\min}(M_1+M_\Theta)\leq\lambda_{\min}(M_1)+\|M_\Theta\|_\infty$, and $J(\Theta_\Phi)\#J(\Theta_\Psi)\leq\frac12J(\Theta)$ by the arithmetic--geometric mean inequality, so $\|M_\Theta\|_\infty\leq\frac12\|\Theta^*(\mathbf{1}_{\mathcal{Y}})\|_\infty$; Proposition~\ref{lambda min formula} gives the inequalities. For a mean, $0\leq J(\Theta_\Phi)\,\sigma\,J(\Theta_\Psi)\leq J(\Theta)$ in the completely positive sense, which gives the trace-norm bound. For the last statement, choose the column matrix $W=[W_1\ W_\Theta]$ of $\Omega'=2\Omega_1+\Theta$ with $W_1W_1^*=2J(\Omega_1)$ and $W_\Theta W_\Theta^*=J(\Theta)$; its columns are linearly independent because $\mathcal{V}_1\cap\mathcal{V}_2=\{0\}$. Then, splitting $\mathbb{C}^r=\mathbb{C}^{r_1}\oplus\mathbb{C}^{r_\Theta}$ according to the columns of $W_1$ and $W_\Theta$, $C_1=\frac12\mathbf{1}_{\mathbb{C}^{r_1}}\oplus0$, $C_2=\frac12\mathbf{1}_{\mathbb{C}^{r_1}}\oplus\mathbf{1}_{\mathbb{C}^{r_\Theta}}$, $A_2=A_1\oplus T_\Theta$ for the derivative $T_\Theta$ of $\Theta_\Phi$, and $T_1=2A_1$, $T_2=2A_1\oplus T_\Theta$.
\end{proof}

\begin{remark}[Physical interpretation]\label{rem:phys-erasure}
Corollary~\ref{cor:independent branch} covers heralded erasure. Regard channels $\Phi_1,\Psi_1$ as maps into $\mathrm{L}(\mathcal{Y}\oplus\mathbb{C}v_\star)$, and let them output a flag $v_\star v_\star^*$ with probability $p$: $\Phi_2=(1-p)\Phi_1+p\tr(\cdot)\,v_\star v_\star^*$ and likewise $\Psi_2$. The corollary applies to $\big((1-p)\Phi_1,(1-p)\Psi_1\big)$ and $\Theta_\Phi=\Theta_\Psi=p\tr(\cdot)\,v_\star v_\star^*$, since $\range J(\Theta)=\mathbb{C}v_\star\otimes\mathcal{X}$ is orthogonal to $\mathcal{Y}\otimes\mathcal{X}\supseteq\range J(\Phi_1+\Psi_1)$. As $\widehat{F}(\Theta_\Phi,\Theta_\Psi)=p=\frac12\|\Theta^*(\mathbf{1}_{\mathcal{Y}\oplus\mathbb{C}v_\star})\|_\infty$ and $\widehat{F}$ is positively homogeneous, its two bounds coincide and give $\widehat{F}(\Phi_2,\Psi_2)=(1-p)\widehat{F}(\Phi_1,\Psi_1)+p$: an incoherent new branch enters exactly with its weight. Leakage coherently coupled to existing directions \cite{WoodGambetta2018} enters instead through the couplings $\varkappa_i$, and can cost a constant at vanishing weight (Example~\ref{ex:coupling}).
\end{remark}

Without the independence assumption, the quantities of Theorem~\ref{thm:different supports} cannot be bounded by the completely bounded distance between the pairs alone. Take $\dim\mathcal{X}=1$, so that completely positive maps are positive semidefinite matrices, and unit vectors $y_1\neq y_2$ with $y_1\approx y_2$. Let $\Phi_1=\Psi_1=y_1y_1^*$ and $\Phi_2=y_1y_1^*$, $\Psi_2=y_2y_2^*$. Then $\widehat F(\Phi_1,\Psi_1)=1$ and $\widehat F(\Phi_2,\Psi_2)=0$, because the geometric mean of two distinct rank-one projections is $0$, although $\|\Psi_1-\Psi_2\|_1$ is small. In the frame $W=[\sqrt3\,y_1\ \ y_2]$ of $\Omega'=3y_1y_1^*+y_2y_2^*$ one finds $C_1=\mathrm{diag}(\frac23,0)$, $T_1=\frac12$ and $T_2=\mathrm{diag}(1,0)$. Hence $(\Omega_2)_{\rm s}=y_2y_2^*$ has size $m_{\rm s}=1$, however close $y_2$ is to $y_1$, and $\tau_{\rm ac}=\frac12$. The jump is therefore recorded by the theorem, through $\tau_{\rm ac}$ and $m_{\rm s}$, and not hidden. It is the support-changing counterpart of the impossibility example of Corollary~\ref{impossibility corollary}.

\subsection{Continuity in the derivative across support changes}

Theorem~\ref{thm:different supports} is a continuity statement in which the relative distance $\delta$ is replaced by the quantities of the Lebesgue decomposition. When the supports differ, $\delta\geq\frac12$ by Lemma~\ref{lem:one parameter}, so no statement in terms of $\delta$ alone can cover this case.

\begin{corollary}[Continuity across support changes]\label{cor:support continuity}
Fix $(\Phi_1,\Psi_1)$, let $(\Phi_2^{(n)},\Psi_2^{(n)})$ be pairs with $\sup_n\|\Omega_2^{(n)*}(\mathbf{1}_{\mathcal{Y}})\|_\infty<\infty$, and let $\tau_{\rm ac}^{(n)},\varkappa_i^{(n)},m_{\rm s}^{(n)},d_{\rm ac}^{(n)},h_{\rm ac}^{(n)}$ be the quantities of Theorem~\ref{thm:different supports} for the pairs $(\Phi_1,\Psi_1)$ and $(\Phi_2^{(n)},\Psi_2^{(n)})$. If
\[
\tau_{\rm ac}^{(n)}\to0,\qquad\varkappa_1^{(n)},\varkappa_2^{(n)}\to0,\qquad m_{\rm s}^{(n)}\to0,\qquad\min\{d_{\rm ac}^{(n)},h_{\rm ac}^{(n)}\}\to0,
\]
then $\widehat F(\Phi_2^{(n)},\Psi_2^{(n)})\to\widehat F(\Phi_1,\Psi_1)$, $\widehat Q_\alpha(\Phi_2^{(n)}\|\Psi_2^{(n)})\to\widehat Q_\alpha(\Phi_1\|\Psi_1)$ for $\alpha\in(0,1)$, and $|||\Phi_2^{(n)}\sigma\Psi_2^{(n)}-\Phi_1\sigma\Psi_1|||_1\to0$ for every mean $\sigma$. The same holds for $\widehat D_f$ for every operator convex $f$ such that $h_f$ is finite on a compact interval containing the spectra of all the in-pair derivatives. If the supports of the sums agree for all $n$, the hypothesis reduces to $\tau^{(n)}\to0$ and $d_T^{(n)}\to0$.
\end{corollary}

\begin{proof}
Every term of $R_\varphi$ is a constant times $\omega^I_\varphi$ at one of $\tau_{\rm ac}^{(n)},\varkappa_1^{(n)},\varkappa_2^{(n)}$, or a constant times $m_{\rm s}^{(n)}$, $\tanh\frac12d_{\rm ac}^{(n)}$ or $h_{\rm ac}^{(n)}$, and the constants are bounded in $n$. Lemma~\ref{modulus to zero} gives $\omega^I_\varphi(\delta)\to0$ as $\delta\to0$. The last statement is Lemma~\ref{lem:support quantities}(a).
\end{proof}

None of the hypotheses can be dropped. In the example following Corollary~\ref{cor:independent branch}, $m_{\rm s}=1$ and $\tau_{\rm ac}=\frac12$ while the pairs are arbitrarily close. The next example shows that the coupling alone can carry a jump, with every other quantity tending to zero.

\begin{example}[Coupling]\label{ex:coupling}
Let $\dim\mathcal{X}=1$ and $\mathcal{Y}=\mathbb{C}^2$, let $P$ be the projection onto $(e_1+e_2)/\sqrt2$, and put $Z_t=\mathrm{diag}(\sqrt t,1)$. Let $J(\Phi_t)=Z_tPZ_t$, $J(\Psi_t)=Z_t(\mathbf{1}_{\mathbb{C}^2}-P)Z_t$, and $J(\Phi_0)=J(\Psi_0)=\frac12e_2e_2^*$. Then $\|J(\Phi_t)-J(\Phi_0)\|_1$ and $\|J(\Psi_t)-J(\Psi_0)\|_1$ are $O(\sqrt t)$. Now $J(\Phi_t)$ and $J(\Psi_t)$ are distinct rank-one operators, so their geometric mean is $0$ and $\widehat F(\Phi_t,\Psi_t)=0$, whereas $\widehat F(\Phi_0,\Psi_0)=\frac12$. Take $(\Phi_t,\Psi_t)$ as the first pair and $(\Phi_0,\Psi_0)$ as the second. Then $J(\Omega')=\mathrm{diag}(t,2)$, $W=\mathrm{diag}(\sqrt t,\sqrt2)$ and $C_1=\mathrm{diag}(1,\frac12)$, so $Q_1=e_1e_1^*$, $Q_R=e_2e_2^*$ and $Q_2=0$. Moreover $T_1=P$, and $T_2=\frac12$ on $\range Q_R$. Hence $\tau_{\rm ac}=0$, $d_{\rm ac}=h_{\rm ac}=0$, $\varkappa_2=0$ and $m_{\rm s}=t\to0$, but $\varkappa_1=|\langle e_1,Pe_2\rangle|=\frac12$ for every $t$. The jump of size $\frac12$ is carried by the coupling, and the term $\|\Omega_1^*(\mathbf{1}_{\mathcal{Y}})\|_\infty\varkappa_1^{1/2}=(1+t)/\sqrt2$ of Theorem~\ref{thm:different supports}(b) accounts for it.
\end{example}

\section{Where continuity fails: a dichotomy}\label{sec:dichotomy}

Example~\ref{ex:coupling} suggests a general mechanism. A pair whose sum has a non-invertible Choi matrix can be approximated by pairs that use the missing Choi directions with vanishing weight but with full coupling, and this changes the derivative in the limit. We show that this is the only obstruction to continuity, and we compute the size of the jumps. The answer is a dichotomy: each quantity is continuous at a pair exactly when the Choi matrix of the sum is invertible or the quantity takes an extreme value there.

Throughout this section, $\Phi,\Psi\in\mathrm{CP}(\mathcal{X},\mathcal{Y})$, $\Omega=\Phi+\Psi\neq0$ has column matrix $W\in\mathrm{L}(\mathbb{C}^r,\mathcal{Y}\otimes\mathcal{X})$ as in Theorem~\ref{fundamental principle}, $T$ is the derivative of $\Phi$ with respect to $\Omega$, and
\[
c_\Omega=\dim(\mathcal{Y}\otimes\mathcal{X})-r=\dim\ker J(\Omega)
\]
is the Choi codimension of $\Omega$. Pairs converge in any norm on the finite-dimensional space of Hermitian-preserving maps, for instance in the completely bounded trace norm. Every operator connection $\sigma$ has a representing function $f_\sigma(x)=a_\sigma+b_\sigma x+\int_{(0,\infty)}\frac{x(1+s)}{x+s}\,dm_\sigma(s)$ with $a_\sigma,b_\sigma\geq0$ and a finite positive measure $m_\sigma$ \cite{KuboAndo1980}. By \eqref{eq:mean kernels},
\[
\varphi_\sigma=\ell_\sigma+\eta_\sigma,\qquad\ell_\sigma(t)=a_\sigma\,t+b_\sigma\,(1-t),\qquad\eta_\sigma=\int_{(0,\infty)}\eta_s\,dm_\sigma(s)\geq0,
\]
where $\eta_\sigma$ vanishes at $0$ and $1$, and is strictly positive on $(0,1)$ unless $m_\sigma=0$. By Remark~\ref{rem:connections}, $J(\Phi)\,\sigma\,J(\Psi)=W\varphi_\sigma(T)W^*$ and $a_\sigma J(\Phi)+b_\sigma J(\Psi)=W\ell_\sigma(T)W^*$. For an operator convex $f$ with $f(0^+)$ and $f'(\infty)$ finite we write $\ell_f(t)=f(0^+)(1-t)+f'(\infty)t$ for the chord of $h_f$, so that $h_f\leq\ell_f$ on $[0,1]$.

We call $\Phi$ and $\Psi$ \emph{mutually singular} if $\range J(\Phi)\cap\range J(\Psi)=\{0\}$. This holds if and only if $T$ is a projection: an eigenvector $u$ of $T$ with eigenvalue in $(0,1)$ gives $Wu\in W\range T\cap W\range(\mathbf{1}_{\mathbb{C}^r}-T)$, and conversely $W$ is injective. By Proposition~\ref{prop:lebesgue cp}, it means that $\Phi=\Phi_{\rm s}$ and $\Psi=\Psi_{\rm s}$. In that case $J(\Phi)\,\sigma\,J(\Psi)=a_\sigma J(\Phi)+b_\sigma J(\Psi)$, and in particular $J(\Phi)\#J(\Psi)=0$. The equivalence of mutual singularity with $J(\Phi)\#J(\Psi)=0$ is the criterion following Theorem~\ref{range theorem}, which for the geometric mean is due to Kian \cite[Theorem~4.3]{kian2026relative}; for states it is the distinct image property of \cite{cree2020fidelity}.

\begin{lemma}[Envelope]\label{lem:envelope}
\begin{enumerate}
\item[(a)] For positive semidefinite $A,B\in\mathrm{L}(\mathcal{V})$ and every connection $\sigma$, $a_\sigma A+b_\sigma B\leq A\,\sigma\,B$.
\item[(b)] If $A_n\to A$ and $B_n\to B$, then every limit point $L$ of $(A_n\,\sigma\,B_n)$ satisfies $a_\sigma A+b_\sigma B\leq L\leq A\,\sigma\,B$.
\item[(c)] $\widehat F$ and $\widehat Q_\alpha$, $\alpha\in(0,1)$, are upper semicontinuous on $\mathrm{CP}(\mathcal{X},\mathcal{Y})\times\mathrm{CP}(\mathcal{X},\mathcal{Y})$. For every operator convex $f$, $\widehat D_f$ is lower semicontinuous. If $f(0^+)$ and $f'(\infty)$ are finite, then $\widehat D_f(\Phi\|\Psi)\leq\lambda_{\max}\big(\Tr_{\mathcal{Y}}(f(0^+)J(\Psi)+f'(\infty)J(\Phi))\big)$.
\end{enumerate}
\end{lemma}

\begin{proof}
(a) For $A>0$, the inequality $f_\sigma(x)\geq a_\sigma+b_\sigma x$ gives $A\,\sigma\,B=A^{1/2}f_\sigma(A^{-1/2}BA^{-1/2})A^{1/2}\geq a_\sigma A+b_\sigma B$. The general case follows by continuity from above.

(b) Since $x(1+s)/(x+s)\leq1+x$, we have $f_\sigma(x)\leq(a_\sigma+m_0)+(b_\sigma+m_0)x$ with $m_0=m_\sigma((0,\infty))$, so $A_n\,\sigma\,B_n\leq(a_\sigma+m_0)A_n+(b_\sigma+m_0)B_n$ is bounded. The lower bound follows from (a), since the positive cone is closed. For $\lambda>0$, monotonicity gives $A_n\,\sigma\,B_n\leq(A_n+\lambda\mathbf{1}_{\mathcal{V}})\,\sigma\,(B_n+\lambda\mathbf{1}_{\mathcal{V}})$. The right-hand side converges to $(A+\lambda\mathbf{1}_{\mathcal{V}})\,\sigma\,(B+\lambda\mathbf{1}_{\mathcal{V}})$, because the defining formula is continuous on positive definite pairs. Hence $L\leq(A+\lambda\mathbf{1}_{\mathcal{V}})\,\sigma\,(B+\lambda\mathbf{1}_{\mathcal{V}})$ for every $\lambda>0$, and continuity from above gives $L\leq A\,\sigma\,B$.

(c) By Proposition~\ref{lambda min formula}, $\widehat F(\Phi,\Psi)=\lambda_{\min}\big(\Tr_{\mathcal{Y}}(J(\Phi)\#J(\Psi))\big)$, and $\lambda_{\min}\circ\Tr_{\mathcal{Y}}$ is continuous and monotone. Applying (b) along a subsequence that realises the $\limsup$ gives upper semicontinuity, and Lemma~\ref{renyi closed form} gives the same for $\widehat Q_\alpha$. $\widehat D_f$ is a supremum over $\rho$ of the functions $(\Phi,\Psi)\mapsto\widehat D_f(\Phi_\rho\|\Psi_\rho)$, which are lower semicontinuous by Lemmas~\ref{lem:unb basic}(2) and~\ref{lem:equal range}; see also \cite[Section~5]{Hiai2019maximal}, \cite[Theorem~A.4]{MosonyiHiai2024} and \cite{Matsumoto2018newf}. For the last inequality, $h_f\leq\ell_f$ gives $Wh_f(T)W^*\leq f(0^+)J(\Psi)+f'(\infty)J(\Phi)$; apply Theorem~\ref{thm:all f}(b).
\end{proof}

\begin{remark}\label{rem:one sided}
Parts (a) and (b) of Lemma~\ref{lem:envelope} are elementary consequences of the Kubo--Ando representation and of continuity from above, and the upper semicontinuity of the Matsumoto fidelity was already used in the proof of Proposition~\ref{lambda min formula}. The lower semicontinuity of maximal $f$-divergences is noted by Matsumoto \cite{Matsumoto2018newf} as a consequence of his dual formula, and is proved in \cite[Section~5]{Hiai2019maximal} and \cite[Theorem~A.4]{MosonyiHiai2024}. Part (b) also shows that one-sided perturbations cannot produce jumps: if $A_n\geq A$ and $B_n\geq B$ for all $n$, then monotonicity gives $A_n\,\sigma\,B_n\geq A\,\sigma\,B$, so every limit point equals $A\,\sigma\,B$, and $A_n\,\sigma\,B_n\to A\,\sigma\,B$.
\end{remark}

\begin{lemma}[Dilation into the Choi kernel]\label{lem:dilation}
Suppose $c_\Omega\geq1$, let $\Upsilon$ be a matrix whose $c_\Omega$ columns form a basis of $\ker J(\Omega)$, and let
\[
\widehat T=\begin{bmatrix}T&X\\X^*&Z\end{bmatrix}\quad\text{on }\mathbb{C}^r\oplus\mathbb{C}^{c_\Omega},\qquad0\leq\widehat T\leq\mathbf{1}_{\mathbb{C}^r\oplus\mathbb{C}^{c_\Omega}}.
\]
For $t>0$ put $W_t=[\,W\ \ \sqrt t\,\Upsilon\,]$, and let $\Phi_t=\Psi_{\widehat T}$ and $\Psi_t=\Psi_{\mathbf{1}_{\mathbb{C}^r\oplus\mathbb{C}^{c_\Omega}}-\widehat T}$ with respect to the Kraus operators given by the columns of $W_t$. Then $\|J(\Phi_t)-J(\Phi)\|_\infty$ and $\|J(\Psi_t)-J(\Psi)\|_\infty$ are $O(\sqrt t)$. For every connection $\sigma$, $J(\Phi_t)\,\sigma\,J(\Psi_t)=W_t\varphi_\sigma(\widehat T)W_t^*\to W[\varphi_\sigma(\widehat T)]_{11}W^*$ as $t\downarrow0$, where $[\cdot]_{11}$ is the compression to $\mathbb{C}^r$. For every operator convex $f$ with $h_f(\widehat T)$ finite, $\widehat D_f(\Phi_t\|\Psi_t)\to\lambda_{\max}\big(\Tr_{\mathcal{Y}}(W[h_f(\widehat T)]_{11}W^*)\big)$.
\end{lemma}

\begin{proof}
The columns of $W_t$ are linearly independent, because $\range\Upsilon\cap\range W=\ker J(\Omega)\cap\range J(\Omega)=\{0\}$. Since $W_tW_t^*=J(\Omega)+t\Upsilon\Upsilon^*$, $W_t$ is the column matrix of $\Phi_t+\Psi_t$, and $\widehat T$ is the derivative of $\Phi_t$. For any block operator $M=[M_{ij}]$ on $\mathbb{C}^r\oplus\mathbb{C}^{c_\Omega}$,
\[
W_tMW_t^*=WM_{11}W^*+\sqrt t\,(WM_{12}\Upsilon^*+\Upsilon M_{21}W^*)+t\,\Upsilon M_{22}\Upsilon^* .
\]
With $M=\widehat T$ and $M=\mathbf{1}_{\mathbb{C}^r\oplus\mathbb{C}^{c_\Omega}}-\widehat T$ this gives the first statement. With $M=\varphi_\sigma(\widehat T)$, Remark~\ref{rem:connections} gives the second. With $M=h_f(\widehat T)$, Theorem~\ref{thm:all f}(b) gives the third, since $\lambda_{\max}\circ\Tr_{\mathcal{Y}}$ is continuous.
\end{proof}

The dilations we use are projections, after Halmos \cite{Halmos1969}. If $\mathcal{U}\subseteq\mathbb{C}^r$ is invariant under $T$, write $T=T_{\mathcal{U}}\oplus T_{\mathcal{U}^\perp}$, $R=(T_{\mathcal{U}}(\mathbf{1}_{\mathcal{U}}-T_{\mathcal{U}}))^{1/2}$, $k=\rank R$, and let $R_0$ be a partial isometry from $\mathcal{U}$ into $\mathbb{C}^{c_\Omega}$ with $R_0^*R_0$ the range projection of $R$, which requires $k\leq c_\Omega$. Put
\begin{equation}\label{eq:halmos}
\widehat T_{\mathcal{U}}=T_{\mathcal{U}^\perp}\oplus\begin{bmatrix}T_{\mathcal{U}}&RR_0^*\\R_0R&R_0(\mathbf{1}_{\mathcal{U}}-T_{\mathcal{U}})R_0^*\end{bmatrix}\quad\text{on }\mathcal{U}^\perp\oplus(\mathcal{U}\oplus\mathbb{C}^{c_\Omega}).
\end{equation}
Since $T_{\mathcal{U}}$, $R$ and $R_0^*R_0$ commute, a direct computation shows that the second summand $\widehat P_{\mathcal{U}}$ is a projection, and its compression to $\mathcal{U}$ is $T_{\mathcal{U}}$. Hence for every continuous $\varphi$ on $[0,1]$ with chord $\ell_\varphi(t)=\varphi(0)(1-t)+\varphi(1)t$, we have $\varphi(\widehat P_{\mathcal{U}})=\ell_\varphi(\widehat P_{\mathcal{U}})$, and
\begin{equation}\label{eq:halmos compression}
[\varphi(\widehat T_{\mathcal{U}})]_{11}=\varphi(T_{\mathcal{U}^\perp})\oplus\ell_\varphi(T_{\mathcal{U}})=\varphi(T)-P_{\mathcal{U}}(\varphi-\ell_\varphi)(T)P_{\mathcal{U}},
\end{equation}
where $P_{\mathcal{U}}$ denotes the orthogonal projection onto $\mathcal{U}$, which commutes with $T$.

\begin{theorem}[Continuity dichotomy]\label{thm:dichotomy}
Let $\Phi,\Psi,\Omega,W,T$ and $c_\Omega$ be as above.
\begin{enumerate}
\item[(a)] If $J(\Omega)$ is invertible, then for every pair $(\Phi',\Psi')$, with $d_{\rm cb}=\max\{|||\Phi-\Phi'|||_1,|||\Psi-\Psi'|||_1\}$ and $\overline{\Omega}=\frac12(\Omega+\Phi'+\Psi')$,
\[
|\widehat F(\Phi,\Psi)-\widehat F(\Phi',\Psi')|\leq\tfrac{11}{4}\,\|\overline{\Omega}^*(\mathbf{1}_{\mathcal{Y}})\|_\infty\,\nu(\Omega)^{-1}d_{\rm cb}^{1/2}.
\]
Theorem~\ref{thm:four maps} and Corollary~\ref{cor:four maps cb} give local moduli for every mean, for $\widehat Q_\alpha$ and for $\widehat D_f$ with bounded $f$. In particular all these quantities are continuous at $(\Phi,\Psi)$.
\item[(b)] Suppose $c_\Omega\geq1$. Let $\mathcal{U}\subseteq\mathbb{C}^r$ be invariant under $T$, with $\rank T_{\mathcal{U}}(\mathbf{1}_{\mathcal{U}}-T_{\mathcal{U}})\leq c_\Omega$, and let $(\Phi_t,\Psi_t)$ be the pairs of Lemma~\ref{lem:dilation} with $\widehat T=\widehat T_{\mathcal{U}}$. Then $\|J(\Phi_t)-J(\Phi)\|_\infty,\|J(\Psi_t)-J(\Psi)\|_\infty=O(\sqrt t)$, and simultaneously for all connections $\sigma$,
\[
J(\Phi_t)\,\sigma\,J(\Psi_t)\longrightarrow W\big(\varphi_\sigma(T)-P_{\mathcal{U}}\eta_\sigma(T)P_{\mathcal{U}}\big)W^*,
\]
and for every operator convex $f$ with $f(0^+),f'(\infty)<\infty$,
\[
\widehat D_f(\Phi_t\|\Psi_t)\longrightarrow\lambda_{\max}\Big(\Tr_{\mathcal{Y}}\big(W(h_f(T)+P_{\mathcal{U}}(\ell_f-h_f)(T)P_{\mathcal{U}})W^*\big)\Big).
\]
\item[(c)] If $c_\Omega\geq\rank T(\mathbf{1}_{\mathbb{C}^r}-T)$, the choice $\mathcal{U}=\range T(\mathbf{1}_{\mathbb{C}^r}-T)$ in (b) makes $\Phi_t$ and $\Psi_t$ mutually singular. Then $\Phi_t\,\sigma\,\Psi_t=a_\sigma\Phi_t+b_\sigma\Psi_t\to a_\sigma\Phi+b_\sigma\Psi$, which is the smallest limit allowed by Lemma~\ref{lem:envelope}. Moreover $\widehat F(\Phi_t,\Psi_t)=\widehat Q_\alpha(\Phi_t\|\Psi_t)=0$, and for bounded $f$, $\widehat D_f(\Phi_t\|\Psi_t)\to\lambda_{\max}\big(\Tr_{\mathcal{Y}}(f(0^+)J(\Psi)+f'(\infty)J(\Phi))\big)$, which is the largest value allowed by Lemma~\ref{lem:envelope}.
\item[(d)] Consequently:
\begin{enumerate}
\item[(i)] for a connection $\sigma$ with $m_\sigma\neq0$, the map $(\Phi',\Psi')\mapsto\Phi'\,\sigma\,\Psi'$ is continuous at $(\Phi,\Psi)$ if and only if $J(\Omega)$ is invertible or $\Phi$ and $\Psi$ are mutually singular;
\item[(ii)] $\widehat F$ is continuous at $(\Phi,\Psi)$ if and only if $J(\Omega)$ is invertible or $\widehat F(\Phi,\Psi)=0$, and the same holds for $\widehat Q_\alpha$, $\alpha\in(0,1)$;
\item[(iii)] for bounded operator convex $f$, $\widehat D_f$ is continuous at $(\Phi,\Psi)$ if $J(\Omega)$ is invertible. If $1\leq c_\Omega$ and $\rank T(\mathbf{1}_{\mathbb{C}^r}-T)\leq c_\Omega$, it is continuous if and only if
\[
\widehat D_f(\Phi\|\Psi)=\lambda_{\max}\big(\Tr_{\mathcal{Y}}(f(0^+)J(\Psi)+f'(\infty)J(\Phi))\big);
\]
\item[(iv)] if $f'(\infty)=+\infty$, $c_\Omega\geq1$ and $\Phi,\Psi$ are not mutually singular, then every neighbourhood of $(\Phi,\Psi)$ contains pairs with $\widehat D_f=+\infty$.
\end{enumerate}
\end{enumerate}
\end{theorem}

\begin{proof}
(a) Since $J(\Omega)$ is invertible, $\range J(\Phi'+\Psi')\subseteq\range J(\Omega)$, and Corollary~\ref{cor:four maps cb} applies with $(\Phi_1,\Psi_1)=(\Phi,\Psi)$. For the other quantities, the same corollary gives $\delta\leq d_{\rm cb}/\nu(\Omega)^2$. When $\delta<\frac12$, Lemma~\ref{lem:one parameter} gives $\tau\leq2\sqrt2\,\delta/\sqrt{1-4\delta^2}$ and $\tanh\frac{d_T}{2}=2\gamma\leq2\delta$, and Theorem~\ref{thm:four maps} applies.

(b) Since $\rank R=\rank T_{\mathcal{U}}(\mathbf{1}_{\mathcal{U}}-T_{\mathcal{U}})\leq c_\Omega$, the dilation \eqref{eq:halmos} exists and $0\leq\widehat T_{\mathcal{U}}\leq\mathbf{1}_{\mathbb{C}^r\oplus\mathbb{C}^{c_\Omega}}$. By Lemma~\ref{lem:dilation} and \eqref{eq:halmos compression} with $\varphi=\varphi_\sigma$, whose chord is $\ell_\sigma$ because $\varphi_\sigma(1)=a_\sigma$ and $\varphi_\sigma(0)=b_\sigma$, the means converge to $W(\varphi_\sigma(T)-P_{\mathcal{U}}\eta_\sigma(T)P_{\mathcal{U}})W^*$. With $\varphi=h_f$, whose chord is $\ell_f$, the divergences converge as stated.

(c) On $\mathcal{U}^\perp$, $T(\mathbf{1}_{\mathbb{C}^r}-T)=0$, so $T_{\mathcal{U}^\perp}$ is a projection, and so is $\widehat T_{\mathcal{U}}$. Since $W_t$ is injective, $\Phi_t$ and $\Psi_t$ are mutually singular, and $\varphi_\sigma(\widehat T_{\mathcal{U}})=\ell_\sigma(\widehat T_{\mathcal{U}})$ gives $J(\Phi_t)\,\sigma\,J(\Psi_t)=a_\sigma J(\Phi_t)+b_\sigma J(\Psi_t)$. The geometric mean and $\#_{1-\alpha}$ have $a_\sigma=b_\sigma=0$, so $\widehat F(\Phi_t,\Psi_t)=\widehat Q_\alpha(\Phi_t\|\Psi_t)=0$ by Proposition~\ref{lambda min formula} and Lemma~\ref{renyi closed form}. Similarly $h_f(\widehat T_{\mathcal{U}})=\ell_f(\widehat T_{\mathcal{U}})$ gives $\widehat D_f(\Phi_t\|\Psi_t)=\lambda_{\max}\big(\Tr_{\mathcal{Y}}(f(0^+)J(\Psi_t)+f'(\infty)J(\Phi_t))\big)$, which converges as stated.

(d)(i) If $J(\Omega)$ is invertible, use (a) and Theorem~\ref{thm:four maps}(a). If $T$ is a projection, then $\eta_\sigma(T)=0$, so the two bounds in Lemma~\ref{lem:envelope}(b) coincide, and every limit point equals $J(\Phi)\,\sigma\,J(\Psi)$; since the sequence is bounded, it converges. Conversely, let $c_\Omega\geq1$ and let $u$ be a unit eigenvector of $T$ with eigenvalue $q\in(0,1)$. Then $\eta_\sigma(q)>0$ because $m_\sigma\neq0$. With $\mathcal{U}=\mathbb{C}u$, (b) produces the limit $W(\varphi_\sigma(T)-\eta_\sigma(q)uu^*)W^*$, which differs from $W\varphi_\sigma(T)W^*$ because $W$ is injective.

(ii) If $J(\Omega)$ is invertible, use (a). If $\widehat F(\Phi,\Psi)=0$, upper semicontinuity (Lemma~\ref{lem:envelope}(c)) and $\widehat F\geq0$ give continuity. Conversely, let $c_\Omega\geq1$ and $\widehat F(\Phi,\Psi)>0$. Put $M=\Tr_{\mathcal{Y}}(Wg_\#(T)W^*)$ and let $x$ be a unit vector with $\langle x,Mx\rangle=\lambda_{\min}(M)=\widehat F(\Phi,\Psi)$. Write $g_\#(T)=\sum_kg_\#(q_k)u_ku_k^*$, the sum over the eigenvalues $q_k\in(0,1)$ of $T$, since $g_\#(0)=g_\#(1)=0$. Since $\langle x,Mx\rangle>0$, some $k$ has $\beta_k:=g_\#(q_k)\langle x,\Tr_{\mathcal{Y}}(Wu_ku_k^*W^*)x\rangle>0$. With $\mathcal{U}=\mathbb{C}u_k$, (b) gives
\[
\widehat F(\Phi_t,\Psi_t)\to\lambda_{\min}\Big(\Tr_{\mathcal{Y}}\big(W(g_\#(T)-g_\#(q_k)u_ku_k^*)W^*\big)\Big)\leq\langle x,Mx\rangle-\beta_k<\widehat F(\Phi,\Psi).
\]
The same argument applies to $\widehat Q_\alpha$, with $h_\alpha$ in place of $g_\#$.

(iii) The first statement is (a), with Theorem~\ref{thm:four maps}(d). Let $\mathcal{E}_f(\Phi,\Psi)=\lambda_{\max}\big(\Tr_{\mathcal{Y}}(f(0^+)J(\Psi)+f'(\infty)J(\Phi))\big)$, a continuous function of the pair. If $1\leq c_\Omega$ and $\rank T(\mathbf{1}_{\mathbb{C}^r}-T)\leq c_\Omega$, then by (c) there are pairs converging to $(\Phi,\Psi)$ along which $\widehat D_f\to \mathcal{E}_f(\Phi,\Psi)$, so continuity forces $\widehat D_f(\Phi\|\Psi)=\mathcal{E}_f(\Phi,\Psi)$. Conversely, if equality holds, lower semicontinuity gives $\liminf\widehat D_f(\Phi'\|\Psi')\geq\widehat D_f(\Phi\|\Psi)$, and Lemma~\ref{lem:envelope}(c) gives $\limsup\widehat D_f(\Phi'\|\Psi')\leq\lim \mathcal{E}_f(\Phi',\Psi')=\mathcal{E}_f(\Phi,\Psi)=\widehat D_f(\Phi\|\Psi)$.

(iv) Let $u$ be an eigenvector of $T$ with eigenvalue in $(0,1)$ and $\mathcal{U}=\mathbb{C}u$. Then $\widehat T_{\mathcal{U}}$ has the eigenvalue $1$, so by Proposition~\ref{prop:lebesgue cp} $\Phi_t$ has a nonzero singular part relative to $\Psi_t$, and Theorem~\ref{thm:all f}(a) gives $\widehat D_f(\Phi_t\|\Psi_t)=+\infty$.
\end{proof}

The approximants in (c) separate the Choi supports of the two maps at cost $O(\sqrt t)$, and every mean then collapses to its affine part. Theorem~\ref{thm:different supports} records exactly this: for $(\Phi_t,\Psi_t)$ against $(\Phi,\Psi)$ one has $\tau_{\rm ac}=0$, $d_{\rm ac}=0$ and $m_{\rm s}=O(t)$, but the couplings do not tend to zero. The condition $c_\Omega\geq\rank T(\mathbf{1}_{\mathbb{C}^r}-T)$ holds whenever the Kraus rank of $\Omega$ is at most $\frac12\dim\mathcal{X}\dim\mathcal{Y}$. When it fails, (b) still gives a nonzero jump whenever (d) asserts discontinuity. Part (iv) is the channel form of the familiar fact that divergences with $f'(\infty)=\infty$ require a support condition. Parts (i)--(iii) have no classical counterpart: if all Choi matrices involved commute, $\widehat F$ is continuous, since $\sqrt{xy}$ is.

The ingredients of Theorem~\ref{thm:dichotomy} are classical, but we are not aware of an earlier description of the points of continuity. Discontinuities at singular pairs are known: the Matsumoto fidelity of two distinct pure states vanishes \cite{cree2020fidelity}, which is the mechanism of Corollary~\ref{impossibility corollary}, and the Belavkin--Staszewski conditional entropy is discontinuous on singular states, with continuity bounds of order $\sqrt\delta/(\lambda_-(d^{-1}-\lambda_-))$ on states on a $d$-dimensional space whose smallest eigenvalue is at least $\lambda_-$, where $0<\lambda_-<d^{-1}$ \cite{BluhmCapelGondolfPerezHernandez2023}. Theorem~\ref{thm:dichotomy}(a) is an analogue for completely positive maps, with a constant governed by $\nu(\Omega)$, which is comparable to the square root of the smallest eigenvalue of $J(\Omega)$ by Lemma~\ref{lem:nu}(2). The ways of approaching a singular pair that are used in the literature are all compatible with Theorem~\ref{thm:dichotomy}: regularisation from above \cite{KuboAndo1980,AndoLiMathias2004} and positive semidefinite perturbations, as in the perturbation bounds for parallel sums of \cite{LuoSongXu2018}, converge by Remark~\ref{rem:one sided}, and, as observed in the introduction of \cite{Matsumoto2018newf}, maximal $f$-divergences of states are continuous along line segments into the interior of their domain. By Corollary~\ref{cor:four maps cb} and Remark~\ref{rem:one sided}, a jump requires perturbations $(\Phi_n,\Psi_n)$ that enlarge the Choi support of the sum and do not satisfy both $\Phi_n\geq\Phi$ and $\Psi_n\geq\Psi$.

\begin{corollary}[States]\label{cor:states dichotomy}
For positive semidefinite $\rho_1,\rho_2$ on $\mathbb{C}^d$, the Matsumoto fidelity $\tr(\rho_1\#\rho_2)$ is continuous at $(\rho_1,\rho_2)$ if and only if $\rho_1+\rho_2$ is invertible or $\tr(\rho_1\#\rho_2)=0$. If $\rank(\rho_1+\rho_2)\leq d/2$, there are mutually singular pairs $(\rho_{1,t},\rho_{2,t})\to(\rho_1,\rho_2)$, which can be normalised to states when $\rho_1,\rho_2$ are states, and $\tr(\rho_{1,t}\#\rho_{2,t})=0$.
\end{corollary}

\begin{proof}
Take $\dim\mathcal{X}=1$ in Theorem~\ref{thm:dichotomy}. Then $\widehat F(\rho_1,\rho_2)=\tr(\rho_1\#\rho_2)$, and $\rank T(\mathbf{1}_{\mathbb{C}^r}-T)\leq\rank(\rho_1+\rho_2)\leq d-\rank(\rho_1+\rho_2)=c_\Omega$. Normalising each map by its trace preserves mutual singularity and convergence.
\end{proof}

For example, $\rho_1=\rho_2=\frac12(e_1e_1^*+e_2e_2^*)$ on $\mathbb{C}^4$ has $\tr(\rho_1\#\rho_2)=1$. Coupling $e_1$ to $e_3$ as in \eqref{eq:halmos} gives $\rho_{1,t}=\frac12(vv^*+e_2e_2^*)$ and $\rho_{2,t}=\frac12(ww^*+e_2e_2^*)$ with $v,w=e_1\pm\sqrt t\,e_3$. These are at trace distance $O(\sqrt t)$ from $\rho_1$, and $\tr(\rho_{1,t}\#\rho_{2,t})=\frac12$, because the mean splits over $\mathbb{C}e_2\oplus\mathrm{span}\{e_1,e_3\}$ and $vv^*\#ww^*=0$. Coupling $e_2$ to $e_4$ as well gives fidelity $0$.

\begin{corollary}[Channels]\label{cor:channels dichotomy}
Let $\Phi$ be a channel. Then $\widehat F(\Phi,\Phi)=1$, and if $J(\Phi)$ is invertible, $\widehat F$ is continuous at $(\Phi,\Phi)$. If there are channels $\Phi_n\to\Phi$ with $\range J(\Phi_n)\cap\range J(\Phi)=\{0\}$, then $\widehat F(\Phi,\Phi_n)=0$ for all $n$. This is the case for every isometric channel $\Phi=V(\cdot)V^*$ with $\dim\mathcal{Y}\geq2$: take $\Phi_n=U_nV(\cdot)V^*U_n^*$ with unitaries $U_n\to\mathbf{1}_{\mathcal{Y}}$ such that $U_nV\notin\mathbb{C}V$.
\end{corollary}

\begin{proof}
$\widehat F(\Phi,\Phi)=\lambda_{\min}(\Tr_{\mathcal{Y}}J(\Phi))=\lambda_{\min}(\Phi^*(\mathbf{1}_{\mathcal{Y}})^\top)=1$. Continuity is Theorem~\ref{thm:dichotomy}(a). If the Choi ranges intersect trivially, the pair $(\Phi,\Phi_n)$ is mutually singular, so $J(\Phi)\#J(\Phi_n)=0$. For an isometric channel the Choi matrices are $|V\rangle\!\rangle\langle\!\langle V|$ and $|U_nV\rangle\!\rangle\langle\!\langle U_nV|$, whose ranges intersect trivially when $U_nV\notin\mathbb{C}V$. Such $U_n=e^{iH/n}$ exist as soon as $HV\notin\mathbb{C}V$ for some Hermitian $H$, which is possible because $\dim\mathcal{Y}\geq2$.
\end{proof}

Thus the geometric channel fidelity, viewed as a function of its second argument, is discontinuous at every isometric channel, although it equals $1$ there. This is the distinct image property of \cite{cree2020fidelity}, applied to the rank-one Choi matrices of isometric channels. Whether channel-valued approximants exist at every channel with non-invertible Choi matrix, as the CP approximants of Theorem~\ref{thm:dichotomy} do, is left open.

\begin{remark}[Physical interpretation]\label{rem:phys-dichotomy}
Corollary~\ref{cor:channels dichotomy} is a caution: an ideal gate $U(\cdot)U^*$ and a mis-rotated gate $U'(\cdot)U'^*$, $U'\notin\mathbb{C}U$, have $\widehat{F}=0$ however close they are in $|||\cdot|||_1$, so near unitaries $\widehat{F}$ is not a usable plug-in figure of merit. Noisy devices often have invertible Choi matrices, and as $J(\Phi+\Psi)\geq J(\Phi)$, full Kraus rank of one of the maps suffices. The generalised amplitude damping channel with damping $p$ and thermal population $N$, a qubit in a bath at nonzero temperature, has four Kraus operators \cite{KhatriSharmaWilde2020}, linearly independent for $p\in(0,1]$ and $N\in(0,1)$ (two are off-diagonal in different positions, two are diagonal and independent as $p\neq0$), so pairs containing it are continuity points by Theorem~\ref{thm:dichotomy}(a). At zero temperature ($N=0<p<1$) the Kraus rank is $2$, and as $\widehat{F}(\Phi,\Phi)=1$, Theorem~\ref{thm:dichotomy}(d)(ii) shows that $\widehat{F}$ is discontinuous at $(\Phi,\Phi)$ under completely positive perturbations.
\end{remark}

\section{Conclusion and outlook}\label{sec:conclusion}

Channel divergences built from Kubo--Ando means of Choi matrices that annihilate pairs with trivially intersecting Choi ranges (Theorem~\ref{range theorem}), such as the geometric mean, are not continuous in the diamond norm (Corollaries~\ref{impossibility corollary} and~\ref{maxf impossibility}).

The Arveson--Radon--Nikodym derivative $T$ of $\Phi$ with respect to $\Omega=\Phi+\Psi$ is the coordinate in which these quantities behave well. Its transpose is an effect on the environment of $\Omega$, so it records how an instrument splits a fixed device into two branches, and $\|T_1-T_2\|_\infty$ measures perturbations relative to $\Omega$ in the completely positive order. Every quantity considered is a scalar function of $T$, namely $g_\sigma$ or $h_f$, read through the Kraus operators of $\Omega$. Continuity in $T$ therefore holds with the dimension-free constant $\|\Omega^*(\mathbf{1}_{\mathcal{Y}})\|_\infty$. The integral representations of means and of operator convex functions give explicit moduli, which are exact for the one-pole kernels and, for operator convex $f$ on $[\epsilon,1-\epsilon]$, within a factor $3$ of optimal (Section~\ref{sec:loewner}). Quantitative estimates for unbounded divergences need a bound on $D_{\max}$ within each pair (without it they are still continuous in $T$ with values in $(-\infty,+\infty]$), and for a general operator convex $f$ the Lebesgue decomposition of the pair decides finiteness. Theorem~\ref{thm:two sided} converts these estimates into continuity in the diamond norm on the splittings of a fixed $\Omega$; for the fidelity this is H\"older continuity of order $\frac12$ with constant $\|\Omega^*(\mathbf{1}_{\mathcal{Y}})\|_\infty/\nu(\Omega)$, and in the counterexample $\nu(\Omega)\to0$. When the sums differ, their change enters through the Thompson metric or, if their Choi supports differ, through the singular parts, the absolutely continuous parts and the couplings of the Lebesgue decomposition of the sums.

The dichotomy of Theorem~\ref{thm:dichotomy} says where continuity holds. The fidelity, $\widehat{Q}_\alpha$ with $\alpha\in(0,1)$, the means and the bounded maximal $f$-divergences are continuous at every pair whose sum has an invertible Choi matrix. At a pair whose sum has a singular Choi matrix each of them, as a function on pairs of completely positive maps, is continuous exactly when it takes an extreme value there ($0$ for $\widehat{F}$ and $\widehat{Q}_\alpha$, $a_\sigma J(\Phi)+b_\sigma J(\Psi)$ for a mean, $\lambda_{\max}\big(\Tr_{\mathcal{Y}}(f(0^+)J(\Psi)+f'(\infty)J(\Phi))\big)$ for $\widehat{D}_f$); for means this requires the mean not to be a weighted arithmetic mean, and for bounded maximal $f$-divergences the `only if' needs a rank condition. For channels this separates noisy devices of full Kraus rank, which are continuity points, from isometric gates, at which $\widehat{F}(\Phi,\cdot)$ jumps from $1$ to $0$ (Corollary~\ref{cor:channels dichotomy}). Our estimates complement the joint lower semicontinuity of the operator-valued maximal $f$-divergences of \cite[Theorem~4.7]{Okayasu2026maxf}: in finite dimensions and in the coordinate $T$, continuity is two-sided and quantitative. They also complement the continuity results of \cite{MosonyiHiai2024} for R\'enyi divergences of states and channels. They are upper bounds: the two-sided estimate of Theorem~\ref{thm:two sided} concerns the maps, not the divergences.

We close with some open questions.
\begin{enumerate}
\item \emph{Constants.} Lemma~\ref{lem:one parameter} is sharp as $\delta\to0$. Is its factor $(1-4\delta^2)^{-1/2}$ optimal for finite $\delta$? The constants of Corollaries~\ref{cor:BS} and~\ref{cor:renyi large alpha} are not optimised, and we do not know whether the bound of Theorem~\ref{main theorem} is attained, although the underlying operator inequality is sharp (Remark~\ref{rem:sharp}).
\item \emph{Singular pairs.} Let $c_\Omega=\dim\ker J(\Omega)$. When $1\leq c_\Omega<\rank T(\mathbf{1}_{\mathbb{C}^r}-T)$, what is the exact set of limit points? At singular pairs where the fidelity vanishes it is continuous; at what rate? Does the dichotomy hold within the set of channels? Corollary~\ref{cor:channels dichotomy} settles this at the pairs $(\Phi,\Phi)$ with $\Phi$ an isometric channel.
\item \emph{Convex $f$.} For merely convex $f$ the maximal $f$-divergence is still defined through reverse tests \cite{Matsumoto2018newf}, but the perspective formula need not hold. Is it still a function of $T$, and of which function?
\item \emph{Operational consequences.} The geometric R\'enyi divergences with $\alpha\in(1,2]$ give converse bounds for channel capacities \cite{FangFawzi2021}, and geometric distinguishability measures limit channel discrimination and estimation \cite{KatariyaWilde2021}. Our estimates show how such bounds respond to relative perturbations, and Theorem~\ref{thm:dichotomy} shows that continuity fails at pairs with singular sum unless the quantity takes its extreme value there. Can this be made quantitative for a specific task?
\item \emph{Other norms and algebras.} By Ando's inequality in unitarily invariant norms, Proposition~\ref{K6}(1) holds in every unitarily invariant norm and Proposition~\ref{K6}(2) becomes $\|g_\#(E)-g_\#(E')\|_p\leq\|E-E'\|_{p/2}^{1/2}$ for $p\geq2$, but statements for maps need norms adapted to the Schatten classes. The framework of Okayasu \cite{Okayasu2026,Okayasu2026maxf} suggests extensions to normal completely positive maps between von Neumann algebras, with $T$ in the commutant of a Stinespring representation; our estimates are already dimension-free.
\end{enumerate}

\appendix
\section{Norm bounds and a variational characterisation of means of completely positive maps}\label{app:norms}
 
Throughout this appendix $\Phi,\Psi\in\mathrm{CP}(\mathcal{X},\mathcal{Y})$, $\Omega=\Phi+\Psi$ has a Kraus decomposition $\Omega(X)=\sum_{k=1}^rV_kXV_k^*$ with linearly independent $V_1,\dots,V_r$, $W=[\vectorise(V_1)\ \cdots\ \vectorise(V_r)]$, and $T$ is the ARN derivative of $\Phi$ with respect to $\Omega$, so that $\Psi=\Psi_{\mathbf{1}_{\mathbb{C}^r}-T}$. For a linear map $\Xi:\mathrm{L}(\mathcal{X})\to\mathrm{L}(\mathcal{Y})$ and $p\in[1,\infty]$ we write
\[
\|\Xi\|_{p\to p}=\sup\big\{\|\Xi(X)\|_p:\ X\in\mathrm{L}(\mathcal{X}),\ \|X\|_p\leq1\big\},
\]
and we write $\|\cdot\|_{\mathrm{ui}}$ for an arbitrary unitarily invariant norm \cite[Chapter~IV]{bhatia1997matrix}. We use the following fact repeatedly: if $\Xi$ is completely positive, then the supremum defining $\|\Xi\|_{p\to p}$ may be restricted to positive semidefinite $X$ \cite{Watrous2005notes}.
 
\subsection{Norms of means}
 
\begin{proposition}\label{app:mean norm}
Let $\sigma$ be a Kubo--Ando mean. For every $p\in[1,\infty]$,
\[
\|\Phi\,\sigma\,\Psi\|_{p\to p}\leq\|g_\sigma(T)\|_\infty\,\|\Omega\|_{p\to p},\qquad \|g_\sigma(T)\|_\infty=\max_{t\in\mathrm{spec}(T)}g_\sigma(t),
\]
and for $p=1$,
\[
\|\Phi\,\sigma\,\Psi\|_{1\to1}=|||\Phi\,\sigma\,\Psi|||_1=\max_{\rho\in\mathrm{D}(\mathcal{X})}\tr\big(g_\sigma(T)^\top\,\Omega^c(\rho)\big),
\]
where $\Omega^c$ is the complementary map of $\Omega$, that is, Theorem~\ref{effect theorem} applied with $\Omega$ in place of $\Phi$.
\end{proposition}
 
\begin{proof}
By Theorem~\ref{fundamental principle}, $\Phi\,\sigma\,\Psi=\Psi_{g_\sigma(T)}$ with $g_\sigma(T)\geq0$. Put $c=\|g_\sigma(T)\|_\infty$. Since $c\mathbf{1}_{\mathbb{C}^r}-g_\sigma(T)\geq0$, the map $c\,\Omega-\Psi_{g_\sigma(T)}=\Psi_{c\mathbf{1}_{\mathbb{C}^r}-g_\sigma(T)}$ is completely positive by Theorem~\ref{thm:arveson}, that is, $0\leq\Phi\,\sigma\,\Psi\leq c\,\Omega$ in the completely positive order. For positive semidefinite $X$ this gives $0\leq(\Phi\,\sigma\,\Psi)(X)\leq c\,\Omega(X)$, hence $\|(\Phi\,\sigma\,\Psi)(X)\|_p\leq c\,\|\Omega(X)\|_p$ by the Weyl monotonicity principle \cite[Section~III.2]{bhatia1997matrix}. Taking the supremum over positive semidefinite $X$ with $\|X\|_p\leq1$, which suffices for both completely positive maps \cite{Watrous2005notes}, gives the first claim. The formula for $\|g_\sigma(T)\|_\infty$ is the spectral theorem, since $g_\sigma\geq0$. For $p=1$, the induced and completely bounded trace norms agree on completely positive maps \cite[Section~3.3.2]{watrous2018theory}, and the formula is the last statement of Theorem~\ref{effect theorem}, applied to $\Omega$ and the positive semidefinite operator $g_\sigma(T)$.
\end{proof}
 
\begin{remark}
If $\Phi$ and $\Psi$ are channels, then $\|\Omega\|_{1\to1}=\|\Omega^*(\mathbf{1}_{\mathcal{Y}})\|_\infty=2$, so $|||\Phi\,\sigma\,\Psi|||_1\leq2\|g_\sigma(T)\|_\infty$. The bound depends on the spectrum of $T$ only through the values of $g_\sigma$ on it; since $g_\sigma$ is not monotone in general (for instance $g_\#(t)=\sqrt{t(1-t)}$), it cannot be replaced by $g_\sigma(\|T\|_\infty)$.
\end{remark}
 
\subsection{A variational characterisation of the geometric mean}
 
The following extremal property of the geometric mean is due to Ando \cite{Ando1979}; see also \cite[Chapter~4]{Bhatia2007positive}.
 
\begin{lemma}\label{app:ando max}
For positive semidefinite $A,B\in\mathrm{L}(\mathcal{V})$,
\[
A\#B=\max\Big\{X=X^*:\ \begin{pmatrix}A&X\\X&B\end{pmatrix}\geq0\Big\},
\]
the maximum being taken in the Loewner order. In particular, the block matrix with off-diagonal entries $A\#B$ is positive semidefinite.
\end{lemma}
 
\begin{proof}
Let $A>0$. By the Schur complement criterion \cite[Chapter~1]{Bhatia2007positive}, the block matrix is positive if and only if $B-XA^{-1}X\geq0$. With $M=A^{-1/2}BA^{-1/2}$ and $Y=A^{-1/2}XA^{-1/2}$ this reads $Y^2\leq M$. Since $t\mapsto t^{1/2}$ is operator monotone, $Y\leq|Y|\leq M^{1/2}$, so $X=A^{1/2}YA^{1/2}\leq A^{1/2}M^{1/2}A^{1/2}=A\#B$; and $Y=M^{1/2}$ is admissible, so $A\#B$ is the maximum. For general $A,B\geq0$, adding $\eta\mathbf{1}_{\mathcal{V}}$ to both diagonal blocks preserves positivity, so every admissible $X$ satisfies $X\leq(A+\eta\mathbf{1}_{\mathcal{V}})\#(B+\eta\mathbf{1}_{\mathcal{V}})$, which decreases to $A\#B$ as $\eta\downarrow0$ \cite{KuboAndo1980}. Conversely, the block matrices with off-diagonal entry $(A+\eta\mathbf{1}_{\mathcal{V}})\#(B+\eta\mathbf{1}_{\mathcal{V}})$ are positive, and letting $\eta\downarrow0$ in the closed positive cone shows that $A\#B$ is admissible.
\end{proof}
 
For an adjoint-preserving linear map $\Xi:\mathrm{L}(\mathcal{X})\to\mathrm{L}(\mathcal{Y})$, define
\[
\mathcal{B}_\Xi:\mathrm{L}(\mathcal{X})\to\mathrm{L}(\mathbb{C}^2\otimes\mathcal{Y}),\qquad
\mathcal{B}_\Xi(X)=\begin{pmatrix}\Phi(X)&\Xi(X)\\\Xi(X)&\Psi(X)\end{pmatrix}.
\]
 
This block-positivity characterisation of the geometric mean of completely positive maps is also obtained by Okayasu \cite{Okayasu2026}, for normal completely positive maps between von Neumann algebras, and by Kian \cite{kian2026relative}; we include a short proof in the present notation.

\begin{proposition}\label{app:variational}
The map $\Phi\#\Psi$ is the largest adjoint-preserving map $\Xi$, in the completely positive order, for which $\mathcal{B}_\Xi$ is completely positive. Equivalently,
\[
J(\Phi)\#J(\Psi)=\max\Big\{X=X^*:\ \begin{pmatrix}J(\Phi)&X\\X&J(\Psi)\end{pmatrix}\geq0\Big\}.
\]
In terms of the ARN derivative, the corresponding maximisation over $\mathrm{L}(\mathbb{C}^r)$ is solved by
\[
\max\Big\{X=X^*:\ \begin{pmatrix}T&X\\X&\mathbf{1}_{\mathbb{C}^r}-T\end{pmatrix}\geq0\Big\}=\big(T(\mathbf{1}_{\mathbb{C}^r}-T)\big)^{1/2}=g_\#(T).
\]
\end{proposition}
 
\begin{proof}
Ordering the tensor factors as $\mathbb{C}^2\otimes(\mathcal{Y}\otimes\mathcal{X})$, the Choi matrix of $\mathcal{B}_\Xi$ is
\[
J(\mathcal{B}_\Xi)=\begin{pmatrix}J(\Phi)&J(\Xi)\\J(\Xi)&J(\Psi)\end{pmatrix},
\]
and $J(\Xi)$ is Hermitian because $\Xi$ preserves adjoints. By Choi's theorem \cite{Choi1975}, $\mathcal{B}_\Xi$ is completely positive if and only if this block matrix is positive semidefinite, and $\Xi_1\leq\Xi_2$ in the completely positive order if and only if $J(\Xi_1)\leq J(\Xi_2)$. The first two statements therefore follow from Lemma~\ref{app:ando max} with $A=J(\Phi)$ and $B=J(\Psi)$, since $J(\Phi\#\Psi)=J(\Phi)\#J(\Psi)$. For the last statement apply Lemma~\ref{app:ando max} with $A=T$ and $B=\mathbf{1}_{\mathbb{C}^r}-T$; since these commute, $T\#(\mathbf{1}_{\mathbb{C}^r}-T)=(T(\mathbf{1}_{\mathbb{C}^r}-T))^{1/2}=g_\#(T)$ by the functional calculus, consistently with $J(\Phi)\#J(\Psi)=W\,g_\#(T)\,W^*$ from Theorem~\ref{fundamental principle}.
\end{proof}
 
\begin{remark}
Since $X\mapsto\lambda_{\min}(\Tr_{\mathcal{Y}}X)$ is monotone in the Loewner order, Propositions~\ref{lambda min formula} and~\ref{app:variational} give
\[
\widehat{F}(\Phi,\Psi)=\max\Big\{\lambda_{\min}(\Tr_{\mathcal{Y}}X):\ X=X^*,\ \begin{pmatrix}J(\Phi)&X\\X&J(\Psi)\end{pmatrix}\geq0\Big\},
\]
which is the semidefinite program for the geometric channel fidelity in \cite[eq.~(20)]{NuradhaWilde2025} (cf.\ \cite{KatariyaWilde2021} for the analogous program for the channel fidelity based on the Uhlmann fidelity); at the level of states it reduces to the semidefinite program for the Matsumoto fidelity \cite{matsumoto2010reverse,cree2020fidelity}.
\end{remark}
 
\subsection{Cauchy--Schwarz bounds}
 
\begin{lemma}\label{app:block cs}
Let $A,B\geq0$ in $\mathrm{L}(\mathcal{V})$ and let $X$ be such that $\begin{pmatrix}A&X\\X^*&B\end{pmatrix}\geq0$. Then
\begin{enumerate}
\item $X=A^{1/2}CB^{1/2}$ for some $C$ with $\|C\|_\infty\leq1$;
\item $\|X\|_{\mathrm{ui}}^2\leq\|A\|_{\mathrm{ui}}\,\|B\|_{\mathrm{ui}}$ for every unitarily invariant norm;
\item in particular $\|A\#B\|_{\mathrm{ui}}^2\leq\|A\|_{\mathrm{ui}}\,\|B\|_{\mathrm{ui}}$.
\end{enumerate}
\end{lemma}
 
\begin{proof}
Part (1) is standard \cite[Chapter~1]{Bhatia2007positive}: for $A,B>0$, positivity of the block matrix is equivalent to $B-X^*A^{-1}X\geq0$, that is, to $\|A^{-1/2}XB^{-1/2}\|_\infty\leq1$; the general case follows by replacing $A,B$ with $A+\eta\mathbf{1}_{\mathcal{V}},B+\eta\mathbf{1}_{\mathcal{V}}$ and passing to a limit point of the contractions $C_\eta$. Part (2) is due to Horn and Mathias \cite{HornMathias1990}: writing $X=Y^*Z$ with $Y=A^{1/2}$ and $Z=CB^{1/2}$, the Cauchy--Schwarz inequality $\|Y^*Z\|_{\mathrm{ui}}^2\leq\|Y^*Y\|_{\mathrm{ui}}\|Z^*Z\|_{\mathrm{ui}}$ for unitarily invariant norms \cite{HornMathias1990,BhatiaDavis1995} gives $\|X\|_{\mathrm{ui}}^2\leq\|A\|_{\mathrm{ui}}\,\|B^{1/2}C^*CB^{1/2}\|_{\mathrm{ui}}$, and $0\leq B^{1/2}C^*CB^{1/2}\leq B$ gives $\|B^{1/2}C^*CB^{1/2}\|_{\mathrm{ui}}\leq\|B\|_{\mathrm{ui}}$ by the Weyl monotonicity principle \cite[Section~III.2]{bhatia1997matrix}. Part (3) follows from part (2) and Lemma~\ref{app:ando max}.
\end{proof}
 
\begin{theorem}\label{app:cs maps}
For all $\Phi,\Psi\in\mathrm{CP}(\mathcal{X},\mathcal{Y})$, every unitarily invariant norm and every $p\in[1,\infty]$,
\[
\|J(\Phi\#\Psi)\|_{\mathrm{ui}}^2\leq\|J(\Phi)\|_{\mathrm{ui}}\,\|J(\Psi)\|_{\mathrm{ui}},\qquad
\|\Phi\#\Psi\|_{p\to p}^2\leq\|\Phi\|_{p\to p}\,\|\Psi\|_{p\to p}.
\]
The second inequality also holds for the completely bounded versions of the $p\to p$ norms, and in particular $|||\Phi\#\Psi|||_1^2\leq|||\Phi|||_1\,|||\Psi|||_1$.
\end{theorem}
 
\begin{proof}
The first inequality is Lemma~\ref{app:block cs}(3) with $A=J(\Phi)$ and $B=J(\Psi)$. For the second, $\mathcal{B}_{\Phi\#\Psi}$ is completely positive by Proposition~\ref{app:variational}, so for positive semidefinite $X$ the block matrix $\mathcal{B}_{\Phi\#\Psi}(X)$ is positive semidefinite, and Lemma~\ref{app:block cs}(2) with the Schatten $p$-norm gives
\[
\|(\Phi\#\Psi)(X)\|_p^2\leq\|\Phi(X)\|_p\,\|\Psi(X)\|_p\leq\|\Phi\|_{p\to p}\,\|\Psi\|_{p\to p}\,\|X\|_p^2 .
\]
Since $\Phi\#\Psi$ is completely positive, it suffices to take the supremum over positive semidefinite $X$ \cite{Watrous2005notes}. For the completely bounded versions, apply the same argument, for an arbitrary complex Euclidean space $\mathcal{V}$, to $\Phi\otimes\mathbf{1}_{\mathrm{L}(\mathcal{V})}$ and $\Psi\otimes\mathbf{1}_{\mathrm{L}(\mathcal{V})}$, using $(\Phi\otimes\mathbf{1}_{\mathrm{L}(\mathcal{V})})\#(\Psi\otimes\mathbf{1}_{\mathrm{L}(\mathcal{V})})=(\Phi\#\Psi)\otimes\mathbf{1}_{\mathrm{L}(\mathcal{V})}$, which holds because $(A\otimes C)\#(B\otimes C)=(A\#B)\otimes C$ for $A,B,C\geq0$ \cite{KuboAndo1980}.
\end{proof}

\section{Acknowledgements}
The author discussed some of these results of this paper with Dr. Javier Alejandro Ch\'avez-Dom\'inguez, who was the key motivator for the author to consider this area of mathematics. I thank him for his continued support and supervision. Some results in this paper were discussed with Hunter Hamby. Any mistakes in this paper are due to the author and neither of the aforementioned people. This research was originally motivated by understanding the papers of Okayasu \cite{Okayasu2026} and Kian \cite{kian2026relative}, but then sprawled into a general study of the Arveson Radon-Nikodym theorem and how this can be used to study norm estimates of channel means. Throughout this process, the works of Bhatia \cite{bhatia1997matrix} and Hiai \cite{hiai2010matrix} have proven invaluable. Once the author noticed the problem with the geometric channel fidelity, stemming from the work of \cite{cree2020fidelity}, they decided to push just how far they could take this problem and the continuity estimates. Hence, during the preparation of this work the author used Claude (Anthropic; models Opus~5.5 and Fable~5.1). The purpose of this at first was to help to build a roadmap of the current literature which exists on the subject. Once trained on the material and what the author wanted to do using the ARN theorem, these models were used in
studying the background literature, searching for references, suggesting proofs and statements, and suggesting edits to the exposition. This paper would not exist without the fundamental work of Arveson \cite{arveson1969subalgebras}, Raginsky \cite{Raginsky2003}, Kian \cite{kian2026relative}, Okayasu \cite{Okayasu2026}, Cree and Sikora \cite{cree2020fidelity}, Mosonyi and Hiai \cite{HiaiMosonyi2017}, and Watrous \cite{watrous2018theory}. All mathematical statements and proofs in
this paper were checked by the author, all references were verified against the original sources, and the author takes full responsibility for the content of the manuscript. I really enjoy math, discussing it with my friends etc. Thanks to my friends at OU for being there to discuss this kind of thing. This is only a work in progress, there are bound to be errors and typos in here. I plan to add more exposition to guide the reader in later versions. 
\bibliographystyle{amsalpha}
\bibliography{Bibliography}
\end{document}